%% file: main.tex
\documentclass{article}

\usepackage[preprint]{neurips_2025}
\setcitestyle{numbers,square,comma}
\usepackage[utf8]{inputenc} 
\usepackage[T1]{fontenc}    

\usepackage{url}            
\usepackage{booktabs}       
\usepackage{graphicx}
\usepackage{amsfonts}       
\usepackage{nicefrac}       
\usepackage{microtype}      
\usepackage{xcolor}         
\usepackage{amsmath,amssymb,amsthm,mathtools}
\usepackage[mono=false]{libertinus}
\usepackage{libertinust1math}
\usepackage{url}
\usepackage{setspace}
\usepackage{bm}
\usepackage{tocloft}
\usepackage{etoc}
\usepackage{csquotes}
\usepackage[backref=page]{hyperref}       
\usepackage{cleveref}
\usepackage[normalem]{ulem}
\usepackage{enumitem}
\usepackage{comment}

\usepackage{subcaption}
\usepackage[most]{tcolorbox}
\usepackage{array}
\setlist[enumerate]{itemsep=2pt, parsep=0pt, topsep=0pt}

\definecolor{royalblue}{RGB}{60,100,200}
\definecolor{royalred}{RGB}{160,40,40}

\hypersetup{
    colorlinks=true,
    linkcolor=royalred,   
    citecolor=royalblue,   
    urlcolor=royalblue
}

\crefname{appendix}{Appendix}{Appendices}
\crefformat{equation}{\textcolor{royalred}{#2Equation~#1#3}}
\crefformat{figure}{\textcolor{royalblue}{#2Figure~#1#3}}
\crefformat{table}{\textcolor{royalred}{#2Table~#1#3}}
\crefformat{section}{\textcolor{royalblue}{#2Section~#1#3}}
\crefformat{appendix}{\textcolor{royalblue}{#2Appendix~#1#3}}
\crefmultiformat{equation}{\textcolor{royalred}{(#2#1#3)}}
  {, \textcolor{royalred}{(#2#1#3)}}
  {, \textcolor{royalred}{(#2#1#3)}}
  {, \textcolor{royalred}{(#2#1#3)}}

\usepackage{tikz}

\newtheorem{theorem}{Theorem}

\numberwithin{equation}{section}

\newtcolorbox{summarybox}{
  enhanced,
  breakable,
  colback=white,
  colframe=black!75,
  boxrule=0.5pt,
  arc=2.5mm,
  left=4mm,
  right=4mm,
  top=2.5mm,
  bottom=2.5mm,
  width=0.93\linewidth
}

\makeatletter
\newcount\versionhours
\newcount\versionminutes
\newcount\versiontmp
\newcommand{\versiontimestamp}{%
  \versionhours=\time
  \divide\versionhours by 60
  \versionminutes=\time
  \versiontmp=\versionhours
  \multiply\versiontmp by 60
  \advance\versionminutes by -\versiontmp
  \number\year-\two@digits{\month}-\two@digits{\day}\space
  \two@digits{\the\versionhours}:\two@digits{\the\versionminutes}%
}
\makeatother

\title{A Brief History of Fr\'echet Distances: \\From Curves and Probability Laws to FID}

\author{%
\textit{Y u l i  \; W u}\\
  \texttt{mail@wuyuli.com} \\
}

\begin{document}
\maketitle
\begingroup
\renewcommand\thefootnote{}
\NoHyper
\footnotetext{Preprint.}
\endNoHyper
\endgroup
\setcounter{footnote}{0}

\begin{abstract}
This note provides a chronological account of Fréchet distances, starting with Maurice Fréchet's 1906 doctoral thesis on distances in abstract sets and tracing the Fréchet distance between polygonal curves and its algorithmic computation in the 1990s. 
It then continues with his 1957 paper on a coupling-based distance between probability laws with a brief glimpse of Wasserstein distance and optimal transport. 
We further attempt to draw connections between the distributional, coupling-based facet of Fréchet distances on probability laws and the geometric facet on curves. 
The note ends with a modern use case, the Fréchet Inception Distance (FID) in the era of deep generative model evaluation, interpretable as the Wasserstein-2 distance between multivariate Gaussians in a learned feature space. An appendix includes \TeX{}ified faithful English translations of Fréchet's 1906 thesis and 1957 paper, and Lévy's 1950 note for reader convenience.
\end{abstract}
\vspace{1cm}
{
\etocsetstyle{section}%
  {\par\fontsize{10}{15}\selectfont}%
  {\noindent}%
  {%
    \etocthelinkednumber\enspace%
    {\etocthelinkedname}%
    \leavevmode\leaders\hbox to .5em{\hss.\hss}\hfill\kern0pt\textcolor{royalblue}{\etocthepage}%
    \par%
  }%
  {\par}

\tableofcontents
}

\newpage

\begin{figure}[t]
    \raggedleft
    \vspace{-1em}
    \includegraphics[width=0.5\linewidth]{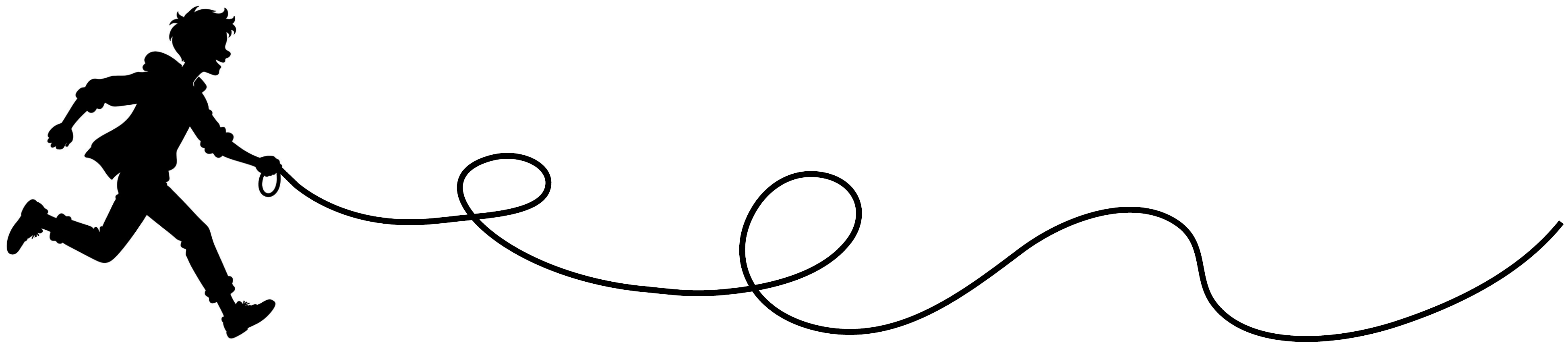}
    \vspace{-2em}
\end{figure}

\section{Introduction}
\label{sec:intro}

The term \emph{Fr\'{e}chet distance} is now often encountered first in one of two late descendants:
the Fr\'{e}chet distance between curves in computational geometry (e.g., \cite{alt1992measuring,alt1995computing,bringmann2014walking,buchin2017four,buchin2019seth,gutschlag2022generalized,cheng2025constant,conradi2025revisiting}), which is the primary entry on Wikipedia \cite{enwiki:1331641232}, or the Fr\'{e}chet Inception Distance (FID) in deep generative modeling \cite{heusel2017gans}.
Both uses can create the impression that Fr\'{e}chet's original concern was already narrowly centered on geometric path comparison or on the comparison of probability laws.
That impression is not entirely false, but it is too narrow for the 1906 thesis \cite{frechet1906quelques} and too compressed for the larger historical development.

The most useful corrective supplied by Fréchet's own writings is a historical one. In 1906, the 28-year-old Maurice Fréchet published his doctoral thesis, supervised by Jacques Hadamard at the École Normale Supérieure, titled \textit{Sur quelques points du calcul fonctionnel (On some points of functional calculus)} \cite{frechet1906quelques}. 
The 1906 thesis is broader and more ambitious.
Its aim is not to define one distinguished distance, but to construct a general language for analysis on sets whose elements may be \emph{``de nature quelconque'' (of any nature)}: numbers, points, curves, functions, surfaces, and so forth.
Half a century later, many years after retiring in 1949 as Chair Professor of the Calculus of Probabilities and Mathematical Physics at the Université de Paris, Fréchet published a paper in 1957, titled \textit{Sur la distance de deux lois de probabilité (On the distance between two probability laws)} \cite{frechet1957distance}. The 1957 paper is narrower and more concrete.
It asks how one should define a distance between two probability laws, and it studies in detail a coupling-based proposal due to his academic sibling and friend Paul L\'{e}vy \cite{barbut2013paul}.
Yet the two papers are deeply continuous in spirit.
The later paper may be read as an application, within probability theory, of the same methodological instinct already visible in the thesis: one should compare complex mathematical objects not by their presentation alone, but by placing them in a common abstract structure in which distance, limit, and continuity make sense. An extensive study and critique on Fréchet's work can be found in \cite{taylor1982study,taylor1985study,taylor1987study}.

Although the term \emph{Fr\'{e}chet distance} is in general not ambiguous because of its polymorphic definitions, I have nevertheless thought it worthwhile to write the present note. 
The motive is rather practical: despite my extensive reliance on large language models in preparing this work, I have often found their answers on this subject less satisfactory and have sometimes had the impression that certain primary sources, Fr\'{e}chet's 1906 thesis \cite{frechet1906quelques} for example, were not part of the training data.
I therefore hope that the note may be of some use, if not per se, then at least through the literature it gathers and the \TeX{}ified versions of Fr\'{e}chet's 1906 thesis \cref{app:frechet1906}, 1957 paper \cref{app:frechet1957} and L\'{e}vy's 1950 note in the collection from Fr\'{e}chet \cref{app:levy1950} that accompany it. I have had especially in mind those who may wish to know more of the background of Fr\'{e}chet distance. If it helps them (including myself), speak of it with greater confidence and exactness, it will have served its purpose.

This note is organized around the two principal historical branches of the subject. We begin with the curve-theoretic line in \cref{sec:curves}, starting from Fr\'{e}chet's 1906 thesis and proceeding to the algorithmic treatment of polygonal curves in the 1990s. We then turn to the probabilistic line in \cref{sec:laws}, from Fr\'{e}chet's 1957 paper on distance between probability laws to the modern Wasserstein framework and the Gaussian formulae later used in FID. The relation between these two lines is taken up in \cref{sec:laws-as-curves}, the modern use case of FID is discussed in \cref{sec:fid}, a few neighboring notions are collected in \cref{sec:other-distances}, and naming questions are discussed last in \cref{sec:naming}. 
The appendix includes English translations from Fr\'{e}chet's 1906 thesis \cite{frechet1906quelques}, 1957 paper \cite{frechet1957distance} and L\'{e}vy's 1950 note in the collection from Fr\'{e}chet \cite{frechet1950recherches}\,\footnote{\,I bought this book second-hand after it had been discarded from EPFL's library, since it was difficult to find both online and locally.}, in addition to two proofs in \cref{app:proof}.
We quote French terms when we worry that English translations may introduce ambiguity or they read cooler.

\section{Fr\'{e}chet Distance on Curves}
\label{sec:curves}

\subsection{Distance, \textit{voisinage} and \textit{\'ecart}}

Fr\'{e}chet's 1906 thesis \cite{frechet1906quelques} concerns functional calculus, a topic that lies beyond the scope of this note. Nevertheless, we quote a comment by Angus E. Taylor in 1982 \cite{taylor1982study} to summarize its significance:

\begin{quote}
  `` It seems to me that the most significant general aspects of Fr\'{e}chet's thesis are two in number. The first is his initiation of abstract point set topology and his development of the abstract theory to an extent that demonstrated conclusively the feasibility of such a procedure. It was Fr\'{e}chet who decisively opened the theory to classes of functions and curves. His examples of application were few and the results he harvested in his thesis were of limited extent. Nevertheless, he called attention to a new point of view in analysis: a perspective on classes of functions as metric spaces. ''
\end{quote}

In this note, we focus on his definitions of \textit{voisinage (neighborhood)} and \textit{\'ecart (gap, distance)}\,\footnote{\,Fr\'echet uses \textit{\'ecart} in his 1906 thesis and \textit{distance} in his 1957 paper. We follow this usage here without imposing a strict distinction between the two terms. Note that A. J. Ward, for example, translates \textit{\'ecart} as \textit{quasi-distance} in \cite{ward1954generalization}.}. Recall the modern definition of a metric (or a distance function) with four axioms: 

A function $d : \mathcal X \times \mathcal X \to \mathbb{R}$ is a metric (or a distance function) if for all $x,y,z \in \mathcal X$:
\begin{equation}
  \label{eq:distance_axioms}
  \begin{aligned}
    d(x,y) &\ge 0,  && \qquad \text{\textit{(non-negativity)}}, \\
    d(x,y) &= 0 \iff x = y,  &&  \qquad \text{\textit{(identity of indiscernibles)}}, \\
    d(x,y) &= d(y,x),  &&  \qquad \text{\textit{(symmetry)}}, \\
    d(x,z) &\le d(x,y) + d(y,z),  &&  \qquad \text{\textit{(triangle inequality)}}.
  \end{aligned}
\end{equation}

Fr\'echet first introduces the notion of \textit{voisinage} in \cite[\S\,27]{frechet1906quelques}:
A \textit{voisinage} between two elements $A$ and $B$ is a number $(A,B)=(B,A)\ge 0$ such that:
\begin{enumerate}
\item $(A,B)=0$ if and only if $A$ and $B$ are identical;
\item there exists a positive function $f(\varepsilon)$ tending to $0$ with $\varepsilon$ such that, whenever $(A,B)\le \varepsilon$ and $(B,C)\le \varepsilon$, one has $(A,C)\le f(\varepsilon)$.
\end{enumerate}
\textit{Voisinage} is a symmetric, positive function defining proximity, but weaker than a metric because it replaces the triangle inequality with a general \textit{small implies small} condition. Related notions in later topology include uniform spaces and neighborhood systems.

Later in his thesis \cite[\S\,49]{frechet1906quelques}, Fr\'echet defines an \textit{\'ecart} of a pair of elements $A,B$ as a number $(A,B)\ge 0$, which enjoys the following two properties: 
\begin{enumerate}
  \item $(A,B)=0$ only if $A$ and $B$ are identical;
  \item if $A,B,C$ are any three elements, one always has $(A,B)\le (A,C)+(C,B)$.
\end{enumerate}
Although Fr\'echet does not specifically claim the symmetry of \textit{\'ecart}, which is part of the defining clause of \textit{voisinage}, he implies that by stating ``Ainsi, \textit{\'ecart est un voisinage} [...]'' (Thus, \textit{\'ecart is a voisinage} [...])\,\footnote{\,The italic emphasis is that of Fr\'echet's thesis.} only after claiming that \textit{\'ecart} satisfies two conditions in the definition of \textit{voisinage}.
Despite the assertion\,\footnote{\,Anecdotally, Fr\'echet's teacher Jacques Hadamard was once displeased that young Fr\'echet made assertion without giving the proof. Hadamard wrote the letter to Fr\'echet: ``[...] mais vous ne devez \uline{en aucune cas} faire de déduction fausse, \underline{\underline{jamais}} \uline{sous quelque prétexte et en quelque circonstance que ce soit}.'' ([...] but you must \uline{under no circumstances} make a false inference, \underline{\underline{never}} \uline{under any pretext and in any circumstance whatsoever}.) The underlining (and the double underlining) is that of Hadamard. \cite{taylor1982study}}, we consider that \textit{\'ecart} is inherited from the properties of \textit{voisinage}, yielding precisely the modern axioms of a distance function. In the curve case, discussed in \cref{sec:curve}, the point is not that definiteness fails, but that the objects being compared are ordered curve-arcs considered up to Fr\'{e}chet's coincidence relation rather than raw parametrizations.

\begin{figure}[t]
     \centering
     \begin{subfigure}[t]{0.3\textwidth}
         \centering
         \includegraphics[width=\textwidth]{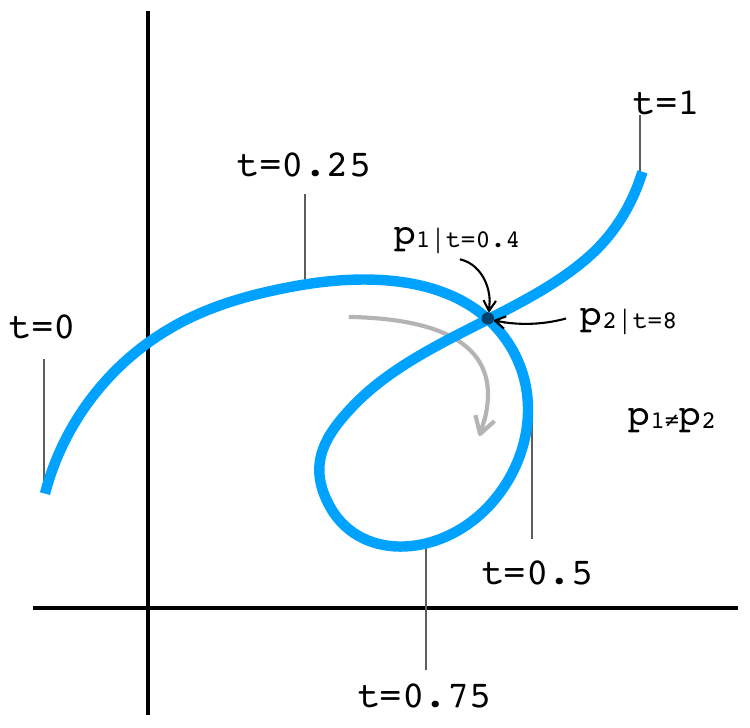}
         \caption{}
         \label{fig:1a}
     \end{subfigure}
     \hfill
     \begin{subfigure}[t]{0.3\textwidth}
         \centering
         \includegraphics[width=\textwidth]{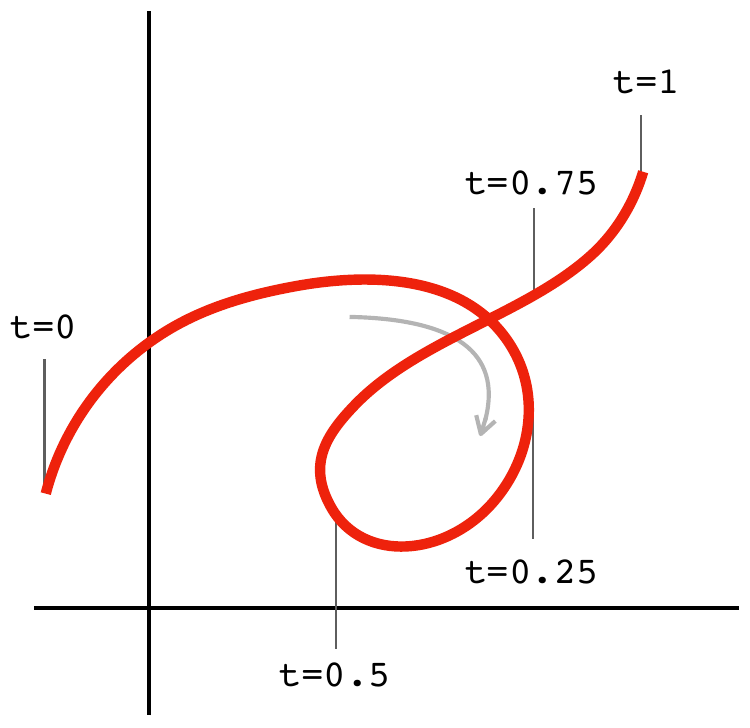}
         \caption{}
         \label{fig:1b}
     \end{subfigure}
     \hfill
     \begin{subfigure}[t]{0.3\textwidth}
         \centering
         \includegraphics[width=\textwidth]{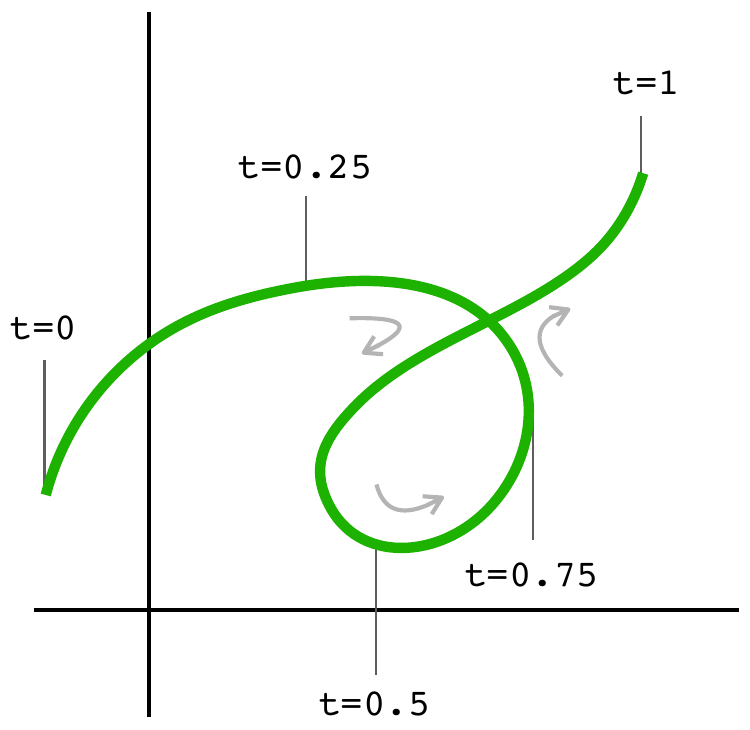}
         \caption{}
         \label{fig:1c}
     \end{subfigure}
        \caption{Fréchet's curve-arcs are ordered and the parameterized trajectories are normalized with $t\in[0,1]$. For instance, curves (a) and (b) coincide despite having different speeds, whereas curves (a) and (c) are distinct. Moreover, two points corresponding to different times $p_1,p_2$ on curve (a) are considered different, even if they share the same spatial coordinate.}
        \label{fig:1}
\end{figure}

\subsection{Curve and \emph{\'ecart} between two curves}
\label{sec:curve}

In \cite[\S\,76]{frechet1906quelques}, Fr\'{e}chet defines a curve-arc by a system
\begin{equation}
x=f(t), \qquad y=g(t), \qquad z=h(t), \qquad a\le t\le b,
\end{equation}
where $f,g,h$ are continuous and are not simultaneously constant on any interval contained in $[a,b]$, unless they are simultaneously constant on the whole interval. A curve-arc is therefore not merely the set of points determined by these equations (see \cref{fig:1}). In Fr\'{e}chet's formulation, the order induced by the parameter is part of the object. Thus two systems
\begin{equation}
x=f(t),\ y=g(t),\ z=h(t), \qquad a\le t\le b,
\end{equation}
and
\begin{equation}
x=F(u),\ y=G(u),\ z=H(u), \,\quad A\le u\le B,
\end{equation}
define the same curve-arc only when they have the same points and these points are encountered in the same order. If the same geometric point is met for two distinct parameter values, then these are counted as two distinct points of the curve (see \cref{fig:1a}).

Accordingly, two curve-arcs coincide if and only if one can establish between their points a one-to-one reciprocal correspondence such that corresponding points coincide and their order is preserved; see \cite[\S\,76]{frechet1906quelques} and \cref{fig:1}. In \cite[\S\,77]{frechet1906quelques}, Fr\'{e}chet shows that this geometric definition agrees with the parametric one. Namely, the two representations define the same curve if and only if there exists a continuous increasing function
$u=\varphi(t)$
from $[a,b]$ onto $[A,B]$ such that
\begin{equation}
f(t)=F(\varphi(t)), \qquad g(t)=G(\varphi(t)), \qquad h(t)=H(\varphi(t)).
\end{equation}

After this reduction, Fr\'{e}chet assumes the parameter interval has been normalized to $[0,1]$.

In \cite[\S\,78]{frechet1906quelques}, he defines the \emph{curve \`ecart}. Given two representations
\begin{equation}
\gamma:\quad x=f(t),\ y=g(t),\ z=h(t), \qquad 0\le t\le 1,
\end{equation}
\begin{equation}
\Gamma:\quad x=F(u),\ y=G(u),\ z=H(u), \quad 0\le u\le 1,
\end{equation}
he forms
\begin{equation}
\delta(t)=\sqrt{[f(t)-F(t)]^{2}+[g(t)-G(t)]^{2}+[h(t)-H(t)]^{2}},
\qquad 0\le t\le 1,
\end{equation}
which is continuous in $t$ and therefore has a maximum
\begin{equation}
d=\max_{0\le t\le 1}\delta(t)\ge 0.
\end{equation}

The distance of the two curves is then defined as the \emph{lower limit} $e$ of the set of values $d$ obtained by taking all possible representations of $\gamma$ and $\Gamma$.

Fr\'{e}chet immediately observes that one may simplify the definition by fixing one arbitrary representation of the first curve and varying only the representation of the second.

In \cite[\S\S\,79--80]{frechet1906quelques}, Fr\'{e}chet proves that this \emph{\'ecart} satisfies the required properties of a distance on curves in his sense. In particular,
\begin{equation}
(\gamma,\Gamma)=0 \iff \gamma=\Gamma,
\end{equation}
and for three curves $\gamma,\Gamma,C$,
\begin{equation}
(\gamma,\Gamma)\le (C,\gamma)+(C,\Gamma).
\end{equation}

In \cite[\S\,81]{frechet1906quelques}, he identifies the induced convergence. A sequence of curves $\gamma_n$ converges to a curve $\gamma$ if and only if there exist representations
\begin{equation}
\gamma_n:\quad x=f_n(t),\ y=g_n(t),\ z=h_n(t), \qquad 0\le t\le 1,
\end{equation}

\begin{equation}
\gamma:\quad x=f(t),\ y=g(t),\ z=h(t), \qquad 0\le t\le 1,
\end{equation}
such that
\begin{equation}
f_n\to f,\qquad g_n\to g,\qquad h_n\to h
\end{equation}
uniformly on $[0,1]$.

Finally, Fr\'{e}chet comments explicitly on effective computation in \cite[Note II, \S\,100]{frechet1906quelques}. He remarks there that, for the purposes of the general theory, the exact numerical value of the distance is less important than the existence of a suitable distance. He also states that the definition itself does not provide an effective method of calculation, and that such a calculation would in general be difficult. His proposed scheme is to show first that, in every continuous correspondence between the two curves, the maximum distance between corresponding points is at least some fixed value $e$, and then to exhibit at least one correspondence for which this maximum is exactly $e$. In modern terms, this already has the form of a lower-bound argument together with a correspondence attaining the bound. He illustrates the procedure on simple examples, including line segments.

\subsection{Comparison with Morse, Ward, and Ewing}

The later literature retains Fr\'{e}chet's basic pattern of minimizing a worst-case discrepancy between corresponding points, but changes the class of curves, the admissible correspondences, or the ambient spaces in which the curves live.

Morse's 1938 paper remains closest to the thesis \cite{morse1938functional}. For parametrized curves $x:[a,b]\to\mathbb R^n$ and $y:[c,d]\to\mathbb R^n$, Morse takes \emph{homeomorphisms} $T:[a,b]\to[c,d]$, i.e., continuous bijections with continuous inverse, preserving endpoints, forms the maximal pointwise distance under $T$, and then takes the greatest lower bound over all such $T$. Relative to the thesis, the main shift is one of presentation: Fr\'{e}chet begins from ordered curve-arcs and only then derives the reparameterization criterion, whereas Morse begins directly with homeomorphisms of the parameter intervals. 

Ward's two 1954 papers generalize the setting more substantially \cite{ward1954generalization,ward1954second}. In the first paper \cite{ward1954generalization}, the points of the curve no longer lie in Euclidean space equipped with its ordinary metric. Instead, Ward starts with a space $E$ carrying a separated uniform structure and with an \emph{\'ecart} or quasi-distance on $E$, and from this data he derives a Fr\'{e}chet-type \emph{\'ecart} on curves. Thus Fr\'{e}chet's 1906 construction for curves in $\mathbb R^3$ with Euclidean point distance becomes, in Ward, a construction for curves in a much more general ambient space. The second Ward paper \cite{ward1954second} enlarges the class of admissible correspondences themselves. Rather than working only with homeomorphisms between parameter intervals, Ward introduces a broader structure based on pairings of intervals, allowing parameter bases that need not all be homeomorphic and may even be disconnected. In this respect Ward is not merely restating Fr\'{e}chet's construction; he is transporting it into a genuinely broader category.

Ewing's textbook presentation \cite[\S\,6]{ewing1969calculus} is closer to the language now standard in analysis and topology. He first defines the Fr\'{e}chet distance between continuous mappings
\begin{equation}
x:[a,b]\to \mathbb{R}^n,\qquad y:[c,d]\to \mathbb{R}^n
\end{equation}
by
\begin{equation}
\rho(x,y)=\inf_h \sup_{t\in[a,b]} \|x(t)-y(h(t))\|,
\end{equation}
where $h:[a,b]\to[c,d]$ ranges over sense-preserving homeomorphisms. Because of this restriction, the curves under discussion are explicitly oriented. Ewing then lets $\mathcal C$ denote the class of continuous curves in $E_n$, each such curve being a class $\{x\}$ of continuous mappings with the property that any two members $x_1,x_2$ of the class satisfy $\rho(x_1,x_2)=0$. For curves $C_1=\{x\}$ and $C_2=\{y\}$ he sets
\begin{equation}
d(C_1,C_2)=\rho(x,y).
\end{equation}

Compared with Fr\'{e}chet's thesis, the geometric content is the same, but the order of presentation is reversed: Fr\'{e}chet begins with the curve and then defines the distance, whereas Ewing defines the distance on mappings first and only afterwards passes to curves as equivalence classes. Taken together, Morse, Ward, and Ewing show that the developments after 1906 are primarily developments of framework. The basic geometric idea remains recognizably the same, while the surrounding formalism becomes progressively more parametric, more topological, or more abstract.

\subsection{Polygonal curves and the Alt--Godau formulation}

Fr\'{e}chet himself already remarks in the 1906 thesis \cite[Note~II, \S\,100]{frechet1906quelques} that \textit{``La définition de l'écart que nous avons donnée ne fait pas prévoir comment on pourrait le calculer effectivement, et ce calcul serait certainement très difficile en général'' (The definition of distance we gave does not indicate how one could compute it effectively, and this computation would certainly be very difficult in general)}. The 1990s computational-geometry literature changes the problem in a different way \cite{alt1992measuring,alt1995computing,godau1991natural}. Here the continuous Fr\'{e}chet principle is retained, but the class of curves is restricted and the main emphasis becomes algorithmic. Alt and Godau also at least popularized if not proposed the now-standard \textit{dog leash} intuition\,\footnote{\,This explains why many papers on curve distances include \textit{dog} in their titles. Interestingly, the acknowledgements of \cite{buchin2017four} (jokingly) remark that \textit{``Fréchet-related papers require a witty title involving a dog''.} Regrettably, I failed to follow the advice.}, already present in the measuring framework of \cite{alt1992measuring}: one imagines a person and a dog walking along the two curves without backtracking, and asks for the shortest leash length that allows both to move from start to finish. In Alt and Godau's 1995 paper \cite{alt1995computing}, a polygonal curve is given by a finite sequence of vertices
\begin{equation}
P=(p_0,p_1,\dots,p_m), \qquad Q=(q_0,q_1,\dots,q_n),
\end{equation}
with the understanding that the curve is the piecewise-linear map obtained by linear interpolation of consecutive vertices. In the notation of the paper, if $P:[0,p]\to V$ and $Q:[0,q]\to V$ are polygonal curves, then the Fr\'{e}chet distance is defined by
\begin{equation}
\delta_F(P,Q)
=
\inf_{\alpha,\beta}
\max_{t\in[0,1]}
\|P(\alpha(t))-Q(\beta(t))\|,
\end{equation}

where $\alpha:[0,1]\to[0,p]$ and $\beta:[0,1]\to[0,q]$ are continuous increasing functions satisfying
\begin{equation}
\alpha(0)=0,\ \alpha(1)=p,\qquad
\beta(0)=0,\ \beta(1)=q.
\end{equation}

Thus the difference from the thesis lies not in the geometric principle but in the model of input: Fr\'{e}chet works with arbitrary continuous curve-arcs in $\mathbb R^3$, while Alt and Godau work with polygonal curves in $\mathbb R^d$ given by finite vertex data.
The distance under discussion is still the continuous Fr\'{e}chet distance, but it is now computed for a discrete representation of the input curves.

\begin{figure}[t]
     \centering
     \begin{subfigure}[t]{0.48\textwidth}
         \centering
         \includegraphics[width=\textwidth]{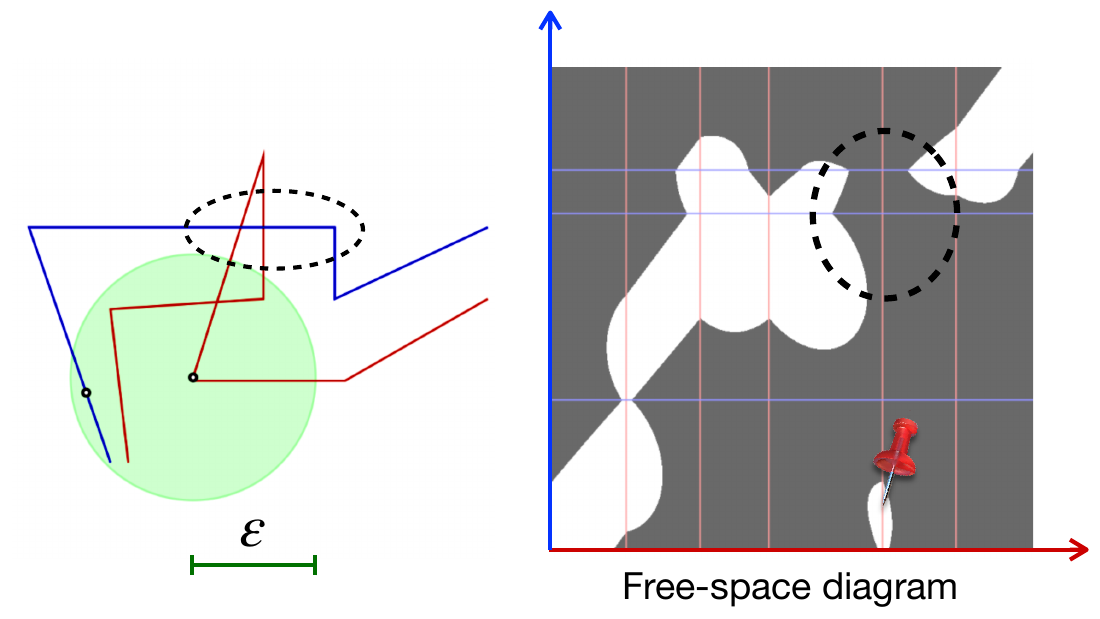}
         \caption{}
         \label{fig:2a}
     \end{subfigure}
     \hfill
     \begin{subfigure}[t]{0.48\textwidth}
         \centering
         \includegraphics[width=\textwidth]{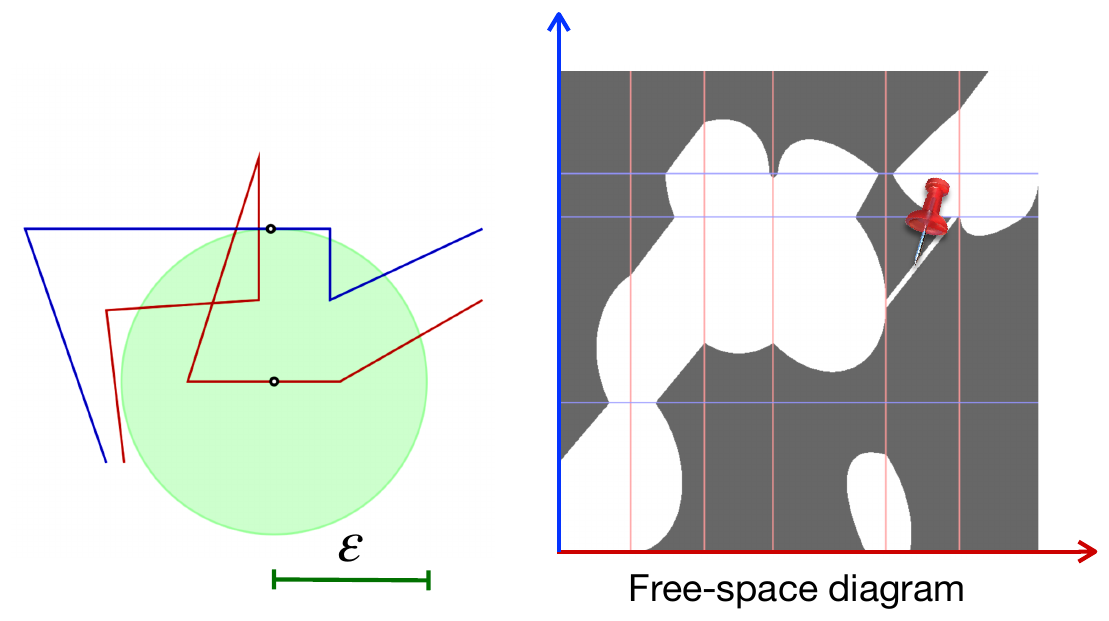}
         \caption{}
         \label{fig:2b}
     \end{subfigure}
        \caption{Two curves together with their free-space diagrams for increasing thresholds $\varepsilon$. The corresponding locations on the curves are indicated by pins in the free-space diagrams. In the free-space diagrams, the $x$-axis represents the normalized red polygonal curve, and the $y$-axis the blue one. The white regions denote the reachable free space for a given $\varepsilon$. In (a), the regions highlighted by dashed circles are disconnected, whereas in (b), a passage emerges as $\varepsilon$ increases. Adapted from \cite{dijk_frechet_diagram}.}
        \label{fig:2}
\end{figure}

Alt and Godau's key contribution is the free-space formulation. For a threshold $\varepsilon>0$, they consider
\begin{equation}
\mathcal F_{\varepsilon}(P,Q)
:=
\{(s,t)\in[0,p]\times[0,q]:\|P(s)-Q(t)\|\le \varepsilon\}.
\end{equation}

The parameter rectangle is decomposed into cells
\begin{equation}
[i-1,i]\times[j-1,j], \qquad 1\le i\le p,\ 1\le j\le q,
\end{equation}
corresponding to pairs of edges of $P$ and $Q$. For polygonal curves, the free space in each cell is convex; in the Euclidean case it is the intersection of the cell with an ellipse (possibly degenerated). The decisive criterion of the paper is:
\begin{equation}
\delta_F(P,Q)\le \varepsilon
\quad\Longleftrightarrow\quad
\text{there exists a continuous path in }\mathcal F_{\varepsilon}(P,Q)
\end{equation}
from $(0,0)$ to $(p,q)$ that is monotone in both coordinates. The distance problem is thereby converted into a reachability problem in parameter space. See \cref{fig:2} for two critical moments of increasing thresholds $\varepsilon$.
For an interactive visualization of this free-space diagram viewpoint online, we refer to \cite{dijk_frechet_diagram}.

This yields two main algorithmic consequences. First, the decision problem of determining whether $\delta_F(P,Q)\le \varepsilon$ can be solved in time $O(pq)$ by propagating reachable intervals along cell boundaries. Second, by combining this decision procedure with a search over critical values, Alt and Godau obtain an algorithm of runtime
\begin{equation}
O(pq\log(pq))
\end{equation}
for computing the Fr\'{e}chet distance of two polygonal curves. This is the point at which Fr\'{e}chet distance becomes a fully algorithmic object.

\section{Fr\'{e}chet Distance on Probability Laws}
\label{sec:laws}

\subsection{Measure-theoretic setting and notation}

To discuss probability laws and optimal transport \cite{peyre2019computational,santambrogiooptimal,villani2009optimal}, it is convenient to fix some standard measure-theoretic notation. Although the present subsection may look somewhat dense at first sight, it is intended as a compact and self-contained reference for the basic definitions used throughout the remainder of the note. This gives a single setting in which one may treat discrete laws, continuous laws, mixtures, pushforwards, and couplings without having to separate special cases at each step. 

Let $(\mathcal X,d)$ be a \emph{metric space}. This means that $\mathcal X$ is a set and 
$d:\mathcal X\times\mathcal X\to\mathbb R$
is a distance function satisfying the four axioms given in \cref{eq:distance_axioms}.
Thus it is more accurate to speak of the distance between any two elements of $\mathcal X$; when $\mathcal X$ is a geometric space, these elements are often called points. For $x\in\mathcal X$ and $r>0$, the open ball $B(x,r)$ of radius $r$ around $x$ is
\begin{equation}
B(x,r):=\{y\in\mathcal X:d(x,y)<r\}.
\end{equation}

A subset $U\subseteq\mathcal X$ is called \emph{open} if, for every $x\in U$, there exists $r>0$ such that $B(x,r)\subseteq U$. A subset $F\subseteq\mathcal X$ is called \emph{closed} if its complement $\mathcal X\setminus F$ is open. The collection of all open subsets is called the topology induced by the metric. In this sense, the metric determines which elements are close, hence which sets count as open neighborhoods, and therefore the topology of the space. Historically, this abstraction is very much in the spirit of Fr\'echet's 1906 thesis \cite{frechet1906quelques}, which is among the first works to study notions such as openness, compactness, and convergence systematically on abstract sets endowed with a distance-like structure, rather than only on subsets of Euclidean space.

To do probability on $\mathcal X$, one needs a class of subsets to which probabilities may be assigned. This is the role of a $\sigma$-algebra. A $\sigma$-algebra $\mathcal A$ on $\mathcal X$ is a collection of subsets of $\mathcal X$ such that $\mathcal X\in\mathcal A$, whenever $A\in\mathcal A$ one also has $\mathcal X\setminus A\in\mathcal A$, and whenever $A_1,A_2,\dots\in\mathcal A$ one also has $\bigcup_{n=1}^\infty A_n\in\mathcal A$. From these properties it follows that $\mathcal A$ is also closed under countable intersections. These are the set-theoretic operations that appear throughout probability: one needs complements to speak about events of the form ``not $A$'', countable unions to speak about ``$A_1$ or $A_2$ or \dots'', and countable intersections for ``$A_1$ and $A_2$ and \dots''. The closure properties of a $\sigma$-algebra are therefore exactly those needed to state the usual rules of probability in a stable way.

The Borel $\sigma$-algebra $\mathcal B(\mathcal X)$ is the $\sigma$-algebra generated by the open subsets of $\mathcal X$, that is, the smallest $\sigma$-algebra containing every open set. Since the open sets encode the topology, the Borel $\sigma$-algebra is the natural measurable structure associated with the metric space. A \emph{measure} on $\mathcal B(\mathcal X)$ is a function $\mu:\mathcal B(\mathcal X)\to[0,\infty]$ that assigns size to Borel sets \footnote{\,Borel sets are named after Émile Borel, a mathematician and the Minister of the Navy of the French Third Republic in 1925. Borel also played an important role in Fr\'{e}chet's academic career. He encouraged Fr\'{e}chet to come back to Paris from Strasbourg to seek positions and supported Fr\'{e}chet's candidacy. In 1941, Fr\'{e}chet succeeded Borel as chair in the Calculus of Probabilities and Mathematical Physics at the Université de Paris.} and is countably additive on pairwise disjoint families. A \emph{probability measure} is a measure with total mass $1$. Thus a \emph{Borel probability measure} is a function $\mu:\mathcal B(\mathcal X)\to[0,1]$ such that $\mu(\mathcal X)=1$ and
\begin{equation}
\mu\!\left(\bigcup_{n=1}^\infty A_n\right)=\sum_{n=1}^\infty \mu(A_n)\,,
\end{equation}
whenever $A_1,A_2,\dots$ are pairwise disjoint Borel sets.

We write
\begin{equation}
\mathcal P(\mathcal X)
:=
\{\mu:\mu \text{ is a Borel probability measure on }\mathcal X\}\,,
\end{equation}
and call \(\mathcal P(\mathcal X)\) the space of Borel probability measures on \(\mathcal X\), that is, the collection of all probability laws defined on the Borel \(\sigma\)-algebra of \(\mathcal X\).

Many results in probability and optimal transport are most naturally formulated when the underlying metric space has two additional regularity properties. A sequence \((x_n)_{n\ge1}\) in \(\mathcal X\) is called a \emph{Cauchy sequence} if, for every \(\varepsilon>0\), there exists \(N\) such that \(d(x_n,x_m)<\varepsilon\) whenever \(n,m\ge N\). The space \((\mathcal X,d)\) is called \emph{complete} if every Cauchy sequence converges to some element of \(\mathcal X\). It is called \emph{separable} if it contains a countable dense subset. A metric space that is both complete and separable is called a \emph{Polish space}. This class of spaces is standard in modern probability theory and provides a convenient setting for much of optimal transport as well \cite{peyre2019computational}. Important examples include \(\mathbb R^d\), separable Banach spaces, and many function spaces. A Banach space is a vector space equipped with a norm and complete w.r.t. the metric induced by that norm.

The reason for using the term \emph{probability law} rather than \emph{probability distribution} is partly historical and partly conceptual. Historically, Fr\'{e}chet himself writes \emph{loi}; translating this as \emph{law} preserves the terminology of the 1957 paper. Conceptually, \emph{distribution} can mean either a random object's law or a concrete representation such as a density, mass function, or cumulative distribution function. The word \emph{law} helps emphasize that the object is the measure itself, independent of any particular representation. This is especially useful once one moves beyond one-dimensional densities to pushforwards, couplings, or probability measures on abstract spaces.

If $(\Omega,\mathcal F,\mathbb P)$ is a \emph{probability space}, meaning that $\Omega$ is a sample space, $\mathcal F$ is a $\sigma$-algebra of events, and $\mathbb P$ is a probability measure on $\mathcal F$, then a measurable map $X:\Omega\to\mathcal X$
is called an \emph{$\mathcal X$-valued random variable}. Measurability means that $X^{-1}(A)\in\mathcal F$ for every Borel set $A\in\mathcal B(\mathcal X)$. The law of $X$ is then the pushforward measure
\begin{equation}
\mathcal L(X)=X_{\#}\mathbb P\in\mathcal P(\mathcal X),
\qquad
X_{\#}\mathbb P(A)=\mathbb P\bigl(X^{-1}(A)\bigr)=\mathbb P(X\in A),
\quad A\in\mathcal B(\mathcal X)\,.
\end{equation}
Here \(\mathbb P\) is a probability measure on the sample space \((\Omega,\mathcal F)\), whereas \(\mathcal L(X)=X_{\#}\mathbb P\) is a probability measure on the target space \(\mathcal X\), and therefore an element of \(\mathcal P(\mathcal X)\).

The notation $X_{\#}\mathbb P$ means that the probability measure $\mathbb P$ on the sample space $\Omega$ is carried through the map $X$ and thereby becomes a measure on the target space $\mathcal X$. In plain terms, the pushforward measure keeps track of how much mass arrives in each subset of the target space after the map has been applied. The law of a random variable is therefore the original probability measure on outcomes, viewed after those outcomes have been converted into values of the random variable. More generally, let \((S,\mathcal A)\) be a measurable space, let \(T:S\to\mathcal Y\) be a measurable map into a metric space \(\mathcal Y\), and let \(\eta\) be a probability measure on \((S,\mathcal A)\). The pushforward measure \(T_{\#}\eta\) is defined by
\begin{equation}
T_{\#}\eta(A)=\eta\bigl(T^{-1}(A)\bigr),
\qquad A\in\mathcal B(\mathcal Y),
\end{equation}
and is a probability measure on the target space \(\mathcal Y\). In the special case where \(S=\mathcal X\) and \(\eta\in\mathcal P(\mathcal X)\), it is again a Borel probability measure on \(\mathcal Y\). In other words, the mass assigned to a set \(A\) after applying \(T\) is exactly the mass that was originally sitting in the preimage \(T^{-1}(A)\). 

For example, let $\eta$ be the uniform probability measure on $[0,1]$ and let $T(x)=x^2$. If we choose the target set \(A=[0,\tfrac14]\subseteq[0,1]\), then 
\begin{equation}
T^{-1}(A)=\{x\in[0,1]:x^2\in[0,\tfrac14]\}=[0,\tfrac12],
\qquad
T_{\#}\eta(A)=\eta([0,\tfrac12])=\tfrac12\,.
\end{equation}
Thus the pushforward $T_{\#}\eta$ is the law of $x^2$ when $x$ is uniformly distributed on $[0,1]$. This is the basic mechanism by which a random variable produces a probability law: one starts with a probability measure on the underlying sample space and pushes it forward to the target space of values. Accordingly, one may describe a probability law either abstractly, as a measure $\mu\in\mathcal P(\mathcal X)$, or concretely, as the law $\mathcal L(X)$ of an $\mathcal X$-valued random variable $X$.

If $\mu,\nu\in\mathcal P(\mathcal X)$, a \emph{coupling} of $\mu$ and $\nu$ is a Borel probability measure $\pi$ on $\mathcal X\times\mathcal X$ whose first marginal is $\mu$ and whose second marginal is $\nu$. Here the first marginal is the measure obtained by forgetting the second coordinate, and the second marginal is obtained by forgetting the first. We denote the set of all such couplings by
\begin{equation}
\Pi(\mu,\nu)
:=
\bigl\{\pi\in\mathcal P(\mathcal X\times\mathcal X):
(\mathrm{pr}_1)_{\#}\pi=\mu,\ (\mathrm{pr}_2)_{\#}\pi=\nu\bigr\}\,,
\end{equation}
where $\mathrm{pr}_1,\mathrm{pr}_2$ are the coordinate projections. In probabilistic language, a coupling is a joint law of a pair $(X,Y)$ whose marginal laws are $\mu$ and $\nu$.

In a nutshell, all this notation is meant to make clear what a probability law is, and to place it in a setting in which distances between probability laws can be defined in a natural and workable way.

\subsection{Fr\'{e}chet's 1957 paper}
\label{sec:frechet1957}

Fr\'{e}chet's 1957 paper \cite{frechet1957distance} is narrower in scope than the 1906 thesis \cite{frechet1906quelques} and is concerned specifically with one-dimensional probability laws represented by their cumulative distribution functions. Before turning to Fr\'{e}chet's own construction, it is helpful to separate it from Paul L\'evy's 1950 note \cite{frechet1950recherches}. In that note, L\'evy gives two \emph{direct} definitions of distance between laws by representing a law as a cumulative-distribution curve \(\Gamma\): first, a max-distance comparison of the intersection points of \(\Gamma\) and \(\Gamma'\) with the lines \(x+y=c\); second, a point-to-curve construction based on the larger of two one-sided maximal distances. He then discusses separately how a distance between random variables might be related to a distance between laws through the gap variable \(Z=|X-Y|\). This is not yet the admissible-joint-law framework used by Fr\'{e}chet\,\footnote{\label{fn:note}\,Fr\'{e}chet writes in 1957 as though L\'evy's 1950 note supplied three definitions of distance between laws. In the 1950 note itself, however, one finds two direct curve-based definitions of law-distance, followed by a separate discussion of how distances between random variables and distances between laws might be related. Moreover, the admissible-bivariate-law scheme that Fr\'{e}chet develops in 1957 is not stated there in the same form.}.

Fr\'{e}chet begins instead by distinguishing the distance between two random numbers in a single trial from what he calls a \emph{global distance}, that is, a single quantity intended to measure their discrepancy across the whole \emph{ensemble} of trials. He then considers two simultaneously realized real-valued random variables \(X\) and \(Y\), with laws \(L\) and \(L'\), and writes \(F\) and \(G\) for their cumulative distribution functions \textit{(les fonctions de répartition)}. His first definition takes the distance between the laws \(L,L'\) to be the lower bound of the chosen global distance over all admissible bivariate laws \(H(x,y)\) having margins \(F\) and \(G\) \cite{frechet1957distance}. In this respect, Fr\'{e}chet's point of departure is closer to L\'evy's comparison between distances on random variables and distances on laws than to L\'evy's two direct curve-geometric definitions.

In modern schematic notation, one may summarize Fr\'{e}chet's first definition as
\begin{equation}
(L,L')=\inf_H ([X],[Y])_H,
\end{equation}
where \(H\) ranges over the admissible bivariate cumulative distribution functions with the prescribed margins. Here \emph{admissible} means exactly that the pair \(X,Y\) has the given separate laws \(L,L'\), or equivalently that the bivariate cumulative distribution function \(H\) has margins \(F\) and \(G\). Fr\'{e}chet explicitly says that one may call \(F\) and \(G\) the \emph{margins} of the correlation table defined by \(H\).
In modern optimal transport language, these admissible bivariate laws are precisely \emph{couplings} of the given marginal laws.

The paper insists that this first definition depends on the prior choice of global distance between random numbers, and Fr\'{e}chet lists several possibilities \cite{frechet1957distance}. In this way the paper begins not with one canonical distance between laws, but with a general scheme that turns a previously chosen global distance between simultaneous random numbers into a distance between their laws.

At the same time Fr\'{e}chet recalls the extremal bounds for admissible bivariate laws. The set of admissible \(H(x,y)\) is exactly the set of bivariate cumulative distribution functions lying between
\begin{equation}
H_0(x,y)=\sup\{F(x)+G(y)-1,0\}
\end{equation}
and
\begin{equation}
H_1(x,y)=\inf\{F(x),G(y)\},
\end{equation}
with endpoints included \cite{frechet1957distance}. A short proof of these bounds is included in \cref{app:frechet-bounds}.

The most detailed part of the article treats the case where the chosen global distance is the quadratic mean deviation\,\footnote{\,Fr\'{e}chet adds (originally in French): ``To be precise, we write $\mathfrak M_H$ to recall that it is the mean \textit{(la moyenne)} when the cumulative distribution function \textit{(la fonction de répartition)} of the pair $(X,Y)$ is $H(x,y)$.''}
\begin{equation}
\sqrt{\mathfrak M_H(X-Y)^2},
\end{equation}
where \(\mathfrak M_H\) indicates the mean taken under the admissible joint law \(H\). In words, this is the root mean square discrepancy between the two simultaneously realized random numbers. The choice is natural because it turns the pointwise discrepancy \(X-Y\) over individual trials into one global quantity over the whole ensemble of trials, and because the square leads to a decomposition in terms of means, quadratic mean deviations, and correlation. Assuming finite second moments, Fr\'{e}chet writes the means of \(X\) and \(Y\) as \(a\) and \(a'\), their quadratic mean deviations as \(\sigma\) and \(\sigma'\), and introduces the reduced variables
\begin{equation}
Z=\frac{X-a}{\sigma},
\qquad
T=\frac{Y-a'}{\sigma'}.
\end{equation}
He then derives the identities
\begin{equation}
\mathfrak M_H(X-Y)^2=(a-a')^2+\sigma^2-2\sigma\sigma' r+\sigma'^2,
\end{equation}
where \(r\) is the linear correlation coefficient of \(X\) and \(Y\). Since \(a,a',\sigma,\sigma'\) are fixed by the separate laws, minimizing the quadratic mean deviation over admissible \(H\) reduces to maximizing \(r\). If \(\rho\) denotes the upper bound of that correlation coefficient over all admissible bivariate laws, then
\begin{equation}
(L,L')=\sqrt{(a-a')^2+\sigma^2-2\sigma\sigma'\rho+\sigma'^2}.
\end{equation}
These correspond to \cite[Eqs.~(5)--(6)]{frechet1957distance}.
This is the route by which the paper passes from a coupling problem to an explicit distance formula.

Fr\'{e}chet then recalls the work of Bass and Dall'Aglio showing that, in the quadratic case, the extremum is attained for \(H\equiv H_1\) under suitable hypotheses. Under the maximal admissible law \(H_1\), the pair \((X,Y)\) is no longer arbitrary. In the continuous strictly monotone case, if \(G(\lambda(x))=F(x)\), then under \(H_1\) one has almost surely (i.e., with probability \(1\)), the functional relation \(Y=\lambda(X)\), and Fr\'{e}chet sets
\begin{equation}
\Delta^2=\mathfrak M_{H_1}(X-Y)^2.
\end{equation}
The quantity \(\Delta\) is therefore the value of the chosen global distance at the distinguished admissible law \(H_1\). He then obtains
\begin{equation}
\Delta^2=\int_{-\infty}^{+\infty}[\lambda(x)-x]^2\,\mathrm{d}F(x),
\end{equation}
and, equivalently,
\begin{equation}
\Delta^2=\int_{-\infty}^{+\infty}[\mu(y)-y]^2\,\mathrm{d}G(y).
\end{equation}
He then rewrites the same quantity as the distance between the two laws \cite[Eq.~(12)]{frechet1957distance}:
\begin{equation}
(L,L')^2=\int_{-\infty}^{+\infty}[\lambda(x)-x]^2\,\mathrm{d}F(x).
\end{equation}
If the laws are written at common probability levels by increasing functions \(l,m:[0,1]\to\mathbb R\), he writes \cite[Eq.~(13)]{frechet1957distance}:
\begin{equation}
\label{eq:frechet1957_13}
(L,L')^2=\int_0^1 [l(t)-m(t)]^2\,\mathrm{d}t.
\end{equation}
The same quantity is thus expressed directly in terms of common probability levels. Later the paper records the corresponding formula for general exponents \(\alpha\ge1\),
\begin{equation}
\label{eq:frechet1957_17}
(L,L')^\alpha=\int_0^1 |l(t)-m(t)|^\alpha\,\mathrm{d}t.
\end{equation}
This is the corresponding \(\alpha\)-version \cite[Eq.~(17)]{frechet1957distance})
At this point one can already see a formal resemblance to the curve-theoretic branch discussed earlier. In both cases one compares two continuous traversals and measures their discrepancy along a parameter. The crucial difference lies in the role of that parameter. For curves in the 1906 sense \cite{frechet1906quelques}, the parametrizations are auxiliary and may be varied independently through admissible reparametrizations; the distance is then obtained by optimizing over those choices. Here, by contrast, the parameter \(t\in[0,1]\) represents a common probability level, so the two functions \(l\) and \(m\) are compared at the same value of \(t\). The coupling information is already encoded in this shared parametrization, rather than being optimized away by separate reparametrizations. This observation will become important in \cref{sec:laws-as-curves}, where probability laws are viewed more explicitly through their cumulative distribution and quantile curves.
These are already the familiar one-dimensional quantile formulae, though reached through the language of admissible bivariate laws and extremal correlation.

Fr\'{e}chet is, however, not fully satisfied with L\'evy's first definition. He calls it attractive (in his words \emph{``très séduisante''}), but remarks that its drawback is the need to solve a minimization problem anew for each chosen global distance. He therefore proposes a new definition, his ``fourth definition'' (see the previous footnote\,\textsuperscript{\ref{fn:note}}), obtained by evaluating the global distance directly at the maximal admissible law \(H_1\):
\begin{equation}
(L,L')=([X],[Y])_{H_1}.
\end{equation}
This is from \cite[Eq.~(15)]{frechet1957distance}. Its advantage is that one no longer has to solve a fresh minimization problem for each chosen global distance: the law-distance is defined directly by evaluating that distance at the distinguished extremal law \(H_1\). In the quadratic case, this recovers exactly the quantity \(\Delta\) above.
The final part of the section proves that this quantity is indeed a distance in the sense stated at the beginning of the paper. Symmetry and definiteness are inherited from the corresponding properties of the underlying global distance. For the triangle inequality, Fr\'{e}chet introduces three laws \(L,L',L''\) with cumulative distribution functions \(F,G,K\) and the trivariate distribution
\begin{equation}
E(x,y,z)=\min\{F(x),G(y),K(z)\},
\end{equation}
whose two-dimensional margins are the relevant \(H_1\)-type bounds. The key comparison is from \cite[Eq.~(16)]{frechet1957distance}. In this way he proves the triangle inequality for the new definition.

Historically, the importance of the 1957 paper is therefore twofold. On the one hand, it derives explicit one-dimensional formulae by studying extremal admissible correlations, especially in the quadratic case. On the other hand, it formulates a general coupling-based pattern: start from a global distance between simultaneous random variables, restrict attention to admissible joint laws with prescribed margins, and then either minimize over all such laws or evaluate directly at the extremal law \(H_1\). In that sense, the paper already contains the essential one-dimensional structure that modern optimal transport later reorganizes in a more systematic form.

\subsection{Wasserstein distance}

The modern framework for these ideas is Wasserstein distance and optimal transport \cite{peyre2019computational,santambrogiooptimal,vaserstein1969markov,villani2009optimal}. The basic question is the following: if one law $\mu$ is thought of as a mass distribution on $\mathcal X$ and another law $\nu$ as a target mass distribution, what is the least possible cost of rearranging the first into the second?

Historically, the older formulation is Monge's problem. One looks for a measurable transport map
\begin{equation}
T:\mathcal X\to\mathcal X
\end{equation}
such that
\begin{equation}
T_{\#}\mu=\nu,
\end{equation}
meaning that the image of the law \(\mu\) under \(T\) (i.e., the law obtained after transporting $\mu$ through the map $T$) is exactly the law $\nu$. If moving one unit of mass from $x$ to $T(x)$ costs $c(x,T(x))$, then the total cost of the map $T$ is
\begin{equation}
\int_{\mathcal X} c(x,T(x))\,\mathrm{d}\mu(x).
\end{equation}

Monge's problem is therefore to minimize this integral over all measurable maps $T$ satisfying $T_{\#}\mu=\nu$. This formulation is geometrically direct, but it is rigid: each point $x$ must be sent to a single destination $T(x)$. In this sense, transport is deterministic.

Kantorovich's formulation is a \emph{relaxation} because it enlarges the admissible class from deterministic maps to arbitrary couplings \cite[\S\,2.3]{peyre2019computational}. If one already has a Monge map $T$, then one can form a coupling by taking the joint law of the pair $(x,T(x))$ when $x$ is distributed according to $\mu$; this is the measure $(\mathrm{id},T)_{\#}\mu$, where $\mathrm{id}$ denotes the identity map. A general coupling, however, need not arise from a single map. The passage from Monge to Kantorovich is therefore a passage from \emph{deterministic} transport to \textit{probabilistic (or fuzzy)} transport \cite[\S\,2.3]{peyre2019computational}: instead of assigning each source point one destination, one allows the mass near $x$ to be distributed across several target locations $y$. In this sense, a coupling encodes not a single motion rule $y=T(x)$ but a joint law describing how source mass and target mass are matched. This is the phenomenon Peyr\'e and Cuturi emphasize as \emph{mass splitting} \cite[\S\,2.3]{peyre2019computational}. Instead of deciding where each point must go, one specifies only a joint law
\begin{equation}
\pi\in\Pi(\mu,\nu)
\end{equation}
on $\mathcal X\times\mathcal X$. Such a coupling may be interpreted as a transport plan: the amount of mass moved from a region near $x$ to a region near $y$ is encoded by the mass that $\pi$ places near $(x,y)$. If moving a unit mass from $x$ to $y$ costs $c(x,y)$, then the total cost of the plan $\pi$ is
\begin{equation}
\int_{\mathcal X\times\mathcal X} c(x,y)\,\mathrm{d}\pi(x,y).
\end{equation}

The Kantorovich problem is to minimize this quantity over all couplings of $\mu$ and $\nu$.

For Wasserstein distance, the cost is the $p$th power of the metric:
\begin{equation}
c(x,y)=d(x,y)^p,\qquad 1\le p<\infty.
\end{equation}

To ensure that this cost is finite, one restricts to laws having finite $p$th moment. Thus, for a fixed point $x_0\in\mathcal X$, one sets
\begin{equation}
\mathcal P_p(\mathcal X)
:=
\left\{
\mu\in\mathcal P(\mathcal X):
\int_{\mathcal X} d(x,x_0)^p\,\mathrm{d}\mu(x)<\infty
\right\}.
\end{equation}

The particular choice of $x_0$ does not matter: if the integral is finite for one point, then it is finite for every other point as well, by the triangle inequality. For $\mu,\nu\in\mathcal P_p(\mathcal X)$, the Wasserstein-$p$ distance is then defined by
\begin{equation}
W_p(\mu,\nu)
=
\left(
\inf_{\pi\in\Pi(\mu,\nu)}
\int_{\mathcal X\times\mathcal X} d(x,y)^p\,\mathrm{d}\pi(x,y)
\right)^{1/p}.
\end{equation}

This is the Kantorovich primal formulation. When $\mathcal X=\mathbb R$ and $d(x,y)=|x-y|$, it is the modern version of the coupling-based construction studied by Fr\'{e}chet in 1957.
It is called \emph{primal} because this formula optimizes directly over transport plans. It asks for the cheapest admissible way to match the mass of $\mu$ with the mass of $\nu$. In discrete problems this becomes the classical linear transport problem: one distributes mass from source bins to target bins and minimizes the total shipping cost.

The same quantity also admits a dual description. For a lower semicontinuous cost $c$, meaning in particular that the cost behaves well under limits of minimizing sequences, Kantorovich duality states that
\begin{equation}
\inf_{\pi\in\Pi(\mu,\nu)}\int c(x,y)\,\mathrm{d}\pi(x,y)
=
\sup_{\substack{f\in L^1(\mu),\,g\in L^1(\nu)\\ f(x)+g(y)\le c(x,y)}}
\left(
\int_{\mathcal X} f\,\mathrm{d}\mu+\int_{\mathcal X} g\,\mathrm{d}\nu
\right).
\end{equation}

The infimum on the left ranges over transport plans; the supremum on the right ranges over pairs of integrable functions constrained by
\begin{equation}
f(x)+g(y)\le c(x,y)\qquad \text{for all }x,y\in\mathcal X.
\end{equation}

The Lebesgue\,\footnote{\,Henri Lebesgue, one of Borel's PhD students, corresponded with Fr\'{e}chet on mathematical questions while Fr\'{e}chet was working on his doctoral thesis; four of their letters were later published \cite{taylor1981quatre}.} integrability conditions $f\in L^1(\mu)$ and $g\in L^1(\nu)$ mean that the first powers of $|f|$ and $|g|$ are integrable with respect to the marginal laws, so that the quantities $\int f\,\mathrm{d}\mu$ and $\int g\,\mathrm{d}\nu$ are well defined. The pair $(f,g)$ is often called a pair of Kantorovich potentials. The constraint means that the two functions may never jointly assign a value larger than the transport cost between $x$ and $y$. The dual problem gives a function-theoretic perspective on the same transport question, replacing transport plans by admissible functions.

For $p=1$, the dual becomes especially concrete. In this case the cost is the ambient metric \(c(x,y)=d(x,y)\). The general dual starts from two functions constrained by
\(
f(x)+g(y)\le d(x,y).
\)
Because the cost is now the metric itself, this constraint can be encoded by a single test function: instead of keeping track of two separate potentials, one may describe the same optimization in terms of one function whose variation between \(x\) and \(y\) is never larger than the distance \(d(x,y)\). That is exactly the \(1\)-Lipschitz condition. The resulting one-function form is the Kantorovich--Rubinstein formula
\begin{equation}
W_1(\mu,\nu)
=
\sup_{\|f\|_{\mathrm{Lip}}\le 1}
\left(
\int_{\mathcal X} f\,\mathrm{d}\mu-\int_{\mathcal X} f\,\mathrm{d}\nu
\right),
\end{equation}
where the supremum ranges over all $1$-Lipschitz functions on $\mathcal X$. A function $f:\mathcal X\to\mathbb R$ is called $1$-Lipschitz if
\begin{equation}
|f(x)-f(y)|\le d(x,y)\qquad \text{for all }x,y\in\mathcal X.
\end{equation}

Thus such a function is allowed to vary, but never faster than the metric itself permits. The dual formula says that \(W_1\) is the largest possible discrepancy between the expectations of a test function whose oscillation is controlled in exactly this way. We refer to \cite[\S6.1]{peyre2019computational} for a detailed derivation.

For $p=2$, the dual is less elementary, but in Euclidean settings it may be written in terms of $c$-convex potentials \cite{peyre2019computational}. Here $c$-convexity means convexity adapted to the cost function \(c(x,y)\), just as ordinary convexity is adapted to the linear cost \(x\cdot y\). The precise formalism is more elaborate, but the guiding point is the same: the transport problem can be studied either by plans or by suitable potentials.

The dual point of view is useful because it replaces an optimization over couplings on $\mathcal X\times\mathcal X$ by an optimization over functions on $\mathcal X$. It also reveals which observables distinguish the two laws. Here a \emph{test function} means a real-valued function \(f\) used to probe the laws through the quantities \(\int f\,\mathrm{d}\mu\) and \(\int f\,\mathrm{d}\nu\). By contrast, an indicator function of an event \(A\subseteq\mathcal X\) takes only the values \(0\) and \(1\), so that \(\int \mathbf 1_A\,\mathrm{d}\mu=\mu(A)\) records only the probability of that event. In particular, the dual for \(W_1\) shows that Wasserstein distance measures how differently the two laws integrate all Lipschitz test functions, and therefore captures not only differences in event probabilities but also the geometry of how mass is arranged in the space.

In one dimension the theory becomes especially explicit. If $F_\mu,F_\nu$ are the cumulative distribution functions of $\mu,\nu$ and $Q_\mu,Q_\nu$ their quantile functions, then following \cite[Remark 2.29]{peyre2019computational}
\begin{equation}
W_p(\mu,\nu)^p
=
\int_0^1 |Q_\mu(t)-Q_\nu(t)|^p\,\mathrm{d}t.
\end{equation}

We revisit \cref{sec:frechet1957}:
\begin{equation*}\tag*{(\ref{eq:frechet1957_13}, revisited)}
(L,L')^2=\int_0^1 [l(t)-m(t)]^2\,\mathrm{d}t,
\end{equation*}
and
\begin{equation*}\tag*{(\ref{eq:frechet1957_17}, revisited)}
(L,L')^\alpha=\int_0^1 |l(t)-m(t)|^\alpha\,\mathrm{d}t.
\end{equation*}
Thus Fr\'{e}chet's formulae from \cite{frechet1957distance}, written here as \cref{eq:frechet1957_13} and \cref{eq:frechet1957_17}, are, in modern notation, precisely the one-dimensional Wasserstein identities for $p=2$ and for general $p=\alpha$. The same family also has a limiting endpoint at \(p=\infty\): one may define
\begin{equation}
W_\infty(\mu,\nu)
:=
\inf_{\pi\in\Pi(\mu,\nu)}
\operatorname*{ess\,sup}_{(x,y)\sim\pi} d(x,y),
\end{equation}
which is the smallest possible worst-case displacement among all couplings. Unlike the finite-\(p\) distances, this endpoint quantity has some distinctive features of its own and is often treated on bounded-support classes such as \(\mathcal P_\infty(\mathcal X)\).
We refer to \cite{champion2008wasserstein,santambrogiooptimal} for more details.

In summary, Wasserstein distance places Fr\'{e}chet's 1957 formulae into a broader framework: probability laws are compared by optimizing over couplings, the one-dimensional quantile identities become special cases of a general transport theory, and the quadratic expression studied by Fr\'{e}chet appears as the \(p=2\) member of this family. The comparison with curve-theoretic Fr\'{e}chet distance will be taken up separately in \cref{sec:laws-as-curves}.

\subsection{Special case of Gaussians}
\label{sec:gaussian-laws}
The Gaussian case is the most important special case for later applications, but it is best understood as one point on a longer line of development. In Fr\'{e}chet's 1957 paper, the quadratic global distance already yields particularly simple formulas whenever the two laws have the same reduced probability law and finite quadratic mean deviation. 
Fr\'{e}chet gives as examples not only the normal law but also the Laplace law and the uniform law; see Eq. (14) and the discussion immediately following it in \cite{frechet1957distance}. Thus the underlying phenomenon is already broader than the Gaussian family: once the laws differ only by location and scale inside one common standardized family, the quadratic distance is determined by those first two parameters.

In Wasserstein's notation, Fr\'{e}chet's one-dimensional quadratic identity for two Gaussian laws
\begin{equation}
\mu=\mathcal N(m_1,\sigma_1^2),\qquad \nu=\mathcal N(m_2,\sigma_2^2)
\end{equation}
reduces to
\begin{equation}
W_2^2(\mu,\nu)=(m_1-m_2)^2+(\sigma_1-\sigma_2)^2.
\end{equation}
This is fully consistent with the quantile formula, since
\begin{equation}
Q_\mu(t)=m_1+\sigma_1\Phi^{-1}(t),\qquad
Q_\nu(t)=m_2+\sigma_2\Phi^{-1}(t),
\end{equation}
and the $L^2([0,1])$ norm of their difference gives the same expression. Peyr\'e and Cuturi summarize this in Remark~2.30 by observing that the one-dimensional Gaussian family may be identified with the half-plane \(\{(m,\sigma):\sigma>0\}\) endowed with the Euclidean metric \cite{peyre2019computational}. In this sense, the quadratic transport geometry of one-dimensional Gaussians is already as flat and explicit as possible.

Dowson and Landau extend this picture from Fr\'{e}chet's 1957 one-dimensional setting to multivariate normal laws \cite{dowson1982frechet}. If
$
\mu=\mathcal N(m_1,\Sigma_1),\, \nu=\mathcal N(m_2,\Sigma_2)
$
are Gaussian laws on $\mathbb R^d$, then the quadratic Wasserstein distance admits the closed form
\begin{equation}
W_2^2(\mu,\nu)
=
\|m_1-m_2\|^2
+
\operatorname{Tr}\!\left(
\Sigma_1+\Sigma_2-2\left(\Sigma_1^{1/2}\Sigma_2\Sigma_1^{1/2}\right)^{1/2}
\right).
\end{equation}
In dimensions \(d>1\), this closed form comes from optimal transport geometry and the special affine-radial structure of the Gaussian family, not from any direct higher-dimensional analogue of the one-dimensional CDF/quantile identities discussed later in \cref{sec:laws-as-curves}.
Dowson and Landau present this explicitly as the Fr\'echet distance between multivariate normal distributions. Their starting point is again Fr\'{e}chet's quadratic definition
$
d(F,G)^2=\inf \mathbb E\|X-Y\|^2
$
over all couplings with the prescribed laws, now interpreted for random vectors. They show that the mean part separates cleanly from the covariance part, and they emphasize that
\begin{equation}
d_\Sigma(\Sigma_1,\Sigma_2)^2
:=
\operatorname{Tr}\!\left(
\Sigma_1+\Sigma_2-2\left(\Sigma_1^{1/2}\Sigma_2\Sigma_1^{1/2}\right)^{1/2}
\right)
\end{equation}
is itself a natural metric on covariance matrices. In modern optimal transport language, this covariance term is the Bures metric; see again \cite[Remark~2.30]{peyre2019computational}.

Why does a closed form appear here? The decisive point is not merely the finiteness of the quadratic mean deviation. Finite second moments are necessary in order for the quadratic transport problem to make sense, but they are not by themselves enough to force an explicit formula in general. What makes the Gaussian and related cases tractable is that one is not comparing arbitrary laws with given moments, but laws inside one fixed \textit{affine-radial} family. In one dimension, Fr\'{e}chet already sees this for laws with the same reduced distribution. In higher dimensions, Dowson and Landau show that the same phenomenon persists for multivariate normal laws.

Their paper also points beyond the Gaussian case. \cite[Eq. (16)]{dowson1982frechet} remains valid for wider multivariate families of \emph{elliptically contoured} laws, and they mention in particular uniform distributions on ellipsoids. These examples are important here because they are bounded. The closed form is therefore not a peculiarity of Gaussian tails: it already appears for compactly supported laws whose level sets are ellipsoids.

Gelbrich's 1990 paper places this formula in the broader Wasserstein framework \cite{gelbrich1990formula}. One of its central messages is that the Gaussian expression is not only an exact formula in special families, but also a universal lower bound controlled only by first and second moments. More precisely, if $\mu,\nu$ are arbitrary laws on $\mathbb R^d$ with means $m_1,m_2$ and covariance matrices $\Sigma_1,\Sigma_2$, then
\begin{equation}\label{eq:gelbrich-lower-bound}
W_2^2(\mu,\nu)
\ge
\|m_1-m_2\|^2
+
\operatorname{Tr}\!\left(
\Sigma_1+\Sigma_2-2\left(\Sigma_1^{1/2}\Sigma_2\Sigma_1^{1/2}\right)^{1/2}
\right).
\end{equation}
Thus among all laws with prescribed means and covariances, the Gaussian pair yields the smallest possible quadratic transport cost. This gives the Gaussian formula a double role: it is exact for Gaussians, and extremal for all laws with the same first two moments.

Gelbrich then proves something stronger for related elliptically contoured families. If two laws share a common radial profile and differ only by center and affine deformation, then equality holds in \cref{eq:gelbrich-lower-bound}. Concretely, one may think of such a family as consisting of densities of the form
\begin{equation}
\rho_i(x)
=
\frac{1}{\sqrt{\det S_i}}\,
h\!\left((x-m_i)^\top S_i^{-1}(x-m_i)\right),
\qquad i=1,2,
\end{equation}
for one fixed non-negative profile \(h:[0,\infty)\to[0,\infty)\) and positive definite shape matrices \(S_1,S_2\). The function \(h\) is the radial profile: it specifies how mass is distributed as a function of generalized radius. The quadratic form \((x-m_i)^\top S_i^{-1}(x-m_i)\) is a generalized squared radius, so the level sets are ellipsoids centered at \(m_i\). Gaussian laws correspond to the exponential profile \(h(t)=c\,e^{-t/2}\), while elliptically bounded laws such as uniform measures on ellipsoids arise from compactly supported profiles \(h\). An elliptically contoured law is therefore not determined by mean and covariance alone; the profile \(h\) matters as well. The closed form is exact when the two laws share the same profile \(h\), in which case the remaining variation is entirely through center and covariance. If the profiles differ, the same expression should in general be regarded only as the Gelbrich lower bound.

This is the precise sense in which finite quadratic mean deviation leads to a closed form here: not for arbitrary distributions, but for families in which second moments interact with a common affine-radial structure. Fr\'{e}chet observes this in one dimension through the common reduced law, Dowson and Landau \cite{dowson1982frechet} work it out for multivariate normals and related elliptic families, Gelbrich \cite{gelbrich1990formula} proves the general lower bound and the corresponding equality cases, and Peyr\'e--Cuturi \cite{peyre2019computational} summarize the result in modern Wasserstein language.
Thus the Gaussian formula belongs to the probability-law branch rather than to the classical Fr\'{e}chet distance on curves. Its connection with Fr\'{e}chet is through the 1957 framework: probability laws are compared via admissible couplings, and in special families the resulting distance admits explicit formulae. This same Gaussian specialization later underlies the Fr\'{e}chet Inception Distance (FID) used in generative-model evaluation \cite{heusel2017gans}; see \cref{sec:fid}.

\section{Probability Laws as Curves}
\label{sec:laws-as-curves}

The most direct way to compare the curve-theoretic and law-theoretic branches is to ask whether a probability law can itself be regarded as a curve. This viewpoint is already explicit in L\'evy's 1950 note, where a law is represented by the curve \(\Gamma\) of its total probability function, that is, its cumulative distribution function, and the first direct definitions compare such curves \cite{frechet1950recherches}. In one dimension, another natural candidate is the quantile function. If $\mu\in\mathcal P(\mathbb R)$ has cumulative distribution function $F_\mu$, define
\begin{equation}
Q_\mu(t):=\inf\{x\in\mathbb R:F_\mu(x)\ge t\},\qquad 0<t<1,
\end{equation}
which is the usual left-continuous quantile. One may then ask whether distances between laws can be recovered from distances between the corresponding quantile curves.

To answer this cleanly, one must keep three distinct constructions separate. The first is the classical Fr\'{e}chet distance on curves, which treats parametrization as auxiliary and optimizes over admissible increasing reparameterizations. The second is Wasserstein distance on laws, which compares quantiles at common probability levels and therefore treats the parameter \(t\) as part of the data. The third is the Fr\'{e}chet distance between planar graphs built from those quantiles. These are related constructions, but they are not generally the same.

The basic obstruction already appears if one treats quantile functions merely as curves in \(\mathbb R\). Let \(\mu=\mathrm{Unif}[0,1]\) and let \(\nu\) be the law with quantile \(Q_\nu(t)=\sqrt t\). Then \(Q_\mu(t)=t\) and \(Q_\nu(t)=\sqrt t\) trace the same ordered image after the increasing reparameterization \(t\mapsto t^2\). As curves in Fr\'{e}chet's 1906 sense, they therefore coincide, and their classical curve Fr\'{e}chet distance is \(0\). The underlying laws, however, are different, and
\begin{equation}
W_p(\mu,\nu)^p=\int_0^1 |t-\sqrt t|^p\,\mathrm{d}t>0.
\end{equation}
Thus the raw curve Fr\'{e}chet distance on quantile functions is too coarse to define a distance on probability laws: on the curve side the parameter may be optimized away, whereas on the law side it is the probability level itself.

\begin{center}
\begin{summarybox}
\textbf{Summary.}
Treating \(Q_\mu\) and \(Q_\nu\) as ordinary curves in \(\mathbb R\) is too coarse for probability laws. Classical Fr\'{e}chet distance may identify them after an increasing reparameterization even when \(\mu\neq\nu\), because on the curve side the parameter is auxiliary, while on the law side it is the probability level.
\end{summarybox}
\end{center}

Once this is made explicit, the exact one-dimensional bridge is straightforward. If the probability level is kept fixed, then the Wasserstein distance from \cref{sec:laws} is precisely the \(L^p\) distance between quantile functions:
\begin{equation}\label{eq:Wp-quantile-curves}
W_p(\mu,\nu)^p
=
\int_0^1 |Q_\mu(t)-Q_\nu(t)|^p\,\mathrm{d}t,
\qquad 1\le p<\infty.
\end{equation}
This is the same one-dimensional quantile identity already stated in \cref{sec:laws}; here the point is that it may be read as a statement about curves, provided one does not allow reparameterization. In this sense, one-dimensional optimal transport compares laws by comparing curves pointwise at common probability levels rather than modulo reparameterization. The quantity in \cref{eq:Wp-quantile-curves} is therefore not a Fr\'{e}chet distance in the 1906 sense, but an \(L^p\) discrepancy on a fixed parameter space.

\begin{figure}[t]
\centering
\includegraphics[width=0.8\linewidth]{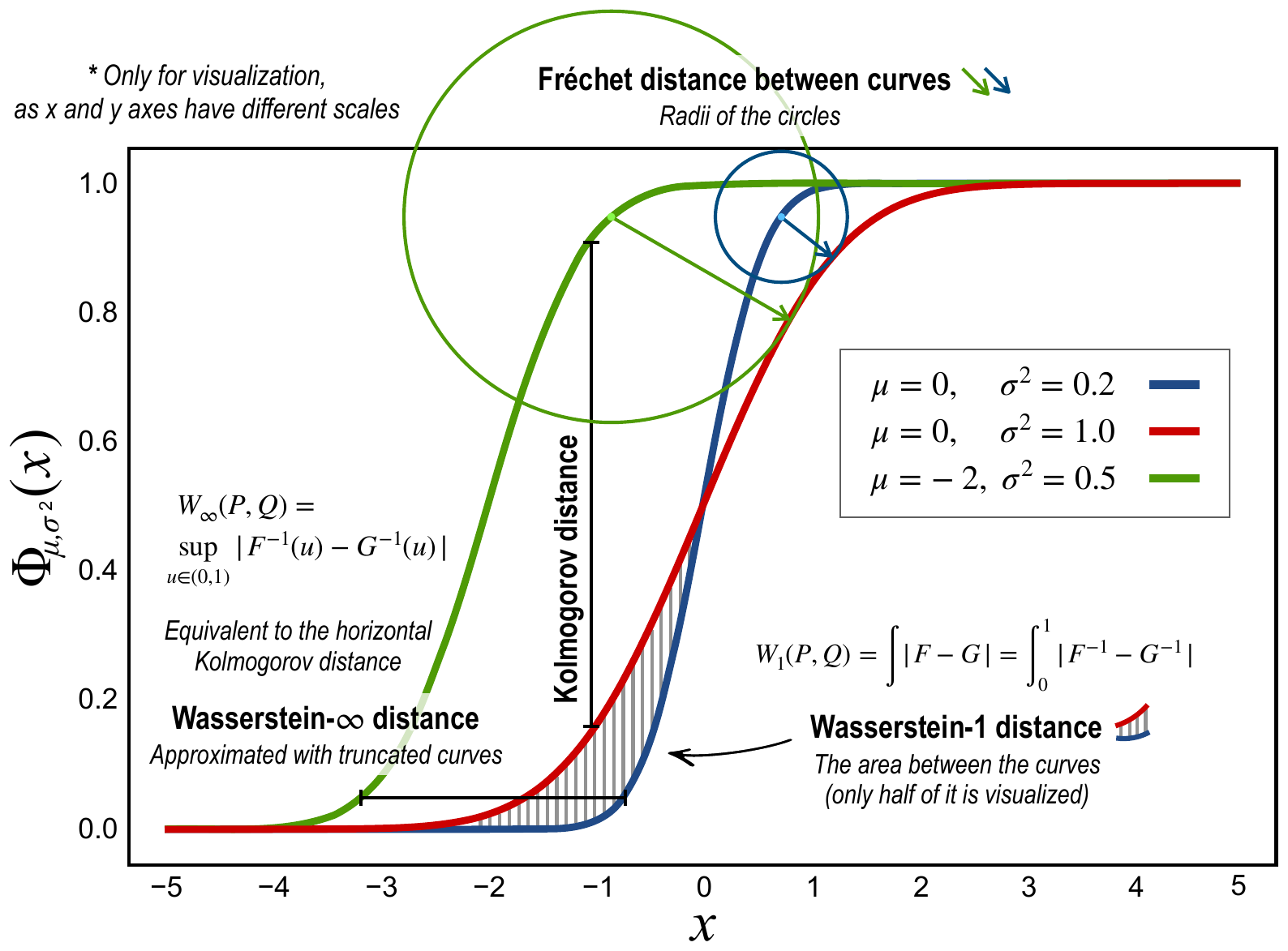}
\caption{Geometric illustration of the CDFs of three one-dimensional Gaussian distributions, together with the Fr\'{e}chet and Kolmogorov distances between the curves, and the Wasserstein-\(1\) and Wasserstein-\(\infty\) distances between the corresponding probability measures. The Kolmogorov distance is shown as the maximum vertical gap between two CDFs, while the Wasserstein-\(1\) distance is illustrated by the area between them, of which only half is visualized. The Wasserstein-\(\infty\) distance is depicted as the horizontal Kolmogorov distance, equivalently the supremum distance between quantile functions, and is approximated here using truncated curves. The Fr\'{e}chet distance between curves is indicated by circles whose radii equal the distance. The figure is schematic only, since the \(x\)- and \(y\)-axes are plotted on different scales.}
\label{fig:3}
\end{figure}

In the limiting case, whenever a bounded-displacement coupling exists, one likewise has
\begin{equation}\label{eq:Winfty-quantile-curves}
W_\infty(\mu,\nu)=\|Q_\mu-Q_\nu\|_{L^\infty(0,1)}.
\end{equation}
On the line, this reflects the fact that the monotone coupling remains optimal for the \(L^\infty\) problem as well \cite{santambrogiooptimal}; in higher dimension, by contrast, the \(W_\infty\) problem has its own specific structure \cite{champion2008wasserstein}. Thus \(W_\infty\) is the closest law-theoretic analogue of a worst-case curve comparison, but it is still not the classical Fr\'{e}chet distance, because the optimization is over couplings of laws rather than over reparameterizations of one and the same curves.

The case \(p=1\) is especially suggestive geometrically. In one dimension,
\begin{equation}\label{eq:W1-cdf-quantile}
W_1(\mu,\nu)
=
\int_{\mathbb R}|F_\mu(x)-F_\nu(x)|\,\mathrm{d}x
=
\int_0^1 |Q_\mu(t)-Q_\nu(t)|\,\mathrm{d}t.
\end{equation}
The first integral is the area between the two cumulative distribution functions when plotted against \(x\), while the second is the area between the two quantile curves when plotted against probability level \(t\). These are not merely analogous pictures: in one dimension they are equal expressions for the same transport distance \cite{peyre2019computational,santambrogiooptimal}. A primal-only proof of \cref{eq:W1-cdf-quantile}, based on a coupling lower bound and the monotone quantile coupling, is recorded in \cref{app:W1-quantile-primal}.

This is also where the one-dimensional picture stops. In higher dimensions there is no comparably direct identity expressing \(W_1\) as an area between multivariate cumulative distribution surfaces, and there is no single scalar quantile function restoring the same monotone matching. Beyond one dimension, Wasserstein distance is governed instead by optimal transport couplings, or equivalently by the Kantorovich--Rubinstein dual formulation.

If one nevertheless wishes to preserve a genuine curve distance, the natural modification is to embed the quantile function as a graph in the plane. Thus, for a law \(\mu\) on \(\mathbb R\), define
\begin{equation}
\Gamma_\mu(t):=(t,Q_\mu(t)),\qquad 0<t<1.
\end{equation}
The first coordinate records the probability level, so the parametrization is now built into the geometry. This brings the two theories much closer: \(\Gamma_\mu\) is an ordered planar curve, while its second coordinate still carries the quantile information relevant for transport.

Under the identity matching one obtains
\begin{equation}
d_F(\Gamma_\mu,\Gamma_\nu)
\le
\sup_{0<t<1}\|(t,Q_\mu(t))-(t,Q_\nu(t))\|,
\end{equation}
where the norm is whichever norm is chosen on the ambient plane. This is only an upper bound coming from the specific matching \(t\mapsto t\); it is not an exact formula for the graph Fr\'{e}chet distance in general. If the ambient metric is the product sup metric
\begin{equation}
d_\infty\bigl((s,x),(t,y)\bigr):=\max\{|s-t|,|x-y|\},
\end{equation}
then the right-hand side reduces to \(\|Q_\mu-Q_\nu\|_{L^\infty(0,1)}=W_\infty(\mu,\nu)\). In this sense, graph embedding places Fr\'{e}chet distance on curves and \(W_\infty\) on laws very close to one another, but only after the probability parameter has been incorporated into the geometric object itself.

Even then, one must keep the distinction clear. The classical Fr\'{e}chet distance between \(\Gamma_\mu\) and \(\Gamma_\nu\) still permits reparameterization of the planar graphs, whereas the Wasserstein comparison keeps the probability level fixed. For that reason, graph Fr\'{e}chet distance is best viewed as a geometric surrogate or companion quantity rather than as an exact restatement of Wasserstein distance. Exact equality occurs only in the degenerate case where the parameter is not optimized away at all: if the graphs are compared pointwise at common \(t\), then the vertical discrepancy is exactly \(W_\infty(\mu,\nu)\), and the corresponding \(L^p\) discrepancy is exactly \(W_p(\mu,\nu)\).

The three viewpoints may therefore be summarized as follows. Raw quantile functions treated as curves in \(\mathbb R\) are too coarse, because increasing reparameterizations can identify distinct laws. Quantile functions compared at fixed probability levels recover Wasserstein distance exactly. Planar graph embeddings produce honest curve objects to which Fr\'{e}chet distance may be applied again, but now only as a nearby geometric construction rather than as an exact transport identity.

\cref{fig:3} illustrates these distinctions on a concrete one-dimensional Gaussian family. As CDF graphs, the three laws may be compared by curve-based quantities such as Fr\'{e}chet distance or by the Kolmogorov sup norm. As probability laws, the same family may be compared by \(W_1\) and \(W_\infty\), which record transport discrepancy rather than only geometric separation between graphs. These quantities can agree qualitatively on which pair is farther apart, but they measure different structures: Fr\'{e}chet and Kolmogorov compare curves, whereas Wasserstein compares laws.

For laws with unbounded support, this distinction becomes even sharper. For Gaussian laws, discussed already in \cref{sec:gaussian-laws}, the quantile functions
\begin{equation}
Q_\mu(t)=m_1+\sigma_1\Phi^{-1}(t),\qquad
Q_\nu(t)=m_2+\sigma_2\Phi^{-1}(t)
\end{equation}
are unbounded as \(t\to0\) or \(t\to1\), and their difference is unbounded unless the two laws have sufficiently closely matched tails, as in the equal-variance pure mean-shift case. The full quantile graphs are therefore often unbounded, and naive curve distances on those graphs are not the right objects. A more reasonable comparison is obtained by truncating to a compact probability interval \([\varepsilon,1-\varepsilon]\), with \(0<\varepsilon<\tfrac12\), and considering
\begin{equation}
\Gamma_\mu^\varepsilon(t):=(t,Q_\mu(t)),\qquad
\Gamma_\nu^\varepsilon(t):=(t,Q_\nu(t)),
\qquad t\in[\varepsilon,1-\varepsilon].
\end{equation}
On this truncated interval,
\begin{equation}
\|Q_\mu-Q_\nu\|_{L^p([\varepsilon,1-\varepsilon])}
\end{equation}
is the truncated quantile discrepancy, while
\begin{equation}
\sup_{t\in[\varepsilon,1-\varepsilon]}|Q_\mu(t)-Q_\nu(t)|
\end{equation}
is its worst-case version. These are truncated comparison quantities on \([\varepsilon,1-\varepsilon]\), not the full Wasserstein distances themselves; they become useful geometric surrogates only after one has separately decided to discard or control the tails.

\begin{center}
\begin{summarybox}
\textbf{Summary.}
For laws on \(\mathbb R\), the exact identity is
\begin{equation}
W_p(\mu,\nu)^p=\int_0^1 |Q_\mu(t)-Q_\nu(t)|^p\,\mathrm{d}t.
\end{equation}

Thus Wasserstein distance may be read as a distance between quantile curves only when the common probability parameter is kept fixed. Embedding quantiles as planar graphs \(\Gamma_\mu(t)=(t,Q_\mu(t))\) brings the law and curve viewpoints closer, but the resulting graph Fr\'{e}chet distance is, in general, a related geometric surrogate rather than an exact restatement of Wasserstein distance.
\end{summarybox}
\end{center}

The overall conclusion is that the curve and law viewpoints meet most naturally in one dimension, through quantiles. At the most faithful level, Wasserstein distance already gives an exact comparison of quantile curves, provided they are compared at the same probability level. At a more geometric level, embedding quantiles as planar graphs produces objects to which the classical Fr\'{e}chet distance can be applied, but then one has moved from an exact law-distance identity to a related curve-distance construction. The two branches therefore coincide exactly when the parameter \(t\) is retained as part of the data, and they diverge as soon as one allows it to be optimized away in the manner of the classical Fr\'{e}chet distance on curves.

\section{Fr\'{e}chet Inception Distance}
\label{sec:fid}

The most visible contemporary use of Fr\'{e}chet's name in machine learning is the Fr\'{e}chet Inception Distance (FID), introduced in \cite{heusel2017gans}. In its standard form, one begins with a feature extractor $\varphi$ defined on images and maps both real and generated samples into a Euclidean feature space $\mathbb R^d$. If the resulting feature laws are denoted by $\mu$ and $\nu$, then FID replaces these laws by Gaussian approximations with matching means and covariances,
\begin{equation}
\mu_G=\mathcal N(m_\mu,\Sigma_\mu),\qquad
\nu_G=\mathcal N(m_\nu,\Sigma_\nu),
\end{equation}
and evaluates the closed-form Gaussian $W_2$ distance
\begin{equation}
\operatorname{FID}(\mu,\nu)
=
\|m_\mu-m_\nu\|_2^2
+
\operatorname{Tr}\!\left(
\Sigma_\mu+\Sigma_\nu
-2(\Sigma_\mu^{1/2}\Sigma_\nu\Sigma_\mu^{1/2})^{1/2}
\right).
\end{equation}

Thus FID belongs to the probability-law branch rather than to the classical curve-theoretic branch. What it inherits from the preceding sections is not the geometric Fr\'{e}chet distance on curves, but the coupling-based and Wasserstein-like viewpoint in which one measures discrepancy between probability laws through a metric on a space of distributions; see in particular \cref{sec:laws}.

This observation matters historically as well as technically. The word ``Fr\'{e}chet'' in FID does not arise from polygonal curves, reparameterizations, or free-space diagrams. It arrives instead through the Gaussian specialization of the distance between laws, as later recognized in the modern optimal transport framework. At the same time, this use case imposes additional approximations that are absent from the mathematical definition of $W_2$: the feature map is fixed externally, the true feature laws are replaced by empirical laws, and these empirical laws are then summarized by Gaussian moments alone. FID is therefore best viewed as a computationally convenient Gaussian proxy for a Wasserstein-type distance in feature space, rather than as an exact law-level distance between the underlying image distributions. This explains both its usefulness and its limitations: the Gaussian formula is explicit and inexpensive, but it retains only second-order information in a chosen representation space, and its numerical value depends strongly on preprocessing, feature extraction, and finite-sample estimation \cite{binkowski2018demystifying,parmar2022aliased}. From the historical perspective of the present note, however, FID still marks a genuine continuation of Fr\'{e}chet's 1957 line, especially as developed in \cref{sec:gaussian-laws}: a distance between laws, expressed through couplings in principle and by explicit formulae in special families.

The later critique literature shows that these limitations are not cosmetic, but structural. First, FID replaces the full feature laws by Gaussians, so it can only see discrepancies that survive at the level of the first two moments. This is mathematically convenient, but it means that distinct non-Gaussian feature laws may receive the same score, and it also helps explain why kernel alternatives such as KID were proposed as nonparametric substitutes with unbiased finite-sample estimators \cite{binkowski2018demystifying}. Second, the empirical FID computed from finitely many generated samples is itself biased: the expected finite-sample score can differ from the population value, the bias can depend on the model being evaluated, and raw FID comparisons at a fixed sample count can therefore reverse model rankings \cite{chong2020effectively}. Third, FID is unexpectedly sensitive to low-level implementation details that are ostensibly outside the intended semantic comparison, including image resizing kernels, antialiasing conventions, and JPEG compression; in extreme cases, these choices can even improve the reported FID after additional image degradation \cite{parmar2022aliased}. Fourth, FID compresses fidelity and diversity into a single scalar, which is one reason later work introduced precision--recall-style diagnostics and density--coverage variants in order to distinguish mode dropping from sample quality more explicitly \cite{kynkaanniemi2019improved,naeem2020reliable,sajjadi2018assessing}. Finally, recent work argues that the very feature space on which FID depends has become less adequate for modern text-to-image models: Inception features can be poorly aligned with the richer semantic content of current models, the Gaussianity assumption is often violated, and FID can disagree with human judgments, fail to track gradual improvements, and behave inconsistently across sample sizes \cite{jayasumana2024rethinking}. Similar practical cautions also arise in domain-specific settings such as retinal image synthesis, where the interpretability of FID depends strongly on the medical relevance of the chosen feature space and on whether clinically meaningful variation is reflected by the score at all \cite{wu2026pragmatic}. Taken together, these critiques do not make FID useless; rather, they clarify its proper status. FID is best regarded as a historically successful, computationally cheap Gaussian proxy whose conclusions are most trustworthy when the feature extractor, preprocessing pipeline, sample regime, and comparison task are all carefully controlled.

\section{Related Distances and Divergences}
\label{sec:other-distances}

The distances discussed so far are not the only ways to compare curves or probability laws. It is therefore useful to place Fr\'{e}chet and Wasserstein distances among a few nearby notions. The point of the present section is not to survey this wider landscape exhaustively, but to indicate which mathematical structures are preserved and which are lost when one moves to neighboring comparisons.
We summarize the related distances and divergences in \cref{tab:related-distances}. We refer to \cite[Fig. 1]{gibbs2002choosing} and \cite[Fig. 8.2]{peyre2019computational} for further diagrams of relationship w.r.t. their bounds between them.

\subsection{Dynamic Time Warping}

Dynamic Time Warping (DTW) is another natural near-neighbor on the curve side \cite{sakoe1978dynamic,hu2023spatio}. For two sampled sequences $(x_1,\dots,x_m)$ and $(y_1,\dots,y_n)$ in a metric space, DTW minimizes an accumulated matching cost over monotone alignments of the index sets. One standard form is
\begin{equation}
\operatorname{DTW}(x,y)
:=
\min_{W\in\mathcal W_{m,n}}
\sum_{(i,j)\in W} d(x_i,y_j),
\end{equation}
where \(\mathcal W_{m,n}\) denotes the set of admissible warping paths from \((1,1)\) to \((m,n)\), usually built from the monotone steps \((1,0)\), \((0,1)\), and \((1,1)\) in the index grid. In contrast with Fr\'{e}chet distance, the objective is additive rather than worst-case: one sums local discrepancies along an admissible warping path instead of minimizing the maximal discrepancy along a continuous matching. DTW is therefore often more forgiving of local oscillations, but it is also less geometric in the sense emphasized by Fr\'{e}chet's \emph{\'ecart}. It is best viewed as a discrete alignment cost for time series, whereas Fr\'{e}chet distance remains closer to a metric on curves as ordered geometric objects.

\subsection{Kullback-Leibler divergence}

For probability laws $\mu$ and $\nu$ on the same measurable space, with $\mu$ absolutely continuous w.r.t. $\nu$, the Kullback--Leibler divergence is
\begin{equation}
\mathrm{KL}(\mu\|\nu)
:=
\int \log\!\left(\frac{\mathrm{d}\mu}{\mathrm{d}\nu}\right)\,\mathrm{d}\mu.
\end{equation}

Introduced by Kullback and Leibler in 1951 \cite{kullback1951information}, this is a \(\phi\)-divergence rather than an optimal transport distance in the terminology of Peyr\'e and Cuturi \cite{peyre2019computational}. Unlike Wasserstein distance, KL divergence is not symmetric and does not satisfy the triangle inequality. It is therefore a divergence rather than a metric. Its strengths are different: it is local in density space, closely tied to likelihood theory, and central in statistics and variational inference. Its weakness is equally clear: it is undefined when $\mu$ is not absolutely continuous w.r.t. $\nu$, and it does not encode a geometric notion of transport or displacement. In that sense, KL divergence occupies almost the opposite end of the spectrum from Wasserstein distance.

One often symmetrizes KL by replacing \(\mathrm{KL}(\mu\|\nu)\) with
\[
\mathrm{KL}(\mu\|\nu)+\mathrm{KL}(\nu\|\mu),
\]
or with its average. This removes the directional asymmetry, but it still does not produce a metric: the triangle inequality generally fails, and the quantity still reflects density mismatch rather than transport geometry.

\subsection{Jensen--Shannon divergence}

A more stable symmetrization of KL is the Jensen--Shannon divergence. Writing
\[
\eta:=\tfrac12(\mu+\nu),
\]
one defines
\begin{equation}
\mathrm{JS}(\mu,\nu)
:=
\frac12\,\mathrm{KL}(\mu\|\eta)
+\frac12\,\mathrm{KL}(\nu\|\eta).
\end{equation}

Introduced by Lin in 1991 \cite{lin1991divergence}, the Jensen--Shannon divergence is always finite as soon as the two laws are dominated by a common reference measure. It is symmetric, bounded, and much less singular when the supports are poorly matched \cite{peyre2019computational}. For this reason it is often preferred when one wants a density-based discrepancy but wishes to avoid the one-sided blow-up of KL. Its square root is in fact a metric, but the comparison remains density-based rather than geometric: as with KL, one does not ask how far mass must move in the ambient space.

\begin{table}[t]
  \captionsetup{skip=5pt}
\centering
\fontsize{8.5}{10}\selectfont
\caption{A summary of the distance and divergence notions discussed in this section.}
\renewcommand{\arraystretch}{1.15}
\begin{tabular}{@{}>{\raggedright\arraybackslash}p{0.13\linewidth}
                >{\raggedright\arraybackslash}p{0.17\linewidth}
                >{\raggedright\arraybackslash}p{0.35\linewidth}
                >{\raggedright\arraybackslash}p{0.27\linewidth}@{}}
\toprule
\textbf{Notion} & \textbf{Domain} & \textbf{How it compares} & \textbf{What it keeps} \\
\midrule
Hausdorff & sets, curve images & worst-case proximity of sets & geometry only \\
DTW & sampled sequences & additive monotone alignment cost & order, only weak geometry \\
Fr\'{e}chet & ordered curves & worst-case monotone alignment & order and geometry \\
Skorokhod & c\`adl\`ag paths & value mismatch plus time distortion & order, with time penalty \\
Wasserstein & probability laws & optimal transport cost & geometry of the ambient space \\
Prokhorov & probability laws & enlargement of sets; weak convergence & weak geometry only \\
KL / JS / TV & probability laws & density or event discrepancy & no transport geometry \\
Hellinger / BC & probability laws & overlap of square-root densities & no transport geometry \\
MMD / Energy & probability laws & test-function or semimetric discrepancy & kernel / semimetric structure \\
Sinkhorn & probability laws & entropically regularized transport & geometry, but smoothed \\
\bottomrule
\end{tabular}
\label{tab:related-distances}
\end{table}

\subsection{Hausdorff distance}

For compact sets $A,B$ in a metric space $(\mathcal X,d)$, the Hausdorff distance is
\begin{equation}
d_H(A,B)
:=
\max\!\left\{
\sup_{a\in A}\inf_{b\in B} d(a,b),
\sup_{b\in B}\inf_{a\in A} d(a,b)
\right\}.
\end{equation}
This set-based metric goes back to Hausdorff's 1914 treatment of metric geometry \cite{hausdorff1914grundzuge}; for a modern discussion in relation to transport, see also \cite{peyre2019computational}. 
When applied to curves via their images, Hausdorff distance forgets parametrization and even order along the curve. This is precisely the feature Fr\'{e}chet sought to avoid in 1906. Two traversals of the same geometric image may have Hausdorff distance zero while representing different ordered curve-arcs, whereas the curve \emph{\'ecart} in Fr\'{e}chet's sense distinguishes them unless they coincide after an admissible increasing reparameterization. Hausdorff distance is therefore a natural point-set comparison, but not a substitute for Fr\'{e}chet distance when temporal or ordered progression matters.

\subsection{Total variation distance}

For probability laws $\mu$ and $\nu$ on the same measurable space, the total variation distance is
\begin{equation}
d_{\mathrm{TV}}(\mu,\nu)
:=
\sup_{A} |\mu(A)-\nu(A)|.
\end{equation}

If $\mu$ and $\nu$ admit densities $f$ and $g$ w.r.t. a common reference measure $\rho$, this becomes
\begin{equation}
d_{\mathrm{TV}}(\mu,\nu)=\frac12\int |f-g|\,\mathrm{d}\rho,
\end{equation}

This is a classical measure-theoretic distance rather than a notion tied to one single founding paper, so it is more natural to cite standard modern references for the present formulation \cite{billingsley1999convergence,peyre2019computational}. Total variation measures the largest discrepancy over measurable events. In that sense it is stronger and more local than Wasserstein distance: it is sensitive to small mismatches in mass allocation, but it ignores the geometry of how far mass would need to move. Two laws can be close in Wasserstein distance because their mass is only slightly displaced, while still being far in total variation if they are supported on disjoint sets. This makes total variation important in probability and statistics, but conceptually quite different from the transport-based comparisons emphasized in \cref{sec:laws,sec:laws-as-curves}.

\subsection{Hellinger distance and Bhattacharyya distance}

If \(\mu\) and \(\nu\) admit densities \(f\) and \(g\) with respect to a common reference measure \(\rho\), the Hellinger distance is
\begin{equation}
H(\mu,\nu)
:=
\left(\frac12\int \bigl(\sqrt{f}-\sqrt{g}\bigr)^2\,\mathrm{d}\rho\right)^{1/2}.
\end{equation}
Equivalently,
\begin{equation}
H(\mu,\nu)^2
=
1-\int \sqrt{fg}\,\mathrm{d}\rho.
\end{equation}

The overlap term
\begin{equation}
\operatorname{BC}(\mu,\nu):=\int \sqrt{fg}\,\mathrm{d}\rho
\end{equation}
is the Bhattacharyya coefficient, and one often defines the Bhattacharyya distance by
\begin{equation}
D_B(\mu,\nu):=-\log \operatorname{BC}(\mu,\nu).
\end{equation}

The Hellinger integral goes back to Hellinger's 1909 work \cite{hellinger1909neue}, while the Bhattacharyya coefficient and distance go back to Bhattacharyya's 1943 paper \cite{bhattacharyya1943measure}. These quantities measure closeness through overlap of square-root densities. They are symmetric and well behaved for many parametric families; for a modern summary in the present context see \cite{peyre2019computational}. Conceptually, however, they still belong to the density-overlap side of the landscape rather than to the transport side: they record how much the two laws overlap pointwise, not how expensive it would be to rearrange one into the other through the ambient geometry.

\subsection{Kolmogorov distance}

For laws $\mu,\nu$ on $\mathbb R$ with CDFs $F_\mu,F_\nu$, the Kolmogorov distance\,\footnote{\,The Soviet mathematician Andrey Kolmogorov corresponded with Fr'{e}chet and met him to discuss various mathematical topics \cite{barbut2013paul}. They are, for example, jointly associated with the Fréchet--Kolmogorov theorem in functional analysis.} is
\begin{equation}
d_K(\mu,\nu):=\sup_{x\in\mathbb R}|F_\mu(x)-F_\nu(x)|.
\end{equation}

This sup-norm comparison is closely tied to Kolmogorov's 1933 work on the empirical determination of a distribution law \cite{Stephens1992,kolmogorov1933sulla}; for the standard modern formulation, see also \cite{Stephens1992,billingsley1999convergence}. In geometric terms, it measures the maximum vertical deviation between the two CDF graphs. Compared with Wasserstein distance, it ignores how far mass must move horizontally; compared with Fr\'{e}chet distance on curves, it is tied to one distinguished representation of the law, namely the CDF. It is therefore highly useful, but it captures a different kind of discrepancy from either transport or curve-reparameterization metrics.

\subsection{Prokhorov distance}

On a metric space $(\mathcal X,d)$, the Prokhorov distance metrizes weak convergence of probability laws and is defined by
\begin{equation}
d_P(\mu,\nu)
:=
\inf\bigl\{\varepsilon>0:\mu(A)\le \nu(A^\varepsilon)+\varepsilon
\text{ and }\nu(A)\le \mu(A^\varepsilon)+\varepsilon \text{ for all Borel }A\bigr\},
\end{equation}
where
\begin{equation}
A^\varepsilon:=\{x\in\mathcal X:\inf_{a\in A} d(x,a)<\varepsilon\}
\end{equation}
is the $\varepsilon$-enlargement of $A$. Introduced by Prokhorov in 1956 \cite{prokhorov1956convergence}, this distance occupies an intermediate conceptual position; for the present formulation as a standard metric on laws, see also \cite{billingsley1999convergence}. Like Wasserstein distance, it is defined on probability laws over a metric space and remembers the underlying geometry through the enlargement $A^\varepsilon$. Unlike Wasserstein distance, it does not measure transport cost directly and does not require moment assumptions. For historical purposes, it is one of the standard metrics used in modern probability theory to discuss convergence in law, whereas Wasserstein adds quantitative displacement information beyond weak convergence alone.

\subsection{Maximum Mean Discrepancy}

Maximum Mean Discrepancy (MMD) compares probability laws through a reproducing-kernel Hilbert space. Given a positive definite kernel $k$ with RKHS $\mathcal H_k$, where \(\mathcal H_k\) is a Hilbert space of functions in which point evaluation is represented by the kernel, one defines
\begin{equation}
\operatorname{MMD}_k(\mu,\nu)
:=
\sup_{\|f\|_{\mathcal H_k}\le 1}
\left(\int f\,\mathrm{d}\mu-\int f\,\mathrm{d}\nu\right).
\end{equation}

Equivalently, it is the RKHS norm of the difference between kernel mean embeddings of $\mu$ and $\nu$ \cite{gretton2006kernel}. Peyr\'e and Cuturi place MMD among integral probability metrics and contrast its sample behavior with Wasserstein distance \cite{peyre2019computational}. MMD is especially attractive in machine learning because it is easy to estimate from samples and depends on a chosen kernel rather than on an ambient transport geometry. It therefore plays a role complementary to Wasserstein distance: one compares expectations of test functions in a function class, while the other compares laws by the cost of moving mass. Kernel Inception Distance (KID) is built on this idea and was proposed partly to avoid some of the finite-sample difficulties of FID \cite{binkowski2018demystifying}.

\subsection{Energy distance}

Another nearby law-distance, emphasized in \cite{peyre2019computational}, is the energy distance. In the modern statistical form relevant here, due to Sz\'ekely and Rizzo \cite{szekely2013energy}, it may be written in Euclidean space as
\begin{equation}
D_E(\mu,\nu)^2
:=
2\,\mathbb E\|X-Y\|
-\mathbb E\|X-X'\|
-\mathbb E\|Y-Y'\|,
\end{equation}

where \(X,X'\sim\mu\) are independent and \(Y,Y'\sim\nu\) are independent. Unlike Wasserstein distance, this quantity is not defined through a transport plan; unlike MMD, it does not begin from a positive definite kernel. Instead, it is built from a negative-type semimetric and can be viewed as a special case of an MMD-type norm. Conceptually it sits between the two: it still compares laws through expectations of pairwise distances, but it does so without solving an optimal matching problem. This makes it attractive in statistics, especially because its empirical behavior is often simpler than that of Wasserstein distance in high dimension.

\subsection{Sinkhorn divergence}

Peyr\'e and Cuturi also emphasize the importance of entropically regularized transport and the associated Sinkhorn divergence \cite{peyre2019computational}. In the form relevant here, introduced by Genevay, Peyr\'e, and Cuturi \cite{genevay2018learning}, one starts from a transport cost \(W_{p,\varepsilon}\) regularized by relative entropy and defines
\begin{equation}
S_\varepsilon(\mu,\nu)
:=
W_{p,\varepsilon}(\mu,\nu)
-\frac12 W_{p,\varepsilon}(\mu,\mu)
-\frac12 W_{p,\varepsilon}(\nu,\nu).
\end{equation}

This correction removes the entropic self-bias and yields a quantity that interpolates between optimal transport and kernel-type discrepancies. In practical terms, Sinkhorn divergences preserve much of the geometry of Wasserstein distance while being far cheaper to compute numerically. For a survey such as the present one, they are worth mentioning because they show how the transport viewpoint can be softened computationally without giving it up entirely.

\subsection{Skorokhod distance}

The Skorokhod distance is a particularly interesting near-neighbor of Fr\'{e}chet distance because it also balances value discrepancy against controlled distortion of the time parameter. It is usually defined for c\`adl\`ag paths, meaning functions that are right-continuous and have left limits at every time. Such path spaces arise naturally in probability theory when one allows jump processes rather than only continuous trajectories. The distance goes back to Skorokhod's 1956 work \cite{skorokhod1956limit}; for the standard modern formulation used on path spaces, see also \cite{billingsley1999convergence}. For c\`adl\`ag paths $x,y:[0,1]\to\mathbb R^d$, one standard form is
\begin{equation}
d_S(x,y)
:=
\inf_{\lambda}
\max\!\left\{
\sup_{t\in[0,1]}\|\lambda(t)-t\|,
\sup_{t\in[0,1]}\|x(t)-y(\lambda(t))\|
\right\},
\end{equation}
where the infimum runs over increasing homeomorphisms $\lambda$ of $[0,1]$. This is not the same as classical curve Fr\'{e}chet distance, but the family resemblance is unmistakable: both optimize over monotone changes of parameter and measure a worst-case mismatch. The difference is that Skorokhod distance penalizes the time distortion itself, whereas Fr\'{e}chet distance quotients it out. This relation has been made algorithmically explicit in later work connecting Fr\'{e}chet-type reasoning with Skorokhod distance \cite{majumdar2015computing}.

\section{Naming Convention}
\label{sec:naming}

Although the naming of a distance is scientifically secondary to its definition and properties, it is still informative. The names \emph{Fr\'{e}chet distance}, \emph{Wasserstein distance}, and \emph{Fr\'{e}chet Inception Distance} do not refer to one single construction. They mark successive stages in which ideas about distance moved between ordered curves, probability laws, optimal transport, Gaussian formulae, and modern generative-model evaluation.

For curves, the name \emph{Fr\'{e}chet distance} is historically well justified. As discussed in \cref{sec:curves}, Fr\'{e}chet's 1906 thesis already contains the essential continuous idea: a curve-arc is an ordered object, admissible increasing reparameterizations preserve that order, and the distance is obtained by minimizing the maximal pointwise discrepancy over such reparameterizations \cite{frechet1906quelques}. Later work changed the setting and the computational status of the notion. In particular, Alt and Godau formulated the polygonal version that became central in computational geometry and gave the free-space algorithmic viewpoint \cite{alt1992measuring,alt1995computing}. Their contribution is therefore indispensable for the modern algorithmic subject, but it does not displace Fr\'{e}chet's role in the underlying continuous notion.

For probability laws, the naming is more delicate. L\'evy's 1950 note already gives direct distances between laws by representing a law as a cumulative-distribution curve \(\Gamma\) and then comparing such curves \cite{frechet1950recherches}. Fr\'{e}chet's 1957 paper starts from a different but related direction: it considers global distances between simultaneously realized random variables and then passes to distances between their laws by ranging over admissible bivariate laws with prescribed margins \cite{frechet1957distance}. Fr\'{e}chet explicitly presents this discussion as connected to L\'evy's definitions, but the 1950 note and the 1957 paper should not be collapsed into the same construction. Modern Wasserstein language makes the coupling structure more systematic, but it can also make the historical argument look heavier than it was in Fr\'{e}chet's one-dimensional setting. Thus the phrase \emph{Fr\'{e}chet distance on probability laws} is defensible when it refers to the coupling-based and \(H_1\)-based development of the 1957 paper, especially the explicit one-dimensional formulae in \cref{sec:frechet1957}; it is less clear-cut if it is used as though Fr\'{e}chet alone originated the whole law-distance problem.

The name \emph{Wasserstein distance} is also a convention with a layered history. R\"uschendorf emphasizes that the quantity now denoted \(W_p\) arose from several directions rather than from a single source \cite{rueschendorfwasserstein}. Kantorovich introduced a minimal \(L^1\)-transport metric in connection with the Monge transportation problem, Kantorovich and Rubinstein developed the dual viewpoint, Gini and Dall'Aglio obtained one-dimensional distribution-function formulae, Fr\'{e}chet studied metric properties of distances between laws, Mallows used the \(L^2\) version in statistics, and Vasershtein introduced the metric in the context from which the English spelling ``Wasserstein'' became common. R\"uschendorf therefore remarks: \textit{``Maybe historically the notion Gini--Dall'Aglio--Kantorovich--Vasersthein--Mallows metric would be correct for this class of metrics''} \cite{rueschendorfwasserstein}. The modern name is convenient and standard, but historically it abbreviates a broader line of work.

Finally, \emph{Fr\'{e}chet Inception Distance} inherits its name from the probability-law branch, not from the curve branch. The formula used in FID is the closed-form quadratic distance between Gaussian laws in an Inception feature space; equivalently, in modern optimal transport language, it is the Gaussian specialization of \(W_2\) \cite{dowson1982frechet,heusel2017gans,peyre2019computational}. Calling it \emph{Fr\'{e}chet} is therefore understandable through the line from Fr\'{e}chet's 1957 paper to Gaussian law distances. Calling it \emph{Wasserstein} would also be mathematically natural, because the same expression is now standardly read as a Wasserstein-2 formula. The established name is thus a matter of lineage and convention rather than a unique mathematical necessity.

\section{Final Remarks}
\label{sec:final}

The historical picture that emerges is broader than the modern terminology sometimes suggests. Fr\'{e}chet's 1906 thesis supplies a metric theory of curves inside a general theory of abstract sets equipped with an \emph{\'ecart}; his 1957 paper supplies a coupling-based distance between probability laws; later computational geometry turns the curve branch into an algorithmic subject; and modern optimal transport provides a unified language in which the law branch becomes part of a wider theory of distances between measures. The contemporary name ``Fr\'{e}chet distance'' now travels across all of these settings, but not always with the same meaning.

What binds the branches together is a shared structural instinct. In each case one begins with objects that are richer than points---curves, functions, random variables, laws---and asks how to compare them in a way that preserves the aspects of structure one cares about. For curves, this is the ordered progression along the path. For laws, it is the possibility of realizing two marginals jointly and measuring the discrepancy there. For FID, it is the comparison of feature-space laws through a Gaussian $W_2$ formula. These are not identical constructions, but they belong to a recognizable family of ideas.

In contemporary applications, this distinction is worth keeping clear. FID is historically closer to Fr\'{e}chet's 1957 distance between laws than to the geometric Fr\'{e}chet distance between curves. Yet the curve-theoretic branch remains valuable, not only in its own right, but also because it clarifies what is special about parametrization, monotonicity, and worst-case alignment. The hope of this note is that reading these branches together makes the name ``Fr\'{e}chet distance'' less confusing, and at the same time more mathematically interesting.

{\small
\bibliographystyle{abbrvnat}
\bibliography{bib.bib}
}

\newpage
\appendix
\crefalias{section}{appendix}
\include{appendix/appendix}

\section{On Some Points of Functional Calculus (translated)}
\label{app:frechet1906}

{
\small

The English translation of
\begin{center}
  \vspace{-.2em}
  \linespread{0.95}\selectfont
\textbf{\textsc{Sur Quelques Points du Calcul Fonctionnel.} }

Memorie e Comunicazioni. Rendiconti del Circolo Matematico di Palermo (1884-1940) 22.1 (1906): 1--72.

By Mr. \textsc{Maurice Fréchet} (Paris)\,\footnote{\,Thesis submitted to the Faculty of Sciences of Paris in order to obtain the degree of Doctor of Science.}

Meeting of April 22, 1906.
\vspace{-1em}
\end{center}

\etocsettocstyle{\subsection*{Contents}}{}

\etocsetstyle{subsection}%
  {\par}%
  {\par\noindent}%
  {\fontsize{9}{9}\selectfont \etocname\hfill\textcolor{royalblue}{\etocthepage}}%
  {\par}

\etocsetstyle{subsubsection}%
  {\par}%
  {\par\hspace{1em}}%
  {\etocname\leavevmode\leaders\hbox to .5em{\hss.\hss}\hfill\kern0pt\textcolor{royalblue}{\etocthepage}}%
  {\par}

\etocsettocdepth{subsubsection}
\localtableofcontents
\clearpage
}
\input{appendix/frechet1906_en_body.tex}

\newpage
\section{Paul Lévy's Note (translated)}
\label{app:levy1950}

{
\small

The English translation of
\begin{center}
  \linespread{0.95}\selectfont
\textbf{\textsc{Distance de deux Variables Aléatoires et Distance entre deux Lois de Probabilité} }

as \textsc{Note B.} in

Fréchet, Maurice. Recherches théoriques modernes sur le calcul des probabilités: livre. Méthode des fonctions arbitraires. Théorie des événements en chaîne dans le cas d'un nombre fini d'états possibles. Avec supplément nouveau et une note de P. Lévy. Gauthier-Villars, 1950.

By \textsc{Paul Lévy}

\end{center}
}
\bigskip

\textit{[From page 333]}
\input{appendix/levy1950_en_body.tex}

\clearpage
\section{On the Distance Between Two Probability Laws (translated)}
\label{app:frechet1957}

{
\small

The English translation of
\begin{center}
\textbf{\textsc{Sur la Distance de deux Lois de Probabilité.} }

Annales de l'Institut de Statistique de l'Université de Paris. Vol. 6. No. 3. 1957.

By \textsc{Maurice Fréchet}
\end{center}

\etocsettocstyle{\subsection*{Contents}}{}

\etocsetstyle{subsection}%
  {\par}%
  {\par\noindent}%
  {\etocname\leavevmode\leaders\hbox to .5em{\hss.\hss}\hfill\kern0pt\textcolor{royalblue}{\etocthepage}}%
  {\par}

\etocsettocdepth{subsection}
\localtableofcontents

}
\bigskip

\input{appendix/frechet1957_en_body.tex}

\clearpage

\begin{figure}[h!]
    \raggedright
    \vspace*{6cm}
    \includegraphics[width=0.5\linewidth]{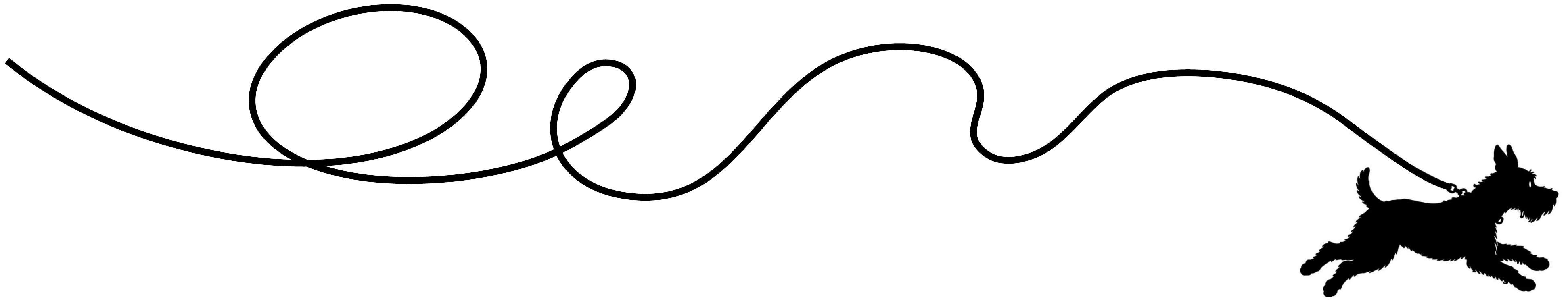}
\end{figure}

{\vspace*{2cm}
\section*{Afterword}

This note comes to an end exactly 120 years after Maurice Fr\'{e}chet published his doctoral thesis. While reading many historical articles for this note, I was repeatedly struck by how naturally many mathematicians moved across languages. Fr\'{e}chet, a native French speaker, knew at least English and German. He was also deeply engaged with Esperanto, having published mathematical papers in Esperanto \cite{frechet1954es} and later served as president of the \textit{Internacia Scienca Asocio Esperantista} (\textit{International Scientific Esperantist Association}) from 1950 to 1953. When he was still a young student, his professor Hadamard wrote to him: \textit{``travaillez l'allemand; une insuffisance en cette langue vous serait une grande gêne plus tard''} (\textit{work on your German; a deficiency in this language would be a great hindrance to you later}) \cite{taylor1982study}. Reading it more than a hundred years later, I definitely can relate.
In Angus E. Taylor's trilogy of articles, \textit{A Study of Maurice Fr\'{e}chet} \cite{taylor1982study,taylor1985study,taylor1987study}, published over several years in the 1980s, the mathematics professor from Berkeley wove French passages from Fr\'{e}chet's letters and German terms into his English prose, often without translation, as though such linguistic openness could simply be taken for granted. These glimpses also suggest how much broader their expertise and intellectual interests were than anything touched on in this note, and how naturally they turned to one another for questions, advice, and encouragement \cite{taylor1981quatre,barbut2013paul}. Reading these works offers not only insight into mathematical history, but also a glimpse of a scholarly world marked by breadth, seriousness, and quiet elegance. And I shouldn't ignore notifications from Duolingo anymore.

\vspace{0.5cm}

\textit{Yuli Wu}

\textit{Sydney, 22 April 2026}

\vspace{1cm}
}

\end{document}

%% file: appendix/appendix.tex

\section{Proofs}
\label{app:proof}

\subsection{Extremal bounds for admissible bivariate laws}
\label{app:frechet-bounds}

Let $F$ and $G$ be one-dimensional cumulative distribution functions, and let $H$ be the cumulative distribution function of a pair $(X,Y)$ whose marginals are $F$ and $G$. Then for every $(x,y)\in\mathbb R^2$ one has
\begin{equation}
\sup\{F(x)+G(y)-1,0\}\le H(x,y)\le \inf\{F(x),G(y)\}.
\end{equation}
Moreover, both endpoint functions
\begin{equation}
H_0(x,y)=\sup\{F(x)+G(y)-1,0\},
\qquad
H_1(x,y)=\inf\{F(x),G(y)\}
\end{equation}
are themselves admissible bivariate cumulative distribution functions with marginals $F$ and $G$.

\begin{proof}
Since $H(x,y)=\mathbb P(X\le x,Y\le y)$, one immediately has
\begin{equation}
H(x,y)\le \mathbb P(X\le x)=F(x),
\qquad
H(x,y)\le \mathbb P(Y\le y)=G(y),
\end{equation}
and therefore
\begin{equation}
H(x,y)\le \inf\{F(x),G(y)\}.
\end{equation}

For the lower bound, observe that
\begin{equation}
\{X\le x,Y\le y\}^{c}=\{X>x\}\cup\{Y>y\}.
\end{equation}
Hence, by the union bound,
\begin{equation}
1-H(x,y)
=\mathbb P(X>x\ \text{or}\ Y>y)
\le \mathbb P(X>x)+\mathbb P(Y>y)
=1-F(x)+1-G(y),
\end{equation}
which gives
\begin{equation}
H(x,y)\ge F(x)+G(y)-1.
\end{equation}
Since also $H(x,y)\ge 0$, we obtain
\begin{equation}
H(x,y)\ge \sup\{F(x)+G(y)-1,0\}.
\end{equation}

It remains to check that the endpoint functions are admissible. For $H_1$, let $U\sim\mathrm{Unif}(0,1)$ and define
\begin{equation}
X=Q_F(U),
\qquad
Y=Q_G(U),
\end{equation}
where $Q_F,Q_G$ are generalized quantile functions. Then $X$ has cumulative distribution function $F$, $Y$ has cumulative distribution function $G$, and
\begin{equation}
\mathbb P(X\le x,Y\le y)
=\mathbb P\bigl(U\le F(x),\ U\le G(y)\bigr)
=\mathbb P\bigl(U\le \inf\{F(x),G(y)\}\bigr)
=\inf\{F(x),G(y)\}.
\end{equation}
Thus the joint cumulative distribution function is exactly $H_1$.

For $H_0$, use the antimonotone coupling
\begin{equation}
X=Q_F(U),
\qquad
Y=Q_G(1-U).
\end{equation}
Then
\begin{equation}
\mathbb P(X\le x,Y\le y)
=\mathbb P\bigl(U\le F(x),\ 1-U\le G(y)\bigr)
=\mathbb P\bigl(1-G(y)\le U\le F(x)\bigr).
\end{equation}
The latter interval has length $F(x)+G(y)-1$ when this quantity is positive, and is empty otherwise. Therefore
\begin{equation}
\mathbb P(X\le x,Y\le y)=\sup\{F(x)+G(y)-1,0\}=H_0(x,y).
\end{equation}
So $H_0$ and $H_1$ are both admissible bivariate cumulative distribution functions with the required marginals.
\end{proof}

\subsection{The one-dimensional \(W_1\) identity}
\label{app:W1-quantile-primal}

Let $\mu,\nu\in\mathcal P_1(\mathbb R)$ with CDFs $F_\mu,F_\nu$ and generalized quantiles $Q_\mu,Q_\nu$.
Recall
\begin{equation}
W_1(\mu,\nu)=\inf_{\pi\in\Pi(\mu,\nu)}\int_{\mathbb R\times\mathbb R}|x-y|\,\mathrm{d}\pi(x,y).
\end{equation}

\begin{theorem}[Quantile and CDF formulas for $W_1$ on $\mathbb R$]
\label{thm:W1-1d-primal}
For $\mu,\nu\in\mathcal P_1(\mathbb R)$,
\begin{equation}
W_1(\mu,\nu)=\int_0^1 |Q_\mu(u)-Q_\nu(u)|\,\mathrm{d}u
=\int_{-\infty}^{+\infty}|F_\mu(x)-F_\nu(x)|\,\mathrm{d}x.
\end{equation}

\end{theorem}

This is the same identity as \cref{eq:W1-cdf-quantile} in the main text. The proof below stays entirely on the primal side: it gives a lower bound for every coupling and then shows that the monotone quantile coupling attains it.

\begin{proof}
\noindent\textbf{Step 1 (a lower bound for every coupling).}
Let $(X,Y)$ be any pair of real-valued random variables with laws $\mu$ and $\nu$.
For every fixed pair of real numbers $x,y$,
\begin{equation}
|x-y|
=\int_{\mathbb R}
\bigl|
\mathbf 1_{\{x\le t\}}-\mathbf 1_{\{y\le t\}}
\bigr|\,\mathrm{d}t.
\end{equation}
Indeed, the integrand is equal to $1$ exactly for $t$ between $x$ and $y$, and is otherwise $0$.
Applying this identity to $(X,Y)$ and using Fubini's theorem gives
\begin{equation}
\mathbb E|X-Y|
=
\int_{\mathbb R}
\mathbb E\bigl|
\mathbf 1_{\{X\le t\}}-\mathbf 1_{\{Y\le t\}}
\bigr|\,\mathrm{d}t.
\end{equation}
For each $t$, the random variables $\mathbf 1_{\{X\le t\}}$ and $\mathbf 1_{\{Y\le t\}}$ take only the values $0$ and $1$, and hence
\begin{equation}
\mathbb E\bigl|
\mathbf 1_{\{X\le t\}}-\mathbf 1_{\{Y\le t\}}
\bigr|
\ge
\bigl|
\mathbb E\mathbf 1_{\{X\le t\}}
-
\mathbb E\mathbf 1_{\{Y\le t\}}
\bigr|
=|F_\mu(t)-F_\nu(t)|.
\end{equation}
Therefore every coupling satisfies
\begin{equation}
\mathbb E|X-Y|
\ge
\int_{\mathbb R}|F_\mu(t)-F_\nu(t)|\,\mathrm{d}t.
\end{equation}
Taking the infimum over all couplings yields
\begin{equation}
\label{eq:W1-cdf-lower-bound}
W_1(\mu,\nu)
\ge
\int_{\mathbb R}|F_\mu(t)-F_\nu(t)|\,\mathrm{d}t.
\end{equation}

\medskip
\noindent\textbf{Step 2 (the quantile coupling attains the bound).}
Now let $U\sim\mathrm{Unif}(0,1)$ and define
\begin{equation}
X:=Q_\mu(U),
\qquad
Y:=Q_\nu(U),
\end{equation}
where $Q_\mu(u)=\inf\{x:F_\mu(x)\ge u\}$ and similarly for $Q_\nu$.
Then $X\sim\mu$, $Y\sim\nu$, so the law of $(X,Y)$ defines a coupling $\pi_\uparrow\in\Pi(\mu,\nu)$.
For every $t\in\mathbb R$, the generalized inverse property gives, up to null endpoint events,
\begin{equation}
\{X\le t\}=\{U\le F_\mu(t)\},
\qquad
\{Y\le t\}=\{U\le F_\nu(t)\}.
\end{equation}
Thus these two events are intervals in the same unit interval. Their symmetric difference has length exactly
$|F_\mu(t)-F_\nu(t)|$, and hence
\begin{equation}
\mathbb E\bigl|
\mathbf 1_{\{X\le t\}}-\mathbf 1_{\{Y\le t\}}
\bigr|
=
|F_\mu(t)-F_\nu(t)|.
\end{equation}
Using the first identity in the proof again,
\begin{equation}
\mathbb E|X-Y|
=
\int_{\mathbb R}|F_\mu(t)-F_\nu(t)|\,\mathrm{d}t.
\end{equation}
Since this value is attained by the coupling $\pi_\uparrow$, the lower bound \eqref{eq:W1-cdf-lower-bound} is sharp:
\begin{equation}
\label{eq:W1-cdf-attained}
W_1(\mu,\nu)
=
\int_{\mathbb R}|F_\mu(t)-F_\nu(t)|\,\mathrm{d}t.
\end{equation}

\medskip
\noindent\textbf{Step 3 (identify the same value with the quantile integral).}
For the same quantile coupling,
\begin{equation}
\mathbb E|X-Y|
=
\mathbb E|Q_\mu(U)-Q_\nu(U)|
=
\int_0^1|Q_\mu(u)-Q_\nu(u)|\,\mathrm{d}u.
\end{equation}
Combining this with \eqref{eq:W1-cdf-attained} completes the proof:
\begin{equation}
W_1(\mu,\nu)=\int_0^1|Q_\mu(u)-Q_\nu(u)|\,\mathrm{d}u=\int_{\mathbb R}|F_\mu(t)-F_\nu(t)|\,\mathrm{d}t.
\end{equation}

\end{proof}

\textbf{Remark}: Vallender~\cite{vallender1974calculation} is a classical reference for explicit Wasserstein computations on the line.

%% file: appendix/frechet1906_en_body.tex
\subsection*{INTRODUCTION.}
\phantomsection\addcontentsline{toc}{subsection}{INTRODUCTION.}

\textbf{1.}
 It is fairly generally agreed that a number $y$ is to be considered a function of the variable $x$ when to every numerical value of $x$ there corresponds a well-determined value of $y$. When one imposes on this correspondence certain conditions of continuity, differentiation, etc., one obtains classes of functions enjoying precise properties. The study of these properties constitutes the principal object of modern Analysis.

The notion of function thus understood has progressively been extended in several directions (for example, with regard to uniformity), and in particular from the point of view of what one is to take as variable. For a long time, functions of two, three, or even $n$ numerical variables have been considered. The other generalizations are more recent. Thus, Mr. \textsc{Le Roux} was led to study functions whose value depends no longer on $n$, but on an infinite sequence of independent variables (\textsc{xviii})\,\footnote{\,Roman numerals refer to the bibliography placed at the end of the Memoir.}. Mr. \textsc{Volterra} (\textsc{xv}) and Mr. \textsc{Arzelà} (\textsc{v}) seem to have been the first to study systematically functions whose value depends on the position and form of a variable line. Mr. \textsc{Hadamard} (\textsc{xi}) considered a particular class of functions whose variable is the form of an ordinary function.

In this Memoir we shall place ourselves at a wholly general point of view encompassing these different cases.

For this purpose, we shall say that a \textit{functional operation} (\textit{op\'eration fonctionnelle}) $U$ is defined on a set $E$ of elements \textit{of any nature} (numbers, curves, points, etc.) when to every element $A$ of $E$ there corresponds a determined numerical value of $U$: $U(A)$. The search for properties of these operations constitutes the object of \textit{Functional Calculus} (\textit{Calcul Fonctionnel}).

\textit{The present work is a first attempt to establish systematically some fundamental principles of Functional Calculus and then apply them to certain concrete examples.}\\

\textbf{2.}
To this end, it was first necessary to generalize the theory of linear sets, which has brought so much progress to the theory of functions of one variable. One might object that for a long time function theory could do without the consideration of point sets. But prior study of sets imposes itself with much greater force in Functional Calculus. Indeed, nothing (at least at first) plays, in Functional Calculus, the role of the interval, whose consideration sufficed for so long for analysts in function theory.

Granted the utility of this prior study of sets, a difficulty arose. The first generalization it was natural to attempt is that of the notion of continuous function. Now, if one wishes to extend it to operations whose variable is an element \textit{of any nature}, one must first know what is to be understood by neighboring elements or by the limit of a sequence of elements. This seems impossible: one is accustomed to giving a special definition of limit for each category of elements considered up to now: points, curves, etc. I overcame the difficulty by a method analogous to the one that permits, in the theory of \textit{abstract} groups, reasoning on a mode of composition not explicitly defined.

After developing this point of view, I observe that almost all classical definitions of limit (but not all) can be translated as follows: for the category of elements considered, one can associate to every pair of elements $A, B$ a number $(A, B) \ge 0$, having properties very analogous to those of the distance between two points, and such that: $\mathrm{1}^{\circ}\, B$ coincides with $A$ if $(A, B)=0$ and $2^{\circ}\, B$ tends to $A$ if $(A, B)$ tends to zero. By adopting this hypothesis, less general but still very broad, one obtains results that are more precise and more numerous.

\textit{The path we have just indicated led us to the generalization of almost all theorems on linear sets and on continuous functions (at least those that can be stated in a way independent of the nature of the elements considered).}

It was first necessary to see how to transform theorem statements so that they retain meaning in the general case. One then had either to transcribe proofs in a more general language, or, when that was not possible, to give new and more general proofs. It turned out that the proofs thus obtained are often as simple, and sometimes even simpler, than the particular proofs they replaced. This no doubt comes from the fact that the way the question is posed compels one to use only truly essential particularities. Besides, in this \textsc{First Part} we needed only the most elementary notions (apart from the definition and properties of the power of a set).

In the \textsc{Second Part}, we applied the general results thus obtained to various categories of elements, each of a determined nature. The only difficulty consisted in showing that each of these categories satisfies the general conditions necessary for validity of the theorems of the \textsc{First Part}. That done, all that remained was to state these theorems in each case. We call attention in particular to the definition we gave of the ``distance'' (\textit{écart}) of two curves, two holomorphic functions, etc. It seems this is the only means by which one can (thus generalizing the notion of interval) free oneself somewhat from consideration of sets.

Along the way, we sharpened the definition of two curves and completed an important theorem of Mr. \textsc{Arzelà} on ``compact'' sets of curves.

Moreover, a table of contents indicates in detail the adopted divisions.

I dare to hope that the new propositions I have been able to obtain, for example in the case of curves, will justify the usefulness of the general considerations from which I deduced them. But if proof of this usefulness for applications remained to be given, it would manifestly follow from the use already made of the classical particular cases; for example by Mr. \textsc{Arzelà} with regard to sets of functions (\textsc{vi, viii}) and by Mr. \textsc{Le Roux} for space of infinitely many dimensions (\textsc{xix}).\\

\textbf{3.}
For greater clarity, the propositions I believe to be new are the only ones printed in italics. Some of them were moreover already presented to the Academy of Sciences on November 21, 1904; January 2, February 27, March 20, and November 27, 1905; or published in another form in a Note in the \textit{Transactions of the American Mathematical Society} (\textsc{xxvii}). Finally, in order not to be too long, I have left aside researches relating to other parts of Functional Calculus that I had published elsewhere (\textsc{xxvii}).

I do not wish to end this \textsc{Introduction} without expressing the gratitude I owe my teachers for the encouragement they have never ceased to give me. I must thank in particular Mr. \textsc{Hadamard}, who kindly took an interest in the development of this work and whose valuable remarks led me to improve its writing in many points.

\newpage
\subsection*{FIRST PART. INTRODUCTION OF THE NOTION OF LIMIT IN ABSTRACT SETS.}
\phantomsection\addcontentsline{toc}{subsection}{FIRST PART. INTRODUCTION OF THE NOTION OF LIMIT IN ABSTRACT SETS.}

\etocsettocstyle{\subsection*{Chapter I.}}{}
\phantomsection\addcontentsline{toc}{subsection}{Chapter I.}
\subsubsection*{General notions on sets of elements of a class ($L$).}
\phantomsection\addcontentsline{toc}{subsubsection}{General notions on sets of elements of a class \texorpdfstring{($L$)}.}

\textbf{4.}
\textit{\textbf{Object of Functional Calculus.}}---Let us consider a set $E$ formed of arbitrary elements (numbers, points, functions, lines, surfaces, etc.), but such that one can distinguish distinct elements. To every element $A$ of this set let us associate a determined number $U(A)$; we thus define what we shall call a uniform \textit{functional operation} on $E$.

The study of these operations is the object of \textit{Functional Calculus}. It is to be presumed that the properties of each operation will depend, on the one hand, on properties of the set on which it is defined, and on the other hand on the nature of the correspondence defining the operation, that is, of the (non-univocal) correspondence between elements of $E$ and certain numbers that are values of the operation. One thus sees that functional calculus necessarily includes a preliminary theory of sets.\\

\textbf{5.}
\textit{\textbf{Set Theory.}}---The results obtained in such a theory will be all the more extensive as one addresses more general sets. To obtain the greatest possible generality, one should therefore first seek results applicable to abstract sets, that is, sets for which one does not specify the nature of their elements.

The most important results obtained so far in this direction are those concerning the notions (introduced by Mr. G. \textsc{Cantor}) of the power of a set (\textit{puissance d'un ensemble}) and of an ordered set.

For the \textsc{First Part} of this work, I shall assume known (besides the most elementary mathematical notions) only these results, and even only the simplest among them\,\footnote{\,See, for example, \textsc{ii}, pages 1, 2; and \textsc{iii}, pages 27, 28.}. In particular, I shall not use the theory of linear sets and continuous functions, which will appear, on the contrary, as a very particular corollary of our general theory.\\

\textbf{6.}
\textit{\textbf{The notion of limit in abstract sets.}}---The best-known results, and in fact the most important ones in set theory, are those deduced from the notion of the limit of a sequence of elements. Now this notion is never introduced except after one has abandoned the study of abstract sets and clearly specified the nature of the elements considered. The same phenomenon occurred when one developed separately theories of groups of motions, substitutions, transformations, etc., where each category of elements gave rise to a mode of composition perfectly defined, but whose definition varied from one category to another. One could arrive at a common theory only by \textit{refraining} from giving a general definition of this mode of composition, while searching for conditions common to particular definitions and retaining only those independent of the nature of the elements considered (\textsc{xxiv}). It is also from the same point of view that one was led to simplify the study of forces, velocities, rotations, etc., by preceding it with a theory of vectors. This idea reappears, somewhat modified, in the increasingly widespread use of ``descriptive'' definitions, especially advocated by Mr. \textsc{Drach} (\textsc{xxi}). Such are, for example, the definitions of the measure of a set proposed by Mr. \textsc{Borel}\,\footnote{\,``In every case, it proceeds from the same fundamental idea: defining the new elements one introduces by means of their essential properties, that is, those strictly indispensable for the reasonings that are to follow'' (I, page 48).} and Mr. \textsc{Lebesgue} (\textsc{iv}, page 102); the definition of area introduced by Mr. \textsc{Hadamard} in his Leçons de Géométrie Élémentaire (\textsc{Armand Colin}, 1898), etc. But there is this difference: in a ``descriptive'' definition one states \textit{characteristic} properties of the being one \textit{wants to define} (\textsc{iv}). By contrast, in the theory of abstract groups, the mode of composition is assumed defined in advance in each particular case; but one voluntarily ignores this definition in order to retain only certain general conditions that it satisfies \textit{but that do not determine it}.

Proceeding in this way, some proofs become more difficult, since one deprives oneself of a more concrete representation. But what is thus lost is largely regained by dispensing with repeating several times, under different forms, the same arguments. One often also gains a clearer view of what in proofs was truly essential, and thus simplifies them by stripping away what depended only on the proper nature of the elements considered. This is what we shall try to do for Functional Calculus, and in particular for the theory of abstract sets.\\

\textbf{7.}
If one examines carefully the various classical definitions of the limit of a sequence of numbers, or of points, or of functions, etc., one immediately notices that all these definitions satisfy two conditions I and II that can be stated independently of the nature of the elements considered. These are the conditions that will first suffice us to generalize certain propositions of Function Theory.

\textit{Henceforth, we shall limit ourselves to the study of sets drawn from a class ($L$) of elements of arbitrary nature but satisfying the following conditions: one knows how to distinguish whether two elements of class ($L$) are distinct or not. In addition, one has been able to give a definition of the limit of a sequence of elements of class ($L$). We therefore assume that, a random infinite sequence of elements (distinct or not) of class ($L$) being chosen, one can state with certainty whether this sequence $A_{1}, A_{2}, \ldots, A_{n}, \ldots$ has or does not have a limit $A$ (moreover unique). The procedure that permits the answer (in other words, the definition of limit) is absolutely arbitrary, subject only to satisfying conditions I and II just mentioned, namely:}

I) \textit{If each element of the infinite sequence $A_{1}, A_{2}, \ldots, A_{n}, \ldots$ is identical to one same element $A$, the sequence certainly has limit $A$.}

II) \textit{If an infinite sequence $A_{1}, A_{2}, \ldots, A_{n}, \ldots$ has limit $A$, every sequence of elements taken from the first in the same order: $A_{n_{1}}, A_{n_{2}}, \ldots, A_{n_{p}}, \ldots$ (thus integers $n_{1}, n_{2}, \ldots, n_{p}$ increase) has limit also equal to $A$.}\\

\textbf{8.}
\textit{\textbf{Usual definitions of point-set theory extended to abstract sets.}}---We are now able, under the hypothesis in which we have placed ourselves, to generalize definitions ordinarily adopted in point-set theory.

So let us consider an arbitrary set $E$ of elements of a class ($L$), in which one knows how to define, in conformity with conditions I and II, the limit of a sequence of elements.

We shall say that an element $A$ of class ($L$) is a \textit{limit element} (\textit{\'el\'ement limite}) of $E$ when there exists an infinite sequence of elements of $E$: $A_{1}, A_{2}, \ldots, A_{n}, \ldots$ which are distinct and tend to $A$. We shall call the \textit{derived set} (\textit{ensemble d\'eriv\'e}) of a set $E$, and denote it by $E^{\prime}$, the set of limit elements of $E$. We shall say that a set is \textit{closed} when it contains its derived set; that it is \textit{perfect} when it coincides with its derived set. Finally, considering a given set $H$ as fundamental set, we shall say that an element $A$ of any set $E$, formed of elements of $H$, is \textit{interior} to $E$ \textit{in the strict sense} when $A$ is an element of $E$ that is not the limit of any sequence of distinct elements $A_{1}, A_{2}, \ldots, A_{n}, \ldots$ belonging to $H$ without belonging to $E$.

These definitions are usual in the particular case of point sets. We shall also introduce, with Mr. \textsc{Lindelöf} (\textsc{ii}, page 4), the notion of condensation element (\textit{élément de condensation}) of a set. But we shall give a definition different from his (which does not lend itself to the present generalization), though equivalent in the case he considers, that of point sets. We shall call a \textit{condensation element} (\textit{élément de condensation}) of a set $E$ a limit element of $E$ that is also a limit element of every set obtained by deleting from $E$ a \textit{countable} infinity (\textsc{ii}, page 2) of elements. One sees that a set can give rise to a limit element only if it possesses an infinity of distinct elements, and that a set can give rise to a condensation element only if it possesses an uncountable infinity of elements.\\

\textbf{9.}
\textit{\textbf{Compact sets.}}---We shall see later that it is useful to introduce a new notion, that of a \textit{compact} set, which plays in the theory of abstract sets the same role as that of a bounded set in point-set theory. We shall say that a set is compact when it contains only a finite number of elements, or when every infinity of its elements gives rise to at least one limit element\,\footnote{\,See below an example of a compact set ($\mathrm{n}^{\circ}$ 86) and a non-compact set ($\mathrm{n}^{\circ}$ 85).}. When a set is both compact and closed, we shall call it an \textit{extremal} set (\textit{ensemble extrémal}); this designation will be justified later ($\mathrm{n}^{\circ}$ 11). The role of an extremal set in abstract set theory is fairly similar to that of an interval in the theory of linear sets. There is also a close analogy between compact sets and bounded point sets. This analogy is made evident by the following remarks, which follow immediately from the definition: \textit{Every set containing a non-compact set is non-compact. Every part of a compact set is compact. Every set formed by a finite number of compact sets is a compact set.} This analogy will moreover reappear constantly below.

\textit{If one considers a sequence of sets $E_{1}, E_{2}, \ldots, E_{n}, \ldots$ formed of elements of one same compact set $E$, each closed, contained in the preceding one, and each possessing at least one element, there is necessarily an element common to all these sets.}

Indeed, there is at least one element $M_{1}$ in $E_{1}$, $M_{2}$ in $E_{2}$, etc. If there is an infinity of distinct elements in the sequence $M_{1}, M_{2}, \ldots$, let these be $M_{n_{1}}, M_{n_{2}}, \ldots$. Since set $E$ is compact, this latter sequence has at least one limit element $M$; as it is contained in $E_{n_{p}}$ from $M_{n_{p}}$ onward, $M$ will be a limit element of $E_{n_{p}}$, and consequently will belong to the closed set $E_{n_{p}}$. Thus $M$ belongs to $E_{n_{p}}$, and therefore to $E_{1}, E_{2}, \ldots, E_{n_{p}}$. As this holds for every $p$, it follows that $M$ is common to all sets $E_{1}, E_{2}, \ldots$ If, on the contrary, the sequence $M_{1}, M_{2}, \ldots$ has only a finite number of distinct elements, then at least one element $M$ is repeated infinitely many times in this sequence, say $M \equiv M_{n_{1}} \equiv M_{n_{2}} \equiv \ldots$. Then $M$ is contained, for every $p$, in $E_{n_{p}}$ and therefore in $E_{1}, E_{2}, \ldots, E_{n_{p}}$. Hence it is common to all sets $E_{1}, E_{2}, \ldots, E_{n}, \ldots$\,\footnote{\,The preceding proof is a characteristic example of the type of proof we shall constantly use below. Every time we deal with an extremal set, we shall always have to extract a sequence of elements of the set tending to an element of the set. To see the transformation needed in ordinary proofs concerning point sets, see the corresponding proof by Mr. \textsc{Baire} (\textsc{iii}, page 18), where he uses subdivision of an interval into half-intervals, which is not generalizable.}.

\subsubsection*{Continuous operations.}
\phantomsection\addcontentsline{toc}{subsubsection}{Continuous operations.}

\textbf{10.}
We can now define continuity of a function on an abstract set. We shall say that a uniform functional operation $V$ on a set $E$ of elements of a class ($L$) is continuous on $E$ if, for any element $A$ of $E$ that is limit of a sequence of elements $A_{1}, A_{2}, \ldots, A_{n}, \ldots$ of $E$, one always has:
$$
V(A)=\lim _{n=\infty} V\left(A_{n}\right) .
$$

By contrast, we still cannot generalize the notion of uniform continuity. We shall see later how it can be introduced by considering classes of elements less general than class ($L$) (see $\mathrm{n}^{\circ}$ 47).\\

\textbf{11.}
We shall first generalize a result independent of the notion of continuity, stated first by \textsc{Weierstrass} for point sets, then by Mr. \textsc{Arzelà} for sets of curves\,\footnote{\,See v, page 347. Mr. \textsc{Arzelà} moreover assumes there that $E$ is a \textit{continuous} set in the sense indicated a little further on ($\mathrm{n}^{\circ}$ 12), a condition that is not necessary.}.

\textsc{\textbf{Theorem.}}---\textit{Given a uniform operation $U$ on an extremal set $E$, there exists at least one element $A$ of $E$ such that the upper limit (\textit{limite sup\'erieure})\,\footnote{\,One calls the upper limit (\textit{limite sup\'erieure}) of a set of numbers a quantity $\mu$ such that every number of the set is at most equal to $\mu$, and such that, whatever $\varepsilon>0$, there is at least one number of the set greater than $\mu-\varepsilon$. When no such number $\mu$ exists, one says the upper limit is infinite.} $\mu$ (finite or not) of $U$ on $E$ is equal to the upper limit of $U$ on every set $K$ of elements of $E$ to which $A$ is interior in the strict sense} [\textit{considering $E$ as the fundamental set} ($\mathrm{n}^{\circ}$ 8)].

The theorem is evident if $E$ contains only a finite number of elements, or if $U$ attains its upper limit at a determined element of $E$.

In the opposite case, one can find, for every $n$, an element $A_{n}$ of $E$ such that $U\left(A_{n}\right)>\alpha_{n}$, where $\alpha_{n}$ equals $\displaystyle \mu-\cfrac{1}{n}$ if $\mu$ is finite, or equals $n$ if $\mu$ is infinite.

One and the same element $B$ cannot satisfy $U(B)>\alpha_{n}$ for an infinity of values of $n$, for otherwise one would have $U(B)=\mu$, contrary to our hypothesis. Therefore there is an infinity of distinct elements in the sequence $A_{1}, A_{2}, \ldots, A_{n}, \ldots$, and since $E$ is compact, this infinity has at least one limit element $A$. In conclusion, there is a sequence $A_{n_{1}}, A_{n_{2}}, \ldots$, extracted from the first and tending to an element $A$ of $E$, since $E$ is closed.

I say that $A$ satisfies the announced condition.

Indeed, if $K$ is a set to which $A$ is interior in the strict sense, the elements of the sequence $A_{n_{1}}, A_{n_{2}}, \ldots$ are all elements of $K$ from a certain rank $q$ onward (otherwise $A$ would be the limit of a sequence of elements of $E$ not belonging to $K$). But then the upper limit $\mu_{\mathrm{1}}$ of $U$ on $K$ is at least equal to $U\left(A_{n_{p}}\right)>\alpha_{n_{p}}$ for every $p>q$. Hence $\mu \ge \mu_{\mathrm{1}}>\alpha_{n_{p}}$, with $\alpha_{n_{p}}$ tending to $\mu$. Therefore $\mu_{\mathrm{1}}=\mu$.

The preceding theorem would be stated in the same way for the lower limit.

One immediately deduces the following important corollary:

\textsc{\textbf{Corollary.}}---\textit{Every operation continuous on an extremal set $E$: $\mathrm{1}^{\circ}$ is bounded on this set; $2^{\circ}$ attains, at at least one element of this set, its upper limit, and at at least one element of this set, its lower limit.}

Thus, when a set $E$ is extremal, every continuous operation attains each of its extrema there. It is this circumstance that led us to adopt the designation extremal set, which will be further justified below ($\mathrm{n}^{\circ}$ 51).

One may give a corollary broader in a certain sense than the preceding one by considering upper semicontinuous operations\,\footnote{\,This definition is a generalization of that proposed by Mr. \textsc{Baire} (\textsc{iii}, page 71) for point sets.}. We shall say that a uniform operation $U$ on a set $E$ is upper semicontinuous on $E$ if, whatever element $A$ of $E$ is the limit of a sequence of distinct elements of $E$: $A_{1}, A_{2}, A_{3}, \ldots, A_{n}, \ldots$, $U(A)$ is at least equal to the greatest limit (\textit{plus grande des limites})\,\footnote{\,Given an infinite sequence of numbers $u_{1}, u_{2}, \ldots, u_{n}, \ldots$, the greatest limit (\textit{plus grande des limites}) of this sequence is a number $\lambda$ such that, whatever $\varepsilon$, one can find an integer $p$ such that $u_{n}<\lambda+\varepsilon$ for $n>p$, and there exists an integer $m>p$ for which $u_{m}>\lambda-\varepsilon$. If there is no number $\lambda$ satisfying these conditions, one says that the greatest limit is infinite.} of the sequence $U\left(A_{1}\right), U\left(A_{2}\right), U\left(A_{3}\right), \ldots, U\left(A_{n}\right), \ldots$.

\textsc{\textbf{Corollary.}}---\textit{Every operation upper semicontinuous on an extremal set $E$: $1^{\circ}$ is bounded above on $E$; $2^{\circ}$ attains, at at least one element of $E$, its upper limit\,\footnote{\,This theorem is used, for example, in the search for geodesics, since the length of a curve is a lower semicontinuous function of the curve itself ($\mathrm{n}^{\circ}$ 95). (See also \textsc{xvii}).}.}

The proof is the same as above; but one can no longer say anything about the lower limit.\\

\textbf{12.}
\textsc{\textbf{Theorem.}} --- \textit{If $U$ is a continuous operation on a continuous set (\textit{ensemble continu}) $E$, $U$ takes, at at least one element of $E$, every value between any two values taken by $U$.}

I call a set continuous (\textit{ensemble continu}) if, given any two of its elements $A, B$, one can extract from $E$ a set $F$ whose every element corresponds to a point of interval $(0,1)$ on an axis $0t$, and conversely. The correspondence is supposed such that if two elements $A_{1}, A_{2}$ of $F$ correspond to two points $t_{1}, t_{2}$, $A_{1}$ tends to $A_{2}$ when $t_{1}$ tends to $t_{2}$. The preceding theorem was proved for line functions by Mr. \textsc{Arzelà} (\textsc{v}, page 348). His proof generalizes immediately to the present case and amounts to considering, on $F$, operation $U$ as a function of $t$.

\subsubsection*{Series of continuous operations.}
\phantomsection\addcontentsline{toc}{subsubsection}{Series of continuous operations.}

\textbf{13.}
\textit{\textbf{Uniform and quasi-uniform convergence.}} --- Let us consider a sequence of continuous operations on a same set $E$; namely:
$$
U_{1}, U_{2}, U_{3}, \ldots, U_{n}, \ldots
$$

When this sequence converges at every point of $E$, the limit defines a uniform operation on $E$, say $U$. It is often useful to know whether this limit is or is not a continuous operation on $E$. There is no reason why this should always be so.

A very general case where continuity of $U$ is guaranteed is that in which convergence of the sequence is uniform. We shall say that a sequence of arbitrary uniform operations on $E$ \textit{converges uniformly} on $E$ to an operation $U$ if, whatever number $\varepsilon>0$, one can find an integer $p$ such that $n>p$ implies $\left|U_{n}(A)-U(A)\right|<\varepsilon$ at every element $A$ of $E$.

It is easily shown that if a sequence of operations $U_{1}, U_{2}, \ldots$, continuous on an arbitrary set $E$, converges uniformly to an operation $U$, then $U$ is continuous on $E$.\\

\textbf{14.}
But the limit may be continuous without convergence being uniform. One was therefore led to seek a less restrictive condition for convergence. Mr. \textsc{Arzelà} succeeded in solving the problem by introducing quasi-uniform convergence. His rather complicated proof applies only to functions of one variable (\textsc{ix}). But Mr. \textsc{Borel} then gave a simpler proof (\textsc{ii}, page 42), which generalizes immediately to sets of elements of a class ($L$), with a slight modification required by our more general hypothesis.

We shall say that a sequence of operations $U_{1}, U_{2}, \ldots, U_{n}, \ldots$, uniform on a set $E$, \textit{converges quasi-uniformly} to an operation $U$ when, given a number $\varepsilon>0$ and an arbitrary integer $N$, one can determine, once and for all, an integer $N^{\prime} \ge N$ such that to each element $A$ of $E$ one can assign an integer $n_{A}$ satisfying
$$
N \le n_{A} \le N^{\prime}, \quad\left|U(A)-U_{n_{A}}(A)\right|<\varepsilon\,.
$$

When convergence is uniform, there is also quasi-uniform convergence. Hence the theorem stated above is a consequence of the following:

\textit{When a sequence of operations $U_{1}, U_{2}, \ldots, U_{n}, \ldots$ continuous on an arbitrary set formed of elements of a class ($L$) converges quasi-uniformly on $E$ to an operation $U$, this operation is continuous on $E$.} The proof is the direct generalization of that of Mr. \textsc{Borel} (\textsc{ii}, page 42). This theorem has a converse, but the converse applies only to an extremal set $E$.

It suffices to generalize Mr. \textsc{Borel}'s second proof (\textsc{ii}, pages 43--44); however, to prove that a certain quantity $n_{A}$ determined at every element $A$ of $E$ is bounded on $E$, one cannot divide, as he does, into intervals of decreasing lengths. But if $n_{A_i}$ is not bounded, one can form, \textit{since $E$ is extremal}, a sequence of elements of $E$: $A_{1}, A_{2}, A_{3}, \ldots$ tending to an element of $E$ and such that $n_{A_{p}}$ tends to infinity with $p$. The rest of the proof follows.

One can nevertheless state a converse applicable to an arbitrary set. Combining it with the direct theorem gives the following proposition:

\textit{For a sequence of operations $U_{1}, U_{2}, \ldots, U_{n}, \ldots$ continuous on one same set $E$ formed of elements of a class ($L$) to converge to an operation continuous on $E$, it is necessary and sufficient that this sequence converge quasi-uniformly on every extremal set formed of elements of $E$.}

It suffices to show that if an element $A$ of $E$ is the limit of a sequence of elements of $E$: $A_{1}, A_{2}, \ldots, A_{n}, \ldots$, one has:
$$
U(A)=\lim _{n=\infty} U\left(A_{n}\right)
$$

Now the set $F$ formed by elements $A, A_{1}, A_{2}, \ldots$ is extremal. One can therefore apply to it the converse already proved in this case.\\

\textbf{15.}
\textit{\textbf{Equally continuous operations.}} --- Let us consider a family $\mathfrak{F}$ of operations continuous on one same set $E$ formed of elements of a class ($L$). In a large number of questions in Analysis, it would be useful to know whether family $\mathfrak{F}$ is such that from every infinity of distinct operations of $\mathfrak{F}$ one can extract a sequence $U_{1}, U_{2}, \ldots$, $U_{n}, \ldots$ that converges uniformly to a limit $U$, necessarily continuous on $E$. We shall say that such a family $\mathfrak{F}$ is \textit{compact}.

Mr. \textsc{Arzelà}, who showed how knowing that a family is compact would simplify many proofs in Analysis, succeeded in solving this question in the case where elements of $E$ are numbers (\textsc{v}, page 1).

For this he uses the notion of \textit{equal continuity} introduced by Mr. \textsc{Ascoli} (\textsc{xxii}). Since we place ourselves at a much more general point of view than those two authors, we shall give a definition of equal continuity independent of the notion of interval. We shall then recognize that it coincides with Mr. \textsc{Arzelà}'s definition in the sets considered by him. Mr. \textsc{Arzelà}'s proof is complicated and does not generalize to our present point of view. We shall give a different proof whose general course is basically the same as Mr. \textsc{Hilbert}'s for existence of geodesics (\textsc{xx}). But Mr. \textsc{Hilbert} had not seen that his \textit{method} could be generalized much farther than that of Mr. \textsc{Arzelà}, whereas his \textit{result} is, on the contrary, an immediate consequence of Mr. \textsc{Arzelà}'s theorem.

We shall say that uniform operations on one same set $E$ formed of elements of a class $(L)$ constitute a family $\mathfrak{F}$ of operations equally continuous at $A$ on $E$ if, given a number $\varepsilon>0$ and any sequence of elements of $E$: $A_{1}, A_{2}, \ldots, A_{n}, \ldots$, having as limit an element $A$ of $E$, one can find an integer $p$ such that inequality $n>p$ implies
$$
\left|U(A)-U\left(A_{n}\right)\right|<\varepsilon,
$$
whatever operation $U$ of family $\mathfrak{F}$.

It is of course understood that number $p$, independent of $U$, may instead vary with the sequence $A_{1}, A_{2}, \ldots, A_{n}, \ldots$. It follows immediately from the definition, as it should, that operations \textit{equally continuous} (\textit{également continues}) on $E$ are in particular each continuous on $E$. Conversely, if operations \textit{in finite number} are continuous on $E$, they are \textit{equally continuous} on $E$. But this remark is no longer true for an infinity of distinct operations.

To reach the goal we have set ourselves, we shall prove a series of propositions interesting in themselves, but involving for $E$ and $\mathfrak{F}$ different hypotheses. Combining them in the case where all these hypotheses are simultaneously satisfied yields the desired result.\\

\textbf{16.}
\textsc{\textbf{$1^{\text {st }}$ Lemma.}} --- The limit $U$ of a convergent sequence $U_{1}, U_{2}, \ldots, U_{n}, \ldots$ of operations equally continuous on a set $E$ is a continuous operation on $E$. Moreover, if $E$ is extremal, convergence is necessarily uniform on $E$.

Indeed, if one considers an arbitrary element $A$ of $E$, limit of a sequence of elements of $E$: $A_{1}, A_{2}, \ldots, A_{n}, \ldots$, one has
$$
\left|U_{q}(A)-U_{q}\left(A_{n}\right)\right|<\varepsilon, \text { for } n>p,
$$
whatever $q$ may be. If one lets $q$ increase indefinitely, one has in the limit
$$
\left|U(A)-U\left(A_{n}\right)\right| \le \varepsilon, \text { for } n>p \, .
$$

This proves that $U$ is continuous at every element $A$ of $E$.

Moreover, one sees that if operations of a family $\mathfrak{F}$ are equally continuous on $E$, so too are operations of family $\mathfrak{F}_1$, formed by the set of limit operations on $E$ of an infinite sequence of operations of $\mathfrak{F}$ (if there are any).

Now suppose $E$ is an extremal set. If equally continuous operations $U_{1}, U_{2}, \ldots$ converge to a limit $U$ on an extremal set $E$, convergence is necessarily uniform. Otherwise one could find a number $\varepsilon>0$ such that, whatever $n$, there exist an element $A_{n}$ of $E$ and an integer $p_{n}>n$ for which
$$
\left|U_{p_{n}}\left(A_{n}\right)-U\left(A_{n}\right)\right|>\varepsilon ;
$$
and since $E$ is extremal, one may assume $A_{1}, A_{2}, \ldots$ tend to an element $A$ of $E$. One then has:
$$\left|U_{p_{n}}\left(A_{n}\right)-U\left(A_{n}\right)\right| \le\left|U_{p_{n}}\left(A_{n}\right)-U_{p_{n}}(A)\right|+\left|U(A)-U\left(A_{n}\right)\right|+\left|U(A)-U_{p_{n}}(A)\right|.$$
Since $U_{p_{n}}(A)$ tends to $U(A)$ and operations $U_{1}, U_{2}, \ldots$ are equally continuous, one can find $k^{\prime}$ and $k^{\prime\prime}$ such that the last term is less than $\displaystyle \cfrac{\varepsilon}{3}$ for $n>k^{\prime}$, and the first two are each less than $\displaystyle \cfrac{\varepsilon}{3}$ for $n>k^{\prime\prime}$. Thus for $n>k^{\prime}+k^{\prime\prime}$:
$$
\left|U_{p_{n}}\left(A_{n}\right)-U\left(A_{n}\right)\right| \le \varepsilon .
$$

Hence we indeed arrive at a contradiction.\\

\textbf{17.}
\textsc{\textbf{$2^{\text {nd }}$ Lemma.}} --- Consider a sequence of operations $U_{1}, U_{2}, \ldots$ convergent on a set $D$. If these operations are equally continuous on set $F$ formed by elements of $D$ and of its derived set $D^{\prime}$, the sequence considered is also convergent on $F$.

It suffices to prove that the sequence is convergent at every element $A$ that is limit of a sequence of elements of $D$: $A_{1}, A_{2}, \ldots, A_{n}, \ldots$ Now, given $\varepsilon>0$, one can find $p$ such that:
$$
\left|U_{q}(A)-U_{q}\left(A_{n}\right)\right|<\cfrac{\varepsilon}{3},
$$
for $n>p$, whatever $q$ may be, since the operations are equally continuous. Take a determined value of $n>p$, for instance $n=p+1$. Since sequence $U_{1}\left(A_{p+1}\right), U_{2}\left(A_{p+1}\right), \ldots$ is convergent, one can find an integer $r$ such that:
$$
\left|U_{r+k}\left(A_{p+1}\right)-U_{r}\left(A_{p+1}\right)\right|<\cfrac{\varepsilon}{3},
$$
whatever $k$ may be.

Then, whatever integer $k$,
$$
\begin{aligned}
\left|U_{r+k}(A)-U_{r}(A)\right| \le & \left|U_{r+k}(A)-U_{r+k}\left(A_{p+1}\right)\right|+\left|U_{r+k}\left(A_{p+1}\right)-U_{r}\left(A_{p+1}\right)\right| \\
& +\left|U_{r}\left(A_{p+1}\right)-U_{r}(A)\right|<\varepsilon .
\end{aligned}
$$

By a theorem of \textsc{Cauchy}, this proves that sequence $U_{1}(A), U_{2}(A), \ldots$ is convergent.\\

\textbf{18. \textsc{$3^{\text {rd }}$ Lemma.}} --- For operations continuous on one same extremal set $E$ to form a compact family $\mathfrak{F}$, they must be equally continuous and collectively bounded.

If these operations were not collectively bounded, there would be one, $U_{n}$, for every $n$, whose absolute value would exceed $n$ at some element $A_{n}$ of $E$.

Since family $\mathfrak{F}$ is compact, one can suppose that $U_{1}, U_{2}, \ldots$ converge uniformly to an operation $U$ continuous on $E$, and that $U_{n}\left(A_{n}\right)$ still tends to infinity. Then one has
$$
\left|U_{n}\left(A_{n}\right)\right| \le\left|U_{n}\left(A_{n}\right)-U\left(A_{n}\right)\right|+\left|U\left(A_{n}\right)\right| .
$$

As $n$ increases indefinitely, the first term of the right-hand side tends to zero; the second remains finite since $U$ is a continuous operation on an extremal set ($\mathrm{n}^{\circ}$ 11). Therefore the left-hand side cannot increase indefinitely.

Likewise, if operations of $\mathfrak{F}$ were not equally continuous, one could find a sequence of elements of $E$: $A_{1}, A_{2}, \ldots$, tending to an element $A$ of $E$, and a number $\varepsilon$ such that, whatever $n$, there exists an operation of $\mathfrak{F}$, $U_{n}$, and an integer $p_{n}>n$ for which:
$$
\left|U_{n}(A)-U_{n}\left(A_{p_{n}}\right)\right|>\varepsilon .
$$

Since family $\mathfrak{F}$ is compact, this amounts to saying that there exists a sequence of elements of $E$: $A_{1}^{\prime}, A_{2}^{\prime}, \ldots$ tending to an element $A$ of $E$, and a sequence of operations of $\mathfrak{F}$: $U_{1}^{\prime}, U_{2}^{\prime}, \ldots$ converging uniformly to an operation $U$, in such a way that:
$$
\left|U_{n}^{\prime}(A)-U_{n}^{\prime}\left(A_{n}^{\prime}\right)\right|>\varepsilon .
$$
(It is enough to extract suitably these two sequences from sequences $A_{p_{1}}, A_{p_{2}}, \ldots$ and $U_{1}, U_{2}, \ldots$.)

But then:
$$
\left|U_{n}^{\prime}(A)-U_{n}^{\prime}\left(A_{n}^{\prime}\right)\right| \le\left|U\left(A_{n}^{\prime}\right)-U_{n}^{\prime}\left(A_{n}^{\prime}\right)\right|+\left|U_{n}^{\prime}(A)-U(A)\right|+\left|U(A)-U\left(A_{n}^{\prime}\right)\right|
$$
and, taking $n$ large enough, one sees as before that one can make the right-hand side less than $\varepsilon$, hence the announced contradiction.

To reach the general theorem, I shall state one last proposition, proved later in another form ($\mathrm{n}^{\circ}$ 66).

If for each integer value of $n$ one considers an infinite sequence $S_{n}$ of numbers $x_{1}^{(n)}, x_{2}^{(n)}, \ldots x_{p}^{(n)}, \ldots$, each, for fixed $p$ and whatever $n$, lying between two fixed numbers $\lambda_{p}$ and $\mu_{p}$, one can extract from sequence $S_{1}, S_{2}, \ldots$ a sequence $S_{n_{1}}, S_{n_{2}}, \ldots$ of increasing indices such that $x_{p}^{\left(n_{q}\right)}$ tends to a limit $x_{p}$, for every fixed number $p$, when $n_{q}$ increases indefinitely.\\

\textbf{19.}
We are now in a position to state the following general theorem:

\textsc{\textbf{Theorem.}} --- \textit{For operations continuous on one same extremal set $E$, formed of elements of a class ($L$), and belonging to one same countable set $D$ or to its derived set $D^{\prime}$, to form a compact family $\mathfrak{F}$, it is necessary and sufficient that operations of $\mathfrak{F}$ be equally continuous and collectively bounded at every element of $E$.}

The condition is necessary by the last Lemma. To prove it is sufficient, it suffices to prove that from every infinite sequence of operations of $\mathfrak{F}$: $U_{1}, U_{2}, \ldots$, one can extract a sequence converging uniformly on $E$ to an operation that will necessarily be continuous on $E$.

Indeed, let $A_{1}, A_{2}, \ldots$ denote sequence of elements of countable set $D$, and consider sequence $S_{n}$ of numbers $x_{1}^{(n)}=U_{n}\left(A_{1}\right), \ldots, x_{p}^{(n)}=U_{n}\left(A_{p}\right), \ldots$ Since operations of $\mathfrak{F}$ are collectively bounded at every element of $D$, these numbers $x_{p}^{(n)}$ are, for each value of $p$ and whatever $n$, between two fixed numbers $\lambda_{p}$ and $\mu_{p}$. By the last proposition recalled above, one can therefore extract from sequence $S_{1}, S_{2}, \ldots$ a sequence $S_{n_{1}}, S_{n_{2}}, \ldots$ of increasing indices such that $x_{p}^{\left(n_{q}\right)}=U_{n_{q}}\left(A_{p}\right)$ converges to a certain limit $x_{p}$, for every $p$. This means that one can extract from sequence $U_{1}, U_{2}, \ldots$ a sequence of operations $U_{n_{1}}, U_{n_{2}}, \ldots, U_{n_{q}}, \ldots$ convergent at every element of $D$.

By the second Lemma, this sequence of operations is also convergent at every element of $E$; by the first Lemma, convergence is uniform, and hence the limit is continuous on $E$.

\textsc{\textbf{1$^\text{st}$ Remark.}} --- If one no longer assumes that set $E$ is extremal, the preceding proof still proves something. It proves that if operations of $\mathfrak{F}$ are equally continuous and collectively bounded at every element of $E \equiv D+D^{\prime}$, then from every infinity of operations of $\mathfrak{F}$ one can extract a sequence tending to an operation continuous on $E$. But one does not know whether convergence is uniform. Indeed, examples can be given where under these conditions convergence is non-uniform. It will always be uniform on every extremal set contained in $E$.

\textsc{\textbf{$2^{\text {nd }}$ Remark.}} --- It might seem that condition $E \equiv D+D^{\prime}$, where $D$ is countable, is a condition introduced artificially for the sole purpose of ensuring correctness of the result. We shall see later that this is on the contrary the general case occurring in applications ($\mathrm{n}^{\mathrm{os}}$ 45, 52).\\

\textbf{20.}
\textbf{\textit{Upper-limit operation of a family of operations.}} --- Let us consider a family $\mathfrak{F}$ of uniform operations on a set $E$. At each element $A$ of $E$, we may define a number $L(A)$ to be the upper limit of values at $A$ of operations of $\mathfrak{F}$. If operations of $\mathfrak{F}$ are collectively bounded at every element of $E$, we thus determine a uniform operation on $E$, which we shall call the upper limit of $\mathfrak{F}$.

\textsc{\textbf{Theorem.}} --- \textit{Given a family $\mathfrak{F}$ of operations collectively bounded and equally continuous on an arbitrary set $E$, the upper limit of $\mathfrak{F}$ is a continuous operation on $E$.}

This theorem was proved by Mr. \textsc{Arzelà} (\textsc{v}, page 9) in the case where elements are points on a line. His proof generalizes immediately to the present case, remembering that our definition of continuity renders use of intervals unnecessary.\\

\textbf{21.}
\textsc{\textbf{Remark.}} --- We have seen above how one can generalize certain known theorems in the classes so broad that we called them classes $(L)$. Nevertheless, the small number of hypotheses we imposed on classes $(L)$ will not permit us to continue this extension. One quickly sees that certain fundamental propositions of the theory of linear sets and continuous functions, whose statement, suitably phrased, can retain a meaning for classes $(L)$, are no longer true for every class $(L)$. Thereby every attempt to generalize the later propositions that rely on them as a basis fails.

\subsubsection*{Derived sets.}
\phantomsection\addcontentsline{toc}{subsubsection}{Derived sets.}

\textbf{22.}
We shall give as an example one of the most important among them: the derived set of a linear set is closed. Is this proposition applicable to sets of elements of an arbitrary class ($L$)? \textit{The answer is negative.}

Let us indeed consider class ($C$) of functions of one variable $x$ defined on a fixed interval $J: 0 \le x \le 1$, and let us agree to say that a sequence of such functions $f_{1}(x), f_{2}(x), \ldots$ tends to a function $f(x)$ if, for each value $x$ in interval $J$, $f_{n}(x)$ has a finite limit $f(x)$. This definition does satisfy conditions I and II of no. 7. Therefore this class ($C$) is a class ($L$). However, the derived set of an arbitrary set of elements of class ($C$) is not necessarily closed.\\

{
\small

\,\footnote{\,In the \textsc{First Part}, subjects requiring somewhat special knowledge are printed in small type.} It is enough to consider set $E$ of continuous functions on $J$. Set $E^{\prime}$ of limit functions of continuous functions is set of functions of classes 0 or 1 in Mr. \textsc{Baire}'s sense (\textsc{iii}, page 126). Now it is known that there are class-2 functions (same reference), that is, there are limit functions of functions of $E^{\prime}$ that do not belong to $E$.

\textbf{23.}
\textit{\textbf{Limits of limit functions.}} --- We shall move away for a moment from our subject to point out a proposition to which the preceding example gives a certain interest.

Consider a set of functions $f_{n}^{(p)}(x)$ defined on interval $J$ and such that: $1^{\circ}$ sequence
$$
f_{n}^{(1)}(x), \quad f_{n}^{(2)}(x), \quad \ldots \quad f_{n}^{(p)}(x), \quad \ldots
$$
has a determined limit $f_{n}(x)$, and this whatever $n$; $2^{\circ}$ sequence
$$
f_{1}(x), f_{2}(x), \ldots, f_{n}(x), \ldots
$$
has determined limit $f(x)$.

From the above, it will not always be possible to find two sequences of integers $n_{1}, n_{2}, \ldots, p_{1}, p_{2}, \ldots$ such that one has $\displaystyle f(x)=\lim _{r=\infty} f_{n_{r}}^{\left(p_{r}\right)}(x)$.

This remark justifies interest in the following theorem:

\textsc{\textbf{Theorem.}} --- Consider functions $f_{n}^{(p)}(x), f_{n}(x), f(x)$ defined on one same interval $J$, whatever integers $n, p$, and such that:
\begin{equation*}
f(x)=\lim _{n=\infty} f_{n}(x), \quad f_{n}(x)=\lim _{p=\infty} f_{n}^{(p)}(x) . \tag{ \textsc{i}}
\end{equation*}

If in general it is impossible to choose integers $n_{r}, p_{r}$ such that equality
$$
f(x)=\lim _{r=\infty} f_{n_{r}}^{\left(p_{r}\right)}(x)
$$
holds at every point of $J$, it is at least always possible to determine them so that this equality holds at every point of $J$ except on a set of points of measure zero\,\footnote{\,For the meaning of this expression, see \textsc{ii}, page 16.}.

To prove this, I shall rely on two preliminary remarks: $1^{\circ}$ Assume, to simplify, that interval $J$ is interval $(0,1)$. If one considers a sequence of sets of points of $J$: $E_{1}, E_{2}, \ldots$, having measures $1-m_{1}, 1-m_{2}, \ldots$, the set of points common to all these sets has measure at least equal to $1-\left(m_{1}+m_{2}+\ldots\right)$. $2^{\circ}$ Mr. \textsc{Borel} proved (\textsc{ii}, page 37) that, whatever fixed number $\varepsilon>0$, the set of points where the remainder of rank $n$ of a convergent series of functions of $x$ on $J$ is in absolute value greater than $\varepsilon$ is a set whose measure tends to zero with $\cfrac{1}{n}$.

Let $\varepsilon$ and $\sigma$ be two arbitrary numbers between 0 and 1; one can find an integer $q$ such that the measure of the set of points where $\left|f-f_{n}\right|<\cfrac{\varepsilon}{2}$ is greater than $1-\cfrac{\sigma}{2}$ for $n>q-1$. Then, with $q$ fixed, one can find a number $r$ such that the measure of the set of points where $\left|f_{q}-f_{q}^{(r)}\right|<\cfrac{\varepsilon}{2}$ is also greater than $1-\cfrac{\sigma}{2}$.
At every point common to these two sets one has $\left|f-f_{q}^{(r)}\right|<\varepsilon$, and set of these common points has measure at least equal to $1-\left(\cfrac{\sigma}{2}+\cfrac{\sigma}{2}\right)=1-\sigma$. A fortiori, the same is true of the set of points such that $\left|f-f_{q}^{(r)}\right|<\varepsilon$.

This being so, let $\varepsilon$ and $\sigma$ take two successive sequences of values $\varepsilon_{1}, \varepsilon_{2}, \ldots ; \sigma_{1}, \sigma_{2}, \ldots$ Whatever $r$, one can determine integers $n_{r}, p_{r}$ such that set $G_{r}$ of points where one has $\left|f-f_{n_{r}}^{\left(p_{r}\right)}\right|<\varepsilon_{r}$ has measure greater than $1-\sigma_{r}$.
Now call $H_{r}$ the set of points common to $G_{r}, G_{r+1}, G_{r+2}, \ldots, G_{r+n}, \ldots$.
Its measure is at least equal to $1-\left(\sigma_{r}+\sigma_{r+1}+\ldots\right)$, and at each of its points one has $\left|f-f_{n_{r+q}}^{\left(p_{r+q}\right)}\right|<\varepsilon_{r+q}$, whatever integer $q$.

Now suppose one has chosen, as is possible, numbers $\varepsilon_{r}, \sigma_{r}$ so that $\varepsilon_{r}$ decrease to zero and the series $\sigma_{1}+\sigma_{2}+\ldots$ converges. For example, suppose $\varepsilon_{r}=\sigma_{r}=\cfrac{1}{r^{2}}$. Then one sees that at all points of $H_{r}$ one has $\left|f-f_{n_{r+q}}^{\left(p_{r+q}\right)}\right|<\cfrac{1}{(r+q)^{2}}$, whatever $q$; therefore at every point of $H_{r}$ sequence $f_{n_{1}}^{\left(p_{1}\right)}, f_{n_{2}}^{\left(p_{2}\right)}, \ldots$ converges to $f$. Hence set $H$ of points where this sequence converges has measure at least equal to that of $H_{r}$ for every $r$, that is, at least equal to $1-\left(\cfrac{1}{r^{2}}+\cfrac{1}{(r+1)^{2}}+\cdots\right)$ for every $r$. In other words, set $H$ has the same measure as $J$.\\

\textbf{24.}
\textsc{\textbf{Remarks.}} --- 1$^{\circ}$ There is an important special case where one can arrange that set $H$ \textit{coincides} with interval $J$. It is the case where $f_{n}^{(p)}(x)$ tends \textit{uniformly} to $f_{n}(x)$, whatever $n$. Indeed, in this case, for every value of $n$ one can choose an integer $p_{n}$ such that:
$$
f_{n}(x)=f_{n}^{\left(p_{n}\right)}(x)+\cfrac{\varepsilon_{n}(x)}{n} \, \text { with } \, \left|\varepsilon_{n}(x)\right|<1 \, ,
$$
whatever $x$. Then one has:
$$
f(x)=\lim _{n=\infty} f_{n}(x) \quad \text { and } \quad 0=\lim _{n=\infty} \cfrac{\varepsilon_{n}(x)}{n},
$$
whence:
$$
f(x)=\lim _{n=\infty} f_{n}^{\left(p_{n}\right)}(x)
$$

$2^{\circ}$ If equalities (\textsc{i}) were not necessarily verified everywhere, but each held on a set of measure $\ge m$ ($m$ fixed), one could arrange that set $H$ also had measure $\ge m$. In particular, if equalities (\textsc{i}) each held except on a set of measure zero, the preceding proof would still show that one can arrange for $H$ to coincide with $J$ except on a set of measure zero.

Apply remark $1^{\circ}$ taking as $f_{n}(x)$ arbitrary continuous functions.
It is known one may take as $f_{n}^{(p)}(x)$ polynomials with uniform convergence of $f_{n}^{(p)}$ to $f_{n}$. Therefore, by $1^{\circ}$, limit functions of continuous functions are also limits of polynomials. Now apply our theorem taking for $f_{n}$ a function of first class; from what we have just said, one may take polynomials for $f_{n}^{(p)}$. Hence class-2 functions are limits of polynomials except on a set of measure zero. Now apply remark $2^{\mathrm{o}}$ taking for $f_{n}$ a class-2 function. From the above, one may take polynomials for $f_{n}^{(p)}$, convergence occurring on $J$ except on a set of measure zero. Hence class-3 functions are also limits of polynomials except on a set of measure zero. More generally, \textit{every function $f(x)$ included in Mr. \textsc{Baire}'s classification (\textsc{iii}, page 126) can be regarded as the limit of a sequence of polynomials except on a set of measure zero}\,\footnote{\,I stated this theorem in the Comptes Rendus of March 20, 1905: \textit{Sur la notion d'écart dans le calcul fonctionnel}. Mr. \textsc{Lebesgue} obtained this same theorem as an application of his definition of the integral.}. If the sequence of polynomials converges everywhere, $f(x)$ is of class 0 or 1. If the polynomial sequence converges everywhere uniformly, $f(x)$ is class 0, that is, continuous. It is even enough that convergence hold everywhere quasi-uniformly.
}

\subsection*{Chapter II.}
\phantomsection\addcontentsline{toc}{subsection}{Chapter II.}
\subsubsection*{Definition of limit by neighborhood \textit{(voisinage)}.}
\phantomsection\addcontentsline{toc}{subsubsection}{Definition of limit by neighborhood \textit{(voisinage)}.}

\textbf{25.}
We already noted above that the extreme generality of classes ($L$) does not allow extension to them of a large number of properties of linear sets. To obtain more numerous properties, one should impose new restrictions on the conception of classes ($L$). They should be chosen so as to satisfy the following conditions: $1^{\circ}$ these restrictions should be stateable independently of the nature of the elements considered; $2^{\circ}$ they should be satisfied by the classes of elements most frequently occurring in applications; $3^{\circ}$ they should provide the sought generalization of theorems on linear sets and continuous functions.\\

\textbf{26.}
Now we observed ($\mathrm{n}^{\circ}$ 22) that if one considers an arbitrary class ($L$), the derived set of a set of elements of this class is not always closed. Yet this is a property it would have been quite natural to assume, translating it as follows: when elements $A_{1}, A_{2}, \ldots$ of a class ($L$) are each limit of a sequence of elements (for example $\displaystyle A_{n}=\lim _{p=\infty} A_{n}^{(p)}$) and tend to a limit $A$, this element $A$ is itself a limit of the latter elements (for example $\displaystyle A=\lim _{q=\infty} A_{n_{q}}^{\left(p_{q}\right)}$). The example we gave ($\mathrm{n}^{\circ}$ 22) shows this property is not a consequence of the definition of classes ($L$). It also proves that it may fail even among classes ($L$) that arise naturally in Function Theory.\\

\textbf{27.}
However, the preceding property is compatible with a large number of definitions of limit (particularly so-called uniform limits). One should therefore either study what happens when one imposes this property as an additional condition on the definition of limit, or introduce an additional condition that would imply this property as a consequence. Now examination of already known cases where this property holds shows that it is the latter circumstance that occurs, in the following form:

Consider a class $(V)$ of elements of arbitrary nature, but such that one can tell whether two of them are identical or not, and such that, moreover, to any two of them $A, B$, one can associate a number $(A, B)=(B, A) \ge 0$ having the two following properties: $1^{\circ}$ The necessary and sufficient condition for $(A, B)$ to be zero is that $A$ and $B$ be identical. $2^{\circ}$ There exists a determined positive function $f(\varepsilon)$ tending to zero with $\varepsilon$, such that inequalities $(A, B) \le \varepsilon$, $(B, C) \le \varepsilon$ imply $(A, C) \le f(\varepsilon)$, whatever elements $A, B, C$. In other words, it is enough that $(A, B)$ and $(B, C)$ be small for $(A, C)$ to be small as well. We shall call \textit{neighborhood (voisinage)} of $A$ and $B$ the number $(A, B)$.

This being so, we may say that a sequence of elements of class ($V$), $A_{1}, A_{2}, \ldots$, tends to an element $A$ if neighborhood $\left(A_{n}, A\right)$ tends to zero with $\displaystyle\cfrac{1}{n}$. If a sequence $A_{1}, A_{2}, \ldots$ has limit $A$, it can have no other, for if $B$ were limit of the same sequence, numbers $\left(A, A_{n}\right)$ and $\left(B, A_{n}\right)$ would be infinitely small with $\displaystyle\cfrac{1}{n}$, hence so would $(A, B)$ (2nd condition). Then $(A, B)$ would be zero and therefore elements $A, B$ would not be distinct (1st condition).

Moreover, this definition of limit does satisfy conditions I and II that we imposed in general on every definition of limit ($\mathrm{n}^{\circ}$ 7), thanks to conditions $1^{\circ}$ and $2^{\circ}$ imposed on definition of neighborhood (voisinage).\\

\textbf{28.}
Nevertheless, not every definition of limit satisfying conditions I and II can be deduced from notion of neighborhood (voisinage). To prove this, it will suffice to prove the following theorem.

\textsc{\textbf{Theorem.}} --- \textit{The derived set of a set of elements of a class ($V$) is a closed set.}

Indeed, returning to the notation used above: since here limit is deduced from neighborhood (voisinage), one can associate to each integer $n$ an integer $q_{n}$ such that $\displaystyle\left(A, A_{q_{n}}\right)<\cfrac{1}{n}$, and an integer $p_{n} \ge n$ such that $\displaystyle\left(A_{q_{n}}, A_{q_{n}}^{\left(p_{n}\right)}\right)<\cfrac{1}{n}$. One therefore has:
$$
\left(A, A_{q_{n}}^{\left(p_{n}\right)}\right)<f\left(\cfrac{1}{n}\right)
$$
and consequently $A_{q_{n}}^{\left(p_{n}\right)}$ tends to $A$ when $n$ tends to infinity.

If now we return to example of class ($C$) given above, we see it provides a class ($L$) where it is impossible to construct any definition of neighborhood (voisinage) such that the definition of limit deduced from it would coincide with the adopted one. For, starting from the latter, one obtains a non-closed derived set.

It is therefore genuinely restrictive to consider in particular classes ($V$) among classes $(L)$, and we may hope thus to obtain new generalizations.

We shall therefore now limit ourselves to the study of classes ($V$), that is, those classes ($L$) where definition of limit is deduced from that of neighborhood (voisinage).\\

\textbf{29.} \textsc{\textbf{Theorem.}} --- \textit{Set $P$ of condensation elements\,\footnote{\,See definitions given above ($\mathrm{n}^{\circ}$ 8). We shall see later ($\mathrm{n}^{\circ}$ 43) a very general case where an arbitrary set is always \textit{condensed}, i.e., where any uncountable infinity of its elements always gives rise to at least one condensation element.} of a condensed set $E$ formed of elements of a class} ($V$) \textit{is a perfect set or null set\,\footnote{\,The following proof is the generalization of the proof recently given by Mr. \textsc{Lindelöf} for the case where elements of $E$ are points of space. See \textsc{Borel} (\textsc{ii}, pages 5, 6). I have insisted only on points of his proof requiring some precautions to be generalized.}}.

Indeed, one can first show that $P$ is a closed set without assuming $E$ is condensed, and even for an arbitrary class ($L$). For if $A$ is any element that is limit of elements $A_{1}, A_{2}, \ldots$ of $P$, it is first a limit element of $E$ since $P$ is contained in $E^{\prime}$, which is closed by the preceding theorem. Moreover, if one removes from $E$ a countable set, elements $A_{1}, A_{2}, \ldots$ remain limit elements of remaining set $E_{1}$ (without necessarily belonging to it), since they are condensation elements of $E$. Therefore $A$ is also a limit element of $E_{1}^{\prime}$, which is also closed. Hence $A$ is again a limit element of $E_{1}$; it is indeed a condensation element.

Now show that every element of $P$ belongs to $P^{\prime}$. Indeed, let $A$ be an element of $P$, and let $E_{n}$ be the set of elements $B$ of $E$ such that $\displaystyle\cfrac{1}{n+1}<(B, A) \le \cfrac{1}{n}$. There are infinitely many such sets that are uncountable; otherwise there would be an integer $q$ such that set of elements $B$ of $E$ satisfying $\displaystyle(A, B) \le \cfrac{1}{q+1}$ is countable, and consequently $A$ would not be a condensation element. Thus among sets $E_{n}$ there are infinitely many that are uncountable: $E_{n_{1}}, E_{n_{2}}, \ldots$ taken in increasing index order. \textit{Since $E$ is condensed}, uncountable set $E_{n_{p}}$ gives rise to at least one condensation element $A_{n_{p}}$. Moreover, $A_{n_{p}}$ is distinct from $A$, since it is limit of elements $B$ such that $(A, B)$ is greater than $\displaystyle\cfrac{1}{n_{p}+1}$. Now one easily sees that $(A_{n_{p}}, A)$ tends to zero, hence $A$ belongs to $P^{\prime}$.\\

\textbf{30.}
\textsc{\textbf{Theorem.}} --- \textit{Let $E$ be a condensed set formed of elements of a class} ($V$). \textit{Set $D$ of elements of $E$ that are not condensation elements of $E$ is countable\,\footnote{\,Same remark as for the preceding theorem.}.}\\

\textbf{31.} \textsc{\textbf{Corollary.}} --- \textit{Every closed and condensed set formed of elements of a class} ($V$) \textit{is the sum of a countable set and a perfect set (or null set) with no common elements.}\\

\textbf{32.} \textsc{\textbf{Theorem.}} --- \textit{Set $D$ of those elements of a compact set $E$ of elements of a class} ($V$) \textit{that do not belong to derived set $E^{\prime}$ of $E$ is countable}\,\footnote{\,For the proof in the case where $E$ is a point set, see for example i, page 36. The second part of that proof relies on the fact that intervals of bounded lengths covering segment $(0,1)$ and whose midpoints do not belong to two intervals are necessarily finite in number. It therefore cannot be generalized here.}.

Indeed, if $A$ is an element of $D$, let $\rho_{A}$ denote the lower limit $\ge 0$ of the neighborhood of $A$ with elements of $E^{\prime}$. Number $\rho_{A}$ is positive; otherwise $A$ would be a limit element of $E^{\prime}$ and therefore would belong to $E^{\prime}$.

I claim that number of elements of $D$ for which one has $\displaystyle\rho_{A}>\cfrac{1}{n}$ is finite. Indeed, otherwise one could take an infinity of them, $A_{1}, A_{2}, \ldots, A_{p}, \ldots$, all distinct and tending to a limit $B$. Then, for $p$ large enough, one would have $\displaystyle\left(A_{p}, B\right)<\cfrac{1}{n}<\rho_{A_{p}}$, which is impossible, since $B$, limit of elements of $E$, would belong to $E^{\prime}$, contrary to the definition of $\rho_{A_{p}}$.

Thus there is a finite number of elements $A$ of $D$: $A_{n}^{(1)}, A_{n}^{(2)}, \ldots, A_{n}^{\left(p_{n}\right)}$ such that $\displaystyle\rho_{A} \ge \cfrac{1}{n}$.

If one considers the countable sequence:
$$
A_{1}^{(1)}, \ldots, A_{1}^{\left(p_{1}\right)}, A_{2}^{(1)}, \ldots, A_{2}^{\left(p_{2}\right)}, \ldots, A_{n}^{(1)}, \ldots, A_{n}^{\left(p_{n}\right)}, \ldots,
$$
one sees that all elements of $D$ are listed in it, each at least once. Therefore $D$ is countable.

In particular, one sees that if $E$ is closed, $E$ decomposes into its derived set $E^{\prime}$ and a countable set $D$.

\textsc{\textbf{Corollary I.}} --- \textit{If derived set $E^{\prime}$ of a compact set $E$ formed of elements of a class} ($V$) \textit{is countable, set $E$ itself is countable.} \\

\textbf{33.}
\textsc{\textbf{Corollary II.}} --- \textit{Any extremal set $F$ formed of elements of a class $(V)$ may be regarded as the derived set of a set of elements of class $(V)$.}
At least, this proposition will be true when class ($V$) is such that one can always find in it an element whose neighborhood with an arbitrarily given element is smaller than any given positive number. This very natural condition is satisfied in all concrete examples we shall study below.

Assume therefore this condition is satisfied; one has $F=D+F^{\prime}$, with $D$ a countable set of elements:
$$
A_{1}, A_{2}, \ldots, A_{n}, \ldots
$$

To each of them, $A_{n}$, corresponds as above a positive number $\rho_{A_{n}}$. Whatever $p$, one can by hypothesis find an element of class ($V$), $B_{n}^{(p)}$, such that $\left(A_{n}, B_{n}^{(p)}\right)<\cfrac{1}{n p}$. We thus have elements $B_{n}^{(p)}$ determined for every integer value of $n$ and $p$. Let $E$ denote set of elements of $F^{\prime}$ and of elements $B_{n}^{(p)}$. I say that $F$ is derived set of $E$. Indeed, every element of $F$ is obviously a limit element of $E$. Conversely, consider a sequence $\Sigma$ of elements of $E$ tending to a limit $C$. I say that $C$ belongs to $F$. Indeed, if this sequence contains an infinity of elements $B_{n}^{(p)}$ corresponding to one same value of $n$, this is evident. Likewise if there is an infinity of elements of $F^{\prime}$, since $F^{\prime}$ is closed and contained in $F$. There remains the case where $\Sigma$ has only finitely many elements of $F^{\prime}$ and only finitely many elements $B_{n}^{(p)}$ for each value of $n$. Then there would be an infinity of distinct elements $B_{n}^{(p)}$ corresponding to increasing values of $n$:
$$
B_{n_{1}}^{\left(p_{1}\right)}, \ldots, B_{n_{q}}^{\left(p_{q}\right)}, \ldots
$$

For $q$ sufficiently large one would have $\left(B_{n_{q}}^{\left(p_{q}\right)}, C\right)<\varepsilon$ and $\left(A_{n_{q}}, B_{n_{q}}^{\left(p_{q}\right)}\right)<\varepsilon$ $\left[\text{since } B_{n_{q}}^{\left(p_{q}\right)} \text{ tends to } C \text{ and one has } \left(A_{n_{q}}, B_{n_{q}}^{\left(p_{q}\right)}\right)<\cfrac{1}{n_{q} p_{q}}<\cfrac{1}{n q}\right]$;
whence:
$\left(A_{n_{q}}, C\right)<f(\varepsilon)$.

Since $f(\varepsilon)$ tends to zero with $\varepsilon$, one sees that $C$ would be a limit of elements of $D$, that is to say in $F^{\prime}$, hence again in $F$.\\

\textbf{34.} Let us call spheroid of center $A$ and radius $\rho$ the set of all elements $B$ such that $(A, B)<\rho$. We shall say that an element $C$ is \textit{interior} to this spheroid if one has $(A, C)<\rho$.

\textsc{\textbf{Theorem.}} --- \textit{If a closed set $F$ belongs to a compact set $E$ formed of elements of a class $(V)$, one obtains $F$ by removing from $E$ elements interior to each spheroid of a certain countable set of spheroids}\,\footnote{\,The following proof is a generalization of the one given for the case where $E$ is a set of points in the plane by Mr. \textsc{Zoretti} (xxiii, page 5).}.

Indeed, let $G$ be complement of $F$, that is, set of elements of $E$ that are not part of $F$. To every element $A$ of $G$, we may associate a number $\rho_{A}$, namely the lower limit of neighborhood of $A$ with each element of $F$. Number $\rho_{A}>0$ cannot be zero; otherwise $A$ would belong to $F$ as a limit of elements of $F$.

We now form a countable sequence of elements of $G$:
$$
A_{1}^{(1)}, \ldots, A_{1}^{\left(p_{1}\right)}, A_{2}^{(1)}, \ldots, A_{2}^{\left(p_{2}\right)}, \ldots, A_{n}^{(1)}, \ldots, A_{n}^{\left(p_{n}\right)}, \ldots
$$
such that $1^{\circ}$ one has $\rho_{A_{n}^{(1)}} \ge \cfrac{1}{n}, \ldots, \rho_{A_{n}^{\left(p_{n}\right)}} \ge \cfrac{1}{n}$ for every $n$; $2^{\circ}$ if $A$ is an element of $G$ such that $\rho_{A} \ge \cfrac{1}{n}$, one has one of inequalities $\left(A, A_{1}^{(1)}\right)<\cfrac{1}{n}, \ldots,\left(A, A_{n}^{\left(p_{n}\right)}\right)<\cfrac{1}{n}$.
To show this is possible, assume sequence formed up to $A_{n-1}^{\left(p_{n-1}\right)}$, and denote in general by $I_{A_{n}^{(q)}}$ set of elements $B$ of $E$ such that $\left(B, A_{n}^{(q)}\right)<\cfrac{1}{n}$. Then take an element $A$ of $G$ (if any exists) such that $\rho_{A} \ge \cfrac{1}{n}$ and that belongs to none of sets $I_{A_{1}^{(1)}}, \ldots, I_{A_{n-1}^{\left(p_{n-1}\right)}}$.
We may take it as $A_{n}^{(1)}$. Then take, if one exists, an element $A$ of $G$ such that $\rho_{A} \ge \cfrac{1}{n}$ and that belongs to none of sets $I_{A_{1}^{(1)}}, \ldots, I_{A_{n-1}^{\left(p_{n-1}\right)}}, I_{A_{n}^{(1)}}$.
We call it $A_{n}^{(2)}$, and so on.
Thus, step by step, we form a sequence $A_{n}^{(1)}, A_{n}^{(2)}, \ldots$ that may contain no term at all, but surely contains only finitely many. Indeed, by construction they are all distinct.
If there were infinitely many, one could extract a sequence $B_{1}, B_{2}, \ldots$ having a limit $B$, and then for $q$ large enough one would have $\left(B_{q}, B_{q+p}\right)<\cfrac{1}{n}$ whatever $p$, that is to say $B_{q+p}$ would be interior to $I_{B_{q}}$, contrary to the hypothesis.

With sequence thus formed as announced, it is then easy to see that $F$ is set of elements of $E$ belonging to none of sets
$$
I_{A_{1}^{(1)}}, I_{A_{1}^{(2)}}, \ldots, I_{A_{1}^{\left(p_{1}\right)}}, I_{A_{2}^{(1)}}, \ldots, I_{A_{n}^{\left(p_{n}\right)}}, I_{A_{n+1}^{(1)}}, \ldots,
$$
which proves proposition.\\

\textbf{35.}
Let us call a \textit{bounded set} a set drawn from a class ($V$) such that neighborhood of any two elements of this set remains below a fixed number. In the case of linear sets, we shall see this definition coincides with that of compact set. In the general case, we can only state the following proposition:

\textsc{\textbf{Theorem.}} --- \textit{Every compact set formed of elements of a class $(V)$ is bounded.}

Indeed, in the opposite case one could find, for every $n$, two elements $A_{n}, B_{n}$ of $E$ such that $\left(A_{n}, B_{n}\right)>n$, and since $E$ is compact one may suppose $A_{n}, B_{n}$ tend to two respective limits $A, B$. Now for $n$ large enough one has:
$$
\left(A_{n}, A\right)<2(A, B), \quad(A, B)<2(A, B)\,,
$$
whence:
$$
\left(A_{n}, B\right)<f[2(A, B)]=k\,.
$$
Thus, for $n$ large enough one has:
$$
\left(A_{n}, B\right)<k \text { and }\left(B_{n}, B\right)<k\,,
$$
whence:
$$
\left(A_{n}, B_{n}\right)<f(k)\,,
$$
which leads to impossible inequality $n<f(k)$.\\

\textbf{36.}
\textsc{\textbf{Theorem.}} --- \textit{Let $E$ be an extremal set formed of elements of a class $(V)$. If there exists an infinite sequence $G$ of sets $I_{1}, I_{2}, \ldots$ such that every element of $E$ is interior in the strict sense to at least one of these sets $I_{n}$, one can extract from $G$ a finite number of these sets forming a family $H$ having the same property as $G$\,\footnote{\,This theorem is a generalization of a theorem by Mr. \textsc{Borel}, stated for linear sets (\textsc{i}, page 42), then extended by him to sets of points in $n$-dimensional space (\textsc{xiv}, page 357). But his proofs do not generalize to the present case. We shall prove this theorem later for certain classes ($V$), even in the case where $G$ is not countable ($\mathrm{n}^{\circ}$ 42).}.}

Indeed, suppose theorem false. Then there would exist at least one element of $E$, say $A_{1}$, that is not interior to $I_{1}$ in the strict sense.
But $A_{1}$ is interior in the strict sense to one of sets $I_{2}, I_{3}, \ldots$; let $I_{q_{1}}$ be first among them ($q_{1}>1$).
There is at least one element of $E$, say $A_{2}$, that is interior in the strict sense to none of sets $I_{1}, I_{2}, \ldots, I_{q_{1}}$.
Call $I_{q_{2}}$ first among sets $I_{q_{1}+1}, I_{q_{1}+2}, \ldots$ to which $A_{2}$ is interior in the strict sense, and so on.
In this way we form an infinite sequence of elements $A_{1}, A_{2}, \ldots$ belonging to $E$, all distinct, and such that $A_{r+1}$ is interior in the strict sense to none of sets $I_{1}, I_{2}, \ldots, I_{q_{1}}, \ldots, I_{q_{r}}$.
This sequence has at least one limit element $A$ belonging to $E$. Let $\displaystyle A \equiv \lim _{p=\infty} A_{n_{p}}$.
Now $A$ is interior in the strict sense to at least one set $I_{k}$ in sequence $G$.
It is therefore impossible that sequence $A_{n_{1}}, A_{n_{2}}, \ldots$ contain an infinity of elements $A_{n_{1}}^{\prime}, A_{n_{2}}^{\prime}, \ldots$ that are not interior to $I_{k}$ in the strict sense. Otherwise one would call $B_{p}$ an element of $E$ equal to $A_{n_{p}}^{\prime}$ if the latter does not belong to $I_{k}$, or else an element of $E$ not belonging to $I_{k}$ and such that $\left(B_{p}, A_{n_{p}}^{\prime}\right)<\cfrac{1}{p}$ in the other case. Then $A$ would be limit of a sequence $B_{1}, B_{2}, \ldots$ of elements of $E$ not in $I_{k}$, which is impossible.
Thus elements of sequence $A_{n_{1}}, A_{n_{2}}, \ldots$ are all interior to $I_{k}$ in the strict sense from a certain rank onward. This is manifestly contradictory with the way we formed sequence $A_{1}, A_{2}, \ldots$.

\subsubsection*{Normal classes ($V$).}
\phantomsection\addcontentsline{toc}{subsubsection}{Normal classes \texorpdfstring{($V$)}.}

\textbf{37.}
After studying properties of sets and operations on classes $(L)$, it proved useful to restrict ourselves to case of classes ($V$) in order to continue our generalizations. Likewise, we shall now once again restrict the scope of our investigations without yet specifying nature of considered elements.

Let us first give a few definitions. We shall say that a sequence of elements $A_{1}, A_{2}, \ldots$ of a class ($V$) satisfies \textsc{Cauchy} conditions when to every number $\varepsilon>0$ one can associate an integer $n$ such that inequality $\left(A_{n}, A_{n+p}\right)<\varepsilon$ holds for every $p$. It is evident that if a sequence tends to a limit, it satisfies \textsc{Cauchy} conditions. On the other hand, if a sequence satisfies \textsc{Cauchy} conditions \textit{and belongs to a compact set}, it has an obviously unique limit element. But if one only knows that sequence satisfies \textsc{Cauchy} conditions, one cannot assert that it tends to a limit; all one can say is that if it has a limit element, it has only one.

We shall then say \textit{that a class $(V)$ admits a generalization of \textsc{Cauchy}'s theorem if every sequence of elements of this class that satisfies \textsc{Cauchy} conditions has a limit element (necessarily unique).}

Next, we shall call separable class a class that can be regarded in at least one way as derived set of a countable set of its own elements.

Finally, it is useful, as we already did ($\mathrm{n}^{\circ}$ 33), to consider among classes ($V$) those such that, near a given element, there exist elements whose neighborhood with it can be made as small as desired without being zero. In other words, every element of class is then a limit element. Since conversely this is true by definition of the class, we shall call such classes perfect\,\footnote{\,A separable class is necessarily perfect, but converse is not obviously true. Hence the two hypotheses should be kept separate, all the more since certain theorems apply to perfect classes without assuming separability.}.

This being so, we shall henceforth restrict ourselves to study of \textsc{normal} classes ($V$), \textit{that is, perfect, separable, and admitting a generalization of \textsc{Cauchy}'s theorem}. This restriction is by no means artificial; it follows directly from comparison of classes ($V$) with linear sets, and we shall see later that all particular classes we examine belong to this general category of normal classes ($V$). This fact might even suggest asking whether every class ($V$) is necessarily normal. It is not, as shown by following examples:\\

{\small
\textbf{38.}
Take a class of elements represented by points on a line, and call neighborhood of two of them length of arc of circle bounded by two corresponding points in an inversion transformation. We thus obtain a perfect class ($V$). However, if one takes a sequence of elements corresponding to points receding to infinity, one clearly gets a \textsc{Cauchy} sequence with no limit elements. One can give an example where neighborhood is not defined in such an artificial way.
It suffices to take class of numbers from viewpoint of arithmeticians who exclude consideration of irrational numbers.
Taking as neighborhood of two rational numbers absolute value of their difference, one again obtains a perfect class ($V$) for which there exists an infinity of sequences satisfying \textsc{Cauchy} conditions without having limit elements.

Now pass to separable classes. One may so qualify linear sets by considering infinite line as derived set of set of points with rational abscissas.
But this is not the case for every perfect class ($V$).

Take, for example, class of functions (continuous or discontinuous) of one variable $x$ on interval $(0,1)$. Calling neighborhood of two elements $f(x), g(x)$ the upper limit (finite or not) of $|f(x)-g(x)|$ on interval $(0,1)$, one sees this gives a perfect class ($V$). This class is not separable. Indeed, if it were, one could form once and for all a sequence of functions $f_{1}(x), f_{2}(x), \ldots$ such that every element of class is uniform limit of a sequence $f_{n_{1}}(x), f_{n_{2}}(x), \ldots$ extracted from the preceding one. Then, class elements being defined by countable conditions, power of class would be at most equal to that of continuum (\textsc{i}, page 126).
Now it is known the contrary holds (\textsc{i}, page 125).
One may note that preceding examples prove that power of a \textit{perfect} set drawn from a class ($V$) is not necessarily power of continuum as for linear sets, but may be smaller (first example) or larger (second example).
}\\

\textbf{39.}
These remarks made, let us now restrict ourselves to study of normal classes ($V$).

First, to understand structure of such classes, we shall generalize a question proposed by Mr. \textsc{Hadamard} (\textsc{xii}). The question proper will be studied later (see $\mathrm{n}^{\circ}$ 58). We extend it to present general case as follows:

Consider a class ($V$). Suppose one has been able to form partial sets $k_{\varepsilon}$ each containing at least one element and such that: $1^{\circ}$ every element of class belongs to one of sets $k_{\varepsilon}$; $2^{\mathrm{o}}$ any two elements of one same set $k_{\varepsilon}$ have neighborhood less than $\varepsilon$. Considering the $k_{\varepsilon}$ as so many individuals, one may say with Mr. \textsc{Hadamard} that set $\mathfrak{K}_{\varepsilon}$, which has these $k_{\varepsilon}$ as elements, ``enumerates'' the class considered. We aim to determine power of $\mathfrak{K}_{\varepsilon}$.

We find here a new analogy with linear sets. If one takes class of points on a line, problem consists in finding power of set of intervals of lengths $<\varepsilon$ that cover infinite line. Now it is evident they can be chosen so that this set is countable, though infinite.\\

\textbf{40.}
\textit{Likewise, given a separable class $(V)$, one can, whatever $\varepsilon$, choose sets $k_{\varepsilon}$ so that $\mathfrak{K}_{\varepsilon}$ is countable.}

Indeed, let $\omega$ and $\eta$ be two numbers such that $f(\omega)<\varepsilon$ and $f(\eta)<\omega$.
Since class is separable, one can extract from it a sequence of elements $A_{1}, A_{2}, \ldots$ whose derived set is the whole class. Let $K_{\varepsilon}^{(p)}$ be set of elements $A$ such that $\left(A, A_{p}\right)<\omega$. Sequence $K_{\varepsilon}^{(1)}, K_{\varepsilon}^{(2)}, \ldots$ is countable; moreover, for any two elements $A, B$ of $K_{\varepsilon}^{(p)}$, one has $(A, B)<\varepsilon$; finally every element $A$ of class belongs, for a given value of $\varepsilon$, to at least one of the $K_{\varepsilon}^{(p)}$. Set of $K_{\varepsilon}^{(p)}$ therefore constitutes desired set $\mathfrak{K}_{\varepsilon}$. One may even add that every element $A$ is interior in strict sense to at least one of sets $K_{\varepsilon}^{(1)}, K_{\varepsilon}^{(2)}, \ldots$. It is enough to take in $A_{1}, A_{2}, \ldots$ an element $A_{n_{1}}$ such that $\left(A, A_{n_{1}}\right)<\eta$; then $A$ is interior in strict sense to $K_{\varepsilon}^{\left(n_{1}\right)}$.\\

\textbf{41.}
The same method proves that, given a \textit{set} $E$ belonging to a separable class $(V)$, one can form a countable infinity $\mathfrak{K}_{\varepsilon}$ of sets $K_{\varepsilon}$ of elements of $E$ such that every element of $E$ is interior in the strict sense to at least one of the $K_{\varepsilon}$, and such that the neighborhood of any two elements of $K_{\varepsilon}$ remains less than $\varepsilon$. One may ask whether there exist sets $E$ such that the set $\mathfrak{K}_{\varepsilon}$ is not only countable, but finite.

\textsc{\textbf{Theorem.}} --- \textit{In order that the set $\mathfrak{K}_{\varepsilon}$ $[$corresponding to a set $E$ taken from a \textsc{normal} class $(V)$$]$ may be chosen, whatever $\varepsilon$, so as to contain only a finite number of elements, it is necessary and sufficient that $E$ be compact.}

The condition is necessary. Indeed, in the opposite case, one could extract from $E$ an infinite sequence $S$ of distinct elements $B_{1}, B_{2}, \ldots$ with no limit elements. If $\mathfrak{K}_{\varepsilon}$ is finite, there will necessarily be a partial set $K_{\varepsilon}$ to which infinitely many elements of the sequence $S$ belong. One then sees that one can form successively an infinity of sequences $S_{n}$ of elements of $S$, such that $S_{n+1}$ is contained in $S_{n}$ and such that the distance between two elements of $S_{n}$ is less than $\cfrac{1}{n}$. If therefore I call $B_{1}^{(p)}, B_{2}^{(p)}, \ldots$ the elements of the sequence $S_{n}$ taken in the same order as in $S$, one sees that the sequence $B_{1}^{(1)}, B_{2}^{(2)}, \ldots, B_{p}^{(p)}, \ldots$ will be a sequence of distinct elements of $S$ such that one has $\left(B_{p}^{(p)}, B_{n+p}^{(n+p)}\right)<\cfrac{1}{p}$ whatever $n$. Since $E$ is taken from a normal class $(V)$, it would therefore be necessary that the sequence $S$ had a limit element.

Conversely, suppose $E$ compact. Then the set $F \equiv E+E^{\prime}$ is extremal. For a given value of $\varepsilon$, we know that one can form a sequence of partial sets $K_{\varepsilon}^{(1)}, K_{\varepsilon}^{(2)}, \ldots$ such that their totality satisfies the conditions of a $\mathfrak{K}_{\varepsilon}$ relative to $F$. In particular, every element of $F$ is interior in the strict sense to one of the sets $K_{\varepsilon}^{(1)}, K_{\varepsilon}^{(2)}, \ldots$. Hence one can extract from this sequence $K_{\varepsilon}^{(1)}, K_{\varepsilon}^{(2)}, \ldots$ a finite number of terms $K_{\varepsilon}^{\left(q_{1}\right)}, \ldots, K_{\varepsilon}^{\left(q_{r}\right)}$ enjoying the same property (see $\mathrm{n}^{\circ}$ 36).
If we now call $H_{\varepsilon}^{\left(q_{k}\right)}$ the set of elements of $E$ that belong to $K_{\varepsilon}^{\left(q_{k}\right)}$ (which is in $F$), one sees that we obtain a finite number of sets $H_{\varepsilon}^{\left(q_{1}\right)}, \ldots, H_{\varepsilon}^{\left(q_{r}\right)}$ formed with elements of $E$ and such that: 1$^{\circ}$ every element of $E$ is interior in the strict sense to at least one of them, 2$^{\circ}$ any two elements of one of them have neighborhood $<\varepsilon$.

It is now easy to generalize the theorem of $\mathrm{n}^{\circ}$ 36.\\

\textbf{42.}
\textsc{\textbf{Theorem.}} --- \textit{Let $E$ be a set of elements of a normal class $(V)$. In order that from every \textsc{countable or non-countable} family $H$\,\footnote{\,In the case where $E$ is a linear set, one obtains the generalization of Mr. \textsc{Borel}'s theorem (see the note in $\mathrm{n}^{\circ}$ 36) extended by Mr. \textsc{Lebesgue} to the case where the family $H$ is non-countable. His proof (\textsc{iv}, page 105) does not generalize to the present case.} of sets $I$ such that every element of $E$ is interior in the strict sense to at least one of them, one can extract a finite number of sets $I$ forming a family $G$ enjoying the same property as $H$, it is necessary and sufficient that $E$ be extremal.}

We have proved ($\mathrm{n}^{\circ}$ 36) that the condition is sufficient in the case where $H$ is countable.

For the general case, let us remark that, if $E$ is extremal, one can form, whatever $\varepsilon$, finitely many sets $K_{\varepsilon}$ satisfying the conditions indicated above. If the theorem is not true for $E$, there is at least one of these sets such that its elements are not interior in the strict sense to a finite number of sets $I$. For $\varepsilon=\cfrac{1}{n}$ let us call this one $K^{(n)}$, and let $A_{n}$ be one of its elements. Since $E$ is extremal, one can extract from the sequence $A_{1}, A_{2}, \ldots$ a sequence $A_{n_{1}}, A_{n_{2}}, \ldots$ that tends to an element $A$ of $E$. $A$ is interior in the strict sense to an interval $I: I_{0}$, and one easily proves that for $p$ large enough every element of $K^{(n_p)}$ is also interior in the strict sense to $I_{0}$, which indeed leads to a contradiction.

The condition is also necessary. Indeed, suppose that $E$ is not closed; there will exist a sequence of elements $B_{1}, B_{2}, \ldots$ of $E$ tending to an element $B$ not belonging to $E$. Now consider the sets $I_{n}$ each formed by the elements $A$ of $E$ such that $(A,B)>\cfrac{1}{n}$. Every element of $E$ is interior in the strict sense to at least one of the $I_{n}$. But if one considers a finite number of these sets $I_{n}$, there will certainly be terms of the sequence $B_{1}, B_{2}, \ldots$ that do not belong to them.

Likewise, suppose that $E$ is not compact; there will exist a sequence $S$ of elements of $E: B_{1}, B_{2}, \ldots$ with no limit elements. If we consider an arbitrary element $A$ of $E$, we can define a number $\rho_{A}>0$ such that $\left(A, B_{n}\right) \ge \rho_{A}$, whatever element $B_{n}$ distinct from $A$ and taken from the sequence $S$; then $A$ will be interior in the strict sense to the set $I_{A}$ of elements $C$ of $E$ such that $(C,A)<\rho_{A}$. It follows that every element $B$ of $E$ is interior in the strict sense to at least one of the intervals $I_{A}$, namely $I_{B}$. If from the family $H$ of intervals $I_{A}$ one extracts a finite number forming a family $G$, it will be impossible that $G$ contain $E$, for in every interval $I_{A}$ there is at most one element of the sequence $S$.\\

\textbf{43.}
\textsc{\textbf{Theorem.}} --- \textit{Every non-countable set $E$ formed of elements of a normal class $(V)$ is condensed.}

Let $E_{1}$ be a non-countable infinity of elements of $E$; it is a matter of proving that $E_{1}$ gives rise to at least one condensation element. Indeed, whatever $\varepsilon$, one can, as we have seen, find a countable sequence of sets $K_{\varepsilon}^{(1)}, K_{\varepsilon}^{(2)}, \ldots, K_{\varepsilon}^{(n)}, \ldots$, such that every element of $E_{1}$ is in one of them and such that the neighborhood of two elements of one of them is $<\varepsilon$. There is therefore at least one of these sets that contains a non-countable infinity of elements of $E_{1}$.
Taking successively $\varepsilon=1, \cfrac{1}{2}, \ldots, \cfrac{1}{n}, \ldots$, we can thus form a sequence of sets $H^{(1)}, H^{(2)}, \ldots$ such that $H^{(n)}$ contains a non-countable infinity of elements of $E_{1}$, is contained in $H^{(n-1)}$, and such that any two elements of $H^{(n)}$ have neighborhood less than $\cfrac{1}{n}$.
Then, if one takes any two elements $A_{n}, A_{n}^{\prime}$ in $H^{(n)}$, each of the sequences $A_{1}, A_{2}, \ldots ; A_{1}^{\prime}, A_{2}^{\prime}, \ldots$ has a limit $A$, and the same one.
Therefore $E_{1}$ gives rise to a limit element $A$, and it is a condensation element, for if one removes from $E_{1}$ a countable infinity of elements, there will still remain elements in each of the $H^{(n)}$, and $A$ will still be a limit element of the remaining part of $E_{1}$.\\

\textbf{44.}
\textsc{\textbf{Remark.}} --- By combining the preceding theorem with that of $\mathrm{n}^{\circ}$ 31, one obtains the following statement:

\textit{
Every closed set $F$ formed of elements of a normal class $(V)$ is the sum of a countable set and of a perfect set or the null set, without common elements.}\\

\textbf{45.}
\textsc{\textbf{Theorem.}} --- \textit{Let $E$ be any set formed of elements of a separable class $(V)$; there exists a countable set of elements of $E$ such that every element of $E$ belongs either to this set $D$ or to its derived set $D^{\prime}$. When $E$ is closed, one has: $E \equiv D+D^{\prime}$. When $E$ is perfect, $E \equiv D^{\prime}$.}

Indeed, whatever $n$, one can form a countable infinity of sets $H_{n}^{(1)}, H_{n}^{(2)}, \ldots, H_{n}^{(p)}, \ldots$ formed of elements of $E$ such that in each of them the neighborhood remains less than $\cfrac{1}{n}$ and such that every element of $E$ is in one of the sets of this sequence. In each of these sets $H_{n}^{(p)}$ there is at least one element $A_{n}^{(p)}$. When one gives $n$ and $p$ arbitrary integer values, one obtains a countable set $D$ of elements $A_{n}^{(p)}$ belonging to $E$. Now it is evident that every element of $E$ belongs to $D$ or to $D^{\prime}$.

The converse is not in general exact; but if $E$ is closed, every element of $D+D^{\prime}$ belongs to $E$; therefore in this case $E \equiv D+D^{\prime}$. Moreover $E^{\prime} \equiv D^{\prime}+D^{\prime\prime} \equiv D^{\prime}$. If even $E$ is perfect, one has $E \equiv E^{\prime}$, whence $E \equiv D^{\prime}$.

Let us moreover observe that, conversely, if $E, D$ are two sets formed of elements of a class $(V)$ such that one has $E \equiv D+D^{\prime}$, the set $E$ is necessarily closed. If $E \equiv D^{\prime}$ ($D$ being a part of $E$), $E$ is perfect.

\subsubsection*{Continuity defined by means of neighborhood.}
\phantomsection\addcontentsline{toc}{subsubsection}{Continuity defined by means of neighborhood.}

\textbf{46.}
\textsc{\textbf{Theorem.}} --- \textit{Let us consider an operation $U$ defined on a set $E$ formed of elements of a class $(V)$. The necessary and sufficient condition for $U$ to be a continuous operation on $E$ is that, if $A$ is any element common to $E$ and $E^{\prime}$, one can make correspond to every number $\varepsilon>0$ a number $\eta_{A}$ such that the inequality $(A,B)<\eta_{A}$ entails $|U(A)-U(B)|<\varepsilon$ for every element $B$ of $E$.}

1$^{\circ}$ Indeed, according to the definition we gave of continuity of $U$, if $U$ is continuous at $A$ and if $A_{1}, A_{2}, \ldots, A_{n}, \ldots$ tend to $A$, $U\left(A_{n}\right)$ tends to $U(A)$. Now let $\varepsilon$ be any positive number; if one could not determine a number $\eta_{A}$ such that the inequality $(A,B)<\eta_{A}$ entails $|U(A)-U(B)|<\varepsilon$, one could determine, whatever $n$, an element $B_{n}$ of $E$ such that:
$$
\left(A, B_{n}\right)<\cfrac{1}{n}, \quad\left|U(A)-U\left(B_{n}\right)\right| \ge \varepsilon .
$$

As $n$ increases indefinitely, the first inequality shows that $B_{n}$ tends to $A$ and, by hypothesis, $U\left(B_{n}\right)$ will tend to $U(A)$, which leads to a contradiction with the second inequality.

$2^{\circ}$ If, whatever $\varepsilon$, one can determine $\eta_{A}$ such that $(A,B)<\eta_{A}$ entails $|U(A)-U(B)|<\varepsilon$, one sees that one can take $p$ large enough so that, if the sequence $B_{1}, B_{2}, \ldots$ tends to $A$, the inequality $n>p$ entails: $\left|U(A)-U\left(B_{n}\right)\right|<\varepsilon$. It suffices to take $p$ large enough so that $n>p$ entails $\left(A,B_{n}\right)<\eta_{A}$.

This proves that $U\left(B_{n}\right)$ tends to $U(A)$.\\

\textbf{47.}
This theorem provides us with a second definition of continuity of an operation. This definition is the one generally adopted for the usual concrete examples. It is moreover less general than the first, since it has meaning only when one can define a neighborhood.

But the introduction of neighborhood is absolutely necessary if one wishes to extend the notion of uniform continuity.

We shall say that an operation $U$ is uniformly continuous on a set $E$ formed of elements of a class $(V)$ if, given a positive number $\varepsilon$, one can choose a positive number $\eta$ such that, for any two elements of $E: A, B$, the inequality
$$
(A, B)<\eta \quad \text { entails } \quad|U(A)-U(B)|<\varepsilon .
$$

It obviously follows from this definition that every operation uniformly continuous on a set $E$ is continuous on $E$. The converse is true in a very general case.

\textsc{\textbf{Theorem.}} --- \textit{Every operation continuous on an \textsc{extremal} set $E$ formed of elements of a class $(V)$ is uniformly continuous on $E$.}

Indeed, in the opposite case, one could find a number $\varepsilon>0$ such that, whatever $n$, there are two elements $A_{n}, B_{n}$ of $E$ satisfying:
$$
\left(A_{n}, B_{n}\right)<\cfrac{1}{n}, \quad\left|U\left(A_{n}\right)-U\left(B_{n}\right)\right| \ge \varepsilon .
$$

From the sequence $A_{1}, A_{2}, \ldots$ one could extract a sequence $A_{p_{1}}, A_{p_{2}}, \ldots$ having as limit an element $A$ of $E$. Then $B_{p_{n}}$ would have the same limit $A$, since, $\left(A_{p_{n}}, A\right)$ and $\left(A_{p_{n}}, B_{p_{n}}\right)$ being infinitesimal with $\cfrac{1}{n}$, so is $\left(A, B_{p_{n}}\right)$. Now one has:
$$
\left|U\left(A_{p_{n}}\right)-U\left(B_{p_{n}}\right)\right| \le\left|U\left(A_{p_{n}}\right)-U(A)\right|+\left|U(A)-U\left(B_{p_{n}}\right)\right| .
$$

The two quantities on the right-hand side tend to zero with $\cfrac{1}{n}$. It is therefore impossible that the left-hand side remain $\ge \varepsilon$.\\

\textbf{48.}
\textbf{\textit{Equally continuous operations.}} --- When one considers continuous operations on a set $E$ formed of elements of a class $(V)$, one can state the conditions defining their equal continuity (\textit{égale continuité}) ($\mathrm{n}^{\circ}$ 15) in a way that brings in neighborhood and is sometimes more convenient.

Let $\mathfrak{F}$ be a family of continuous operations on an arbitrary set $E$ formed of elements of a class $(V)$. If this family is such that, to every number $\varepsilon>0$, one can make correspond a number $\eta>0$ so that one has:
$$
|U(A)-U(B)|<\varepsilon
$$
for every operation $U$ of $\mathfrak{F}$ and for every pair of elements $A, B$ of $E$ satisfying $(A,B)<\eta$, the operations of $\mathfrak{F}$ are equally continuous on $E$. This follows immediately from the definition of $\mathrm{n}^{\circ}$ 15; one sees at the same time that each operation of $\mathfrak{F}$ will be \textit{uniformly} continuous.

Conversely, this condition will be satisfied by every family $\mathfrak{F}$ of equally continuous operations on an \textit{extremal} set $E$ formed of elements of a class $(V)$. Indeed, in the opposite case, one could find a number $\varepsilon>0$ such that, whatever $n$, there exist two elements $A_{n}, B_{n}$ of $E$ and an operation $U_{n}$ of $\mathfrak{F}$ for which:
$$
\left(A_{n}, B_{n}\right)<\cfrac{1}{n}, \quad\left|U_{n}\left(A_{n}\right)-U_{n}\left(B_{n}\right)\right| \ge \varepsilon .
$$

One may even suppose that the sequence $A_{1}, A_{2}, \ldots, A_{n}, \ldots$ has a limit $A$ in $E$, since $E$ being extremal it would suffice otherwise to replace $A_{1}, A_{2}, \ldots$ by a suitably extracted sequence from the preceding one. Then $\left(A,A_{n}\right)$ and $\left(A_{n},B_{n}\right)$ will be infinitesimal with $\cfrac{1}{n}$, hence also $\left(A,B_{n}\right)$. It follows that the sequence $B_{1}, B_{2}, \ldots$ also tends to $A$. But, since the operations of $\mathfrak{F}$ are equally continuous, one can determine two fixed numbers $p, p^{\prime}$ such that one has:
$$
\begin{array}{ll}
\left|U_{q}\left(A_{n}\right)-U_{q}(A)\right|<\cfrac{\varepsilon}{2} & \text { for } n>p \,,\\
\left|U_{q}\left(B_{n}\right)-U_{q}(A)\right|<\cfrac{\varepsilon}{2} & \text { for } n>p^{\prime}\,,
\end{array}
$$
and this whatever the integer $q$. Hence one has, for $n>p+p^{\prime}$,
$$
\left|U_{n}\left(A_{n}\right)-U_{n}\left(B_{n}\right)\right| \le\left|U_{n}\left(A_{n}\right)-U_{n}(A)\right|+\left|U_{n}(A)-U_{n}\left(B_{n}\right)\right|<\cfrac{\varepsilon}{2}+\cfrac{\varepsilon}{2}=\varepsilon,
$$
whence the announced contradiction.

The conditions I have just indicated are those used by Mr. \textsc{Arzelà} to define equal continuity in the cases he considered. One sees that our results are somewhat more extensive, even when considering the same elements as he does, since our definition of $\mathrm{n}^{\circ}$ 15 includes cases where there is no equal continuity in Mr. \textsc{Arzelà}'s sense, and which arise when $E$ is not extremal (see an example in $\mathrm{n}^{\circ}$ 54).

In the case where $E$ is formed of elements of a separable class, the condition imposed on $E$ in the general theorem ($\mathrm{n}^{\circ}$ 19), namely that one has $E \equiv D+D^{\prime}$, $D$ being countable, is satisfied automatically. That theorem therefore becomes:

\textsc{\textbf{Theorem.}} --- \textit{In order that continuous operations on one and the same extremal set $E$, formed of elements of a separable class $(V)$, form a compact family $\mathfrak{F}$, it is necessary and sufficient that the operations of $\mathfrak{F}$ be, at every element of $E$, equally continuous and bounded.
}\\

\textbf{49.}
\textbf{\textit{Introduction of distance (\'ecart).}} --- When we apply the general results of the \textsc{First Part} to concrete examples, we shall first recognize that, in each case, one can associate to every pair of elements $A, B$ a number $(A,B)\ge 0$, which we shall call \textit{the distance (\'ecart) of the two elements}, and which enjoys the following two properties: \textit{a)} The distance (\'ecart) $(A,B)$ is zero only if $A$ and $B$ are identical. \textit{b)} If $A, B, C$ are any three elements, one always has $(A,B)\le (A,C)+(C,B)$.

When one can define the distance (\'ecart) of any two elements of a certain class, we shall say that this is a class $(E)$.

It is easy to see that the number thus defined satisfies the conditions imposed in the definition of neighborhood. Indeed, condition 1$^{\circ}$ of $\mathrm{n}^{\circ}$ 27 is automatically satisfied, and condition 2$^{\circ}$ is satisfied, by taking for example $f(\varepsilon)=2\varepsilon$, if one takes into account the present condition \textit{b)}.

Thus, \textit{distance (\'ecart) is a neighborhood} enjoying a particular property, or, if one prefers, every class $(E)$ is a class $(V)$. In most proofs of known theorems, property \textit{b)} of distance (\'ecart) enters the reasoning. However, the theory developed in this Chapter shows that it is not indispensable and that it is enough to use neighborhood, without thereby having to complicate the reasoning notably.\\

\textbf{50.}
One exception must nevertheless be made for the theorem we are now about to establish; the hypothesis that one is operating on a class $(E)$ does in fact enter essentially into the proof. In spite of this, it is nevertheless likely that the statement remains true for all classes $(V)$.

This theorem is the converse of the theorem we gave at the beginning ($\mathrm{n}^{\circ}$ 11) concerning the maximum of a continuous operation. I call attention both to the mode of proof and to the converse itself, which seems to me in certain cases to give precise form to a general principle stated somewhat vaguely by Mr. \textsc{Hilbert} (\textsc{xx}, page 137):

``Every problem of the calculus of variations possesses a solution, provided certain suitably chosen restrictive hypotheses concerning the nature of the given boundary conditions are fulfilled, and necessarily also provided that what one understands by the word \textit{solution} undergoes a generalization conformable to the meaning, to the nature of things.''\\

\textbf{51.}
\textsc{\textbf{Theorem.}} --- \textit{The necessary and sufficient condition for every operation continuous on a set $E$ of elements of a class $(E)$, $1^{\circ}$ to be bounded on this set, $2^{\circ}$ to attain there its upper limit, is that this set $E$ be extremal.}

We have already proved that the condition is sufficient.

We shall now show that if the set $E$ were not extremal, one could form at least two operations continuous on $E$: one unbounded, the other bounded but not attaining its upper limit. Observe that it is enough to obtain the first $U$, for then one may take for the second:
$$
V \equiv \cfrac{U^{2}}{1+U^{2}},
$$
which never attains its upper limit: $1$.

1$^{\circ}$ First suppose that $E$ is not closed. There is then an element $A$, limit of elements of $E$ and not belonging to $E$. Then it will suffice to take, for the value of $U(B)$ at each element $B$ of $E$:
$$
U(B)=\cfrac{1}{(A, B)}\,.
$$

This operation is obviously unbounded. It is well defined at every element $B$ of $E$; moreover it is continuous, for if $B$ and $C$ are two elements of $E$, one has:
$$
|U(B)-U(C)|=\left|\cfrac{(A, B)-(A, C)}{(A, B)(A, C)}\right| .
$$

Now, taking $B$ fixed and $\varepsilon$ as small as one wishes, and in particular smaller than $(A,B)$, since one has
$$
|(A, C)-(A, B)|<(B, C),
$$
the inequality $(B,C)<\varepsilon$ entails
$$
|(A, B)-(A, C)|<\varepsilon, \quad (A,B)\cdot (A,C)>(A,B)[(A,B)-\varepsilon] .
$$

Therefore, if $\omega$ is a number $>0$ as small as one wishes, it suffices to choose $\varepsilon$ so that
$$
\varepsilon<(A, B) \quad \text { and } \quad \cfrac{\varepsilon}{(A, B)[(A, B)-\varepsilon]}<\omega\,,
$$
for the inequality
$$
(B, C)<\varepsilon\, \text { to entail }\, |U(B)-U(C)|<\omega .
$$

$2^{\circ}$ Let us likewise consider the case where $E$ is non-compact. Then one could find in $E$ an infinite sequence $S$ of distinct elements of $E$
$$
A_{1}, A_{2}, \ldots, A_{n}, \ldots
$$
such that the sequence $S$ has no limit element. To form an operation $U$ continuous and unbounded on this set $E$, we first make a few remarks on the sequence $S$.

Let $\varepsilon_{p}$ be the lower limit of the distances from $A_{p}$ to all the other elements of the sequence $S$. The number $\varepsilon_{p}\ge 0$ cannot be zero, otherwise one could extract from $S$ a sequence tending to $A_{p}$.

We then decompose $E$ into partial sets in the following way. Let $E_{p}$ be the set of elements $B$ of $E$ such that $(B,A_{p})\le \alpha_{p}$, where $\alpha_{p}$ denotes the smaller of the two numbers $\cfrac{1}{p}$ and $\cfrac{\varepsilon_{p}}{3}$. Let now $B$ be an element of $E_{p}$ and $C$ an element of $E_{q}$; one has
\begin{equation*}
\left(A_{p}, A_{q}\right) \le\left(A_{p}, B\right)+\left(C, A_{q}\right)+(B, C) . \tag{I}
\end{equation*}

Now, whatever $p$ and $q$, by the very definition of $\varepsilon_{p}$ and $\varepsilon_{q}$,
$$
\left(A_{p}, A_{q}\right) \ge \varepsilon_{p} \quad \text { and } \quad\left(A_{p}, A_{q}\right) \ge \varepsilon_{q},
$$
whence
$$
\left(A_{p}, A_{q}\right) \ge \cfrac{\varepsilon_{p}+\varepsilon_{q}}{2} .
$$
On the other hand, by the definition of $E_{p}$ and $E_{q}$:
$$
\left(A_{p}, B\right) \le \cfrac{\varepsilon_{p}}{3}, \quad\left(A_{q}, C\right) \le \cfrac{\varepsilon_{q}}{3} ;
$$
the inequality therefore becomes
$$
\cfrac{\varepsilon_{p}+\varepsilon_{q}}{2} \le(B, C)+\cfrac{\varepsilon_{p}+\varepsilon_{q}}{3}.
$$
Hence
$$
(B, C) \ge \cfrac{\varepsilon_{p}+\varepsilon_{q}}{6}.
$$

Thus, the distance between any two elements, one in $E_{p}$, the other in $E_{q}$, remains greater than a fixed positive number. This proves that: $1^{\circ}$ $E_{p}$ and $E_{q}$ have no common element; $2^{\circ}$ no limit element of $E_{p}$ belongs to $E_{q}$, nor conversely. Finally, if one considers a sequence of elements $B_{1}, B_{2}, \ldots$ belonging respectively to sets $E_{p_{1}}, E_{p_{2}}, \ldots$ all distinct, it is impossible for such a sequence to have a limit. Indeed, in that case one could first, if necessary considering only a subsequence extracted from the given sequence, suppose that the indices $p_{1}, p_{2}, \ldots$ increase. Now, calling $C$ the limit of this sequence of elements, one would have
$$
\left(C, A_{p_{n}}\right) \le\left(C, B_{n}\right)+\left(B_{n}, A_{p_{n}}\right)\,,
$$
and since by hypothesis $\left(B_{n}, A_{p_{n}}\right) \le \cfrac{1}{p_{n}}$, one sees that both terms on the right-hand side would tend to zero with $\cfrac{1}{n}$. But this is impossible, since then there would be a sequence $A_{p_{1}}, A_{p_{2}}, \ldots$ extracted from $S$ and tending to a limit $C$.

This being so, we form the required operation $U$ as follows. We take
$$
U(B)=\left[\alpha_{p}-\left(B, A_{p}\right)\right] \cfrac{p}{\alpha_{p}}\,,
$$
when $B$ is any element of $E_{p}$, and we take $U(B)=0$ when $B$ is in none of the sets $E_{1}, E_{2}, \ldots$, i.e. is in the set $G \equiv E-E_{1}-E_{2}-\cdots$. One sees first that $U$ is a well-determined operation at every element of $E$. It is unbounded, since
$$
U\left(A_{p}\right)=p .
$$

Yet it is continuous on $E$. It suffices to show that if distinct elements of $E$: $C_{1}, C_{2}, \ldots$ tend to an element $C$ of $E$, one has
$$
U(C)=\lim _{n=\infty} U\left(C_{n}\right) .
$$

Indeed, observe first that there are not, by what precedes, infinitely many elements of the sequence $C_{1}, C_{2}, \ldots$ belonging to distinct sets $E_{p}$. Then:

$1^{\circ}$ Either $C$ is in $G$, i.e. one has, whatever $p$:
$$
\left(C, A_{p}\right)>\alpha_{p} .
$$
Now:
$$
\left(C_{n}, A_{p}\right) \ge\left(C, A_{p}\right)-\left(C, C_{n}\right)=\alpha_{p}+\left[\left(C, A_{p}\right)-\alpha_{p}-\left(C, C_{n}\right)\right] ;
$$
on the other hand, for $n$ large enough:
$$
\left(C, C_{n}\right)<\left(C, A_{p}\right)-\alpha_{p},
$$
whence, combining these two inequalities:
$$
\left(C_{n}, A_{p}\right)>\alpha_{p} .
$$

Therefore, from some index onward, $C_{n}$ is in $G$. Then one has
$$
U\left(C_{n}\right)=0, \quad U(C)=0,
$$
hence
$$
U(C)=\lim _{n=\infty} U\left(C_{n}\right) ;
$$

$2^{\circ}$ Or else $C$ is in $E_{p}$ and one has $\left(C, A_{p}\right) \le \alpha_{p}$. Then two cases occur:

\textit{a)} $\left(C, A_{p}\right)<\alpha_{p}$; one sees as before that from some index onward $C_{n}$ is constantly in $E_{p}$. Then:
$$
\left|U(C)-U\left(C_{n}\right)\right|=\left|\left(C_{n}, A_{p}\right)-\left(C, A_{p}\right)\right| \cfrac{p}{\alpha_{p}} \le\left(C, C_{n}\right) \cfrac{p}{\alpha_{p}} \,.
$$
Therefore:
$$
U(C)=\lim _{n=\infty} U\left(C_{n}\right) \,.
$$

\textit{b)} $\left(C, A_{p}\right)=\alpha_{p}$. Then, from some index onward, there are in the sequence $C_{1}, C_{2}, \ldots$ only elements of $G$: $C_{n_{1}}, C_{n_{2}}, \ldots$ and elements of $E_{p}$: $C_{q_{1}}, C_{q_{2}}, \ldots$. One sees as before that:
$$
U(C)=\lim _{r=\infty} U\left(C_{q_{r}}\right) \,.
$$
On the other hand:
$$
U(C)=0, \quad U\left(C_{n_{r}}\right)=0,
$$
hence:
$$
U(C)=\lim _{r=\infty} U\left(C_{n_{r}}\right) .
$$
Whence, combining:
$$
U(C)=\lim _{n=\infty} U\left(C_{n}\right) .
$$
Thus $U$ is indeed an unbounded operation, continuous at every element of $E$.

\newpage
\subsection*{SECOND PART. APPLICATIONS OF THE GENERAL THEORY.}
\phantomsection\addcontentsline{toc}{subsection}{SECOND PART. APPLICATIONS OF THE GENERAL THEORY.}

\textbf{52.}
The method to be employed for applying the general theory to particular cases obtained by specifying the nature of the elements indicates itself. In all the examples that follow, we begin by defining the class of elements on which we shall operate; in particular, we give a rule allowing one to distinguish distinct elements, which is not always as simple as one might think (see $\mathrm{n}^{\circ}$ 76--77). Then we show that the ordinary definition of the limit of these elements may be regarded as falling under the definition we gave by means of distance (a particular case of neighborhood (voisinage), $\mathrm{n}^{\circ}$ 49). Finally we prove that the class $(E)$ thus obtained is normal ($\mathrm{n}^{\circ}$ 37). \textit{This done, there remains only to repeat the general theorems of the \textsc{First Part}, stated in the language corresponding to the class considered.} Finally, in each case, we give a necessary and sufficient condition for a set to be compact ($\mathrm{n}^{\circ}$ 9). This amounts to replacing the general definition of compact set by a definition of less abstract appearance, but particular to each case.

\subsection*{Chapter III.}
\phantomsection\addcontentsline{toc}{subsection}{Chapter III.}
\subsubsection*{Linear sets and functions of one variable.}
\phantomsection\addcontentsline{toc}{subsubsection}{Linear sets and functions of one variable.}

\textbf{53.}
Let us take as elements the points of a line (or the numbers representing their abscissas). One knows how to distinguish distinct elements of this class. If one calls the distance between two of them their metric distance, one sees that this definition indeed satisfies conditions \textit{a)}, \textit{b)}, imposed on the notion of distance, $\mathrm{n}^{\circ}$ 49. Moreover, the ordinary definition of the limit of a sequence of points coincides with that deduced from this definition of distance.

The class $(E)$ thus defined is normal. Indeed, Cauchy's theorem ($\mathrm{n}^{\circ}$ 37) was proved precisely in this case. On the other hand, one knows that the set of points of a line may be considered as the derived set of the countable set of points with rational abscissas.

One also knows that every bounded linear set containing infinitely many points gives rise to at least one limit point. Conversely, we proved ($\mathrm{n}^{\circ}$ 35) that every compact set taken from a class $(V)$ is bounded. Hence, in the case of linear sets, the notion of compact set coincides with that of bounded set.

\textit{If now we apply our general theorems, we recover most of the classical theorems of the theory of linear sets and continuous functions of one variable.} (See for example I, II).\\

\textbf{54.}
We should even remark that we thereby obtain slightly more. Indeed, in classical theory one is naturally led to restrict oneself to sets of points of one same interval, and one forgets that unbounded sets sometimes have different properties. On the contrary, we have seen in the general theory that the particular properties of compact sets intervene constantly in proofs, and in an explicit manner.

Thus, for example, consider functions $f_{1}(x), f_{2}(x), \ldots$ continuous \textit{for every value} of $x$ and converging, for every $x$, to $f(x)$. Under what condition is $f(x)$ continuous? Mr. \textsc{Arzelà} ($\mathrm{n}^{\circ}$ 14) did not treat this case; if one generalizes his result without precaution, one might be tempted to say: for $f(x)$ to be continuous, it is necessary and sufficient that, given two positive numbers $\varepsilon$ and $N$, one can find a third $N^{\prime}\ge N$ such that, for every value of $x$, there exists an integer $n_{x}$ satisfying
$$
N \le n_{x} \le N^{\prime}, \quad\left|f_{n_{x}}(x)-f(x)\right|<\varepsilon\,.
$$

Now that would be an error; to see this, it suffices to take $f_{n}(x)=e^{-\frac{n}{1+x^{2}}}$. The difficulty is removed exactly as in the general case. For $f(x)$ to be continuous, it is necessary and sufficient that convergence be quasi-uniform ($\mathrm{n}^{\circ}$ 14) \textit{on every bounded interval}.

Likewise, one might be tempted to generalize another result of Mr. \textsc{Arzelà} ($\mathrm{n}^{\circ}$ 15) by stating thus the necessary and sufficient condition for continuous functions $f(x)$, for all $x$, to be such that from any infinite subcollection one can extract a convergent sequence: it is necessary and sufficient that, $1^{\circ}$ these functions be bounded as a whole, $2^{\circ}$ to every number $\varepsilon>0$ one can make correspond a number $\eta$ such that the inequality $\left|x^{\prime}-x^{\prime\prime}\right|<\eta$ entails, for all functions $f(x)$:
$$
\left|f\left(x^{\prime}\right)-f\left(x^{\prime\prime}\right)\right|<\varepsilon .
$$

Now neither of these two conditions is necessary, as shown by the example of the functions $f_{n}(x)=x^{2}+\cfrac{1}{n}$, where $n$ takes all positive integer values. On the contrary, from our general theorems we can conclude that the necessary and sufficient condition is that the functions be bounded and equally continuous at each point $x$ taken separately (see $\mathrm{n}^{\circ}$ 15). Moreover, if these conditions are satisfied, convergence is uniform on every bounded segment.\\

\textbf{55.}
\textsc{\textbf{Remark.}} --- Most classical theorems on linear sets and functions of one variable extend immediately to sets of points in space of $2,3,\ldots,n,\ldots$ dimensions and to functions of $2,3,\ldots,n,\ldots$ variables. For most of them, this was in fact the only generalization that had been made. The analogy thus observed is much deeper from our general point of view.

To obtain these theorems, it suffices to proceed as above by calling the distance of two elements $\left(x_{1}, x_{2}, \ldots, x_{n}\right)$ and $\left(x_{1}^{\prime}, x_{2}^{\prime}, \ldots, x_{n}^{\prime}\right)$ the quantity
$$
\sqrt{\left(x_{1}-x_{1}^{\prime}\right)^{2}+\cdots+\left(x_{n}-x_{n}^{\prime}\right)^{2}} \,.
$$

\subsection*{Chapter IV.}
\phantomsection\addcontentsline{toc}{subsection}{Chapter IV.}
\subsubsection*{Sets of continuous functions and functionals.}
\phantomsection\addcontentsline{toc}{subsubsection}{Sets of continuous functions and functionals.}

\textbf{56.}
Let us take as variable elements the functions of $x$ uniformly continuous on a fixed interval $J$. We shall of course consider as distinct two functions whose difference is not everywhere zero on $J$.

Here there are two classical definitions of limit; we shall consider only the so-called uniform limit\,\footnote{\,One says that a function $f_n(x)$ tends uniformly to $f(x)$ on $J$ if to every number $\varepsilon>0$ one can make correspond an integer $p$ such that the inequality $n>p$ entails: $|f_n(x)-f(x)|<\varepsilon$ \textit{for every value} of $x$ in the interval $J$ (see $\mathrm{n}^{\circ}$ 13).}. We may regard this definition as a consequence of the following definition of the distance between two functions $f(x), g(x)$ uniformly continuous on $J$: \textit{This distance $(f,g)$ is the maximum of $|f(x)-g(x)|$ on $J$}. This very natural definition seems to have been used systematically for the first time by \textsc{Weierstrass}. It satisfies the general conditions \textit{a)}, \textit{b)} of $\mathrm{n}^{\circ}$ 49. The class $(E)$ thus defined is normal. Indeed, one first knows that it admits a generalization of Cauchy's theorem relative to convergence of a sequence. On the other hand, one can form, once and for all, a sequence $S$ of functions uniformly continuous on $J$:
$$
f_{1}(x), f_{2}(x), \ldots
$$
such that every continuous function on $J$ is its uniform limit. For example, one may take to form the sequence $S$ the evidently countable set of continuous functions that take rational values at each division point of $J$ into $q$ equal parts ($q$ arbitrary) and are linear in each of these divisions.
Then, given a continuous function $f(x)$ on $J$, if one takes in this set the function $f_{n_q}(x)$ that takes at each division point a value equal to the lower approximation of $f(x)$ within $\cfrac{1}{q}$, one easily sees that the sequence $f_{n_{1}}(x), f_{n_{2}}(x), \ldots, f_{n_{q}}(x), \ldots$ extracted from $S$ converges uniformly to $f(x)$\,\footnote{\,One could also proceed as in $\mathrm{n}^{\circ}$ 72 by using the expansion of a continuous function into a series of polynomials (\textsc{ii}, page 50).}.

This remark may be stated as follows, setting
$$
u_{1}(x)=f_{1}(x), \ldots, u_{q}(x)=f_{q}(x)-f_{q-1}(x), \ldots:
$$

\textit{
One can form, once and for all, a series of functions uniformly continuous on $J$:
$$
u_{1}(x)+u_{2}(x)+\cdots+u_{q}(x)+\cdots
$$
such that by grouping its terms suitably one can make it tend uniformly to any continuous function on $J$.}\\

\textbf{57.}
According to our definition of limit, Arzelà's theorem on continuous functions (which is deduced in the previous chapter from our general theorems) allows us to assert that the necessary and sufficient condition for a set of uniformly continuous functions on $J$ to be compact is that these functions be, at each point of $J$, bounded and equally continuous. We now have no further difficulty in applying our general theorems to sets of continuous functions on $J$, \textit{remembering that the limit in question in what follows is the \textsc{uniform} limit.}

\textsc{\textbf{Theorem.}} --- \textit{Given any set $E$ of uniformly continuous functions on $J$, one can extract from $E$ a countable sequence of elements such that every element of $E$ belongs to this sequence or is one of its limit elements} ($\mathrm{n}^{\circ}$ 45). One also knows that the derived set of $E$ is closed\,\footnote{\,We know ($\mathrm{n}^{\circ}$ 22) that this would no longer be so if one took the definition of limit in its ordinary sense, i.e. without supposing it necessarily uniform.}.

\textsc{\textbf{Theorem.}} --- \textit{If this same set $E$ is non-countable, it gives rise to at least one condensation element} ($\mathrm{n}^{\circ}$ 43). Let us note that the general definition of condensation elements can here be translated as follows: an element $f(x)$ is a condensation element of $E$ if, whatever $\varepsilon>0$, there is a non-countable infinity of elements $\varphi(x)$ of $E$ such that $f(x)-\varepsilon<\varphi(x)<f(x)+\varepsilon$.

\textsc{\textbf{Theorem.}} --- \textit{Let $F$ be a closed set of uniformly continuous functions on an interval $J$. $1^{\circ}$ One can always decompose this set into a countable set and a perfect set or the null set} ($\mathrm{n}^{\circ}$ 44). \textit{$2^{\circ}$ If $F$ is part of a set $B$ of functions which are, at each point of $J$, equally continuous and bounded, one can extract from $B$ a sequence of elements $f_{1}(x), f_{2}(x), \ldots$ and form a sequence of positive numbers $\varepsilon_{1}, \varepsilon_{2}, \ldots$ such that one obtains $F$ by removing from $B$ the functions $f(x)$ that satisfy \textsc{one} of the following inequalities:}
$$
f_{1}(x)-\varepsilon<f(x)<f_{1}(x)+\varepsilon, \ldots, f_{n}(x)-\varepsilon<f(x)<f_{n}(x)+\varepsilon, \ldots \quad \left(\mathrm{n}^{\mathrm{o}} \;34\right).
$$

\textsc{\textbf{Theorem.}} --- \textit{Let $B$ be a \textsc{closed} set of functions equally continuous and bounded at every point of an interval $J$. If there exists a family $H$ of sets $I$ of functions of $B$, such that every function of $B$ is interior in the strict sense to at least one of these sets, one can extract from $H$ a finite number of sets $I$ forming a family $G$ enjoying the same property as $H$. The converse is true} ($\mathrm{n}^{\circ}$ 42).\\

\textbf{58.}
By specializing the question solved in $\mathrm{n}^{\circ}$ 39, one obtains the problem posed by Mr. \textsc{Hadamard} in these terms: ``Let $E$ be the set of continuous functions between 0 and 1, having fixed values at the endpoints, or satisfying analogous conditions. Divide $E$ into partial sets $k_{\varepsilon}$ such that any two functions belonging to $E$ have distance less than $\varepsilon$. Considering the $k_{\varepsilon}$ as so many individuals, one may say that the set $\mathfrak{K}_{\varepsilon}$, which has these $k_{\varepsilon}$ as elements, \textit{enumerates} the set $E$. It is this set $\mathfrak{K}_{\varepsilon}$ whose properties, and in particular whose power, should be studied.''

The general theory applied here allows us to give this problem a complete solution as regards power, and indeed it is immaterial whether we assume or not that function values are prescribed at the endpoints. The answer is as follows: \textit{Whatever a set $E$ of uniformly continuous functions on an interval $J$, one can always choose its decomposition into sets $k_{\varepsilon}$ so that $\mathfrak{K}_{\varepsilon}$ is countable} ($\mathrm{n}^{\circ}$ 40). \textit{For $\mathfrak{K}_{\varepsilon}$ not only to be countable but to consist of a finite number of terms whatever $\varepsilon$, it is necessary and sufficient that the functions belonging to $E$ be, at each point of $J$, bounded and equally continuous} ($\mathrm{n}^{\circ}$ 41). In both cases, one can choose the $k_{\varepsilon}$ so that each element of $E$ is interior in the strict sense to at least one of them. Moreover, we know an effective means of determining the $k_{\varepsilon}$ in the general case, independently of $E$.\\

\textbf{59.}
Operations acting on continuous functions are specifically called \textit{functionals} by Mr. \textsc{Hadamard} (\textsc{xi}). An important example of a continuous functional is the linear functional:
$$
U(f)=\int_{a}^{b} m(x) f(x)\, d x
$$
where $m(x)$ is a bounded integrable function on the interval $(a,b)$.

We may again apply here our general theorems on operations.

\textsc{\textbf{Theorem.}} --- Let $U$ be a functional continuous on a closed set $B$ of functions of $x$ that are bounded and equally continuous at every point of an interval $J$. This functional: $1^{\circ}$ is bounded on $B$; $2^{\circ}$ attains there at least once its upper limit ($\mathrm{n}^{\circ}$ 11). \textit{Conversely, if this circumstance occurs for every continuous functional on $B$, $B$ is a closed set of functions equally continuous and bounded at every point of $J$} ($\mathrm{n}^{\circ}$ 51).

\textsc{\textbf{Theorem.}} --- For a sequence of functionals $U_{1}, U_{2}, \ldots$, continuous on an arbitrary field $E$ of functions of $x$ uniformly continuous on $J$, to converge to a continuous functional on $E$, it is necessary and sufficient that convergence be quasi-uniform on every closed set formed of equally continuous and bounded functions belonging to $E$ ($\mathrm{n}^{\circ}$ 14).

\textsc{\textbf{Theorem.}} --- Consider a family $\mathfrak{K}_\varepsilon$ of functionals continuous on a closed set of functions equally continuous and bounded at every point of an interval $J$. In order that from every infinite subset of functionals of $\mathfrak{K}_\varepsilon$ one can extract a sequence that converges uniformly, it is necessary and sufficient that the functionals of $\mathfrak{K}_\varepsilon$ be bounded and equally continuous at every element of the function set ($\mathrm{n}^{\circ}$ 48).

\subsection*{Chapter V.}
\phantomsection\addcontentsline{toc}{subsection}{Chapter V.}
\subsubsection*{Sets and functions of points of space of countably infinite dimension.}
\phantomsection\addcontentsline{toc}{subsubsection}{Sets and functions of points of space of countably infinite dimension.}

\textbf{60.}
A great number of the elements that occur in mathematics are each completely determined by an infinite sequence of real or complex numbers. For example, a \textsc{Taylor} series is determined by the sequence of its coefficients in the form:
\begin{equation*}
f(x)=a_{0}+a_{1}x+\cdots+a_{n}(x)+\cdots \tag{\textsc{i}}
\end{equation*}

A continuous function of $x$ between $0$ and $1$ is determined by the sequence of values it takes at points with rational abscissas, by Mr. \textsc{Borel}'s interpolation formula (\textsc{ii}, p.~80):
$$
\varphi(x)=\lim _{q=\infty} \sum_{p=0}^{p=q} \varphi\left(\cfrac{p}{q}\right) M_{p, q}(x) \,.
$$

Likewise, an irrational number is defined by a sequence of rational numbers; and so on.

One may therefore regard the numbers in the sequence that defines each of these elements as the coordinates of that element considered as a point of a space $(E_{\omega})$ of countably infinite dimension. There are several advantages in proceeding thus. First, the advantage that always appears when one uses geometric language, so conducive to intuition through the analogies it gives rise to. Next, this allows one to obtain general properties applicable to all these categories of elements: those resulting not from their own nature, but from their determination by a countable set of numbers. It remains understood, moreover, that not every infinite sequence of numbers necessarily corresponds to one of these elements. For example, expression (\textsc{i}) defines $f(x)$ only if the sequence $a_{0}, a_{1}, \ldots, a_{n}, \ldots$ satisfies certain conditions.\\

\textbf{61.}
We shall now therefore abstract from the origin of the sequence of numbers that defines one of these elements, and simply consider this sequence:
$$
x_{1}, x_{2}, \ldots, x_{n}, \ldots
$$
as defining a determined point $x$ of space.

We shall say that two points $x, x^{\prime}$ of $(E_{\omega})$ coincide when their coordinates are respectively equal, and only in this case.

In order to apply in the present case the general theorems obtained above, it will suffice to show that \textit{the class of points of $(E_{\omega})$ defines a normal class $(E)$.
}\\

\textbf{62.}
\textit{\textbf{Distance between two points of}} $\left(E_{\omega}\right)$. --- One could define in many ways the distance between two points:
$$
\begin{array}{cll}
x: & x_{1}, & x_{2}, \ldots, x_{n}, \ldots \\
x^{\prime}: & x_{1}^{\prime}, & x_{2}^{\prime}, \ldots, x_{n}^{\prime}, \ldots
\end{array}
$$

But in order to generalize usefully the case of space of finite dimension, one must arrange things so that, if the distance $(x,x^{\prime})$ tends to zero, each of the quantities $x_{1}-x_{1}^{\prime},\, x_{2}-x_{2}^{\prime},\, \ldots, x_{n}-x_{n}^{\prime}, \ldots$ tends to zero.

An important question arises, namely whether it would be useful to choose this definition so that this infinite sequence of differences tends \textit{uniformly} to zero. We shall see later ($\mathrm{n}^{\circ}$ 68) that it is preferable not to impose this condition, natural though it may seem. Be that as it may, we shall adopt the following definition of distance, whose choice will be justified later.

We shall call the distance of the two points $x, x^{\prime}$ the well-defined number:
$$
\left(x, x^{\prime}\right)=\cfrac{\left|x_{1}-x_{1}^{\prime}\right|}{1+\left|x_{1}-x_{1}^{\prime}\right|}+\cfrac{1}{2!}\, \cfrac{\left|x_{2}-x_{2}^{\prime}\right|}{1+\left|x_{2}-x_{2}^{\prime}\right|}+\cdots+\cfrac{1}{n!}\, \cfrac{\left|x_{n}-x_{n}^{\prime}\right|}{1+\left|x_{n}-x_{n}^{\prime}\right|}+\cdots
$$

The right-hand side is obviously convergent whatever $x_{1}, x_{1}^{\prime}; x_{2}, x_{2}^{\prime}; \ldots$ One sees that it is always $>0$ whenever $x$ and $x^{\prime}$ are distinct. Moreover, for every $n$:
$$
\cfrac{\left|x_{n}^{\prime}-x_{n}^{\prime \prime}\right|}{1+\left|x_{n}^{\prime}-x_{n}^{\prime \prime}\right|} \le \cfrac{\left|x_{n}-x_{n}^{\prime}\right|}{1+\left|x_{n}-x_{n}^{\prime}\right|}+\cfrac{\left|x_{n}-x_{n}^{\prime \prime}\right|}{1+\left|x_{n}-x_{n}^{\prime \prime}\right|}\,,
$$
as is easy to verify.

Hence for any three points $x, x^{\prime}, x^{\prime\prime}$:
$$
\left(x^{\prime}, x^{\prime \prime}\right) \le\left(x, x^{\prime}\right)+\left(x, x^{\prime \prime}\right)\,.
$$

Thus the number we have just defined does satisfy the two general conditions we had imposed on the definition of distance ($\mathrm{n}^{\circ}$ 49).\\

\textbf{63.}
\textit{\textbf{Limit of a point of}} $\left(E_{\omega}\right)$. --- From the definition of distance follows the definition of the limit of a point in space $(E_{\omega})$.

One says that a point $x^{(n)}$ of space $(E_{\omega})$ tends to the point $x$ of $(E_{\omega})$, as $n$ increases indefinitely, if the distance $\left(x^{(n)},x\right)$ tends to zero with $\cfrac{1}{n}$. But one may give an equivalent definition involving the coordinates directly, by means of the following theorem:

\textsc{\textbf{Theorem.}} --- \textit{The necessary and \textsc{sufficient} condition for a point of $(E_{\omega})$}
$$x^{(n)}:  x_{1}^{(n)}, \, x_{2}^{(n)}, \ldots, x_{p}^{(n)}, \ldots$$
\textit{to tend to a point}
$$x: x_{1}, \, x_{2}, \ldots, x_{p}, \ldots,$$
\textit{is that, for every $p$, $x_{p}^{(n)}$ tends to $x_{p}$ as $n$ increases indefinitely.}

Indeed, suppose first that this latter condition is satisfied. I shall show that $(x,x^{(n)})$ tends to zero. For this, observe that:
\begin{equation*}
\left(x, x^{(n)}\right) \le \cfrac{\left|x_{1}-x_{1}^{(n)}\right|}{1+\left|x_{1}-x_{1}^{(n)}\right|}+\cdots+\cfrac{1}{p!}\, \cfrac{\left|x_{p}-x_{p}^{(n)}\right|}{1+\left|x_{p}-x_{p}^{(n)}\right|}+r_{p}, \tag{2}
\end{equation*}
where $r_{p}$ denotes the expression independent of $n$:
$$
r_{p}=\cfrac{1}{(p+1)!}+\cfrac{1}{(p+2)!}+\cdots
$$

This last quantity tends to zero with $\cfrac{1}{p}$; one can therefore fix a number $p$ independent of $n$ such that $r_{p}<\cfrac{\varepsilon}{2}$. With $p$ thus \textit{fixed}, it is clear that the sum of the $p$ preceding terms in inequality (2) tends to zero as $n$ increases indefinitely. One can therefore find a number $q$ such that their sum is less than $\cfrac{\varepsilon}{2}$ for $n>q$. In summary, we obtain a number $q$ such that:
$$
\left(x^{(n)}, x\right)<\varepsilon, \quad \text { for } n>q ;
$$
which proves that the assumed condition was indeed sufficient.

Let us prove it is necessary. The expression of $(x,x^{(n)})$ shows that, for every value of $p$:
$$
\cfrac{\left|x_{p}-x_{p}^{(n)}\right|}{1+\left|x_{p}-x_{p}^{(n)}\right|}<p!\left(x, x^{(n)}\right) .
$$

Hence the left-hand side tends to zero when, with $p$ fixed, $n$ increases indefinitely. This can happen only if $x_{p}^{(n)}$ tends to $x_{p}$ for every fixed value of $p$.\\

\textbf{64.}
This proposition shows that the class $(E_{\omega})$ admits a generalization of Cauchy's theorem:

\textsc{\textbf{Theorem.}} ---
\textit{The necessary and sufficient condition for an infinite sequence of points of space $(E_{\omega})$:}
$$
x^{(1)}, \quad x^{(2)}, \ldots, x^{(n)}, \ldots
$$
\textit{to have a limit is that, for every number $\varepsilon>0$, one can determine a number $n$ such that one has}
$$
\left(x^{(n+p)}, x^{(n)}\right)<\varepsilon,
$$
\textit{whatever the integer $p$.}

The condition is obviously necessary.

Conversely, suppose the condition satisfied. One then has:
$$
\cfrac{\left|x_{q}^{(n)}-x_{q}^{(n+p)}\right|}{1+\left|x_{q}^{(n)}-x_{q}^{(n+p)}\right|}<\varepsilon q!
$$
Whence:
$$
\left|x_{q}^{(n)}-x_{q}^{(n+p)}\right|<\cfrac{\varepsilon q!}{1-\varepsilon q!}\,,
$$
assuming $\varepsilon<\cfrac{1}{q!}$. Once $q$ is fixed, one can choose $\varepsilon$ so that
$$
\varepsilon<\cfrac{1}{q!} \quad \text { and } \quad \cfrac{\varepsilon q!}{1-\varepsilon q!}<\omega\,,
$$
$\omega$ being given in advance; then $\varepsilon$ determines $n$, and one has:
$$
\left|x_{q}^{(n)}-x_{q}^{(n+p)}\right|<\omega,
$$
whatever $p$. By a well-known theorem of \textsc{Cauchy}, it follows that the sequence of numbers:
$$
x_{q}^{(1)}, \quad x_{q}^{(2)}, \ldots, x_{q}^{(n)}, \ldots
$$
has a limit $x_{q}$ as $n$ increases indefinitely. And this holds for every $q$. By the preceding theorem, this proves that the sequence $x^{(1)}, \ldots, x^{(n)}, \ldots$ has a determined limit $x: x_1, \ldots, x_q, \ldots$

\subsubsection*{Sets of points of space $(E_{\omega})$.}
\phantomsection\addcontentsline{toc}{subsubsection}{Sets of points of space \texorpdfstring{$(E_{\omega})$}.}

\textbf{65.}
From the fact that one can define the distance between two points of space $(E_{\omega})$ of countably infinite dimension, it follows that the points of this space form a class $(E)$. We further complete the assimilation to the general theory by the following theorem:

\textsc{\textbf{Theorem.}} --- \textit{The points of space $(E_{\omega})$ form a normal class $(E)$.}

We already know they form a class $(E)$ admitting a generalization of Cauchy's theorem on convergent sequences. This class $(E)$ is evidently perfect. It remains only to show that it may be regarded as the derived set of a countable sequence of its elements. For this, note first that the set of rational numbers of either sign is a countable set that can be written as a sequence:
$$
c_{1}, \, c_{2}, \ldots, c_{n}, \ldots
$$

Consider then the set $D$ of points of space $(E_{\omega})$ having coordinates of the form:
$$
x_{1}=c_{n_{1}}, \quad x_{2}=c_{n_{2}}, \ldots, x_{p}=c_{n_{p}}, \quad x_{p+1}=0, \quad x_{p+2}=0, \ldots,
$$
where $p, n_{1}, n_{2}, \ldots, n_{p}$ are arbitrary positive integers. This set is obviously countable, and one easily sees that every point of $(E_{\omega})$ is the limit of a sequence of elements of $D$.\\

\textbf{66.}
\textit{\textbf{Compact sets.}} --- The notion of compact set takes on a very simple geometric meaning in the case of a set of points of $(E_{\omega})$. To indicate this, we first say that \textit{a set of points of $(E_{\omega})$ is contained in a finite domain} if one can find a sequence of numbers $\ge 0$: $M_{1}, M_{2}, \ldots$ such that for every point $(x_{1}, \ldots, x_{n}, \ldots)$ of the set:
$$
\left|x_{1}\right| \le M_{1}, \quad\left|x_{2}\right| \le M_{2}, \ldots,\left|x_{n}\right| \le M_{n}, \ldots
$$

\textsc{\textbf{Theorem.}} --- \textit{The necessary and sufficient condition for a set $E$ of points of space $(E_{\omega})$ to be compact is that this set be contained in a finite domain.}

Indeed, suppose $E$ compact; if it were not in a finite domain, one could find $n$ such that the set of coordinates of rank $n$ of points of $E$ is unbounded. In other words, one would have that, for every $p$, one can find a point $x^{(p)}$ of $E$ such that, calling $x_{n}^{(p)}$ its coordinate of rank $n$, one has:
$$
\left|x_{n}^{(p)}\right|>p .
$$

Now this is impossible if $E$ is compact; for from the sequence of points of $E$: $x^{(1)}, x^{(2)}, \ldots$ one could extract a sequence $x^{\left(p_{1}\right)}, x^{\left(p_{2}\right)}, \ldots$ having a limit $x$. And then the numbers $x_{n}^{\left(p_{1}\right)}, x_{n}^{\left(p_{2}\right)}, \ldots$ which increase indefinitely like $p_{1}, p_{2}, \ldots$ would have a finite limit $x_{n}$.

Conversely, suppose $E$ is in a finite domain. Then either there is only a finite number of elements, and in that case $E$ is compact by definition, or there is an infinite number of distinct elements. Consider an infinite set $E_{1}$ contained in $E$. The set of its coordinates of rank $1$ is bounded by hypothesis. Therefore: either infinitely many of them are equal to one same number $x_{1}$, or one can find infinitely many that are distinct and tend to a determined number $x_{1}$. In both cases, one can extract from $E_{1}$ a sequence of points of $\left(E_{\omega}\right)$:
$$
x^{(1)}, \quad x^{(2)}, \ldots, x^{(n)}, \ldots,
$$
such that the sequence $x_{1}^{(1)}, x_{1}^{(2)}, \ldots, x_{1}^{(n)}, \ldots$ has a determined limit $x_{1}$. From this first sequence, one can extract another sequence of points arranged in the same order:
$$
x^{(1,2)}, \quad x^{(2,2)}, \ldots, x^{(n, 2)}, \ldots,
$$
such that the second-rank coordinates have a determined limit $x_{2}$, while first-rank coordinates continue tending to $x_{1}$. And so on. One arrives at forming a sequence of points of $E_{1}$:
$$
x^{(1, p)}, \quad x^{(2, p)}, \ldots, x^{(n, p)}, \ldots
$$
contained in the sequences previously formed and such that coordinates of rank at most $p$ have respectively determined limits: $x_{1}, x_{2}, \ldots, x_{p}$.

Now consider the sequence of distinct points of $E_{1}$:
\begin{equation*}
x^{(1)}, \quad x^{(2,2)}, \ldots, x^{(n, n)}, \ldots \tag{$\Sigma$}
\end{equation*}
From rank $p$ onward, and this for every $p$, this sequence is formed of elements of the sequence $x^{(1,p)}, x^{(2,p)}, \ldots, x^{(n,p)}, \ldots$ taken in the same order. Therefore its coordinates of rank $p$ tend, for every $p$, to a determined number $x_{p}$. In other words, sequence ($\Sigma$) is a sequence of elements of $E_{1}$ that are distinct and have a determined limit: the point of coordinates
$$
x_{1}, \quad x_{2}, \ldots, x_{p}, \ldots
$$

The proposition is thus proved\,\footnote{\,Mr. \textsc{Montel} has long possessed the sufficient condition, which served him to establish certain propositions stated in a Note in the Comptes Rendus (\textsc{xxv}).}.

\subsubsection*{Applications of the general theorems.}
\phantomsection\addcontentsline{toc}{subsubsection}{Applications of the general theorems.}

\textbf{67.}
There is now no difficulty in applying to sets of points of space $(E_{\omega})$ the theorems stated for sets drawn from a normal class $(E)$.

\textsc{\textbf{Theorem.}} --- \textit{Consider an arbitrary set $E$ of points of space $(E_{\omega})$ of countably infinite dimension: $1^{\circ}$ one can extract from $E$ a countable set of points, $D$, such that every point of $E$ belongs either to $D$ or to its derived set $D^{\prime}$; $2^{\circ}$ the derived set $E^{\prime}$ of $E$ is closed; $3^{\circ}$ the set $E$ is countable or else gives rise to at least one condensation point.}

\textsc{\textbf{Theorem.}} --- \textit{If a set $E$ of points of space $\left(E_{\omega}\right)$ is in a finite domain, the set of those of its points that are not limit points of $E$ is countable.}

\textsc{\textbf{Theorem.}} --- \textit{Let $E$ be a set of points of space $(E_{\omega})$ that is closed and contained in a finite domain $D$: $1^{\circ}$ there exists a countable set $G$ of spheroids such that one obtains $E$ by removing from $D$ the points interior to each of these spheroids; $2^{\circ}$ $E$ can be decomposed into a countable set of points and a perfect set or null set without common points with the former; $3^{\circ}$ there exists at least one set $K$ of points of $(E_{\omega})$ whose derived set coincides with $E$.}

\textsc{\textbf{Theorem.}} --- \textit{For a set $E$ of points of space $(E_{\omega})$ to be closed and contained in a finite domain, it is necessary and sufficient that from every family $G$ of sets $I$ of points of $E$, such that every point of $E$ is interior in the strict sense to at least one of the sets $I$, one can extract a finite number of these sets $I$ forming a family $H$ enjoying the same property as $G$.}

\textsc{\textbf{Theorem.}} --- \textit{Every perfect set of points of $(E_{\omega})$ has the power of the continuum.}

Here one may also solve a question analogous to that posed by Mr. \textsc{Hadamard} for sets of continuous functions ($\mathrm{n}^{\circ}$ 39).

Let $E$ be a set of points of space $(E_{\omega})$; decompose $E$ into partial sets $k_{\varepsilon}$ such that any two points of one of them have distance less than $\varepsilon$, and let $\mathfrak{K}_\varepsilon$ be the set of the $k_{\varepsilon}$.
\textit{Whatever $E$ and $\varepsilon$, one can always carry out this decomposition so that $\mathfrak{K}_\varepsilon$ is countable. For it to be possible, whatever $\varepsilon$, to decompose $E$ into a finite number of sets $k_{\varepsilon}$, it is necessary and sufficient that $E$ be contained in a finite domain.}\\

\textbf{68.}
\textsc{\textbf{Remark.}} --- This is the place to observe that one might have been led to choose another definition of the limit of a sequence of points of space $(E_{\omega})$. At first sight, it seems that it would have been more natural, while also better suited to applications, to adopt the following definition: a sequence of points $x^{(1)}, x^{(2)}, \ldots, x^{(n)}, \ldots$ tends to a point $x$ when, given a number $\varepsilon>0$, one can find an integer $q$ such that:
$$
\left|x_{p}^{(n)}-x_{p}\right|<\varepsilon, \quad \text { for } \quad n>q,
$$
whatever the rank $p$ of the coordinates considered of $x^{(n)}$ and of $x$.

There would certainly be reason to study a theory of sets and continuity in space $\left(E_{\omega}\right)$ founded on this definition, which may sometimes present advantages (\textsc{xviii}). However, it lends itself less well to applications of our general theorems. Indeed, one sees that in order to deduce this definition from a suitable definition of distance, one is led to call the distance between two points $x, x^{\prime}$ the upper limit of $\left|x_{p}-x_{p}^{\prime}\right|$ when $p$ takes all possible integer values. This definition seems much simpler than the one we adopted; yet it has the drawback of not applying to any two arbitrary elements. More exactly, it could give an infinite value to the distance between two points whose coordinates are finite (but unbounded).

Another drawback of this method is that one no longer sees clearly what geometric meaning should be attached to the notion of compact set. Indeed, it is easy to form an example of a set of points of $(E_{\omega})$ that is contained in a finite domain and yet has no limit element in the sense of the second definition. One might perhaps hope that there is a limit element when one requires the set to be not merely infinite but non-countable. Or again by restricting the definition we gave of a finite domain and saying that the domain will be finite if all coordinates are, in absolute value, less than one \textit{same} fixed number. None of these restrictions can, however, restore validity to the theorem in question. It suffices indeed to consider the set of all distinct points whose coordinates are always equal to $0$ or $1$. This set has the power of the continuum (it is therefore not countable), and all coordinates of its points are bounded as a whole; yet it has no limit point in the sense of the second definition. On the contrary, it has limit points in the sense of the first.

\subsubsection*{Functions of an infinite sequence of independent variables.}
\phantomsection\addcontentsline{toc}{subsubsection}{Functions of an infinite sequence of independent variables.}

\textbf{69.}
In accordance with our general definitions, we shall say that a function of an infinite sequence of variables $x_{1}, x_{2}, \ldots, x_{n}, \ldots$:
$$
f\left(x_{1}, x_{2}, \ldots, x_{n}, \ldots\right)
$$
is continuous if $f\left(x_{1}^{(p)}, x_{2}^{(p)}, \ldots, x_{n}^{(p)}, \ldots\right)$ has as limit $f\left(x_{1}, x_{2}, \ldots, x_{n}, \ldots\right)$ when $x_{1}^{(p)}, x_{2}^{(p)}, \ldots, x_{n}^{(p)}, \ldots$ tend simultaneously, but independently, to $x_{1}, x_{2}, \ldots, x_{n}, \ldots$ as $p$ increases indefinitely. Here again we may apply our general theorems:

\textsc{\textbf{Theorem.}} --- \textit{If a function of a point of space $(E_{\omega})$ of countably infinite dimension is continuous on a bounded and closed set of points of this space, this function: $1^{\circ}$ is bounded; $2^{\circ}$ attains at least at one point of $E$ its upper limit. Conversely, if this property holds for all continuous functions on a set $E$, this set is bounded and closed.}

\textsc{\textbf{Theorem.}} --- \textit{Every function continuous on an extremal set of points of space $(E_{\omega})$ is uniformly continuous on this set.}

\textsc{\textbf{Theorem.}} --- \textit{The necessary and sufficient condition for a series of functions continuous at every point of an extremal set of points of space $(E_{\omega})$ to converge to a continuous function is that the convergence of this series be quasi-uniform.}

\textsc{\textbf{Theorem.}} --- \textit{Let $\mathfrak{K}_\varepsilon$ be a family of continuous functions on a bounded and closed set of points of space $(E_{\omega})$. In order that from every infinite subset of functions of $\mathfrak{K}_\varepsilon$ one can extract a sequence that converges uniformly, it is necessary and sufficient that the functions of $\mathfrak{K}_\varepsilon$ be equally continuous and bounded at every point of the set.}

\subsection*{Chapter VI.}
\phantomsection\addcontentsline{toc}{subsection}{Chapter VI.}
\subsubsection*{Holomorphic functions inside the same domain.}
\phantomsection\addcontentsline{toc}{subsubsection}{Holomorphic functions inside the same domain.}

\textbf{70.}
Let us consider, in the complex $z$-plane, an arbitrary domain $\mathfrak{A}$ with simple or multiple boundary. We shall study sets whose elements are functions holomorphic at every point \textit{interior} to $\mathfrak{A}$. To avoid restrictions in the statements we are going to obtain, we shall not concern ourselves with what happens on the boundary of $\mathfrak{A}$. In other words, we shall regard as identical two functions holomorphic inside $\mathfrak{A}$ that have the same value at every interior point of $\mathfrak{A}$.

To apply our general theorems, we shall show that one can define a distance. Suppose one has determined a sequence of bounded domains $\mathfrak{A}_{1}, \mathfrak{A}_{2}, \ldots, \mathfrak{A}_{n}, \ldots$, each interior to the next and to $\mathfrak{A}$, in such a way that every domain $\mathfrak{A}'$ interior to $\mathfrak{A}$ is also interior to the domains of this sequence from some index onward\,\footnote{\,For example, if $\mathfrak{A}$ is bounded by a simple convex contour, one would take for $\mathfrak{A}_{n}$ the homothetic image of $\mathfrak{A}$ with ratio $\cfrac{n}{n+1}$ with respect to a fixed point taken independently of $n$ inside $\mathfrak{A}$.}. To define the distance between two functions $f$ and $g$ holomorphic inside $\mathfrak{A}$, let $u_{n}$ be the finite maximum of $|f(z)-g(z)|$ when $z$ varies in $\mathfrak{A}_{n}$ or on its boundary. The series
$$
\cfrac{u_{1}}{1+u_{1}}+\cfrac{1}{2!}\, \cfrac{u_{2}}{1+u_{2}}+\cdots+\cfrac{1}{n!}\, \cfrac{u_{n}}{1+u_{n}}+\cdots
$$
is certainly convergent; its sum shall by definition be the distance $(f,g)$ of $f$ and $g$ in $\mathfrak{A}$. One sees that if $|f(z)-g(z)|$ has a finite upper bound $u$ when $z$ varies arbitrarily inside $\mathfrak{A}$, the distance will be less than $\cfrac{eu}{1+u}$. But the distance has a finite value $<e$ even in the opposite case.

Using the remarks made for points of space $(E_{\omega})$ in $\mathrm{n}^{\circ}$ 62, one again sees that the definition just given does satisfy the conditions imposed on distance ($\mathrm{n}^{\circ}$ 49).

In accordance with the general theory, we shall therefore say that, given two functions holomorphic inside $\mathfrak{A}$: $f(z)$ and $f_{n}(z)$, the function $f(z)$ is the limit of $f_{n}(z)$ if their distance tends to zero with $\cfrac{1}{n}$.

To understand the meaning of this definition, it suffices to note that \textit{the necessary and sufficient condition for the distance $(f,f_{n})$ to tend to zero is that $f_{n}(z)$ tend uniformly to $f(z)$ in every domain $\mathfrak{A}'$ interior to $\mathfrak{A}$.}

Indeed, if this condition is fulfilled, the maximum $u_{p}^{(n)}$ of $|f(z)-f_{n}(z)|$ in $\mathfrak{A}_{p}$ tends to zero with $\cfrac{1}{n}$, whatever fixed integer $p$. By an argument used above ($\mathrm{n}^{\circ}$ 63), this suffices for the distance
$$
\left(f, f_{n}\right)=\cfrac{u_{1}^{(n)}}{1+u_{1}^{(n)}}+\cdots+\cfrac{1}{p!}\, \cfrac{u_{p}^{(n)}}{1+u_{p}^{(n)}}+\cdots
$$
also to have limit zero. Conversely, if the distance tends to zero, the same is true of $u_{p}^{(n)}$ for every fixed integer $p$, by the inequality
$$
\cfrac{u_{p}^{(n)}}{1+u_{p}^{(n)}}<p!\left(f, f_{n}\right)\,.
$$
Hence $f_{n}(z)$ tends uniformly to $f(z)$ in each domain $\mathfrak{A}_{p}$ and therefore in every domain $\mathfrak{A}'$ interior to $\mathfrak{A}$. But it is not certain that the convergence of $f_{n}(z)$ to $f(z)$ is uniform \textit{in} $\mathfrak{A}$. One would see an example by taking for $\mathfrak{A}$ the disk $|z|\le1$ and setting $f_{n}(z)\equiv z^{n}$, $f(z)\equiv0$. It nevertheless seems preferable to adopt our definition rather than the one according to which one would regard $f(z)$ as limit of $f_{n}(z)$ only when $f_{n}(z)$ tends to $f(z)$ uniformly \textit{throughout the whole domain} $\mathfrak{A}$. For one would be led not always to regard a function such as $\cfrac{1}{1-z}$ as the limit of its \textsc{Taylor} expansion in its disk of convergence.\\

\textbf{71.}
Thus, we shall establish our theory of sets of functions holomorphic inside $\mathfrak{A}$ by considering one function as the limit of another when the value of the second tends to the first uniformly inside every domain interior to $\mathfrak{A}$.

In order to apply our general theorems of the \textsc{First Part}, we shall prove that \textit{the class of functions holomorphic inside a same domain $\mathfrak{A}$ is a normal class $(E)$.}

We know it can be viewed as a class $(E)$. Let us show it admits a generalization of \textsc{Cauchy}'s theorem.

\textsc{\textbf{Theorem.}} --- \textit{For a sequence of functions}
$$
f_{1}(z), \quad f_{2}(z), \ldots, f_{n}(z), \ldots
$$
\textit{holomorphic inside a domain $\mathfrak{A}$ to be uniformly convergent in every domain $\mathfrak{A}'$ completely interior to $\mathfrak{A}$, it is necessary and sufficient that to every number $\varepsilon>0$ one can make correspond an integer $n$ such that, whatever $p$:}
$$
\left(f_{n}, f_{n+p}\right)<\varepsilon\,.
$$

There is no difficulty in showing the condition is necessary. To show it is sufficient, let $u_{q}^{(n,r)}$ be the maximum modulus of $|f_{n}(z)-f_{r}(z)|$ in $\mathfrak{A}_{q}$. Now take a domain $\mathfrak{A}'$ interior to $\mathfrak{A}$. One can take $q$ large enough so that $\mathfrak{A}'$ is also completely interior to $\mathfrak{A}_{q}$. With this number $q$ fixed, choose $\varepsilon>0$ less than $\cfrac{\omega}{q!(1+\omega)}$.

By hypothesis, one can find an integer $n$ such that
$$
\left(f_{n+p}, f_{n}\right)<\varepsilon,
$$
whatever $p$. Now,
$$
\cfrac{1}{q!}\, \cfrac{u_{q}^{(n,n+p)}}{1+u_{q}^{(n,n+p)}}<\left(f_{n}, f_{n+p}\right).
$$
Hence
$$
\cfrac{1}{q!}\, \cfrac{u_{q}^{(n,n+p)}}{1+u_{q}^{(n,n+p)}}<\cfrac{1}{q!}\, \cfrac{\omega}{1+\omega}\,,
$$
or
$$
u_{q}^{(n,n+p)}<\omega\,.
$$
In other words, given $\omega>0$, one can find an integer $n$ such that at every point of $\mathfrak{A}_{q}$:
$$
\left|f_{n}(z)-f_{n+p}(z)\right|<\omega,
$$
whatever integer $p$. Therefore the sequence $f_{1}(z), f_{2}(z), \ldots$ is indeed uniformly convergent in $\mathfrak{A}'$.\\

\textbf{72.}
Let us now prove that functions holomorphic inside the same domain $\mathfrak{A}$ form a separable class.

\textsc{\textbf{Theorem.}} --- \textit{One can determine, once and for all, a sequence $S$ of polynomials $P_{1}(z), P_{2}(z), \ldots, P_{n}(z), \ldots$ such that every function holomorphic inside $\mathfrak{A}$ is the limit of at least one sequence extracted from $S$, and this uniformly in every bounded domain $\mathfrak{A}'$ interior to $\mathfrak{A}$.}

We rely on the well-known theorem that every function $f(z)$ holomorphic inside $\mathfrak{A}$ is the sum of a series of polynomials (with uniform convergence in every domain $\mathfrak{A}'$ interior to $\mathfrak{A}$). Then, whatever $n$, one can determine a polynomial $Q_{n}(z)$ such that \textit{throughout the whole domain} $\mathfrak{A}_{n}$, boundary included,
$$
\left|f(z)-Q_{n}(z)\right|<\cfrac{1}{2n}\,.
$$

But by replacing the real and imaginary parts of the coefficients of $Q_{n}(z)$ by rational approximations within $\cfrac{1}{q}$, one can always form (for $q$ large enough) a polynomial $R_{n}(z)$ with rational coefficients such that in the bounded domain $\mathfrak{A}_{n}$ and on its boundary,
$$
\left|Q_{n}(z)-R_{n}(z)\right|<\cfrac{1}{2n}\,.
$$

Hence one has $|f(z)-R_{n}(z)|<\cfrac{1}{n}$ in $\mathfrak{A}_{n}$, and therefore $R_{n}(z)$ tends uniformly to $f(z)$ in every bounded domain $\mathfrak{A}'$ interior to $\mathfrak{A}$. Moreover, the sequence $R_{1}(z), R_{2}(z), \ldots$ is extracted from the set $E$ of polynomials of the form:
$$
\cfrac{a_{0}+a_{1}z+a_{2}z^{2}+\cdots+a_{k}z^{k}}{q}+i\,\cfrac{b_{0}+b_{1}z+b_{2}z^{2}+\cdots+b_{k}z^{k}}{q},
$$
where $q$ and $k$ are positive integers; $a_{0}, a_{1}, \ldots, a_{k}, b_{0}, b_{1}, \ldots, b_{k}$ are positive, negative, or zero integers. It plainly suffices to prove that the set $E$ is countable. Moreover this set is the sum of the sets $E_{q,k}$ consisting of those polynomials corresponding to fixed values of $q$ and $k$. It therefore suffices to prove (\textsc{ii}, p.~2) that each set $E_{q,k}$ is countable. But this follows from the fact that the polynomials of $E_{q,k}$ for which
$$
|a_{0}|+|a_{1}|+\cdots+|a_{k}|+|b_{0}|+\cdots+|b_{k}|=r
$$
are finite in number for each value of $r$, and that one obtains $E_{q,k}$ by taking successively $r=1,2,\ldots$.\\

\textbf{73.}
\textsc{\textbf{Theorem.}} --- \textit{The necessary and sufficient condition for a set of functions holomorphic inside $\mathfrak{A}$ to be a compact set is that the functions of this set be bounded in modulus as a whole in every bounded domain $\mathfrak{A}'$ interior to $\mathfrak{A}$.}

The condition is necessary by an argument entirely analogous to that of $\mathrm{n}^{\circ}$ 18, considering the bounded domain $\mathfrak{A}'$ as an extremal set of points of the plane.

To prove that the condition is sufficient, it suffices to prove that if one considers an infinity $I$ of functions holomorphic inside $\mathfrak{A}$ and bounded in every domain $\mathfrak{A}'$ interior to $\mathfrak{A}$, one can extract from $I$ a sequence that converges uniformly in each of these domains $\mathfrak{A}'$\,\footnote{\,This theorem was proved by Mr. \textsc{Arzelà} (\textsc{x}, p.~6) when one adds to our hypothesis the conditions: $1^{\circ}$ that the functions of $I$ are bounded not only in the $\mathfrak{A}'$, but throughout $\mathfrak{A}$ (which is less general); $2^{\circ}$ that the derivatives of these functions are also bounded in $\mathfrak{A}$, which is unnecessary. Mr. \textsc{Montel} stated without proof (\textsc{xxv}) that if functions are regular in $\mathfrak{A}$, \textit{continuous on the boundary}, and bounded as a whole in $\mathfrak{A}$, one can extract a sequence converging uniformly in $\mathfrak{A}$.}.

Now one can extract from $\mathfrak{A}$ a sequence $S$ of points \textit{interior} to $\mathfrak{A}$: $\alpha_{1}, \alpha_{2}, \ldots, \alpha_{n}, \ldots$, such that every point interior to $\mathfrak{A}$ belongs to $S$ or to its derived set ($\mathrm{n}^{\circ}$ 45). Let $r_{n}$ be the lower limit of the distance from $\alpha_{n}$ to points not interior to $\mathfrak{A}$. Let $\boldsymbol{\Gamma}_{n}$ be a circle centered at $\alpha_{n}$ with radius $\rho_{n}=r_{n}\left(\cfrac{n}{n+1}\right)$. Every point interior to $\mathfrak{A}$ is interior to at least one of the circles $\boldsymbol{\Gamma}_{n}$; and conversely, every point of $\boldsymbol{\Gamma}_{n}$, or of its boundary, is \textit{interior} to $\mathfrak{A}$. In each circle $\boldsymbol{\Gamma}_{n}$ any function $f(z)$ of $I$ admits a \textsc{Taylor} expansion converging uniformly in $\boldsymbol{\Gamma}_{n}$, boundary included:
$$
f(z)=a_{0,n}+a_{1,n}(z-\alpha_{n})+\cdots+a_{p,n}(z-\alpha_{n})^{p}+\cdots
$$

Then to each function $f(z)$ one can associate a point $A$ of space $(E_{\omega})$ of countably infinite dimension, namely the point whose coordinates are the countable set of numbers
\begin{equation*}
a_{0,1}, a_{1,1}, a_{2,1}, \ldots ; a_{0,2}, a_{1,2}, a_{2,2}, \ldots ; a_{0,n}, a_{1,n}, \ldots . \,\footnote{\,These numbers are not arranged according to the order of a single index as in Chapter V. But it is enough to know that one could do so.} \tag{1}
\end{equation*}

Hence to the set $I$ of functions $f(z)$ corresponds a set $G$ of points $A$ of $(E_{\omega})$. This set is bounded, for one has, whatever point $A$,
$$
\left|a_{p,n}\right|<\cfrac{M_{n}}{\left(\rho_{n}^{\prime}\right)^{p}}\,,
$$
where $M_{n}$ denotes a number exceeding the modulus of every function $f(z)$ of $I$ at every point of a circle $C_{n}'$ centered at $\alpha_{n}$ with radius $\rho_{n}'>\rho_{n}$ and $<r_{n}$. Therefore the set $G$ is compact ($\mathrm{n}^{\circ}$ 66), and one can take a sequence of its points $A_{1}, A_{2}, \ldots$ having a limit $B$. Denoting
$$
a_{0,1}^{(q)},\, a_{1,1}^{(q)}, \ldots, a_{p,n}^{(q)}, \ldots
$$
the coordinates of $A_{q}$ and
$$
b_{0,1},\, b_{1,1}, \ldots, b_{p,n}, \ldots
$$
those of $B$, one has:
$$
b_{p,n}=\lim_{q=\infty}a_{p,n}^{(q)},
$$
whatever $p,n$. In particular:
$$
\left|b_{p,n}\right|\le\cfrac{M_{n}}{\left(\rho_{n}^{\prime}\right)^{p}}\,;
$$
hence, for each value of $n$, the series
$$
b_{0,n}+b_{1,n}(z-\alpha_{n})+\cdots+b_{p,n}(z-\alpha_{n})^{p}+\cdots
$$
converges to a function $\varphi_{n}(z)$ holomorphic inside $C_{n}'$. Let, on the other hand, $f_{1}(z), f_{2}(z), \ldots$ be the functions of $I$ corresponding to $A_{1}, A_{2}, \ldots$. I shall prove that this sequence of functions has a limit. Indeed, in every circle $\boldsymbol{\Gamma}_{n}$:
$$
\begin{gathered}
\left|f_{q}(z)-\varphi_{n}(z)\right|<\left|a_{0,n}^{(q)}-b_{0,n}\right|+\cdots+\left|a_{p,n}^{(q)}-b_{p,n}\right|\left(\rho_{n}\right)^{p} \\
+2M_{n}\left[\left(\cfrac{\rho_{n}}{\rho_{n}^{\prime}}\right)^{p}+\left(\cfrac{\rho_{n}}{\rho_{n}^{\prime}}\right)^{p+1}+\cdots\right] .
\end{gathered}
$$
Taking $p$ sufficiently large, the last term is $<\cfrac{\varepsilon}{2}$; with this $p$ chosen and fixed, the sum of the first $p+1$ terms tends to zero with $\cfrac{1}{q}$, so one can find $h$ such that this sum remains below $\cfrac{\varepsilon}{2}$ for $q>h$. In summary, $f_{q}(z)$ tends uniformly to $\varphi_{n}(z)$ in $\boldsymbol{\Gamma}_{n}$, boundary included.

Since every point interior to $\mathfrak{A}$ is interior to at least one circle $\boldsymbol{\Gamma}_{n}$, one first sees that one has indeed extracted from $I$ a sequence $f_{1}, f_{2}, \ldots$ convergent at every interior point of $\mathfrak{A}$. Moreover, the limit $f(z)$ is holomorphic at every interior point of $\mathfrak{A}$ since it coincides with $\varphi_{n}(z)$ in $\boldsymbol{\Gamma}_{n}$ for every $n$. Finally, convergence is uniform in every bounded domain $\mathfrak{A}'$ interior to $\mathfrak{A}$.
Indeed, the points of this domain $\mathfrak{A}'$ (boundary included) form a bounded and closed set, each point of which is interior to at least one of the circles $\boldsymbol{\Gamma}_{n}$; one can therefore extract a finite number of these circles having the same property (see $\mathrm{n}^{\circ}$ 36). Since $f_{q}(z)$ converges uniformly in a \textit{finite} number of circles $\boldsymbol{\Gamma}_{n}$, it also converges uniformly in the domain $\mathfrak{A}'$ they cover.

It would moreover be incorrect to assert that one can choose the sequence $f_{1}(z), f_{2}(z), \ldots$ so that it converges uniformly \textit{throughout the whole domain} $\mathfrak{A}$ (see the preceding note). This would be incorrect even if one supposed that the functions of $I$ are both holomorphic and bounded as a whole in $\mathfrak{A}$, \textit{boundary included}. To see this, it suffices again to take $\mathfrak{A}$ to be the disk $|z|\le1$ and $I$ the functions $z^{n}$.\\

\textbf{74.}
\textit{\textbf{Application of the general theorems.}} --- We know that functions holomorphic inside the same domain $\mathfrak{A}$ form a normal class $(E)$. We may therefore apply to them the general results of the \textsc{First Part}, taking care to keep for the definition of the limit of a sequence of holomorphic functions in $\mathfrak{A}$ the sense specified above ($\mathrm{n}^{\circ}$ 71).

\textsc{\textbf{Theorem.}} --- \textit{Let $E$ be any set of functions holomorphic inside one and the same domain $\mathfrak{A}$ with simple or multiple boundary: $1^{\circ}$ the set of limit elements of $E$ is closed; $2^{\circ}$ either $E$ is countable, or $E$ gives rise to at least one condensation element ($\mathrm{n}^{\circ}$ 8); $3^{\circ}$ one can extract from $E$ a countable sequence $S$ of functions such that every function of $E$ belongs to $S$ or is the limit of at least one sequence extracted from $S$.}

\textsc{\textbf{Theorem.}} --- \textit{If the functions of the preceding set $E$ are bounded in modulus as a whole in every bounded domain $\mathfrak{A}'$ interior to $\mathfrak{A}$, the set of functions of $E$ that are not limits of functions of $E$ is countable.}

\textsc{\textbf{Theorem.}} --- \textit{Let $F$ be a closed set of holomorphic functions bounded in modulus as a whole in every bounded domain $\mathfrak{A}'$ interior to $\mathfrak{A}$: $1^{\circ}$ there exists at least one set $H$ of functions holomorphic inside $\mathfrak{A}$ whose set of limit elements coincides with $F$; $2^{\circ}$ one can decompose $F$ into a countable set and a perfect set or null set.}

\textsc{\textbf{Theorem.}} --- \textit{Every perfect set of functions holomorphic inside $\mathfrak{A}$ has the power of the continuum.}

One may also answer a question analogous to that of Mr.~\textsc{Hadamard} ($\mathrm{n}^{\circ}$ 39, 58) and give a generalization of Mr.~\textsc{Borel}'s theorem ($\mathrm{n}^{\circ}$ 42) for sets of holomorphic functions.\\

\textbf{75.}
Likewise, let us consider functional operations ($\mathrm{n}^{\circ}$ 4) whose variable is a holomorphic function (see \textsc{xvi, xxvi}).

\textsc{\textbf{Theorem.}} --- \textit{Let $F$ be the previously considered set of holomorphic functions. Every real continuous operation on holomorphic functions of $F$: $1^{\circ}$ is uniformly continuous on $F$; $2^{\circ}$ is bounded and attains its upper limit for at least one function of $F$.}

\textit{Conversely, if every real continuous operation on a set $E$ of functions holomorphic inside one and the same domain $\mathfrak{A}$ attains its upper limit for at least one function of $E$, the set $E$ is closed and compact.}

\textsc{\textbf{Theorem.}} --- \textit{Consider a series of continuous operations on a closed set $F$ of holomorphic functions bounded in every bounded domain $\mathfrak{A}'$ interior to $\mathfrak{A}$. For the sum of this series to be a continuous operation on $F$, it is necessary and sufficient that convergence be quasi-uniform.}

\textsc{\textbf{Theorem.}} --- \textit{Let $\mathfrak{F}$ be a family of continuous operations on the same set $F$ of holomorphic functions. In order that from every infinite subset of operations of $\mathfrak{F}$ one can extract a sequence converging uniformly, it is necessary and sufficient that the operations of $\mathfrak{F}$ be equally continuous and bounded at every element of $F$.}

\subsection*{Chapter VII.}
\phantomsection\addcontentsline{toc}{subsection}{Chapter VII.}
\subsubsection*{Sets of continuous curves and line functions.}
\phantomsection\addcontentsline{toc}{subsubsection}{Sets of continuous curves and line functions.}

\textbf{76.}
\textit{\textbf{Definition of a curve-arc.}} --- Let $f,g,h$ be three functions of $t$, continuous on the same interval $I: a\le t\le b$, and which are not simultaneously constant on any interval contained in $I$, or else are simultaneously constant on \textit{the whole} interval $I$. The formulas
\begin{equation*}
x=f(t), \quad y=g(t), \quad z=h(t) \qquad\qquad (a\le t\le b) \tag{1}
\end{equation*}
define, in three-dimensional space, a set $E$ of points $(x,y,z)$. In the second case this set $E$ contains only one point; in the first, it is uncountable and then also perfect, bounded, and connected. If one takes into account the order in which the points of $E$ appear when $t$ increases, one obtains an ordered sequence (\textsc{iii}, p.~27) called a \textit{curve-arc}. Thus there is between the set $E$ and the curve-arc $\gamma$ it generates the same difference as between a combination and an arrangement. In other words, consider a system of three functions $F(u),G(u),H(u)$ satisfying the same conditions as $f,g,h$; the formulas
\begin{equation*}
x=F(u), \quad y=G(u), \quad z=H(u) \qquad\qquad A \le u \le B \tag{2}
\end{equation*}
define a curve-arc $\boldsymbol{\Gamma}$, which is to be regarded as coinciding with $\gamma$ only if not only $\gamma$ and $\boldsymbol{\Gamma}$ are formed by the same set of points, but these points are also encountered in the same order when one makes either $t$ or $u$ increase in (1) or (2). This latter condition is independent of the first; to see it, one need only choose, as is easy, formulas (1) and (2) so as to represent two figure-eight arcs: \textit{QRSQTPQ} and \textit{QSRQTPQ}, formed by the same points but not traversed in the same order.

Note that, as in the preceding example, one and the same geometric point may be encountered twice on a curve-arc, that is, may correspond to two different values of $t$. We shall nevertheless regard these two values of $t$ as corresponding to two distinct points of $\gamma$ (multiple points). This convention being made, we may say that, for two curves $\gamma$ and $\boldsymbol{\Gamma}$ to coincide, it is necessary and sufficient that one can establish between their points a one-to-one reciprocal correspondence such that corresponding points coincide and their order is preserved\,\footnote{\,Note that this definition does satisfy the fundamental conditions: $1^{\circ}$ if $\gamma$ is identical with $\boldsymbol{\Gamma}$, then $\boldsymbol{\Gamma}$ is identical with $\gamma$; $2^{\circ}$ if two curves are identical with a same third one, they coincide.}.\\

\textbf{77.}
In general, one is content to say that two curves (1) and (2) coincide if it is possible to find a function $u=\varphi(t)$ continuous and constantly increasing, going from $A$ to $B$ as $t$ increases from $a$ to $b$, and such that
\begin{equation*}
f(t) \equiv F[\varphi(t)], \quad g(t) \equiv G[\psi(t)], \quad h(t)=H[\varphi(t)] \qquad\qquad(a \le t \le b) . \tag{3}
\end{equation*}
This definition has the drawback of not having a direct relation with the intuitive notion we have of identity of two curves, a notion that seems clearly expressed in the definition I gave above. Moreover, note that if condition (3) is fulfilled, the curves do indeed coincide in the first sense I indicated. But the converse is less evident. To prove it, observe that if one can establish between (1) and (2) this correspondence in which order is preserved, one can evidently obtain identities such as (3), in which $\varphi(t)$ is a well-determined function for each value of $t$, and which goes without ever decreasing from $A$ to $B$ as $t$ increases from $a$ to $b$. It remains to show that $\varphi(t)$ is a continuous increasing function. Now if it were otherwise: either $\varphi(t)$ would be constant on some interval, and then $f(t),g(t),h(t)$ would be simultaneously constant on that interval, or $\varphi(t)$ would jump abruptly from values $<\varphi(t_{0})-k$ to values $>\varphi(t_{0})+k$ when $t$ passes through $t_{0}$, and then order would not be preserved.

Thus, the geometric definition we have given completely coincides with the ordinary definition. This one allows, in particular, always to suppose that the parameter interval of variation is $(0,1)$; it suffices to make the transformation $u=\cfrac{t-a}{b-a}$. We shall always assume henceforth that this has been done.\\

\textbf{78.}
\textit{\textbf{Definition of distance (\'ecart).}} --- Now that we can distinguish distinct curves, we may consider the set of all these curves as forming a class of elements to which we shall try to apply our general theorems. To this end, we shall prove that one can give for this class a definition of distance (\'ecart).

Consider indeed two arbitrary representations:
\begin{equation*}
x=f(t), \quad y=g(t), \quad z=h(t) \qquad\qquad\qquad\; (0 \le t \le 1), \tag{4}
\end{equation*}
\begin{equation*}
x=F(u), \quad y=G(u), \quad z=H(u) \qquad \qquad \quad (0\le u \le 1),\tag{5}
\end{equation*}
of two curves $\gamma$ and $\boldsymbol{\Gamma}$, and form the expression
$$
\delta(t)=\sqrt{[f(t)-F(t)]^{2}+[g(t)-G(t)]^{2}+[h(t)-H(t)]^{2}} \quad\quad(0 \le t \le 1) ;
$$
this is a continuous function of $t$ with a maximum $d\ge0$. We shall call \textit{distance (\'ecart) of the two curves} the lower limit $e\ge0$ of the set of values of $d$ obtained by taking for (4) and (5) all possible representations of $\gamma$ and $\boldsymbol{\Gamma}$. One can simplify the computation of the distance by noting that it suffices to take for $e$ the lower limit of the values of $d$ obtained by taking for $\gamma$ one fixed arbitrary representation (4) and varying only the representation of $\boldsymbol{\Gamma}$. Indeed, any arbitrary simultaneous representation of $\gamma$ and $\boldsymbol{\Gamma}$ can be written, from the above,
$$
\begin{array}{llll}
x=f[\varphi(t)], & y=g[\varphi(t)], & z=h[\varphi(t)] \qquad\qquad & (0 \le t \le 1), \\
x=F[\psi(u)], & y=G[\psi(u)], & z=H[\psi(u)] \qquad\qquad & (0 \le u \le 1),
\end{array}
$$
where $\varphi(t)$ and $\psi(u)$ are two continuous increasing functions. Then, for such a representation, the value of $\delta$ is
$$
\delta_{1}(t)=\sqrt{\{f[\varphi(t)]-F[\psi(t)]\}^{2}+\{g[\varphi(t)]-G[\psi(t)]\}^{2}+\{h[\varphi(t)]-H[\psi(t)]\}^{2}}
$$
and if $\varphi_{1}(u)$ is the inverse function of $\varphi(t)$, setting $\theta(u)=\psi\left[\varphi_{1}(u)\right]$, one has
$$
\delta_{1}\left[\varphi_{1}(u)\right]=\sqrt{\left\{f(u)-F[\theta(u)]\right\}^{2}+\left\{g(u)-G[\theta(u)]\right\}^{2}+\{h(u)-H[\theta(u)]\}^{2}} .
$$

But the maximum of $\delta_{1}(t)$ is obviously the same as that of $\delta_{1}\left[\varphi_{1}(u)\right]$, which takes the same values. Now the latter is indeed obtained by keeping the same representation of $\gamma$ and varying only that of $\boldsymbol{\Gamma}$.\\

\textbf{79.}
This being established, it is easy to prove that \textit{the definition just given does satisfy the general conditions a), b) for distance} ($\mathrm{n}^{\circ}$ 49). Indeed, it is first obvious that if $\gamma$ and $\boldsymbol{\Gamma}$ coincide, one of the values of $d$ is zero and therefore $(\gamma,\boldsymbol{\Gamma})=0$. Conversely, if this condition is satisfied: either one of the values of $d$ is zero, and then $\gamma$ and $\boldsymbol{\Gamma}$ surely coincide; or else, whatever $n$, one can determine a function $\varphi_{n}(t)$ continuous and increasing from $0$ to $1$, such that for every value of $t$ between $0$ and $1$,
\begin{equation*}
\sqrt{\left\{f(t)-F\left[\varphi_{n}(t)\right]\right\}^{2}+\left\{g(t)-G\left[\varphi_{n}(t)\right]\right\}^{2}+\left\{h(t)-H\left[\varphi_{n}(t)\right]\right\}^{2}}<\cfrac{1}{n} . \tag{6}
\end{equation*}
Hence
$$
f(t)=\lim _{n=\infty} F\left[\varphi_{n}(t)\right], \, g(t)=\lim _{n=\infty} G\left[\varphi_{n}(t)\right], \, h(t)=\lim _{n=\infty} H\left[\varphi_{n}(t)\right] ; \, (0 \le t \le 1),
$$
with uniform convergence. To prove the theorem, it suffices to prove that one can extract from the sequence $\varphi_{1}(t), \varphi_{2}(t), \ldots$ a sequence converging uniformly to a function $\varphi(t)$ (necessarily continuous and increasing from $0$ to $1$).
For this, we shall show that the functions $\varphi_{n}(t)$ are equally continuous and bounded ($\mathrm{n}^{\circ}$ 57).
They are obviously bounded: $0\le\varphi_{n}(t)\le1$. If they were not equally continuous, one could determine a number $\varepsilon>0$ such that, whatever $n$, there are a function $\varphi_{p_{n}}$ and two numbers $t_{n}, t_{n}'$ in the interval $(0,1)$ satisfying
\begin{equation*}
\left|\varphi_{p_{n}}\left(t_{n}\right)-\varphi_{p_{n}}\left(t_{n}^{\prime}\right)\right|>\varepsilon, \quad 0<\left|t_{n}-t_{n}^{\prime}\right|<\cfrac{1}{n} . \tag{7}
\end{equation*}
And since the numbers $t_{n}, \varphi_{p_{n}}(t_{n}), \varphi_{p_{n}}(t_{n}')$ are bounded as a whole, one may suppose they are chosen so as to converge (as $n$ increases indefinitely) to respective limits denoted $\tau_{0}, u_{0}, u_{0}'$. Then inequalities (7) show that $t_{n}'$ also tends to $\tau_{0}$ and that $|u_{0}-u_{0}'|\ge\varepsilon$. This being so, let $\lambda$ be any number between $u_{0}$ and $u_{0}'$; for $n$ large enough, $\varphi_{p_{n}}(t_{n})$ and $\varphi_{p_{n}}(t_{n}')$, which tend to $u_{0}$ and $u_{0}'$, will contain $\lambda$ between them. The continuous function $\varphi_{p_{n}}(t)$ therefore takes the value $\lambda$ at least once for a value $\tau_{n}$ of $t$ between $t_{n}$ and $t_{n}'$. One therefore has, by (6),
$$
\left|f\left(\tau_{n}\right)-F(\lambda)\right|<\cfrac{1}{p_{n}}, \quad\left|g\left(\tau_{n}\right)-G(\lambda)\right|<\cfrac{1}{p_{n}}, \quad\left|h\left(\tau_{n}\right)-H(\lambda)\right|<\cfrac{1}{p_{n}} \,.
$$
Since $\tau_{n}$ obviously tends to $\tau_{0}$, it follows that
$$
f\left(\tau_{0}\right)=F(\lambda), \quad g\left(\tau_{0}\right)=G(\lambda), \quad h\left(\tau_{0}\right)=H(\lambda),
$$
whatever value of $\lambda$ between $u_{0}$ and $u_{0}'$. We thus reach the announced contradiction, since $F(u),G(u),H(u)$ are not simultaneously constant on any interval. One should make an exception for the case where $\boldsymbol{\Gamma}$ reduces to a point, but then the converse is obvious.\\

\textbf{80.}
To show that condition b) of $\mathrm{n}^{\circ}$ 49 is satisfied, consider three curves $\gamma$, $\boldsymbol{\Gamma}$, and $C$, having representations (4), (5), and
$$
x=a(v), \quad y=b(v), \quad z=c(v) \qquad\qquad(0 \le v \le 1).
$$
To obtain their respective distances, it suffices to take the maxima of the quantities
$$
\begin{aligned}
\delta(t) & =\sqrt{\{f[\varphi(t)]-F[\psi(t)]\}^{2}+\{g[\varphi(t)]-G[\psi(t)]\}^{2}+\{h[\varphi(t)]-H[\psi(t)]\}^{2}}\,, \\
\delta_{1}(t) & =\sqrt{\{a(t)-f[\varphi(t)]\}^{2}+\{b(t)-g[\varphi(t)]\}^{2}+\{c(t)-h[\varphi(t)]\}^{2}}\,, \\
\delta_{2}(t) & =\sqrt{\{a(t)-F[\psi(t)]\}^{2}+\{b(t)-G[\psi(t)]\}^{2}+\{c(t)-H[\psi(t)]\}^{2}}\,,
\end{aligned}
$$
as $t$ increases from $0$ to $1$, then vary the form of $\varphi(t)$ and $\psi(t)$, continuous increasing functions of $t$. One obviously has
$$
\delta(t) \le \delta_{1}(t)+\delta_{2}(t) ;
$$
one easily deduces for the corresponding maxima $d,d_{1},d_{2}$:
$$
d \le d_{1}+d_{2}
$$
and, by varying $\varphi$ and $\psi$ independently:
$$
(\gamma, \boldsymbol{\Gamma}) \le (C,\gamma)+(C,\boldsymbol{\Gamma}) .
$$

\textbf{81.}
\textit{\textbf{Limit of a sequence of curves.}} --- The definition of distance yields ($\mathrm{n}^{\circ}$ 49, 27) a definition of limit: \textit{The curve $\gamma_{n}$ tends to $\gamma$ when the distance $(\gamma_{n},\gamma)$ tends to zero with $\cfrac{1}{n}$}. This definition can be expressed analytically as follows:

The necessary and sufficient condition for a curve $\gamma_{n}$ to have as limit a curve $\gamma$, as $n$ increases indefinitely, is that there exist a representation of $\gamma_{n}$:
\begin{equation*}
x=f_{n}(t), \quad y=g_{n}(t), \quad z=h_{n}(t) \qquad\qquad(0 \le t \le 1), \tag{8}
\end{equation*}
and a representation of $\gamma$:
\begin{equation*}
x=f(t), \quad y=g(t), \quad z=h(t) \qquad\qquad (0 \le t \le 1) \tag{4}
\end{equation*}
such that $f_{n},g_{n},h_{n}$ converge uniformly to $f,g,h$.

Indeed, if this condition is satisfied, it is obvious that $(\gamma_{n},\gamma)$ tends to zero. Conversely, choose a fixed representation (4) of $\gamma$; for each $n$, one may choose a representation (8) of $\gamma_{n}$ such that
$$
\sqrt{\left[f(t)-f_{n}(t)\right]^{2}+\left[g(t)-g_{n}(t)\right]^{2}+\left[h(t)-h_{n}(t)\right]^{2}}<\left(\gamma_{n}, \gamma\right)+\cfrac{1}{n} \,.
$$

Then, if $(\gamma_{n},\gamma)$ tends to zero, one sees that $f_{n},g_{n},h_{n}$ converge uniformly to $f,g,h$.

We are thus led back to the ordinary definition of the limit of a curve. In particular, it follows that if $\gamma_{n}$ tends to $\gamma$, then for $n$ large enough it lies entirely in the volume generated by a sphere of radius $\varepsilon$ whose center runs along $\gamma$.\\

\textbf{82.}
The definition we gave of distance seems rather arbitrary; let us try to show why it is preferable to others that may appear more natural.

For a long time one has had the idea of defining a number determined by two curves $C_{1}$ and $C_{2}$ in such a way that this number is zero when the two curves coincide and is very small when the two curves are infinitely close.

Mr.~\textsc{Arzelà} thus considers what he calls \textit{the shortest distance of the two curves} (\textsc{v}, p.~343), namely the lower limit $\lambda$ of the distance between two points, one arbitrary on $C_{1}$, the other arbitrary on $C_{2}$. Such a number indeed satisfies the two preceding conditions. But it does not satisfy the converse conditions: for $\lambda$ to be zero it is not necessary that the two curves coincide, it is enough that they intersect. The two curves may be very different in shape and the number $\lambda$ as small as one wishes, even zero. Hence the consideration of $\lambda$ is not sufficient for Mr.~\textsc{Arzelà} to define the limit of a curve.

Before Mr.~\textsc{Arzelà}, \textsc{Weierstrass} had introduced a quantity he called the \textit{neighborhood of two curves}. The definition he gave escapes the previous objection, in the sense that neighborhood is always small when the two curves are close, and only then. But it is defined only for neighboring curves. \textsc{Weierstrass} called neighborhood of two neighboring curves any number $\ge0$ smaller than the maximum distance of two corresponding points with the same abscissa or very close abscissas. His definition, perfectly sufficient for his aim\,\footnote{\,Namely, to bring greater rigor into the calculus of variations.}, lacked precision from our point of view. Moreover, he could not think of using it in set theory, which was still in its beginnings. On the contrary, Mr.~\textsc{Arzelà} seems to have been the first to perceive the importance of defining a well-determined number such as his \textit{shortest distance}, defined for any two curves.

Instead of the preceding definitions whose insufficiency we have recognized, one might also have thought of introducing a number $\mu_{1,2}$, the upper limit of the minimum distance from a point of $C_{1}$ to all points of $C_{2}$; or again a number $\rho_{1,2}$, the lower limit of the maximum distance from a point of $C_{1}$ to all points of $C_{2}$. These numbers satisfy part of the conditions we require of distance, but not all.

First, they are well defined even when $C_{1}$ and $C_{2}$ are not very close. Moreover, if $C_{1}$ coincides with $C_{2}$, one has $\mu_{1,2}=0$. But it may happen that $\mu_{1,2}=0$ without $C_{1}$ coinciding with $C_{2}$; it suffices that their set of points be the same. Conversely, it may happen that $C_{1}$ coincides with $C_{2}$ without $\rho_{1,2}=0$, for in the latter case $C_{2}$ would reduce to a point. These numbers $\mu_{1,2},\rho_{1,2}$ therefore cannot define a distance of two curves. But since they are easier to compute than distance $e$, they may sometimes be useful to give an idea of the value of $e$. One sees easily that one always has
$$
\lambda \le \mu_{1,2} \le e \quad \text{and} \quad \lambda \le \mu_{2,1} \le e.
$$

These inequalities are useful in certain questions (see, for example, $\mathrm{n}^{\circ}$ 85).\\

\textbf{83.}
\textit{\textbf{Limit of a curve.}} --- The theorem given in $\mathrm{n}^{\circ}$ 81 on the limit of a sequence of curves may be stated as follows. For $\gamma_{n}$ to have $\gamma$ as limit, it is necessary and sufficient that, for each value of $n$, one can establish between $\gamma_{n}$ and $\gamma$ a one-to-one reciprocal correspondence preserving order, in such a way that if $M_{n}$ on $\gamma_{n}$ corresponds to $M$ on $\gamma$, then $M_{n}$ tends to $M$, and uniformly. Assuming this condition holds for one correspondence, one may ask whether it would also hold for any other one-to-one order-preserving correspondence. In other words, let
\begin{equation*}
x=\varphi_{n}(t), \quad y=\psi_{n}(t), \quad z=\chi_{n}(t) \qquad\qquad (0 \le t \le 1) \tag{$8'$}
\end{equation*}
be any representation of $\gamma_{n}$. For a same value of $t$, according to (8), ($8'$), (4), there correspond two points $M_{n},M_{n}'$ on $\gamma_{n}$ and one point $M$ on $\gamma$. The question is what happens to $M_{n}'$ in the limit.

It is easy to see that, with $M$ fixed and $n$ increasing indefinitely: $1^{\circ}$ the corresponding point $M_{n}'$ has at least one limit point; $2^{\circ}$ every limit point of $M_{n}'$ is a point of $\gamma$ (which may be distinct from $M$). But the converse is not true. In other words, when $M$ takes all possible positions on $\gamma$, the set of limit points of the corresponding points $M_{n}'$ lies on $\gamma$, but does not always fill curve $\gamma$. If, for example, one takes
$$
\varphi_{n}(t)\equiv f_{n}\left[\alpha_{n}(t)\right], \quad \psi_{n}(t)\equiv g_{n}\left[\alpha_{n}(t)\right], \quad \chi_{n}(t)\equiv h_{n}\left[\alpha_{n}(t)\right]
$$
with
$$
\alpha_{n}(t)=t e^{n(t-1)},
$$
formulas ($8'$) do indeed give a representation of $\gamma_{n}$ since $\alpha_{n}(t)$ is continuous and strictly increasing from $0$ to $1$. However, one sees that for $t\neq1$, $\varphi_{n}(t),\psi_{n}(t),\chi_{n}(t)$ tend to $f(0),g(0),h(0)$, and for $t=1$, $\varphi_{n}(1),\psi_{n}(1),\chi_{n}(1)$ tend to $f(1),g(1),h(1)$. Thus the points $M_{n}'$ all tend to one or the other endpoint of $\gamma$.

This shows in particular that if one considers a set $E$ of curves and takes for each an arbitrary parametric representation, then when a sequence of curves of $E$ has a limit, the corresponding analytic representations need not have a limit.\\

\textbf{84.}
\textit{\textbf{The set of curves forms a normal class $(E)$.}} --- To reach application of our general theorems, it remains only to prove that the class of curves, which from the foregoing is already a class $(E)$, is also a normal class $(E)$.

I. I say it is separable. Indeed, we saw ($\mathrm{n}^{\circ}$ 56) that one can form, once and for all, a sequence $S$ of continuous functions
$$
A_{1}(t), A_{2}(t), \ldots, A_{n}(t), \ldots,
$$
such that every continuous function is one of the limit elements of this sequence. Then, if one considers a representation (4) of a curve $\gamma$, one can find three functions of this sequence $A_{p_{n}}(t),A_{q_{n}}(t),A_{r_{n}}(t)$ such that
$$
\left|f(t)-A_{p_{n}}(t)\right|<\cfrac{1}{n}, \quad\left|g(t)-A_{q_{n}}(t)\right|<\cfrac{1}{n}, \quad\left|h(t)-A_{r_{n}}(t)\right|<\cfrac{1}{n} \qquad(0 \le t \le 1)
$$
and these functions can always be chosen not simultaneously constant on a same interval. This proves that one can regard $\gamma$ as the limit curve of a sequence of curves extracted from the set $D$ of curves
$$
x=A_{p}(t), \quad y=A_{q}(t), \quad z=A_{r}(t) \qquad\qquad(0 \le t \le 1),
$$
where $A_{p},A_{q},A_{r}$ are three continuous functions extracted from sequence $S$, not simultaneously constant on a same interval. Now $D$ is clearly countable. We thus reach the following conclusion, which (besides its later usefulness) is very curious in itself:

\textsc{\textbf{Theorem.}} --- \textit{One can draw, once and for all, in three-dimensional space, a countable sequence of continuous curves}\,\footnote{\,From the method of construction of sequence $S$, one may even suppose that these continuous lines are all unicursal curves, or all polygonal lines each with a finite number of sides.}
$$
\gamma_{1}, \gamma_{2}, \ldots, \gamma_{n}, \ldots,
$$
\textit{such that every curve in space is the limit of at least one suitably extracted sequence from the first.}

II. Let us now show that \textsc{Cauchy}'s theorem on convergence of a sequence extends to a sequence of curves:

\textsc{\textbf{Theorem.}} --- \textit{The necessary and sufficient condition for a sequence of curves $C_{1},C_{2},\ldots$ to tend to a limit curve is that, for every $\varepsilon>0$, one can find an integer $n$ satisfying $\left(C_{n},C_{n+p}\right)<\varepsilon$ for every integer $p$.}

The condition is obviously necessary. To prove it is sufficient, it suffices to show that if it is satisfied, the set is compact. Indeed, if one can extract from the \textsc{Cauchy} sequence $C_{1},C_{2},\ldots$ a sequence $C_{n_{1}},C_{n_{2}},\ldots$ (with increasing indices) tending to a curve $C$, then the given sequence also converges to $C$.

To form the sequence $C_{n_{1}},C_{n_{2}},\ldots$, note that by hypothesis one can always choose $n_{q}$ large enough so that, for every integer $p$,
$$
\left(C_{n_{q}}, C_{n_{q+p}}\right)<\cfrac{1}{q^{2}}\,.
$$

Given a representation
$$
x=f_{q}(t), \quad y=g_{q}(t), \quad z=h_{q}(t) \qquad\quad\qquad\qquad(0 \le t \le 1)
$$
of $C_{n_{q}}$, we know one can always choose that of $C_{n_{q+1}}$
$$
x=f_{q+1}(t), \quad y=g_{q+1}(t), \quad z=h_{q+1}(t) \qquad\qquad (0 \le t \le 1),
$$
so that
$$
\sqrt{\left[f_{q}(t)-f_{q+1}(t)\right]^{2}+\left[g_{q}(t)-g_{q+1}(t)\right]^{2}+\left[h_{q}(t)-h_{q+1}(t)\right]^{2}}<\left(C_{n_{q}}, C_{n_{q+1}}\right)+\varepsilon,
$$
with $\varepsilon$ as small as one wishes. Take in particular $\varepsilon=\cfrac{1}{q^{2}}$. One then sees one can choose successively the representations of $C_{n_{1}},C_{n_{2}},\ldots$ so that, for every $q$,
$$
\left|f_{q}(t)-f_{q+1}(t)\right|<\cfrac{2}{q^{2}}, \quad\left|g_{q}(t)-g_{q+1}(t)\right|<\cfrac{2}{q^{2}}, \quad\left|h_{q}(t)-h_{q+1}(t)\right|<\cfrac{2}{q^{2}} \,.
$$

Then, given $\omega>0$, if one chooses $q_{1}$ such that
$$
\cfrac{2}{q_{1}^{2}}+\cfrac{2}{\left(q_{1}+1\right)^{2}}+\cdots<\omega
$$
(as is always possible), one easily sees that, for every integer $p$,
$$
\left|f_{q_{1}}(t)-f_{q_{1}+p}(t)\right|<\omega, \quad\left|g_{q_{1}}(t)-g_{q_{1}+p}(t)\right|<\omega, \quad\left|h_{q_{1}}(t)-h_{q_{1}+p}(t)\right|<\omega \,.
$$
This shows that the functions $f_{q}(t),g_{q}(t),h_{q}(t)$ converge uniformly to three continuous functions $f(t),g(t),h(t)$ as $q$ increases indefinitely. If these three functions are constant from $0$ to $1$, then $C_{n_{1}},C_{n_{2}},\ldots$ tend to a point. If these three functions are not constant on $(0,1)$, then to show $C_{n_{1}},C_{n_{2}},\ldots$ tend to a limit curve, one would have to prove that $f(t),g(t),h(t)$ are not simultaneously constant on any interval. Or rather, one would have to prove that one can choose the representations of $C_{n_{1}},C_{n_{2}},\ldots$ so that they converge to continuous functions not simultaneously constant on any interval.

For this, we use the fact that if $f,g,h$ are not constant from $0$ to $1$, one can form a function $a(t)$ increasing (never decreasing) from $0$ to $1$, and constant only on intervals where $f,g,h$ are simultaneously constant\,\footnote{\,See \textsc{Note I} at the end of the Memoir.}.

Then, setting $u=a(t)$, to each value of $u$ corresponds exactly one triple of values of $f(t),g(t),h(t)$, say $F(u),G(u),H(u)$. Thus there are three functions $F,G,H$, determined for each value of $u$, such that
$$
f(t)\equiv F[a(t)], \quad g(t)\equiv G[a(t)], \quad h(t)\equiv H[a(t)] \qquad(0 \le t \le 1) .
$$
First let us show that the formulas
$$
x=F(u), \quad y=G(u), \quad z=H(u) \qquad\qquad\qquad\qquad(0 \le u \le 1)
$$
define a curve. By the hypothesis on $a(t)$, $F,G,H$ cannot be simultaneously constant on a same interval. I claim they are continuous; that is, if there exists a sequence of values of $u$: $u_{1},u_{2},\ldots,u_{n},\ldots$ converging to a limit $u_{0}$, one has for example
$$
F\left(u_{0}\right)=\lim_{n=\infty}F\left(u_{n}\right) .
$$
Indeed, from the bounded set of numbers $F(u_{1}),F(u_{2}),\ldots$, one can form a sequence $F(u_{1}'),F(u_{2}'),\ldots$ converging to a determined limit $\alpha_{0}$; it suffices to prove that however this sequence is formed, one always has $\alpha_{0}=F(u_{0})$. Now, to each number $u_{n}'$, the equation $u=a(t)$ corresponds at least one number $t_{n}'$ such that $u_{n}'=a(t_{n}')$. Since numbers $t_{1}',t_{2}',\ldots$ are bounded, one can extract from their sequence a sequence $t_{p_{1}}',t_{p_{2}}',\ldots$ with limit $t_{0}$; and since $a(t)$ is continuous:
$$
a\left(t_{0}\right)=\lim_{n=\infty}a\left(t_{p_{n}}'\right)=\lim_{n=\infty}u_{p_{n}}'=u_{0}.
$$

But $f(t)$ is also continuous, hence
$$
F\left(u_{0}\right)=F\left[a\left(t_{0}\right)\right]=f\left(t_{0}\right)=\lim_{n=\infty}f\left(t_{p_{n}}'\right)=\lim_{n=\infty}F\left(u_{p_{n}}'\right)=\alpha_{0}.
$$

Thus $F(u)$ is continuous; likewise $G,H$.

Now set
$$
a_{n}(t)=\left(1-\cfrac{1}{n}\right)a(t)+\cfrac{t}{n}.
$$
For each $n$, $a_{n}(t)$ is continuous and increasing in $t$ [as $n$ tends to infinity, this function converges uniformly to $a(t)$]. Therefore the inverse function $t=b_{n}(u)$ is also continuous and increasing. Hence, setting
$$
F_{n}(u)=f_{n}\left[b_{n}(u)\right], \quad G_{n}(u)=g_{n}\left[b_{n}(u)\right], \quad H_{n}(u)=h_{n}\left[b_{n}(u)\right],
$$
curve $C_{n}$ is also represented by
$$
x=F_{n}(u), \quad y=G_{n}(u), \quad z=H_{n}(u) \qquad\qquad (0 \le u \le 1)\,.
$$

It remains only to prove that $F_{n}(u),G_{n}(u),H_{n}(u)$ converge uniformly to $F(u),G(u),H(u)$. Now
$$
\left|F\left[a_{n}(t)\right]-f_{n}(t)\right| \le \left|F\left[a_{n}(t)\right]-F[a(t)]\right|+\left|f(t)-f_{n}(t)\right|\,.
$$

The two terms on the right converge uniformly to zero; hence so does the left. But setting $u=a_{n}(t)$, this becomes $|F(u)-F_{n}(u)|$, and its maximum, for a given $n$, remains the same. Therefore the maximum of the left-hand member tends to zero. The same holds for $|G(u)-G_{n}(u)|$ and for $|H(u)-H_{n}(u)|$.\\

\textbf{85.}
\textit{\textbf{Compact sets of curves.}} --- We saw ($\mathrm{n}^{\circ}$ 35) that every compact set formed of elements of a class $(V)$ is bounded. We can here put this statement in a more geometric form by introducing the notion of finite domain; we shall call this a set of points such that the distance between any two of them remains less than a fixed number. Then one easily sees that the necessary and sufficient condition for a set of curves to be bounded (that is, for the distance between any two curves of the set to remain below a fixed number) is that all curves of the set be contained in a same finite domain. To prove this, it is enough to use that the distance of two arbitrary curves is at most the upper limit of the distance from an arbitrary point on one to an arbitrary point on the other, and at least the minimum of the distance from an arbitrary fixed point on one to a variable point on the other ($\mathrm{n}^{\circ}$ 82). Thus we see that \textit{every compact set of curves is contained in a finite domain}. But we now come to the greatest difference between point-sets (even in countably infinite dimensions) and sets of curves. Indeed, contrary to what happens for point-sets, \textit{a set of curves lying in one same finite domain is not necessarily compact}. To see this it suffices, for example, to consider the set $E$ of curves
$$
C_{n} \qquad\qquad x=2\pi t,\quad y=\sin 2^{n}\pi t,\quad z=0 \qquad\qquad(0 \le t \le 1)\,,
$$
all contained in the same circle $x^{2}+y^{2}\le5\pi^{2}$.

It is impossible to extract from this set a sequence of curves having a limit, since the distance between any two of them is at least $1$. Indeed, in any continuous correspondence between any two of them $C_{n}$ and $C_{p}$, one can always find a point $M$ on $C_{n}$ whose distance to the corresponding point is $\ge1$. To see this, note that on sinusoid $C_{n}$ there are $2^{n-1}$ points $A_{1},A_{2},\ldots,A_{2^{n-1}}$ with ordinate $1$, and $2^{n-1}$ points $A_{1}',A_{2}',\ldots,A_{2^{n-1}}'$ with ordinate $-1$. Suppose $n>p$; it suffices to show that, in any continuous correspondence, at least one of points $A_{1},A_{2},\ldots$ of $C_{n}$ corresponds to a point of $C_{p}$ with nonpositive ordinate, or at least one of points $A_{1}',A_{2}',\ldots$ of $C_{n}$ corresponds to a point of $C_{p}$ with nonnegative ordinate; then there are indeed two corresponding points at distance at least $1$.
Now assume there exists a correspondence $G$ where this is not so, and let $b_{1},b_{2},\ldots,b_{1}',b_{2}',\ldots$ be the points of $C_{p}$ corresponding to $A_{1},A_{2},\ldots$ and to $A_{1}',A_{2}',\ldots$ in correspondence $G$.
The ordinates of $b_{1},b_{2},\ldots$ would be positive, those of $b_{1}',b_{2}',\ldots$ negative, and one would meet these points in the order $b_{1},b_{1}',b_{2},b_{2}',\ldots$ while traversing $C_{p}$.
Hence there would be an intersection point of $C_{p}$ with $Ox$ at both ends and on each of the distinct arcs $b_{1}b_{1}', b_{1}'b_{2}, b_{2}b_{2}', b_{2}'b_{3}, \ldots, b_{2^{n-1}}b_{2^{n-1}}'$, i.e. at least $2^{n}+1$ intersections of $C_{p}$ with $Ox$; impossible, since $C_{p}$ meets $Ox$ at $2^{p}+1$ points and $p<n$.\\

\textbf{86.}
\textit{\textbf{Condition for a set of curves to be compact.}} --- The foregoing shows that one must seek an additional condition allowing one to assert that a set of curves contained in a finite domain is compact. Mr.~\textsc{Arzelà} obtained this sufficient condition by immediate application of his theorem on equally continuous functions (\textsc{v}, p.~344). It follows indeed from that theorem that, when one considers a system $R$ of simultaneous representations
$$
C \qquad\qquad\qquad\qquad x=f(t), \quad y=g(t), \quad z=h(t) \qquad\qquad\qquad(0 \le t \le 1)
$$
of a set $E$ of curves $C$, one may assert that this set of curves is compact if $R$ is formed of functions $f,g,h,\ldots$ equally continuous and bounded as a whole.

This extremely important theorem lends itself to many applications. But it should be completed from two points of view. First, it would be good to state, in geometric form, a proposition applying to geometric elements such as curves. Next, Mr.~\textsc{Arzelà} gave no converse to his theorem. Yet this converse could be understood in two ways: one might first think that if one considers a system $R$ of simultaneous representations of curves of a set $E$, the functions of $R$ will necessarily be bounded as a whole and equally continuous when $E$ is compact. \textit{This statement is false}. Indeed, consider a nonclosed curve $C$:
\begin{equation*}
x=f(t), \quad y=g(t), \quad z=h(t) \tag{$0\le t\leq1$}
\end{equation*}
and the system $R$ of representations
$$
x=\cfrac{1}{n}+f\left(e^{-n\left(\frac{1-u^{2}}{u^{2}}\right)}\right), \quad y=\cfrac{1}{n}+g\left(e^{-n\left(\frac{1-u^{2}}{u^{2}}\right)}\right), \quad z=\cfrac{1}{n}+h\left(e^{-n\left(\frac{1-u^{2}}{u^{2}}\right)}\right),
$$
which gives, for each $n$, an admissible representation of a curve $C_{n}$. The set $E$ of curves $C_{n}$ is compact, because another admissible representation of $C_{n}$ is
$$
x=\cfrac{1}{n}+f(t), \quad y=\cfrac{1}{n}+g(t), \quad z=\cfrac{1}{n}+h(t) \qquad\qquad (0 \le t \le 1)
$$
and one sees that $C_{n}$ tends to $C$ as $n$ increases indefinitely.

Yet the functions of $R$ are not equally continuous, for whatever positive numbers $\varepsilon$ and $\eta$, one can always find a number $u$ and an integer $p$ such that
$$
0<1-u<\eta \qquad\left|f_{n}(u)-f_{n}(1)\right|>\varepsilon,
$$
for $n>p$, setting
$$
f_{n}(u)\equiv\cfrac{1}{n}+f\left(e^{-n\left(\frac{1-u^{2}}{u^{2}}\right)}\right)\,.
$$

The second way of stating the converse one might be tempted to give to \textsc{Arzelà}'s theorem (and which this time is correct) is the following: \textit{if a set of curves is compact, there always exists at least one system of simultaneous representations of these curves}
$$
x=f(t),\; y=g(t),\; z=h(t) \qquad\qquad(0 \le t \le 1)
$$
\textit{such that the functions $f,g,h,\ldots$ are bounded as a whole and equally continuous.}

One could derive this proposition directly, but we shall split the proof into two parts, one of which will then allow us to state, \textit{in geometric form}, the condition for a set of curves to be compact.\\

\textbf{87.}
\textit{\textbf{Sets of uniformly divisible curves.}} --- To reach the announced proposition, we begin by transforming Mr.~\textsc{Arzelà}'s condition into another that can be stated geometrically, i.e. without introducing any fixed parametric representation of the curves in the set.

We first call \textit{oscillation of an arc} (\textit{oscillation d'un arc}) $PQ$ of a curve $AB$ the maximum distance between any two points belonging to that arc. We then say that curves of a set $E$ are \textit{uniformly divisible} when, for every $\varepsilon>0$, one can find an integer $n$ such that every curve of $E$ can be decomposed into $n$ consecutive arcs whose oscillations are all less than $\varepsilon$.

It is easy to see that \textit{if a set of curves is compact, the curves of that set are uniformly divisible}. Indeed, otherwise one could find $\varepsilon>0$ such that, for every $n$, there exists at least one curve $C_{n}$ of the set that cannot be decomposed into $n$ arcs whose oscillations are less than $\varepsilon$. But since the set is compact, one can extract from $C_{1},C_{2},\ldots$ a sequence $C_{n_{1}},C_{n_{2}},\ldots,C_{n_{p}},\ldots$ having a limit $C$, and one knows ($\mathrm{n}^{\circ}$ 81) one can choose representations of $C_{n_{p}}$ and $C$:
$$
\left.\begin{array}{ll}
x=F_{p}(t), \quad y=G_{p}(t), \quad z=H_{p}(t) \\
x=F(t), \;\quad y=G(t), \;\,\quad z=H(t)
\end{array}\;\right\} \qquad\qquad 0 \le t \le 1
$$
so that $F_{p},G_{p},H_{p}$ converge uniformly to $F,G,H$. Then, for every $\varepsilon$, one can choose $p$ large enough such that for all $t$:
$$
\left|F_{p}(t)-F(t)\right|<\cfrac{\varepsilon}{9}\,, \quad\left|G_{p}(t)-G(t)\right|<\cfrac{\varepsilon}{9}\,, \quad\left|H_{p}(t)-H(t)\right|<\cfrac{\varepsilon}{9} \,.
$$

On the other hand, divide $C_{n_{p}}$ into $n_{p}$ successive arcs determined by values of $t$:
$$
0,\, \cfrac{1}{n_{p}},\, \cfrac{2}{n_{p}},\, \cfrac{3}{n_{p}}, \ldots, 1.
$$
If $t,t'$ belong to one of these intervals, then $|t-t'|\le\cfrac{1}{n_{p}}$. Now one can choose $p$ large enough so that $|t-t'|\le\cfrac{1}{n_{p}}$ entails
$$
\left|F(t)-F\left(t^{\prime}\right)\right|<\cfrac{\varepsilon}{9}\,, \quad\left|G(t)-G\left(t^{\prime}\right)\right|<\cfrac{\varepsilon}{9}\,, \quad\left|H(t)-H\left(t^{\prime}\right)\right|<\cfrac{\varepsilon}{9} \,;
$$
then for $p$ large enough, in each of the $n_{p}$ arcs into which $C_{n_{p}}$ is divided:
$$
\left|F_{p}(t)-F_{p}\left(t^{\prime}\right)\right|<\left|F_{p}(t)-F(t)\right|+\left|F\left(t^{\prime}\right)-F_{p}\left(t^{\prime}\right)\right|+\left|F(t)-F\left(t^{\prime}\right)\right|<\cfrac{\varepsilon}{3}
$$
and similarly for $G,H$; hence
$$
\sqrt{\left|F_{p}(t)-F_{p}\left(t^{\prime}\right)\right|^{2}+\left|G_{p}(t)-G_{p}\left(t^{\prime}\right)\right|^{2}+\left|H_{p}(t)-H_{p}\left(t^{\prime}\right)\right|^{2}}<\varepsilon .
$$

Thus one reaches a contradiction, since $C_{n_{p}}$ would then be divided into $n_{p}$ arcs whose oscillations are at most $\varepsilon$.\\

\textbf{88.}
So, to prove the converse by which we complete Mr.~\textsc{Arzelà}'s theorem, it will suffice to prove the following proposition:

\textsc{\textbf{Theorem.}} --- \textit{The necessary and sufficient condition for the curves of a set $E$ to be uniformly divisible and contained in one same finite domain is that there exist at least one system of simultaneous representations of the curves of $E$:}
\begin{equation*}
x=f(t), \quad y=g(t), \quad z=h(t) \qquad\qquad(0 \le t \le 1) \tag{9}
\end{equation*}
\textit{such that the functions $f,g,h,\ldots$ are equally continuous and bounded as a whole.}

It is easy to see the condition is sufficient. Let us show it is necessary. For this, start from a first system of simultaneous representations of the curves of the set, namely for curve $C$:
\begin{equation*}
x=F(v), \quad y=G(v), \quad z=H(v) \qquad\qquad (0 \le v \le 1) \tag{10}
\end{equation*}

We may suppose $F,G,H$ are constant on no interval, because for curves in the set reduced to points, oscillation is obviously zero. This being so, let us try to form representation (9) from representation (10). In other words, let us determine for each curve $C$ a function $v=\psi(t)$, continuous and increasing from $0$ to $1$, so that all functions $F[\psi(t)],G[\psi(t)],H[\psi(t)]$ are equally continuous [they will always be bounded as a whole whatever $\psi(t)$].

Now, by hypothesis, for each curve $C$ one can determine an increasing sequence of $v$-values:
$$
0,\quad u_{1}^{(p)},\quad u_{2}^{(p)},\ldots,u_{n_{p}}^{(p)},\quad 1,
$$
such that oscillations of the corresponding arcs of $C$ are all below $\cfrac{1}{p}$, with quantities $u_{q}^{(p)}$ varying with $C$, but their number $n_{p}$ remaining the same for all curves $C$. By inserting, if necessary, new numbers in the sequence $S_{p}:0,u_{1}^{(p)},u_{2}^{(p)},\ldots,1$, one sees easily that one may suppose the following additional conditions: sequence $S_{p+1}$ contains sequence $S_{p}$ and has the same number of terms in each interval determined by $S_{p}$.

This done, let each number $u_{q}^{(p)}$ correspond to $\tau_{q}^{(p)}=\cfrac{q}{n_{p}}$; it is easy to see that if two numbers $u:u_{q}^{(p)},u_{q'}^{(p')}$ ($p\neq p'$ necessarily) are equal, so are the corresponding $\tau$'s, and that if $u_{q}^{(p)}<u_{q'}^{(p')}$ then $\tau_{q}^{(p)}<\tau_{q'}^{(p')}$.\\

\textbf{89.}
We now define, for each curve $C$, the continuous function $t=\varphi(v)$ inverse to the desired function $v=\psi(t)$. I claim it suffices to define it for values of $v$ that are limits of quantities $u_{q}^{(p)}$. In other words, if $v$ is any number in interval $(0,1)$, one can find a sequence of numbers $u_{q_{1}}^{(p_{1})},u_{q_{2}}^{(p_{2})},\ldots$ tending to $v$. Indeed, otherwise there would be $\varepsilon$ such that, for every $p$, each arc $(v-\varepsilon,v)$ and $(v,v+\varepsilon)$ is contained in one of the arcs $(u_{q}^{(p)},u_{q+1}^{(p)})$. Then oscillations of arcs $(v-\varepsilon,v)$ and $(v,v+\varepsilon)$ would be below $\cfrac{1}{n_{p}}$ for every $p$; consequently $F(v),G(v),H(v)$ would be constant from $v-\varepsilon$ to $v+\varepsilon$, contrary to the hypothesis.

Now let $u_{q}^{\prime(p)},u_{q}^{\prime\prime(p)}$ be the $u$'s respectively smaller and larger than $v$. The corresponding numbers $\tau_{q}^{\prime(p)},\tau_{q}^{\prime\prime(p)}$ form two \textsc{Dirichlet} classes; that is, every number of the first class is less than every number of the second, and one can find two numbers, one in each class, whose difference is as small as one wishes. There is therefore one and only one number $t$ between these two classes. This is the number we take as value of $\varphi(v)$. In the case where $v$ is one of the numbers $u_{q}^{(p)}$, the value of $\varphi(v)$ is obviously the corresponding number $\tau_{q}^{(p)}$.\\

\textbf{90.}
One sees immediately that this function $\varphi(v)$, well-defined for every value of $v$ between $0$ and $1$, is increasing and rises from $0$ to $1$. It is also continuous. Indeed, let $\omega$ be any positive number; taking $p$ large enough, one has $n_{p}>\cfrac{2}{\omega}$. With this $p$ chosen, let $u_{q}^{(p)}$ be the greatest number of sequence $S_{p}$ that is less than $v$; then the two numbers $v-u_{q}^{(p)}$, $u_{q+2}^{(p)}-v$ are positive; let $\eta$ be the greater of the two. I claim that for $|v-v'|<\eta$:
$$
\left|\varphi(v)-\varphi\left(v^{\prime}\right)\right|<\omega .
$$

Indeed, $v$ and $v'$ then lie between $u_{q}^{(p)}$ and $u_{q+2}^{(p)}$, hence:
$$
\left|\varphi(v)-\varphi\left(v^{\prime}\right)\right|<\varphi\left(u_{q+2}^{(p)}\right)-\varphi\left(u_{q}^{(p)}\right)=\cfrac{2}{n_{p}}<\omega .
$$

Thus we have indeed formed, for each curve $C$, a continuous function $t=\varphi(v)$ that constantly increases from $0$ to $1$. Then the inverse function $v=\psi(t)$ is also increasing and continuous, and one can form a system $R$ of simultaneous representations:
$$
x=f(t), \quad y=g(t), \quad z=h(t) \qquad\qquad (0\le t \le 1),
$$
by setting:
$$
f(t) \equiv F[\psi(t)], \quad g(t) \equiv G[\psi(t)], \quad h(t) \equiv H[\psi(t)] .
$$

The functions of system $R$ are indeed equally continuous; indeed, from the construction of $\psi(t)$, the oscillation of the curve is less than $\cfrac{1}{p}$ in each arc bounded by the values of $t: 0, \cfrac{1}{n_{p}}, \cfrac{2}{n_{p}}, \ldots, 1$, and the oscillations of $f(t),g(t),h(t)$ are less than those of $C$. Hence one sees that, for every $p$, one can determine a number $\cfrac{1}{n_{p}}$ independent of the curve $C$ considered, such that the inequality $|t'-t''|<\cfrac{1}{n_{p}}$ entails:
$$
\left|f\left(t^{\prime}\right)-f\left(t^{\prime \prime}\right)\right|<\cfrac{1}{p}\,, \quad\left|g\left(t^{\prime}\right)-g\left(t^{\prime \prime}\right)\right|<\cfrac{1}{p}\,, \quad\left|h\left(t^{\prime}\right)-h\left(t^{\prime \prime}\right)\right|<\cfrac{1}{p} \,.
$$

\textbf{91.}
Combining the preceding statements, one obtains the necessary and sufficient condition for a set of curves to be compact in a form that involves no parametric representation of curves in that set.

\textsc{\textbf{Theorem.}} --- \textit{The necessary and sufficient condition for a set of curves to be compact is that the curves of this set be contained in one and the same finite domain and be equally divisible.}\\

\textbf{92.}
\textit{\textbf{Application to rectifiable curves.}} --- Call a polygonal line inscribed in a curve a broken line having the same endpoints as the curve and whose vertices are points of the curve encountered in the same order when one traverses the curve or the broken line. One proves (\textsc{iv}, p.~59) that, when the maximum length of sides of an inscribed polygonal line in an arc $AB$ tends to zero, the perimeter of that polygonal line tends to the upper limit of lengths of polygonal lines inscribed in $AB$, and this limit is called \textit{the length of the curve}. One says the curve is \textit{rectifiable} when this limit is finite.

The preceding theorem allows one to assert that a set of rectifiable curves contained in one and the same finite domain is compact whenever lengths of these curves are bounded as a whole\,\footnote{\,A particular case of this theorem was proved directly by Mr.~\textsc{Hilbert} (\textsc{xx}). The generalization was then carried out by Mr.~\textsc{Lebesgue} in his Thesis (\textsc{xvii}).}. Indeed, suppose one divides each curve of the set into $n$ arcs of equal length. The length of an arbitrary curve $C$ in the set is less than a fixed number $L$ independent of $C$; therefore each equal arc has length less than $\cfrac{L}{n}$. Now oscillation of an arc is obviously less than the length of that arc. Therefore, by taking $n$ large enough $\left(n>\cfrac{L}{\varepsilon}\right)$, each curve $C$ will be divided into arcs of length less than $\varepsilon$. Then curves $C$ are uniformly divisible and their set is compact.\\

\textbf{93.}
One might ask whether the sufficient condition just stated is also necessary. The answer is negative: \textit{one can give an example of rectifiable curves lying in one and the same finite domain and forming a compact set, without their lengths being bounded as a whole.}

Let indeed
$$
x=f(t), \quad y=g(t), \quad z=0 \qquad\qquad\qquad (0 \le t \le 1)\,,
$$
be a fixed representation of a non-rectifiable curve $C$, but whose arcs $\cfrac{1}{n}\le t\le1$ are rectifiable for every $n$. It will suffice to consider curves $C_{n}$:
$$
\begin{gathered}
x=\cfrac{1}{n}+f\left[\cfrac{1}{n}+u\left(1-\cfrac{1}{n}\right)\right], \quad y=\cfrac{1}{n}+g\left[\cfrac{1}{n}+u\left(1-\cfrac{1}{n}\right)\right], \quad z=0 \\
(0 \le u \le 1) .
\end{gathered}
$$

These curves form a compact set, for it is obvious that $C_{1},C_{2},\ldots$ tend to $C$. They are rectifiable and contained in a same finite domain. Their lengths, which are those of arcs $\cfrac{1}{n}\le t\le1$ of curve $C$, are nevertheless not bounded as a whole.

Moreover, it is very easy to construct a curve such as the one whose existence we assumed; one may take, for example, the hyperbolic spiral obtained by setting:
$$
f(t)\equiv t\cos\cfrac{1}{t}, \quad g(t)\equiv t\sin\cfrac{1}{t}\,.
$$

\textit{One can even give an example of curves contained in one and the same finite domain and forming a compact set, without any one of them being rectifiable.}

It suffices to take again the curve $C$ just considered, and consider curves $C_{n}$:
$$
x=\cfrac{1}{n}+f(t), \quad y=\cfrac{1}{n}+g(t), \quad z=0 \qquad\qquad(0 \le t \le 1) .
$$

\subsubsection*{Application of the general theorems.}
\phantomsection\addcontentsline{toc}{subsubsection}{Application of the general theorems.}
We may now apply results of the First Part to sets of curves and line functions.

\textbf{94.}
\textit{\textbf{Sets of curves.}}

\textsc{\textbf{Theorem.}} --- \textit{Let $E$ be an arbitrary set of arcs of continuous curves. $1^{\circ}$ The derived set of $E$ is closed. $2^{\circ}$ One can extract from $E$ a countable sequence $D$ of curves of $E$, such that every curve of $E$ belongs to $D$, or is a limit curve of $D$. $3^{\circ}$ Either $E$ is countable, or there exists a curve $C_{1}$ of $E$ such that, for every value of $\varepsilon$, there is an uncountable infinity of curves of $E$ entirely contained in the volume generated by a sphere of radius $\varepsilon$ whose center runs along $C_{1}$.}

\textsc{\textbf{Theorem.}} --- \textit{Let $G$ be a set of uniformly divisible curves ($\mathrm{n}^{\circ}$ 87) contained in one and the same finite domain: the elements of $G$ that are not limits of curves of $G$ form a countable set.}

\textsc{\textbf{Theorem.}} --- \textit{If the preceding set $G$ is also closed, in other words if $G$ is an arbitrary extremal set: $1^{\circ}$ one can decompose it into a countable set and a perfect set or null set with no common elements; $2^{\circ}$ there exists a set of curves whose derived set coincides with $E$.}

\textsc{\textbf{Theorem.}} --- \textit{Every perfect set of curves has the power of the continuum.}

One may also state the theorems of $\mathrm{n}^{\circ}$ 40, 42, 43, 45, \ldots taking continuous curves as elements.\\

\textbf{95.}
\textit{\textbf{Line functions.}} --- Operations acting on arcs of curves have received from Mr.~\textsc{Volterra} the name \textit{line functions} (\textit{fonctions de lignes}).

Mr.~\textsc{Volterra} succeeded in showing how one might define the derivative of a line function, and from this viewpoint obtained extremely interesting results (\textsc{xv}).

In an article in the Annales scientifiques de l'École Normale (\textsc{xxvii}), I gave some properties of line functions whose derivatives, in Mr.~Volterra's sense, have a given special form.

I shall leave aside here this type of investigation, which departs somewhat from our present object: the study of properties common to \textit{the most general} continuous operations.

General results apply again to line functions.

\textsc{\textbf{Theorem.}} --- Every line function continuous on an extremal set of curves: $1^{\circ}$ is bounded on this set, $2^{\circ}$ attains each extremum there. \textit{Conversely, if all line functions continuous on one and the same set $E$ of curves, $1^{\circ}$ are each bounded, $2^{\circ}$ each attain their extremum on at least one curve of the set, this set is necessarily extremal.}

\textsc{\textbf{Theorem.}} --- \textit{Let $E$ be a closed set of curves contained in one and the same finite domain and uniformly divisible: $1^{\circ}$ for a sequence of line functions $U_{1},U_{2},\ldots$ continuous on $E$ to converge to a line function continuous on $E$, it is necessary and sufficient that the sequence converge quasi-uniformly. $2^{\circ}$ for line functions continuous on $E$ to form a family $\mathfrak{F}$ such that from every infinite subset of $\mathfrak{F}$ one can extract a sequence of line functions converging uniformly on $E$, it is necessary and sufficient that line functions of $\mathfrak{F}$ be bounded as a whole and equally continuous on $E$.}

\subsubsection*{Sets and functions of surfaces.}
\phantomsection\addcontentsline{toc}{subsubsection}{Sets and functions of surfaces.}

\textbf{96.}
One may develop for surfaces a theory that would lead to propositions entirely analogous to those above. But at the beginning there are precautions to take, leading to some length. Therefore, in order not to lengthen this Memoir, we defer this point to a later publication.

\subsection*{NOTE I.}
\phantomsection\addcontentsline{toc}{subsection}{NOTE I.}
\subsubsection*{Construction of a nondecreasing continuous function that is constant only on given intervals.}
\phantomsection\addcontentsline{toc}{subsubsection}{Construction of a nondecreasing continuous function that is constant only on given intervals.}

\textbf{97.}
Consider a set $G$ of intervals $I$ without common points, situated on the $x$-axis in the fundamental segment $(0,1)$. We shall show that one can form a continuous function $\varphi(x)$ increasing without ever decreasing from $0$ to $1$, and constant only on each of the intervals $I$\,\footnote{\,We use this proposition in $\mathrm{n}^{\circ}$ 84.}.

Indeed, first note that set $G$ is countable, for it is the union of sets $G_{n}$ each formed by intervals $I$ whose lengths are greater than $\cfrac{1}{n+1}$ and at most equal to $\cfrac{1}{n}$. Now there is obviously a finite number ($\le n$) of intervals $I$ in $G_{n}$, since the $I$ have no common point.

One can therefore arrange these intervals into a countable sequence $I_{1},I_{2},I_{3},\ldots$; from the foregoing, length of $I_{n}$ tends to zero with $\cfrac{1}{n}$, unless there are only finitely many such intervals. In the latter case there is no difficulty in forming $\varphi(x)$.\\

\textbf{98.}
Suppose then there are infinitely many intervals $I$. The set $P$ of points of segment $(0,1)$ that are not interior in the strict sense to any interval $I$ is a perfect set (\textsc{ii}, p.~12).

$1^{\circ}$. First examine the case where $P$ contains points $0$ and $1$ and is dense in no interval between $0$ and $1$. One then knows (\textsc{ii}, p.~14) that one can establish a one-to-one reciprocal correspondence between all intervals $I$, on one hand, and all positive rational numbers $<1$, on the other, in such a way that corresponding elements are arranged in the same manner on $Ox$.
Call $u_{n}$ the rational number thus corresponding to $I_{n}$. Now let $x$ be any point of segment $(0,1)$, other than $0$ and $1$.
Let $I'$ be any interval $I$ not completely to the right of $x$, and $I''$ any interval $I$ not entirely to the left of $x$. If $u'$ and $u''$ denote corresponding rational numbers, one sees that every rational number between $0$ and $1$ must be either a $u'$ or a $u''$, and one always has $u'\le u''$.
Thus one defines two classes of numbers having exactly one number between them, which we call $\varphi(x)$.
If in addition we set $\varphi(0)=0$, $\varphi(1)=1$, we see that $\varphi(x)$ is defined at every point of segment $(0,1)$. In particular, at every point of an interval $I_{n}$, one has $\varphi(x)=u_{n}$.
So function $\varphi(x)$ is constant on each interval $I_{n}$. I now claim that if $x<x'$, one necessarily has $\varphi(x)<\varphi(x')$, provided $x$ and $x'$ do not belong to a same interval $I$. Indeed, it suffices to note that under this hypothesis there is at least one interval $I_{n}$ entirely contained in $(x,x')$. Then, by construction of $\varphi(x)$:
$$
\varphi(x)<u_{n}<\varphi\left(x^{\prime}\right) .
$$

Thus $\varphi(x)$ is nondecreasing, goes from $0$ to $1$, and is constant only on each interval $I$. Hence one can define, at each point $x$, the upper limit $\varphi(x-0)$ of values of $\varphi(x')$ for $x'<x$, and one has $\varphi(x-0)\le\varphi(x)$. Now $\varphi(x-0)$ is at least equal to the upper limit of values of $\varphi(x')$ when $x'$ lies in some interval $I$ and remains $<x$. Therefore $\varphi(x)=\varphi(x-0)$. In other words, $\varphi(x)$ is left-continuous; similarly, it is right-continuous. In conclusion, function $\varphi(x)$ is continuous and satisfies the announced conditions.

$2^{\circ}$. Now suppose $P$, still containing $0$ and $1$, may be dense in certain intervals. Then one can form a second sequence of intervals $K_{1},K_{2},\ldots$ whose points, endpoints included, are all points of $P$, and such that set $Q$ of points not belonging in the strict sense to either the $K$ or the $I$ is dense in no interval. Then place in each interval $K_{n}$ a sequence of disjoint intervals: $K_{n}^{(1)},K_{n}^{(2)},\ldots$, such that set of points of $K_{n}$ not belonging to any of intervals $K_{n}^{(1)},K_{n}^{(2)},\ldots$ has measure zero (\textsc{ii}, p.~16)\,\footnote{\,For example, in $K_{n}$ one may place a set of intervals similar to the complement [in segment $(0,1)$] of the Cantor set cited below.}. Finally let $R$ be set of points not belonging in the strict sense to any of intervals $I_{1},I_{2},\ldots; K_{1}^{(1)},K_{1}^{(2)},\ldots; K_{2}^{(1)},K_{2}^{(2)},\ldots$

By their construction, these intervals have no common point (whereas an $I$ and a $K$ could share an endpoint). Set $R$ is therefore perfect, contains $0$ and $1$, and is dense in no interval. One can therefore form a nondecreasing continuous function $\varphi_{1}(x)$ constant only on intervals $I_{n}$ or $K_{n}^{(p)}$. Let, on the other hand, $\varphi_{2}(x)$ be the measure of set of points of $P$ between $0$ and $x$.
This is a nondecreasing continuous function; moreover this function is constant on intervals $I$, and certainly increasing on intervals $K_{n}^{(p)}$.
Therefore function
$$
\varphi(x)=\cfrac{\varphi_{1}(x)+\varphi_{2}(x)}{\varphi_{1}(1)+\varphi_{2}(1)}
$$
is nondecreasing continuous, goes from $0$ to $1$, and is constant only on intervals $I$; thus it answers the question.

$3^{\circ}$. Finally, if $P$ does not contain both $0$ and $1$, in other words if one of intervals $I$ ends at $0$ or at $1$, one is reduced to previous cases by removing that interval.\\

\textbf{99.}
There is an interesting special case where one can form function $\varphi(x)$ explicitly: when $P$ is Cantor's set of measure zero (\textsc{ii}, p.~12). Consider any number $x$ between $0$ and $1$, and write its value in base-$3$ numeration\,\footnote{\,We retain the original French-style decimal comma notation (e.g., $0,a_{1}a_{2}\ldots$) in this translated text.}:
$$
x=0, a_{1}a_{2}a_{3}\ldots
$$

Numbers $a_{1},a_{2},\ldots$ can only take values $0,1,2$, and one may always suppose that a number written as:
$$
0,\, a_{1}\, a_{2}\, \ldots a_{n}\, 0\, 2\, 2\, 2 \ldots
$$
or
$$
0,\, b_{1}\, b_{2}\, \ldots b_{p}\, 2\,0\,0\,0 \ldots
$$
is written in respectively equivalent form:
$$
0,\, a_{1}\, a_{2}\, \ldots a_{n}\, 1\,0\,0\,0 \ldots
$$
or
$$
0,\, b_{1}\, b_{2}\, \ldots b_{p}\, 1\,2\,2\,2 \ldots
$$

This being so, we always set:
$$
\varphi\left(0, a_{1}a_{2}\ldots\right)=0', b_{1}b_{2}\ldots,
$$
where $0',\ldots$ indicates passage to base-$2$ system, taking generally for $b_{n}$ a number equal to $0$ if one of numbers $a_{1},a_{2},\ldots,a_{n-1}$ equals $1$, and otherwise equal to $1$ if $a_{n}=2$, and equal to $a_{n}$ if $a_{n}\neq2$. In other words, we set:
$$
\begin{gathered}
\varphi\left(\cfrac{a_{1}}{3}+\cfrac{a_{2}}{3^{2}}+\cdots+\cfrac{a_{n}}{3^{n}}+\cdots\right) \\
=\cfrac{\alpha(a_{1})}{2}+\cfrac{\alpha(a_{2})\beta(a_{1})}{2^{2}}+\cdots+\cfrac{\alpha(a_{n})\beta(a_{1})\ldots\beta(a_{n-1})}{2^{n}}+\cdots,
\end{gathered}
$$
with conditions:
$$
\alpha(0)=0, \quad \alpha(1)=\alpha(2)=1, \quad \beta(1)=0, \quad \beta(0)=\beta(2)=1 .
$$

It is easy to see that one thus assigns the same value of $\varphi(x)$ to all numbers $x$ such that
$$
x=0, c_{1}c_{2}\ldots c_{n}\,1\,\lambda_{1}\lambda_{2}\ldots,
$$
where $c_{1},c_{2},\ldots,c_{n}$ remain fixed and all differ from $1$.

These points form an interval; the correspondence thus established between this interval and the constant value of $\varphi(x)$ is precisely that indicated by Mr.~\textsc{Cantor} to prove that set $P$ has the power of the continuum.

\subsection*{NOTE II.}
\phantomsection\addcontentsline{toc}{subsection}{NOTE II.}
\subsubsection*{On the effective computation of the distance \textit{(\'ecart)} between two curves.}
\phantomsection\addcontentsline{toc}{subsubsection}{On the effective computation of the distance \textit{(\'ecart)} between two curves.}

\textbf{100.}
The theory developed in the \textsc{First Part} shows that once existence of a distance (\'ecart) is established, the very value of this distance is no longer of very great importance. In other words, it is enough to know one can define a number satisfying general conditions of distance ($\mathrm{n}^{\circ}$ 49), without concern for how it might be computed effectively. This is all the more true since one could choose infinitely many other definitions of distance by means of which the general theory could develop exactly in the same way, and which would not change limit curves of a set of curves. For example, if $e$ is the distance between two curves as we defined it, one could everywhere replace it by $3e$, or by $\sqrt{e}$, or by $\cfrac{e}{1+e}$. But the primitive definition has the advantage of giving, for the distance of two curves reduced to two points, a quantity equal to distance of those two points; and for two curves, one fixed and the other going off to infinity, a number tending to infinity.

The definition of distance we gave does not indicate how one could compute it effectively, and this computation would certainly be very difficult in general. One can nevertheless give simple examples where this computation is possible. In general one should proceed as follows. One proves that, in every continuous correspondence between the two curves, the maximum $d$ of distance between corresponding points is greater than or equal to a fixed number $e$. One then shows there exists at least one such correspondence where this maximum $d$ is exactly equal to $e$.

For example, consider two line segments $\overline{AB},\overline{A'B'}$, and let $\Delta$ be the greater of lengths $AA',BB'$. One always has $d\ge\Delta$, and one has $d=\Delta$ when correspondence is such that two corresponding points $M,M'$ divide $\overline{AB}$ and $\overline{A'B'}$ in the same ratio. Therefore distance between the two segments equals $\Delta$.

One may also be led to define distance between two closed curves. One then considers it as lower limit of all distances obtained by considering them in all possible ways as curve-arcs each with two coincident endpoints. One would then easily see that if one considers two circles of radii $R_{1},R_{2}$, centers $O_{1},O_{2}$, lying in the same plane and oriented in the same way, their distance is
$$
O_{1}O_{2}+\left|R_{1}-R_{2}\right|.
$$

A less particular example is as follows. Consider in one same plane two \textit{convex} curves $C_{1}$ and $C_{2}$, where curve $C_{1}$ is assumed to have continuous tangent and be entirely contained in $C_{2}$. Under these conditions, one can see that distance between $C_{1}$ and $C_{2}$ is equal to the maximum segment intercepted by $C_{2}$ on the outward normal to $C_{1}$.

In these examples, one verifies immediately that vanishing of distance is the necessary and sufficient condition for the two curves to coincide, which often allows prediction of probable value of distance.\\

Paris, April 2, 1906.\\

\textsc{Maurice Fréchet.}

\newpage
\subsection*{CITED AUTHORS.}
\phantomsection\addcontentsline{toc}{subsection}{CITED AUTHORS.}
{\small 

\noindent\textbf{Collection of Monographs on Function Theory}

\noindent Published under the direction of \textsc{Émile Borel} (Paris, Gauthier-Villars):
\begin{enumerate}
    \item[(\textsc{i})] \textsc{Borel}, \emph{Leçons sur la théorie des fonctions. Éléments de la théorie des ensembles et applications}, 1898.
    \item[(\textsc{ii})] \textsc{Borel}, \emph{Leçons sur les fonctions de variables réelles et les développements en séries de polynômes}, 1905.
    \item[(\textsc{iii})] \textsc{Baire}, \emph{Leçons sur les fonctions discontinues}, 1905.
    \item[(\textsc{iv})] \textsc{Lebesgue}, \emph{Leçons sur l'intégration et la recherche des fonctions primitives}, 1904.
\end{enumerate}

\medskip

\noindent\textbf{Articles by Mr.~\textsc{Arzelà}}
\begin{enumerate}
    \item[(\textsc{v})] \emph{Funzioni di linee}, \emph{Rendiconti della R.~Accademia dei Lincei}, 1889, vol.~V, 1\textsuperscript{o} semestre, pp.~342--348.\\
    \emph{Sulle funzioni di linee}, \emph{Memorie della R.~Accademia delle Scienze dell'Istituto di Bologna, sezione delle Scienze Fisiche e Matematiche}, s.~V, t.~V (1895), pp.~55--74.
    
    \item[(\textsc{vi})] \emph{Sull'integrabilità delle equazioni differenziali ordinarie}, \emph{Ibid.}, s.~V, t.~V (1895), pp.~75--88.\\
    \emph{Sull'esistenza degli integrali nelle equazioni differenziali ordinarie}, \emph{Ibid.}, s.~V, t.~VI (1896), pp.~33--42.
    
    \item[(\textsc{vii})] \emph{Sulle serie di funzioni}, \emph{Ibid.}, s.~V, t.~VIII (1899--1900), pp.~3--58, 91--134.
    
    \item[(\textsc{viii})] \emph{Sull'inversione di un sistema di funzioni}, \emph{Rendiconto delle sessioni della R.~Accademia delle Scienze dell'Istituto di Bologna}, nuova serie, vol.~VII (1902--1903), pp.~182--201.\\
    \emph{Sulle serie di funzioni ugualmente oscillanti}, \emph{Ibid.}, nuova serie, vol.~VIII (1903--1904), pp.~143--154.\\
    \emph{Sul principio di Dirichlet}, \emph{Ibid.}, nuova serie, vol.~I (1896--1897), pp.~71--84.
    
    \item[(\textsc{ix})] \emph{Sulle serie di funzioni di variabili reali}, \emph{Ibid.}, nuova serie, vol.~VII (1902--1903), pp.~22--32.
    
    \item[(\textsc{x})] \emph{Sulle serie di funzioni analitiche}, \emph{Ibid.}, nuova serie, vol.~VII (1902--1903), pp.~33--42.
\end{enumerate}

\medskip

\noindent\textbf{Miscellaneous}
\begin{enumerate}
    \item[(\textsc{xi})] \textsc{Hadamard}, \emph{Sur les opérations fonctionnelles}, \emph{Comptes Rendus de l'Académie des Sciences}, t.~CXXXVI (1\textsuperscript{er} semestre 1903), pp.~351--354.
    
    \item[(\textsc{xii})] \textsc{Hadamard}, \emph{Sur certaines applications possibles de la théorie des ensembles}, \emph{Verhandlungen des ersten Internationalen Mathematiker-Kongresses} (1897), pp.~201--202.
    
    \item[(\textsc{xiii})] \textsc{Borel}, \emph{Contribution à l'analyse arithmétique du continu}, \emph{Journal de Mathématiques pures et appliquées}, 5\textsuperscript{e} série, t.~IX (1903), pp.~329--375.
    
    \item[(\textsc{xiv})] \textsc{Borel}, \emph{Sur quelques points de la théorie des fonctions}, \emph{Annales scientifiques de l'École Normale supérieure}, 3\textsuperscript{e} série, t.~XII (1895), pp.~9--55.
    
    \item[(\textsc{xv})] \textsc{Volterra}, \emph{Sur une généralisation de la théorie des fonctions d'une variable imaginaire}, \emph{Acta Mathematica}, t.~XII (1889), pp.~233--286.
    
    \item[(\textsc{xvi})] \textsc{Pincherle}--\textsc{Amaldi}, \emph{Le operazioni distributive e le loro applicazioni all'Analisi} (Bologna, 1901).
    
    \item[(\textsc{xvii})] \textsc{Lebesgue}, \emph{Intégrale, Longueur, Aire}, \emph{Annali di Matematica pura ed applicata}, s.~III, t.~VII (1902), pp.~231--359.
    
    \item[(\textsc{xviii})] \textsc{Le Roux}, \emph{Les fonctions d'une infinité de variables indépendantes}, \emph{Nouvelles Annales de Mathématiques}, 4\textsuperscript{e} série, t.~IV (1904), pp.~448--458.
    
    \item[(\textsc{xix})] \textsc{Le Roux}, \emph{Recherches sur les équations aux dérivées partielles}, \emph{Journal de Mathématiques pures et appliquées}, 5\textsuperscript{e} série, t.~IX (1903), pp.~403--455.
    
    \item[(\textsc{xx})] \textsc{Hilbert}, \emph{Sur le principe de Dirichlet} (traduction de M.~Laugel), \emph{Nouvelles Annales de Mathématiques}, 3\textsuperscript{e} série, t.~XIX (1900), pp.~337--344.
    
    \item[(\textsc{xxi})] \textsc{Drach}, \emph{Essai sur une théorie générale de l'intégration et sur la classification des transcendantes}, \emph{Annales scientifiques de l'École Normale supérieure}, 3\textsuperscript{e} série, t.~XV (1898), pp.~243--384.
    
    \item[(\textsc{xxii})] \textsc{Ascoli}, \emph{Le curve limiti di una varietà data di curve}, \emph{Memorie della classe di Scienze Fisiche, Matematiche e Naturali della R.~Accademia dei Lincei}, vol.~XVIII (1883), pp.~521--586.
    
    \item[(\textsc{xxiii})] \textsc{Zoretti}, \emph{Sur les fonctions analytiques uniformes qui possèdent un ensemble parfait discontinu de points singuliers}, \emph{Journal de Mathématiques pures et appliquées}, 6\textsuperscript{e} série, t.~I (1905), pp.~1--51.
    
    \item[(\textsc{xxiv})] \textsc{de Séguier}, \emph{Éléments de la théorie des groupes abstraits} (Paris, Gauthier-Villars, 1904).
    
    \item[(\textsc{xxv})] \textsc{Montel}, \emph{Sur les suites de fonctions analytiques}, \emph{Comptes Rendus de l'Académie des Sciences}, t.~CXXXVIII (1\textsuperscript{er} semestre 1904), pp.~469--471.
    
    \item[(\textsc{xxvi})] \textsc{Bourlet}, \emph{Sur les opérations en général et les équations différentielles linéaires d'ordre infini}, \emph{Annales scientifiques de l'École Normale supérieure}, 3\textsuperscript{e} série, t.~XIV (1897), pp.~133--190.
    
    \item[(\textsc{xxvii})] \textsc{Fréchet}, \emph{Sur les fonctions de lignes fermées}, \emph{Annales scientifiques de l'École Normale supérieure}, 3\textsuperscript{e} série, t.~XXI (1904), pp.~557--571.\\
    \emph{Sur une extension de la méthode de Jacobi--Hamilton}, \emph{Annali di Matematica pura ed applicata}, s.~III, t.~XI (1905), pp.~187--199.\\
    \emph{Sur les opérations linéaires}, \emph{Transactions of the American Mathematical Society}, t.~V (1904), pp.~493--499; t.~VI (1905), pp.~134--140.
\end{enumerate}
}

%% file: appendix/levy1950_en_body.tex





\textbf{Direct definitions of the distance between two laws.} --- The preceding definition has the drawback that it may be very difficult to determine the distance between two laws given by their total probability functions (\textit{fonctions des probabilités totales}). It is therefore useful to have direct definitions. We shall indicate two of them.\\

$1^{\circ}$ The first 
was the one I indicated, or 
rather used (without pronouncing the word distance), in 1925 
(\emph{Calcul des probabilités}, pp. 199--200). Let us define the law $\bm{L}$ corresponding to a random variable $\bm{X}$ by the curve $\bm{\Gamma}$ representing the total probability function
$$
y=\operatorname{Pr}[\bm{X}<x]
$$
with, however, the convention that to each value of $x$ for which this function is discontinuous, we associate the segment
$$
\operatorname{Pr} [\bm{X}<x] \leqq y \leqq \operatorname{Pr} [\bm{X} \leqq x] .
$$

The curve is then continuous, and is manifestly cut at one point and one point only by any line parallel to the straight line
$$
x+y=0
$$

If two laws $\bm{L}$ and $\bm{L}^{\prime}$ are thus represented by two curves $\bm{\Gamma}$ and $\bm{\Gamma}^{\prime}$ which the line $x+y=c$ cuts respectively at $\bm{A}$ and $\bm{A}^{\prime}$, the distance $\left(\bm{L}, \bm{L}^{\prime}\right)=\left(\bm{\Gamma}, \bm{\Gamma}^{\prime}\right)$ will be the maximum of $\bm{AA}^{\prime}$ as $c$ varies from $-\infty$ to $+\infty$. This maximum is certainly attained, because of the continuity of the curves $\bm{\Gamma}$ and $\bm{\Gamma}^{\prime}$.


To justify this definition, it is enough to observe that the triangle inequality (2) is a consequence of the inequality
\begin{equation*}
\bm{A} \bm{A}^{\prime} \leqq \bm{A}\bm{A}^{\prime\prime} + \bm{A}^{\prime}\bm{A}^{\prime \prime} \, , \tag{3}
\end{equation*}
written for the value of $c$ which makes $\bm{AA}^{\prime}$ a maximum, hence equal to $(\bm{L}, \bm{L}^\prime)$. The line $x+y=c$ cuts at $\bm{A}^{\prime \prime}$ the curve $\bm{\Gamma}^{\prime \prime}$ corresponding to $\bm{L}^{\prime \prime}$, and indeed one has
$$
\bm{A}\bm{A}^{\prime \prime} \leqq\left(\bm{L}, \bm{L}^{\prime \prime}\right), \quad \bm{A}^{\prime} \bm{A}^{\prime \prime} \leqq\left(\bm{L}^{\prime}, \bm{L}^{\prime \prime}\right),
$$
and (2) follows.\\

$2^{\circ}$ Let us now give a second possible definition, which in fact includes an infinity of definitions, since it uses the notion of distance between two points $\bm{A}$ and $\bm{A}^{\prime}$ of the plane, and it is not necessary for this distance to be the Euclidean distance; we shall suppose only that it tends to zero at the same time as the Euclidean distance. In any case, the distance $(\bm{A}, \bm{\Gamma}^{\prime})$ from $\bm{A}$ to $\bm{\Gamma}^{\prime}$ will be the minimum of $\bm{AA}^{\prime}$ as $\bm{A}^{\prime}$ ranges over $\bm{\Gamma}^{\prime}$, and the distance $\left(\bm{L}, \bm{L}^{\prime}\right)=\left(\bm{\Gamma}, \bm{\Gamma}^{\prime}\right)$ will be the greater of the following two numbers: the maximum of $(\bm{A}, \bm{\Gamma}^{\prime})$ as $\bm{A}$ ranges over $\bm{\Gamma}$, and the maximum of $(\bm{A}^{\prime}, \bm{\Gamma})$ as $\bm{A}^{\prime}$ ranges over $\bm{\Gamma}^{\prime}$.

Let us show that inequality (2) is indeed satisfied. Let $\bm{A}$ be any point of $\bm{\Gamma}$, let $\bm{A}^{\prime \prime}$ be the point of $\bm{\Gamma}^{\prime \prime}$ nearest to $\bm{A}$ (or one of those points, if there are several), and let $\bm{A}^{\prime}$ be the point of $\bm{\Gamma}^{\prime}$ nearest to $\bm{A}^{\prime \prime}$. One has
$$
\left(\bm{A}, \bm{\Gamma}^{\prime}\right) \leqq \bm{A} \bm{A}^{\prime} \leqq \bm{A} \bm{A}^{\prime \prime}+\bm{A}^{\prime} \bm{A}^{\prime \prime} \leqq\left(\bm{L}, \bm{L}^{\prime \prime}\right)+\left(\bm{L}^{\prime}, \bm{L}^{\prime \prime}\right)\,.
$$

Since the same upper bound applies to $(\bm{A}, \bm{\Gamma})$, whatever $\bm{A}^{\prime}$ on $\bm{\Gamma}^{\prime}$ may be, inequality (2) follows.

This definition of $(\bm{\Gamma}, \bm{\Gamma}^{\prime})$ can be applied to arbitrary curves, but would not be without drawbacks. It would lead, for example, if an ellipse had a very small minor axis, to regarding the complete ellipse and the half-ellipse situated on a determined side of the major axis as two curves that differ very little. M. Fréchet gave a definition of the distance between two curves which applies to the most general Jordan curves and avoids the drawback we have pointed out. We do not need it here, the preceding definition being sufficient for the very special curves we have to consider.

\textbf{A new comparison between the distance between two laws and that between two random variables.} --- We shall consider here only those definitions of the distance $(\bm{X}, \bm{Y})$ which depend only on the nature of the random variable $\bm{Z=|X-Y|}$; then
\begin{equation*}
(\bm{X}, \bm{Y})=(\bm{0}, \bm{Z}) . \tag{4}
\end{equation*}

Now, we have remarked that, since the notion of correlation disappears when one of the variables ceases to be random, there is no longer any need in that case to distinguish between the nearness of the two variables and that of the two corresponding laws. Perhaps one can even, if $\bm{L}_{0}$ and $\bm{L}_{1}$ denote respectively the laws corresponding to zero and to the nonnegative variable $\bm{Z}$, realize the equality of the two distances in question:
\begin{equation*}
(\bm{0}, \bm{Z})=\left(\bm{L}_{0}, \bm{L}_{1}\right) . \tag{5}
\end{equation*}

This remark leads one to pose the following two problems:\\

\textbf{First problem.} --- \textit{Given an acceptable definition of $\left(\bm{L}, \bm{L}^{\prime}\right)$ (that is, one satisfying the triangle inequality), is the definition of $(\bm{X}, \bm{Y})$ deduced from it by formulas (4) and (5) acceptable?}\\

\textbf{Second problem.} --- \textit{Given an acceptable definition of $(\bm{X}, \bm{Y})$, can one give an acceptable definition of $\left(\bm{L}, \bm{L}^{\prime}\right)$ which, in the case of the distance $\left(\bm{L}_{0}, \bm{~L}_{1}\right)$, reduces to that resulting from formulas (4) and (5)?}\\

Here we shall content ourselves with posing these general questions and indicating one particular way of realizing conditions (4) and (5). One has only to take for $(\bm{A}, \bm{A}^{\prime})$, not the Euclidean distance between these two points, but the sum
$$
\left|x-x^{\prime}\right|+\left|y-y^{\prime}\right|\,,
$$
$x, y$ being the coordinates of $\bm{A}$, and $x^{\prime}$ and $y^{\prime}$ those of $\bm{A}^{\prime}$, and then to define $\left(\bm{L}, \bm{L}^{\prime}\right)$ starting from $\left(\bm{A}, \bm{A}^{\prime}\right)$ as we did in item $2^{\circ}$ of the preceding paragraph. One easily verifies that the distance $(\bm{L}_{0}, \bm{L}_{1})$ is then the distance from the point $x=0,\, y=1$ to the curve $\bm{\Gamma}_{1}$ corresponding to $\bm{L}_{1}$, that is to say the minimum, with $\varepsilon$ varying from $0$ to $+\infty$, of the sum
$$
\varepsilon+\operatorname{Pr} \left[ | \bm{X}-\bm{Y} > \varepsilon \right],
$$
which comes to the same thing as saying the infimum of
$$
\varepsilon+\operatorname{Pr} \left[|\bm{X}-\bm{Y}| \geqq \varepsilon \right]
$$

This is precisely the first of the definitions of $(\bm{X}, \bm{Y})$ indicated in this work (p. 205).\\

\textbf{The space of random variables.} --- We reserved for the end a remark with which it would have been more logical to begin; 
but we did not want to risk discouraging the reader at the outset by a notion rather abstract and difficult to grasp clearly: 
if the notion of the distance between two random variables is quite clear, that of the space of random variables is much less so. 
There exist \textit{species of random variables}; 
but to speak of \textit{the space of random variables}, 
implying by this an allusion to a space containing all conceivable random variables, is as illusory as speaking of the set of all conceivable sets.

I shall explain myself more clearly. When we speak, for example, of the space of continuous functions, we may regard it as pre-existing, each point corresponding to a well-determined function, and we never run the risk of having to consider a continuous function that is not represented by a point of that space. 
On the contrary, whatever random variables we may already have considered, 
however great the cardinality of the set they form, nothing prevents us from considering a new random variable independent of the preceding ones. 
\textit{The space of random variables} is a construction that is never completed.

The definition given by M. Fréchet (p. 203) is, moreover, perfectly correct. In speaking of \textit{a certain category of trials}, he supposes the construction of the space of random variables to be fixed at a determined instant. It seemed to me that it was not useless to indicate more explicitly the difficulty he thus set aside.

I conclude with a remark concerning the case, especially interesting in practice, where one considers only random variables that are functions of a countable infinity of independent variables. It is known (see for example P. Lévy, \emph{Bull. Sc. Math.}, 1931, p. 87) that all these choices reduce to the choice of a single variable $t$ chosen between 0 and 1 (with a uniform distribution of probability over this interval), and all the random variables considered are measurable functions of $t$. One may then consider the space of random variables one wishes to study as mapped onto the space of measurable functions (and take in these two spaces definitions of distance that correspond to one another). 
This mapping, possible in an infinity of ways, does not eliminate the difficulty pointed out; for it is enough to adjoin to the preceding ones a new random variable independent of the others for everything to have to be begun again; the mode of mapping considered will have to be replaced by another.


%% file: appendix/frechet1957_en_body.tex
\noindent SUMMARY: The author indicates an explicit expression for the distance between two probability laws (\textit{lois de probabilit\'e}), according to Paul L\'evy's first definition. He also indicates a convenient modification of that definition.

\subsection*{INTRODUCTION}
\phantomsection\addcontentsline{toc}{subsection}{INTRODUCTION}

Mr. Paul L\'evy solved so many problems that he did not always have time to exploit all of his ideas.

For example, he gave\,\footnote{\,See p.~329 of the 2nd edition of our work: ``Recherches th\'eoriques modernes sur le Calcul des Probabilit\'es'' (First book), published by Gauthier-Villars, 1950.} three interesting definitions of the distance ``between two probability laws,'' leaving to readers the task of using them.

I thought that a good way to pay tribute to the very rich work of Mr. Paul L\'evy would consist in continuing his study.

After a brief comparison of his three definitions, I shall study the first more particularly.

My study was made easier by the results stated in my Note to the C.R.\,\footnote{\,Volume 242, 1956, pp.~2426--2428.} ``On correlation tables whose margins are given'' and by the application made of it by Bass\,\footnote{\,C.R. 21 February 1955, volume 240, no.~8.} and by Dall'Aglio\,\footnote{\,``Sugli estremi dei momenti delle funzione di ripartizione doppia'' (Ann. Scuole Norm. Sup. di Pisa, Sc. Fis. a Mat., series III, Vol.~X, 1956, pp.~35--74).}.

\noindent \underline{\textbf{Reminder}:}

Let us recall that:

A distance is defined on a set $E$ if to every pair $a,b$ of elements of $E$ one associates a real number, called the distance between $a$ and $b$ and denoted by $(a,b)$, in such a way that the following conditions hold:

$1^{\circ}\; (a,b)=(b,a) \ge 0;$

$2^{\circ}\;(a,b) \neq 0$ if $a$ and $b$ are regarded as distinct in $E$, and conversely;

$3^{\circ}\;$ For any three elements $a,b,c$ of $E$, one has the ``triangle inequality'':
$$
(a,b) \le (a,c)+(c,b).
$$

Moreover, to say that $a_n$ of $E$ \underline{converges} to $a$ of $E$ is to say that the distance $(a_n,a)$ tends to zero with $\frac{1}{n}$.

\newpage
\noindent \underline{\textbf{Distinction between two ``distances''}:}

Let us apply this definition to the Calculus of Probabilities. In this theory there are two distances to be distinguished. Let $X,Y$ be two random elements determined \underline{simultaneously} in each trial of a given category.

One may define a distance $(X,Y)$ between $X$ and $Y$ in a given trial. But there is also interest in defining a ``global distance'' (\textit{distance globale}) between $X$ and $Y$ in the \underline{whole set} of trials. We shall denote this global distance by:
$$
(\,[X],\,[Y] \,).
$$

Each time, one must make sure that conditions $1^{\circ},2^{\circ},3^{\circ}$ are satisfied. But in general, as regards condition $2^{\circ}$, one agrees to regard as non-distinct two random elements that are only ``almost surely'' identical. For example, we have proved that
$$
\operatorname{Prob}(X \neq Y)
$$
is a global distance.

A global distance more often considered is recalled in the note below. To confine ourselves to the case where $X,Y$ are two random \underline{numbers}, let us observe that, in order to determine a global distance, it is not enough to know the respective cumulative distribution functions:
$$
F(x)=\operatorname{Prob}(X<x);\; G(y)=\operatorname{Prob}(Y<y)
$$
which determine the probability laws $L,L'$ of $X$ and $Y$, but also the joint cumulative distribution function
$$
H(x,y)=\operatorname{Prob}(X<x \text{ and } Y<y)
$$
of the ordered pair $X,Y$\,\footnote{\,For example, one often takes as the ``global distance'' of $X$ and $Y$ the quantity
$$
\sqrt{\mathfrak{M}_{H}(X-Y)^{2}}=\sqrt{\int_{-\infty}^{+\infty}\int_{-\infty}^{+\infty}(x-y)^2\,d_x d_y H(x,y)}.
$$}.

On the contrary, the distance $(L,L')$ between the two laws $L,L'$ must be determined independently of $H(x,y)$ by $F(x)$ and $G(y)$ alone.

Paul L\'evy's last two definitions of $(L,L')$ are purely analytical. They present themselves as modifications of traditional definitions in Analysis of the distance between two functions $F$ and $G$, modifications motivated by the possible discontinuities of the two functions $F$ and $G$.

The first definition proposed by Mr. Paul L\'evy is, on the contrary, a definition that could legitimately appear arbitrary and artificial (if it were stated in purely analytical form) to a mathematician unfamiliar with Probability Theory. Indeed, it appeals to the notion of random number, and it is precisely when one effectively appeals to this notion that it acquires an intuitive and natural meaning.

It will be curious to observe that this first definition leads to a rather simple explicit expression for the distance, contrary to what one might have expected.

\noindent \underline{\textbf{Reminder}:}

The first definition proposed by Mr. Paul L\'evy is as follows:

The distance between two probability laws $L,L'$, whose cumulative distribution functions are $F(x),G(y)$, is the infimum (\textit{borne inf\'erieure}) of the ``global distances'' of two random numbers $X,Y$ whose joint cumulative distribution function $H(x,y)$ is arbitrary, provided it gives to $X,Y$, separately, the cumulative distribution functions $F(x),G(y)$.

\newpage
\noindent \underline{\textbf{Margins}:}

If one wishes, one may say that $F(x),G(y)$ are the ``margins'' (\textit{marges}) of the ``correlation table'' defined by $H(x,y)$.

It will be valuable for us to use the result stated in my Note to the C.R. of 14 May 1956 recalled above, according to which the set of functions $H(x,y)$ is none other than the set of joint cumulative distribution functions lying between the cumulative distribution functions $H_0(x,y)$ and $H_1(x,y)$ (endpoints included), with $H_0 \le H_1$, defined by
$$
\begin{aligned}
& H_0(x,y)=\operatorname{Sup}[(F(x)+G(y)-1),0],\\
& H_1(x,y)=\operatorname{Inf}[F(x),G(y)].
\end{aligned}
$$

Let us observe that in this first definition there appears a global distance between $X$ and $Y$ that can be interpreted in different ways. We gave a first example on p.~2 and a second in Note (1), p.~2. Another example: in order to obtain a global distance between $X$ and $Y$, namely
$$
(\,[X],\,[Y] \,),
$$
which tends to zero when $X$ tends to $Y$ in probability and conversely, we proposed taking for this distance the infimum, as the positive number $\varepsilon$ varies, of the sum
$$
\varepsilon+\operatorname{Prob}(|X-Y|>\varepsilon).
$$

One may also take as global distance of $X$ and $Y$ the mean value of $f(|X-Y|)$, where $f$ is a continuous, increasing, bounded function defined for $x\ge 0$, equal to zero for $x=0$, and such that $f(a+b)\le f(a)+f(b)$ (more general conditions actually suffice; see Parts $2^{\circ}$ and $3^{\circ}$ of my article published in 1957 in the \emph{Giornale Italiano degli Attuari}). As examples of $f(x)$ I gave the functions $\dfrac{x}{1+x}$, $|\mathrm{th}\,x|$, and $|1-e^{-x}|$.\,\footnote{\,For these two definitions, see pp.~205 and 226 of the 2nd edition of my ``Recherches th\'eoriques modernes sur le Calcul des Probabilit\'es'', Gauthier-Villars (1950).}

\noindent \textbf{\underline{The conditions for the existence of a distance:}}

For his first definition of the distance between two laws $L,L'$, L\'evy shows that this distance satisfies condition $3^{\circ}$ above. It is evident that condition $1^{\circ}$ is satisfied. The question is less simple for condition $2^{\circ}$.

It is nevertheless clear that if two laws $L,L'$ are identical, one of the pairs $X,Y$ satisfying these two laws respectively will be the pair $X,X$, whose distance is zero. The infimum of the distances $([X],[Y])$ will therefore be zero, and hence $(L,L')$ will be zero. The most delicate point is the converse statement. When $(L,L')=0$, are $L$ and $L'$ necessarily identical?

Let us specify by the notation $([X],[Y])_H$ the fact that this is the global distance of $X$ and $Y$ when the joint cumulative distribution function of their pair is $H(x,y)$.

In the case where this distance \underline{attains} its infimum $(L,L')$ for at least one particular function $H$, say $K(x,y)$, then by the definition of $([X],[Y])_H$, if this quantity is zero, $X$ and $Y$ are non-distinct (either identical in every trial or ``almost surely'' equal), and then $F(x)\equiv G(y)$, that is, $L\equiv L'$. But if the infimum is not attained, the reasoning is less simple.

To say that $(L,L')=0$ is to say that for every integer $n$ there exists an \underline{admissible} pair $X_n,Y_n$ (\underline{that is to say}, such that $X_n,Y_n$ have respective laws $L,L'$) for which the distance
\begin{equation*}
\left([X_n],[Y_n]\right)<\frac{1}{n}. \tag{1}
\end{equation*}

But the choice adopted for this last distance was not prescribed in advance. And we have shown that one can conceive many such distances satisfying conditions $1^{\circ},2^{\circ},3^{\circ}$ above and conforming to our intuition.

For example, we have shown that one can define very different distances $([X],[Y])$ having the following common property:

The necessary and sufficient condition for
$$
\lim_{n\to\infty}([X_n],[X])=0
$$
is that $X_n$ converges ``in probability'' to $X$ as $n\to\infty$.

Let us adopt such a definition in stating the definition of the distance between two laws $L,L'$.

Starting from (1), we have, for every $\varepsilon>0$,
\begin{align*}
G(x-\varepsilon) &= \operatorname{Prob}[Y_n<x-\varepsilon] \le \operatorname{Prob}[X_n<x] + \operatorname{Prob}[|X_n-Y_n|>\varepsilon] \tag{2}\\
&=F(x)+\operatorname{Prob}(|X_n-Y_n|>\varepsilon).
\end{align*}

For the different definitions we have given of a distance between two random numbers compatible with convergence in probability, one also has:
\begin{equation*}
([X],[Y])=([X-Y],0). \tag{3}
\end{equation*}

Hence, from (1) and (3), one has
$$
([X_n-Y_n],0)<\frac{1}{n},
$$
therefore $X_n-Y_n$ converges to zero in probability, that is,
\begin{equation*}
0=\lim_{n\to\infty}\operatorname{Prob}(|X_n-Y_n|>\varepsilon) \tag{4}
\end{equation*}
for every $\varepsilon>0$. Then from (2), letting $n$ increase to infinity, one has
$$
G(x-\varepsilon)\le F(x) \qquad \text{for every } \varepsilon>0.
$$
And since $G(x)$ is left-continuous,
$$
G(x)\le F(x).
$$
But in the same way one would obtain
$$
F(x)\le G(x),
$$
therefore $F(x)=G(x)$.

Thus, although it was not evident that condition $2^{\circ}$ should be satisfied, this is now proved, at least when the definition of $(L,L')$ is applied in the case where the distance between two random numbers satisfies the conditions above.

\noindent {\textbf{\uline{Case where $( [X],[Y] )$ is equal to the quadratic mean deviation (\textit{\'ecart quadratique moyen}) of $X$ and $Y$:}}}

In what follows, we first take as global distance of two random numbers $X$ and $Y$ their quadratic mean deviation (\textit{\'ecart quadratique moyen})
$$
\sqrt{\mathfrak{M}_{H}(X-Y)^2}\,.
$$
(To be precise, we write $\mathfrak{M}_H$ to recall that it is the mean when the joint cumulative distribution function of the pair $(X,Y)$ is $H(x,y)$.)

Then the distance $\Delta$ between the probability laws $L,L'$ characterized by the cumulative distribution functions $F(x)$ and $G(y)$ can be represented by the formula
$$
\Delta^2=\underset{H}{\operatorname{Inf}}\,\mathfrak{M}_H(X-Y)^2.
$$

Let us observe that
$$
(X-Y)^2 \le 2(X^2+Y^2), \text{ whence } \mathfrak{M}_H(X-Y)^2 \le 2\mathfrak{M}_H X^2+2\mathfrak{M}_H Y^2.
$$

Hence if $\mathfrak{M}_H(X-Y)^2$ is infinite, then at least one of the quantities
$$
\int_{-\infty}^{+\infty}x^2\,dF(x)=\mathfrak{M}_H X^2, \qquad \int_{-\infty}^{+\infty}y^2\,dG(y)=\mathfrak{M}_H Y^2,
$$
is infinite.

To avoid complications, we shall suppose that $F(x)$ and $G(x)$ are the cumulative distribution functions of random variables each having a \underline{finite quadratic mean deviation}. Then $\mathfrak{M}_H(X-Y)^2$ will necessarily be finite. Moreover, $X$ and $Y$, taken separately, also have, under the same hypothesis, finite and determined means:
$$
a=\mathfrak{M}X \qquad , \qquad a'=\mathfrak{M}Y;
$$
let $\sigma,\sigma'$ also denote their quadratic mean deviations. Then
$$
Z=\frac{X-a}{\sigma} \qquad , \qquad T=\frac{Y-a'}{\sigma'}
$$
will be the corresponding \underline{reduced} random variables for $X$ and $Y$. That is to say, the means of $Z$ and $T$ are zero and their quadratic mean deviations are equal to unity. One has
$$
\begin{gathered}
\mathfrak{M}_H(X-Y)^2=\mathfrak{M}(a-a'+\sigma Z-\sigma' T)^2 \\
=(a-a')^2+2(a-a')(\sigma\mathfrak{M}Z-\sigma'\mathfrak{M}T)+\sigma^2\mathfrak{M}Z^2-2\sigma\sigma'\mathfrak{M}ZT+\sigma'^2\mathfrak{M}T^2
\end{gathered}
$$
which, from what precedes, reduces to\,\footnote{\,This expression had already been given in the case of discrete random variables by Salvemini in: ``L'indice quadratico di dissomiglianza era distribuzione gaussiane''. Atti della XIV Reunione Scientifica, Societa Italiana di Statistica, 1954.}
\begin{equation*}
(a-a')^2+\sigma^2-2\sigma\sigma' r+\sigma'^2
\end{equation*}
where, when $\sigma\sigma'\neq 0$, $r$ denotes the (linear) correlation coefficient of $X$ and $Y$. One knows that $|r|\le 1$. One may therefore write $r=\cos\theta$, whence
\begin{equation*}
\sqrt{\mathfrak{M}_H(X-Y)^2}=\sqrt{(a-a')^2+\left(\sigma^2-2\sigma\sigma'\cos\theta+\sigma'^2\right)}. \tag{5}
\end{equation*}

This quantity may therefore be described as the hypotenuse of a right triangle, one side of which has length $|a-a'|$, the difference between the means of $X$ and $Y$, while the other side is the third side of a triangle whose other two sides have lengths $\sigma$ and $\sigma'$, the angle between them being equal to $\theta$ (or again, the cosine of that angle being equal to the linear correlation coefficient of $X$ and $Y$).

One observes that in expression (5), $a,a',\sigma,\sigma'$ are \underline{determined when $F$ and $G$ are known}, and \underline{only} $r=\cos\theta$ depends on $H(x,y)$. The minimum of $\mathfrak{M}_H(X-Y)^2$ therefore occurs when $r=\cos\theta$ is as large as possible. More precisely, the distance between the two probability laws is therefore
\begin{equation*}
(L,L')=\sqrt{(a-a')^2+\left(\sigma^2-2\sigma\sigma'\rho+\sigma'^2\right)} \tag{6}
\end{equation*}
where $\rho$ is the upper bound of the linear correlation coefficient between $X$ and $Y$ as the probability law of the pair $X,Y$ varies in such a way that the separate probability laws of $X$ and $Y$ remain fixed.

Let us also observe that
$$
r=\frac{\iint (x-a)(y-a')\,dH(x,y)}{\sqrt{\left[\int (x-a)^2\,dF\right]\left[\int (y-a')^2\,dG\right]}}\,.
$$

Set $h(z,t)=H(a+\sigma z,a'+\sigma' t)$. Then $r=\dfrac{\iint \sigma\sigma'zt\,dh(z,t)}{\sqrt{\ldots}}$, whence
\begin{equation*}
r=\iint_{P}zt\,dh(z,t). \tag{7}
\end{equation*}

Thus $r$ is also the linear correlation coefficient of $Z$ and $T$, and is expressed by formula (7), where $h(z,t)$ is a cumulative distribution function whose margins are given and equal to
$$
\varphi(z)\equiv F(a+\sigma z) \qquad \psi(t)\equiv G(a'+\sigma' t).
$$

The whole problem is reduced to finding the upper bound of the expression
\begin{equation*}
r=\int_{-\infty}^{+\infty}\int_{-\infty}^{+\infty}zt\,dh(z,t) \tag{8}
\end{equation*}
when $h$ is a cumulative distribution function whose margins $\varphi(z)$ and $\psi(t)$ are given. It is known that $r=\cos\theta=+1$ in the case where there is an increasing \underline{linear functional relation}
$$
Y=mX+p \qquad \text{with } m>0,
$$
between $X$ and $Y$. In that case one has more simply
\begin{equation*}
(L,L')=\sqrt{(a-a')^2+(\sigma-\sigma')^2}. \tag{9}
\end{equation*}

In the 1954 article cited above in the note, Salvemini had conjectured that the minimum of
$$
J=\int_{-\infty}^{+\infty}\int_{-\infty}^{+\infty}(x-y)^2\,d_x d_y H(x,y)
$$
is attained when $H$ is identical with the function $H_1(x,y)$ cited above.

In 1955, Bass stated in the note cited above the equivalent result (according to formula (6)) that the maximum of the linear correlation coefficient is obtained for $H\equiv H_1$. His proof (which he kindly communicated to me in writing) assumes that $X$ and $Y$ \underline{are bounded}.

The following year, in his article cited above, Dall'Aglio proved that $J$ is minimal for $H=H_1$, under the \underline{more general} hypothesis that
$$
\begin{aligned}
& 0=\lim_{x\to-\infty} x^{\gamma}F(x)=\lim_{y\to-\infty} y^{\gamma}G(y),\\
& 0=\lim_{x\to+\infty} x^{\gamma}[1-F(x)]=\lim_{y\to+\infty} y^{\gamma}[1-G(y)]
\end{aligned}
$$
for $\gamma\ge 2$.

He even gives the following expression for the minimum $M$:
\begin{align*}
M &=2\int_{-\infty}^{+\infty}dy\int_{y}^{+\infty}\operatorname{Sup}\{G(y)-F(x),0\}\,dx \,+ \tag{10}\\
&\quad +\,2\int_{-\infty}^{+\infty}dx\int_{x}^{+\infty}\operatorname{Sup}\{F(x)-G(y),0\}\,dy
\end{align*}

But one can give a more convenient expression for this minimum (which is equal to $\Delta^2$).

One can arrive at this other expression for $\Delta^2$ in two ways. Let us compute $\Delta^2$ in the case where $F(x),G(y)$ are everywhere continuous and neither of them is constant on any interval. Then the equation
$$
F(x)=G(y)
$$
represents a curve $\Gamma$ that is intersected in one and only one point by every line parallel to either axis, and on which $y$ increases with $x$ and conversely. One has $H_1(x,y)=F(x)$ above $\Gamma$ (toward positive $y$), and $H_1(x,y)=G(y)$ below $\Gamma$.

Consider a sequence of numbers $\ldots,x_{-h},\ldots,x_{-1},x_0,x_1,\ldots,x_k,\ldots$ increasing from $-\infty$ to $+\infty$; the corresponding numbers $y_j$, such that
$$
F(x_j)=G(y_j),
$$
will also increase from $-\infty$ to $+\infty$. The curve $\Gamma$ will be contained in the sequence $S$ of rectangles
$$
x_j\le x<x_{j+1} \qquad ; \qquad y_j\le y<y_{j+1}.
$$

Above these rectangles, $H_1(x,y)=F(x)$ and below them $H_1(x,y)=G(y)$. Consequently, the integral
$$
\iint_B (x-y)^2\,dx\,dy\,H_1(x,y)
$$
extended over the region $B$ exterior to the rectangles of $S$ is zero, and one has
$$
\int_{-\infty}^{+\infty}\int_{-\infty}^{+\infty}(x-y)^2\,d^2H_1(x,y)=\iint_S (x-y)^2\,d^2H_1(x,y).
$$

Therefore,
$$
\begin{gathered}
\mathfrak{M}_{H_1}(X-Y)^2=\lim \sum_{i=-\infty}^{i=+\infty}\sum_{j=-\infty}^{j=+\infty}(x_i-y_j)^2\Delta_{ij}H_1 \\
=\lim\sum_i (x_i-y_i)^2\left[H_1(x_{i+1},y_{i+1})-H_1(x_{i+1},y_i)-H_1(x_i,y_{i+1})+H(x_i,y_i)\right].
\end{gathered}
$$
Now,
$$
\begin{aligned}
& H_1(x_{i+1},y_{i+1})=F(x_{i+1})=G(y_{i+1}),\\
& H_1(x_i,y_i)=F(x_i)=G(y_i),\\
& H_1(x_{i+1},y_i)=G(y_i), \qquad H_1(x_i,y_{i+1})=G(y_i)=F(x_i).
\end{aligned}
$$

Hence the bracket is equal to
$$
F(x_{i+1})-F(x_i),
$$
and consequently
$$
\mathfrak{M}_{H_1}(X-Y)^2=\lim\sum_i (x_i-y_i)^2\left[F(x_{i+1})-F(x_i)\right].
$$

Since the equation
\begin{equation*}
F(x)=G(y) \tag{11}
\end{equation*}
has a continuous increasing solution in $y$,
$$
y=\lambda(x), \text{ one sees that }
$$
$$
\Delta^2=\mathfrak{M}_{H_1}(X-Y)^2=\int_{-\infty}^{+\infty}[x-\lambda(x)]^2\,dF(x).
$$
On the other hand, solving (11) for
$$
x=\mu(y)
$$
one could also write
$$
\Delta^2=\int_{-\infty}^{+\infty}[\mu(y)-y]^2\,dG(y).
$$

In the same way, the upper bound $\rho$ of $r$ is determined by the formula
$$
\sigma\sigma'\rho=\int_{-\infty}^{+\infty}(x-a)[\lambda(x)-a']\,dF(x)=\int_{-\infty}^{+\infty}(y-a')[\mu(y)-a] \,dG(y).
$$

One may obtain these expressions even more easily by observing that if the joint cumulative distribution function is $H_1(x,y)$, there is almost surely a functional relation between the two random variables $X$ and $Y$, namely,
$$
Y=\lambda(X).
$$
Hence,
$$
\mathfrak{M}_{H_1}(X-Y)^2=\mathfrak{M}_{H_1}[\lambda(X)-X]^2
$$

\begin{equation*}
(L,L')^2=\int_{-\infty}^{+\infty}[\lambda(x)-x]^2\,dF(x) \tag{12}
\end{equation*}

Setting $W=F(X)=G(Y)$, $X$ and $Y$ are two continuous increasing functions of $W$:
$$
X=l(W) \qquad ; \qquad Y=m(W)
$$
and $W$ has a probability density constantly equal to 1 on $(0,1)$. One then has
\begin{equation*}
(L,L')^2=\int_0^1[l(t)-m(t)]^2\,dt \tag{13}
\end{equation*}

In the previous notation,
$$
\Delta^2=(a'-a)^2+\left(\sigma^2-2\sigma\sigma'\rho+\sigma'^2\right)
$$
where
$$
\rho=\int_{-\infty}^{+\infty}\int_{-\infty}^{+\infty}zt\,dh(z,t)
$$
one will have
$$
\rho=\int_{-\infty}^{+\infty}zv(z)\,d\varphi(z)
$$
if $t=v(z)$ is equivalent to $\varphi(z)=\psi(t)$.

\noindent \underline{\textbf{Remark}:}

Dall'Aglio's formula (10) and formulas (12), (13) solve the problem posed by Salvemini in the article cited above.

\noindent \underline{\textbf{Example}:}

Suppose that $X$ and $Y$ are such that their ``reduced'' variables $Z$ and $T$ have the same probability law. Then among the pairs $X,Y$ there will be a pair for which $Z=T$. Then there is an increasing linear functional relation between $X$ and $Y$ for $H=H_1$, and one again has, as in formula (9),

\begin{equation*}
\Delta=\sqrt{(a-a')^2+(\sigma-\sigma')^2}\,. \tag{14}
\end{equation*}

Thus one obtains a particularly simple result. Formula (14) applies, for example, when the margins have a so-called normal distribution. But it also applies when the margins both satisfy the first Laplace law, or are both uniform, etc.

More generally, it also applies when $X$ and $Y$ have the same reduced probability law with finite quadratic mean deviation.

\noindent \underline{\textbf{Generalization}:}

Dall'Aglio showed that if one takes
$$
([X],[Y])_H=\sqrt[s]{\mathfrak{M}_H|X-Y|^s},
$$
the minimum is, for $s\ge 1$, attained for $H\equiv H_1$.

As above, one sees that if $F(x)$ and $G(y)$ are continuous and increasing, one then has
$$
(L,L')^s=\int_{-\infty}^{+\infty}|\lambda(x)-x|^s\,dF(x) \; \text{ or } \; (L,L')^s=\int_0^1|l(t)-m(t)|^s\,dt.
$$

\subsection*{A FOURTH DEFINITION OF THE DISTANCE BETWEEN TWO PROBABILITY LAWS}
\phantomsection\addcontentsline{toc}{subsection}{A FOURTH DEFINITION OF THE DISTANCE BETWEEN TWO PROBABILITY LAWS}

The first of the three definitions of $(L,L')$ due to L\'evy is very seductive. For it makes explicit, in a very simple and very natural way, the intuitive notion that we may have of the distance $(L,L')$ between two probability laws $L,L'$.

Its greatest drawback is that, since it varies with the definition adopted as starting point for the global distance of two random numbers, it requires each time the \underline{prior} determination of the \underline{minimum} of that distance when the correlation between the two numbers varies.

\noindent \underline{\textbf{A fourth definition}:}

We shall therefore propose a fourth definition of $(L,L')$ in order to \underline{evade} this difficulty. We therefore do not regard our definition as better, but only as being of more immediate application. We shall moreover see later that the two definitions are equivalent in a broad case. They are probably equivalent in still more general cases, if not in all cases.

\noindent \underline{\textbf{New definition of $(L,L')$:}}

Let us once more take as our starting point one of the admissible definitions of the ``global'' distance $([X],[Y])_H$ between two random numbers $X,Y$ determined simultaneously in each trial of a given category, corresponding to the joint cumulative distribution function $H(x,y)$ of the pair $X,Y$.

This being so, we shall take as the new definition of the distance between two laws $L,L'$ the quantity
\begin{equation*}
(L,L')=([X],[Y])_{H_1} \tag{15}
\end{equation*}
where $X,Y$ obey the two laws $L,L'$. Thus we call the distance between two probability laws $L,L'$ the value taken by a ``global'' distance between two random numbers $X,Y$, obeying separately the laws $L,L'$, when the joint cumulative distribution function $H(x,y)$ of the pair $X,Y$ attains its maximum value $H_1(x,y)$ specified by formula (2).

It must now be proved that this quantity $(L,L')$ is indeed ``a distance''; that is, that it satisfies the three well-known conditions recalled on pp.~1--2 under numbers $1^{\circ},2^{\circ},3^{\circ}$. This will provide a partial \emph{a posteriori} justification for the definition we have just laid down dogmatically.

By hypothesis, for every correlation function $H(x,y)$ admitting as margins the cumulative distribution functions $F(x),G(y)$ of the laws $L,L'$, the global distance between $[X]$ and $[Y]$ satisfies the conditions
$$
\begin{aligned}
  &1^{\circ}\; ([X],[Y])_H=([Y],[X])_H \ge 0,\\
  &2^{\circ}\; ([X],[Y])_H \text{ is positive if } X \text{ and } Y \text{ are ``distinct,'' and conversely.}
\end{aligned}
$$
(One generally regards as non-distinct two random numbers that are ``almost surely'' equal, that is, such that $\operatorname{Prob}(X=Y)=1$.)

This is also the necessary and sufficient condition for $F(x)\equiv G(x)$. Hence, applying this to the case $H\equiv H_1$,
$$
\begin{aligned}
& 1^{\circ}\; (L,L')=(L',L) \ge 0,\\
& 2^{\circ}\; (L,L')=0 \text{ if } L\equiv L', \text{ and conversely.}
\end{aligned}
$$

As regards condition $3^{\circ}$, we must bring in three probability laws $L,L',L''$, corresponding to three cumulative distribution functions $F(x),G(y),K(z)$. Let us introduce a particular correlation between three random numbers $X,Y,Z$ obeying respectively the three laws $L,L',L''$. For this purpose, we choose the correlation defined by the joint cumulative distribution function
$$
E(x,y,z)=\min[F(x),G(y),K(z)]\,.
$$

One clearly sees that if the system $X,Y,Z$ has joint cumulative distribution function $E(x,y,z)$, then $X,Y,Z$ indeed have $F(x),G(y),K(z)$ as their respective cumulative distribution functions. For example,
$$
\begin{aligned}
& \operatorname{Prob}(X<x)=\operatorname{Prob}(X<x,Y<+\infty,Z<+\infty)\\
& =E(x,+\infty,+\infty)=\min(F(x),+1,+1)=F(x).
\end{aligned}
$$
One further has
$$
E(x,y,+\infty)=\min[F(x),G(y),+1]=H_1(x,y)
$$
and likewise the maximized cumulative distribution functions of $(X,Z)$ and $(Y,Z)$ are
$$
\begin{aligned}
& H_1'(x,z)=E(x,+\infty,z),\\
& H_1''(y,z)=E(+\infty,y,z)\,.
\end{aligned}
$$
This being so,
\begin{equation*}
([X],[Y])_{H_1}=([X],[Y])_{E(x,y,+\infty)}=([X],[Y])_{E(x,y,z)} \tag{16}
\end{equation*}

Now, by hypothesis, for any three-variable joint cumulative distribution function $U(x,y,z)$ having the given margins $F,G,K$, one has
$$
([X],[Y])_U \le ([X],[Z])_U+([Y],[Z])_U.
$$

Therefore, taking in particular $E$ for $U$ and using the two equalities analogous to (4), one obtains
$$
([X],[Y])_{H_1} \le ([X],[Z])_{H_1'}+([Y],[Z])_{H_1''}
$$
that is,
$$
(L,L')\le (L,L'')+(L',L'').
$$

\noindent \underline{\textbf{Equivalence}:}

It follows from the proofs of Bass and Dall'Aglio that \underline{in the case} where one has taken
$$
([X],[Y])_H=\sqrt{\mathfrak{M}_H(X-Y)^2}
$$
and under the very general restrictions which, moreover, may not even be necessary, the value of $(L,L')$ remains the same whether one uses Paul L\'evy's first definition or our own.

\noindent \underline{\textbf{General computation of $(L,L')$:}}

Consider the case where $F(x)$ and $G(y)$ are functions everywhere \underline{continuous} and \underline{increasing}. Then the curve $\Gamma$:
$$
F(x)=G(y)
$$
separates the region $R$ where $H_1=F(\le G)$ from the region $Q$ where $H_1=G(\le F)$. If one sets
$$
F(x)=G(y)=t,
$$
one sees, as above, that $x$ and $y$ are continuous increasing functions of $t$, namely,
$$
x=l(t) \qquad\qquad y=m(t).
$$

When the joint cumulative distribution function is $H_1(x,y)=\min[F(x),G(y)]$, the second difference of $H_1$ is zero in the interior of $R$ and in the interior of $Q$, that is, the point $(X,Y)$ lies almost surely on $\Gamma$.

In other words, in this case there is an almost certain functional relation between $X$ and $Y$, namely,
$$
F(x)=G(y).
$$

The random number $T$, equal to the common value of $F$ and $G$ on $\Gamma$, remains between 0 and 1 with a probability density constantly equal to 1.

It follows that
$$
(L,L')=([l(T)],[m(T)]).
$$

\noindent \underline{\textbf{Example}:}

For example, if one starts from the expression
$$
([X],[Y])_H=\sqrt[\alpha]{\int_{-\infty}^{+\infty}\int_{-\infty}^{+\infty}|x-y|^\alpha\,d_x d_y H(x,y)}=\sqrt[\alpha]{\mathfrak{M}_H|X-Y|^\alpha}
$$
with
$$
\alpha\ge 1,
$$
one has
$$
(L,L')^\alpha=\mathfrak{M}|l(T)-m(T)|^\alpha, \qquad \text{whence}
$$
$$
(L,L')^\alpha=\int_0^1|l(t)-m(t)|^\alpha\,dt.
$$

One can also solve
$$
F(x)=G(y)
$$
for $y=\lambda(x)$,
where $\lambda(x)$, like $F$ and $G$, will be a continuous increasing function. One may then write
$$
(L,L')^\alpha=\int_{-\infty}^{+\infty}|x-\lambda(x)|^\alpha\,dF(x).
$$

\subsection*{GINI'S DISSIMILARITY INDEX}
\phantomsection\addcontentsline{toc}{subsection}{GINI'S DISSIMILARITY INDEX}

For two equally numerous sequences of nondecreasing numbers,
$$
a_1\le a_2\le \ldots\le a_n \qquad , \qquad b_1\le b_2\le \ldots\le b_n,
$$
this index is defined by
$$
I_\alpha=\sqrt[\alpha]{\frac{1}{n}\sum_{h=1}^{n}|a_h-b_h|^\alpha} \qquad \begin{aligned}
& \alpha=1: \text{ simple index}\\
& \alpha=2: \text{ quadratic index}
\end{aligned}
$$
But one may take any $\alpha\ge 1$.

The $a$'s not being necessarily distinct, their distinct values are repeated $r_1$ times, \ldots, $r_n$ times. They may be regarded as the values taken by a statistical magnitude $X$ in $n$ trials. The distinct values of $X$, $x_1<x_2<\ldots$, have frequencies
$$
f_1=\frac{r_1}{n}, \quad f_2=\frac{r_2}{n}, \ldots \quad \text{The sequence } \left\{\begin{array}{l}
x_1,x_2,\ldots\\
f_1,f_2,\ldots
\end{array}\right.
$$
thus defines a frequency law $L$. Likewise, the sequence of distinct values $y_1,y_2,\ldots$ taken by the $b_j$ with frequencies $\varphi_1=\dfrac{s_1}{n},\varphi_2=\dfrac{s_2}{n},\ldots$ defines a law $L'$.

By permuting the $b_h$, one obtains a sequence $b_1',b_2',\ldots$ that is no longer necessarily nondecreasing. If each trial associated some $b_h'$ with an $a_h$, one would obtain a frequency law $H$ for the pair $X,Y$. By associating $b_h$ with $a_h$, one obtains a particular law $H'$ for this pair. We shall show that $H'=H_1$, so that Gini's dissimilarity index becomes identical with what our distance between $L$ and $L'$ reduces to in the particular case where $L,L'$ are the laws of two statistical magnitudes (each taking only finitely many values).

Indeed, suppose first that
$$
a_1<x\le a_n \qquad , \qquad b_1<y\le b_n.
$$
Then there exist $h$ and $j$ such that
$$
a_h<x\le a_{h+1} \qquad , \qquad b_j<y\le b_{j+1}.
$$
One has
$$
H'(x,y)=\operatorname{Prob}(X<x,Y<y)=\operatorname{Prob}(X<a_{h+1},Y<b_{j+1})
$$
and since $H'$ corresponds to the case where $b_j$ is associated with $a_j$, if $j<h$ then
$$
H'(x,y)=\operatorname{Prob}(Y<b_{j+1})=\operatorname{Prob}(Y<y)=G(y)
$$
and
$$
F(x)=\operatorname{Prob}(X<a_{h+1})\ge \operatorname{Prob}(Y<b_{j+1})=\operatorname{Prob}(Y<y).
$$

Therefore, if $j\le h$,
$$
\begin{gathered}
H'(x,y)=G(y)\le F(x),\\
H'(x,y)=\min[F(x),G(y)]=H_1(x,y)\,.
\end{gathered}
$$

Repeating the same reasoning for $h\le j$, one sees that in both cases
\begin{equation*}
H'(x,y)=H_1(x,y). \tag{1$'$}
\end{equation*}

In the case where $x\le a_1$, one obviously has
$$
H'(x,y)=F(x)=0\le G(y),
$$
and consequently (1$'$) still holds; similarly if $y\le b_1$, and by analogous reasoning for $x>a_n$ or $y>b_n$. Thus in all cases\,\footnote{\,Here we merely reproduce, in somewhat greater detail, Salvemini's reasoning, although he does not deal with the notion of distance.}
\begin{equation*}
H'(x,y)\equiv H_1(x,y).
\end{equation*}

The simple $(\alpha=1)$ or quadratic $(\alpha=2)$ dissimilarity index, cleverly chosen by Gini, thus turns out at the same time to be a ``distance'' between two frequency laws. More precisely, it falls within the more general category of the fourth definition of the distance of two absolutely arbitrary probability laws (discrete, continuous, or intermediate) specified on p.~11. It also has the interest of providing a simpler mode of computation than what would result directly from the formula
$$
I_\alpha=\sqrt[\alpha]{\mathfrak{M}_{H_1}|X-Y|^\alpha}\,.
$$